\documentclass[10pt,a4paper]{article}
\usepackage[latin1]{inputenc}
\usepackage[T1]{fontenc}
\usepackage{indentfirst}
\usepackage{amsmath,amssymb,amsthm,amsfonts}
\usepackage{calligra,mathrsfs}
\usepackage{graphicx,color}
\usepackage{xcolor}
\usepackage[all,cmtip]{xy}
\usepackage{bbm}
\usepackage{hyperref}
\usepackage{mathtools}
\usepackage{verbatim}
\usepackage{stmaryrd}
\usepackage[shortlabels]{enumitem}
\usepackage{bm}
\usepackage{url}
\usepackage{comment}
\usepackage{slashed}
\usepackage{stmaryrd}
\usepackage{caption}
\usepackage{subcaption}
\usepackage{scalerel}
\usepackage{multicol}
\usepackage{tabularx}
\usepackage{tikz}
\usetikzlibrary{patterns,patterns.meta}
\usepackage{subcaption}

\usepackage{tikz-cd}
\usetikzlibrary{matrix}
\usetikzlibrary{positioning}
\usetikzlibrary{arrows}
\usetikzlibrary{backgrounds}

\usepackage{caption}
\hypersetup{colorlinks=true,linkcolor=black,urlcolor=lcolor,citecolor=black}

\definecolor{amethyst}{rgb}{0.6, 0.4, 0.8}
\definecolor{alizarin}{rgb}{0.82, 0.1, 0.26}
\definecolor{darkorchid}{rgb}{0.6, 0.2, 0.8}
\definecolor{lcolor}{rgb}{0.6,0.3,0.3}
\definecolor{caribbeangreen}{rgb}{0.0, 0.8, 0.6}
\definecolor{flamingopink}{rgb}{0.99, 0.56, 0.67}
\definecolor{hollywoodcerise}{rgb}{0.96, 0.0, 0.63}
\definecolor{airforceblue}{rgb}{0.36, 0.54, 0.66}
\definecolor{applegreen}{rgb}{0.55, 0.71, 0.0}
\definecolor{ballblue}{rgb}{0.13, 0.67, 0.8}
\definecolor{brinkpink}{rgb}{0.98, 0.38, 0.5}
\definecolor{bluegray}{rgb}{0.4, 0.6, 0.8}
\definecolor{impx}{rgb}{0.96, 0.0, 0.63}
\definecolor{junglegreen}{rgb}{0.16, 0.67, 0.53}
\definecolor{applegreen}{rgb}{0.55, 0.71, 0.0}
\definecolor{darkpastelgreen}{rgb}{0.01, 0.75, 0.24}
\definecolor{parisgreen}{rgb}{0.31, 0.78, 0.47}
\definecolor{ufogreen}{rgb}{0.24, 0.82, 0.44}
\definecolor{bluebell}{rgb}{0.64, 0.64, 0.82}
\definecolor{cornflowerblue}{rgb}{0.39, 0.4, 0.93}
\definecolor{pastelred}{rgb}{1.0, 0.41, 0.38}
\definecolor{carminepink}{rgb}{0.92, 0.3, 0.26}
\definecolor{darkpastelpurple}{rgb}{0.59, 0.44, 0.84}
\definecolor{cerulean}{rgb}{0.0, 0.48, 0.65}
\definecolor{darkcerulean}{rgb}{0.03, 0.27, 0.49}
\definecolor{ceruleanblue}{rgb}{0.16, 0.32, 0.75}
\definecolor{green(munsell)}{rgb}{0.0, 0.66, 0.47}
\definecolor{bleudefrance}{rgb}{0.19, 0.55, 0.91}
\definecolor{antiquefuchsia}{rgb}{0.57, 0.36, 0.51}
\definecolor{deeplilac}{rgb}{0.6, 0.33, 0.73}
\definecolor{ferngreen}{rgb}{0.31, 0.47, 0.26}
\definecolor{qcolor}{rgb}{0.19,0.55,0.91}
\definecolor{hcolor}{rgb}{0.9,0.2,0.5}
\definecolor{scolor}{rgb}{0.9,0.5,0.2}

\definecolor{triangleHor}{rgb}{0.93, 0.53, 0.18}
\definecolor{triangleVer}{rgb}{0.0, 0.5, 1.0}

\newcommand{\tc}[2]{#2}
\newcommand{\NEW}[1]{\tc{lcolor}{#1}}

\newcommand{\step}{\vskip 3mm}

\newcommand{\p}{\partial}
\newcommand{\Z}{\mathbbm{Z}}
\newcommand{\N}{\mathbbm{N}}
\newcommand{\R}{\mathbbm{R}}

\newcommand{\C}{\mathbbm{C}}
\newcommand{\one}{\mathbbm{1}}
\newcommand{\eps}{\epsilon}
\newcommand{\vareps}{\varepsilon}
\DeclareMathOperator{\Hom}{Hom}

\DeclareMathOperator{\so}{so}

\DeclareMathOperator{\End}{End}
\DeclareMathOperator{\Der}{Der}

\renewcommand{\div}{\textnormal{div}}
\newcommand{\DIV}{\textnormal{DIV}}
\newcommand{\curl}{\textnormal{curl}}

\newcommand{\supp}{\textnormal{supp}}

\renewcommand{\subset}{\subseteq}
\renewcommand{\supset}{\supseteq}

\newcommand{\RNum}[1]{\mathrm{\uppercase\expandafter{\romannumeral #1\relax}}}

\newcommand{\proofheader}[1]{\textit{#1}}
\newcommand{\claimheader}[1]{#1}

\DeclareMathOperator{\Lie}{\mathcal{L}}
\renewcommand{\diamond}{\mathcal{D}}
\newcommand{\diamondaux}{\smash{\slashed{\mathcal{D}}}}
\newcommand{\nullinfaux}{\smash{\slashed{\nullinf}}}
\newcommand{\diamonddata}{\underline{\diamond\mkern-4mu}\mkern4mu }
\newcommand{\Udata}{\underline{U\mkern-4mu}\mkern4mu }

\newcommand{\diamondy}{\mathcal{D}'}

\newcommand{\gS}{g_{S^3}}

\newcommand{\gmink}{\eta}
\newcommand{\nullinf}{\new{\mathscr{I}}}
\newcommand{\pnullinf}{\mathscr{I}_{-}}
\newcommand{\fnullinf}{\mathscr{I}_{+}}
\newcommand{\fpnullinf}{\mathscr{I}_{\pm}}

\newcommand{\fptimeinf}{i_{\pm}}
\newcommand{\nullgen}{\tc{caribbeangreen}{h}}
\newcommand{\ynullgen}{\tc{caribbeangreen}{\smash{\slashed{h}}}}

\newcommand{\s}{\tc{caribbeangreen}{\mathfrak{s}}}
\newcommand{\muW}{\mu_{\Dspop}}
\newcommand{\muWs}{\mu_{\Dspop}'}
\newcommand{\muDu}{\mu_{\Dspdata}}

\newcommand{\cyl}{\mathbb{E}}
\newcommand{\gcyl}{g_{\cyl}}
\newcommand{\mucyl}{\mu_{\cyl}}
\newcommand{\mucylform}{\tilde\mu_{\cyl}}
\newcommand{\mucylS}{\mu_{S^3}}

\newcommand{\Kil}{\tc{cornflowerblue}{\mathfrak{K}}}
\newcommand{\KilBasis}{\tc{blue}{\zeta}}
\newcommand{\KilEl}{\tc{blue}{\zeta}}
\newcommand{\ConfKil}{\tc{cornflowerblue}{\mathfrak{K}_{\textnormal{conf}}}}
\newcommand{\transl}{\tc{cornflowerblue}{\mathfrak{t}}}
\newcommand{\boosts}{\tc{cornflowerblue}{\mathfrak{b}}}
\newcommand{\dens}[1]{|\Omega|^{#1}}
\newcommand{\I}{\mathcal{I}}
\newcommand{\It}{\tilde{\mathcal{I}}}
\newcommand{\gx}{\mathfrak{g}}
\newcommand{\lx}{\mathfrak{L}}
\newcommand{\anchor}{\rho}
\newcommand{\anchorg}{\rho_{\gx}}
\newcommand{\anchorL}{\rho_{\lx}}
\newcommand{\OmegaC}{\Omega_{\C}}
\newcommand{\otimesCinf}{\otimes_{C^\infty}}
\newcommand{\dg}{d_{\gx}}
\newcommand{\dI}{d_{\I}}

\newcommand{\ddR}{d}
\newcommand{\injaux}{\tilde{\inj}}
\newcommand{\dl}{d_{\lx}}
\newcommand{\uo}{u_0}
\newcommand{\uI}{u_{\I}}
\newcommand{\up}{u_{+}}
\newcommand{\um}{u_{-}}
\newcommand{\upm}{u_{\pm}}
\newcommand{\upbar}{\bar{u}_{+}}

\newcommand{\Scal}{S}
\newcommand{\Scalg}{\Scal^{\gx}}
\newcommand{\ScalL}{\Scal^{\lx}}
\newcommand{\ScalI}{\Scal^{\I}}

\newcommand{\gxdata}{\underline{\mathfrak{g}}}
\newcommand{\Pconstraints}{\underline{P}}
\newcommand{\PconstraintsG}{\underline{P}_{\textnormal{G}}}

\newcommand{\Thom}{\mathrm{T}_{\mathrm{h}}}
\newcommand{\ghom}{g_{\mathrm{h}}}
\newcommand{\Iinter}{\tc{cornflowerblue}{\mathsf{i}}}

\newcommand{\Iipr}[3]{\langle #2,#3 \rangle_{\I^{#1}}}
\newcommand{\Iiprpm}[3]{\langle #2,#3 \rangle_{\I_{\pm}^{#1}}}
\newcommand{\Iiprp}[3]{\langle #2,#3 \rangle_{\I_{+}^{#1}}}
\newcommand{\Iiprm}[3]{\langle #2,#3 \rangle_{\I_{-}^{#1}}}
\newcommand{\Oipr}[3]{\langle #2,#3 \rangle_{\Omega^{#1}}}

\DeclareMathOperator{\sgn}{sgn}

\newcommand{\eGT}{\mathbf{t}}
\newcommand{\eGB}{\mathbf{b}}
\newcommand{\eGI}{\mathbf{i}}
\newcommand{\eGO}{\mathbf{o}}
\newcommand{\eGg}{\mathbf{e}}
\newcommand{\eT}{\smash{\slashed{\mathbf{t}}}}
\newcommand{\eB}{\smash{\slashed{\mathbf{b}}}}
\newcommand{\eO}{\smash{\slashed{\mathbf{o}}}}
\newcommand{\eI}{\smash{\slashed{\mathbf{i}}}}
\newcommand{\eg}{\smash{\slashed{\mathbf{e}}}}

\newcommand{\cyclind}{\textnormal{cycl}}
\newcommand{\Ieps}{}

\newcommand{\jj}{\tc{orange}{j}}

\newcommand{\fctfol}{\tc{brinkpink}{\phi}}

\newcommand{\framechange}[1]{\tc{darkpastelpurple}{#1}}

\newcommand{\sig}{\tc{blue}{\lambda}}
\newcommand{\epsO}{\tc{orange}{\eps_0}}

\newcommand{\Csmallmain}{\tc{lcolor}{\eps}}
\newcommand{\Cinmain}{\tc{lcolor}{b}}
\newcommand{\CinmainHigher}{\tc{lcolor}{b'}}
\newcommand{\slCinmainHigher}{\NEW{\tilde{b}}}
\newcommand{\Clargemain}{\tc{lcolor}{C}}
\newcommand{\ClargemainHigher}{\tc{lcolor}{C'}}

\newcommand{\Clargecpt}{\Clarge_{\textnormal{cpt}}}
\newcommand{\Csmallcpt}{\Csmall_{\textnormal{cpt}}}

\newcommand{\CFsob}{\tc{amethyst}{C_0}}
\newcommand{\Clk}{\tc{amethyst}{C_1}}
\newcommand{\Ca}{\tc{amethyst}{C_2}}
\newcommand{\Caprime}{\tc{amethyst}{C_2'}}

\newcommand{\zzeta}{\tc{lcolor}{\mathfrak{z}}}
\newcommand{\zzetadata}{\tc{orange}{\underline{\smash{\mathfrak{z}}}}}
\newcommand{\zz}{\tc{lcolor}{z}}
\renewcommand{\tt}{\tc{lcolor}{t}}
\newcommand{\ttcoord}{\tc{lcolor}{\mathfrak{t}}}
\newcommand{\nn}{\tc{magenta}{n}}

\renewcommand{\frame}[1]{\tc{ballblue}{F_{#1}}}
\newcommand{\dualframe}[1]{\tc{ballblue}{F_{#1}^*}}
\DeclareMathOperator{\Ric}{Ric}
\newcommand{\Omegafut}{\Omega_{\vee}^1}

\newcommand{\cone}{\Gamma}
\newcommand{\zzmax}{\zz_{m}}

\newcommand{\spaceinf}{i_0}
\newcommand{\timeinf}{i_+}

\newcommand{\new}[1]{\tc{red}{#1}}
\newcommand{\orient}[1]{\tc{magenta}{#1}}

\newcommand{\SPAN}{\textnormal{span}}

\newcommand{\diag}{\textnormal{diag}}

\newcommand{\intermult}{\iota}
\newcommand{\tint}{\textstyle\int}
\newcommand{\tsum}{\textstyle\sum}

\newcommand{\Bil}{\beta}
\newcommand{\gxG}{\mathfrak{g}_{\textnormal{G}}}
\newcommand{\gxdataG}{\gxdata_{\textnormal{G}}}
\newcommand{\OmegaG}{\Omega_{\textnormal{G}}}
\newcommand{\BilG}{\mathrm{B}_{\mathrm{G}}}
\newcommand{\IG}{\I_{\textnormal{G}}}
\newcommand{\OmegaBil}{\Bil_{\Omega}}
\newcommand{\OmegaTBil}{\Bil_{\transl}}
\newcommand{\OmegaBBil}{\Bil_{\boosts}}

\newcommand{\IBil}{\Bil_{\I}}
\newcommand{\gBil}{\Bil_{\gx}}

\newcommand{\ngG}{\tc{lcolor}{n}}
\renewcommand{\ng}{\tc{lcolor}{m}}

\newcommand{\etay}{\eta'}

\newcommand{\V}{V}
\newcommand{\Vd}{V_*}
\newcommand{\Vdpm}{V_{\pm}}
\newcommand{\Vdp}{V_{+}}

\newcommand{\Xd}{X_*}

\newcommand{\homMfd}{M}
\newcommand{\homMfddata}{\underline{M}}
\newcommand{\Mcpt}{\mathcal{C}}

\newcommand{\NN}{\tc{gray}{N}}
\newcommand{\Cpos}{\tc{bluegray}{q}}

\newcommand{\muM}{\mu_{\homMfd}}
\newcommand{\muMform}{\tilde\mu_{\homMfd}}
\newcommand{\muMz}{\mu'_{\homMfd}}

\newcommand{\CcomTang}[1]{\kappa^{\tang}_{#1}}
\newcommand{\CcomTransv}[1]{\kappa^{0}_{#1}}
\newcommand{\Ccomp}{\slashed{\kappa}}

\newcommand{\current}{\mathbf{j}}

\newcommand{\countzero}{n_0}

\newcommand{\indset}{\mathsf{I}}

\newcommand{\indI}{\tc{orange}{I}}
\newcommand{\indK}{\tc{orange}{K}}
\newcommand{\indJ}{\tc{orange}{J}}

\renewcommand{\a}{\tc{cornflowerblue}{a}}
\newcommand{\AmatLin}{\a}
\newcommand{\AmatBil}{\tc{cornflowerblue}{A}}
\newcommand{\Amat}{\tc{cornflowerblue}{a}_{\tc{cornflowerblue}{0}}}
\newcommand{\LMat}{\tc{cornflowerblue}{L}}
\newcommand{\Fvec}{\tc{cornflowerblue}{F}}
\newcommand{\BSHS}{\tc{cornflowerblue}{B}}
\newcommand{\Fconst}{\mathfrak{F}}

\newcommand{\udata}{\underline{u}}

\newcommand{\kprop}{\mathsf{P}}

\newcommand{\kk}{k}

\newcommand{\kerr}{\mathrm{K}}
\newcommand{\kerrdata}{\underline{\mathrm{K}}}

\newcommand{\Cb}{C_b}
\newcommand{\sCb}{\smash{\slashed{C}}_b}
\newcommand{\nosCb}{\tc{brinkpink}{C}_b}
\newcommand{\Hb}{H_b} 
\newcommand{\sHb}{\smash{\slashed{H}}_b}
\newcommand{\nosHb}{\tc{brinkpink}{H}_b}

\newcommand{\sC}{\smash{\slashed{C}}}
\newcommand{\nosC}{\tc{brinkpink}{C}}
\newcommand{\sH}{\smash{\slashed{H}}}
\newcommand{\nosH}{\tc{brinkpink}{H}}
\renewcommand{\H}{\tc{lcolor}{H}}

\newcommand{\Ht}{H_{\tang}}
\newcommand{\Ct}{C_{\tang}}

\newcommand{\sCX}{\smash{\slashed{\tc{cornflowerblue}{C}}}} 
\newcommand{\CX}{\tc{cornflowerblue}{C}} 
\newcommand{\sHX}{\smash{\slashed{\tc{cornflowerblue}{H}}}} 
\newcommand{\HX}{\tc{cornflowerblue}{H}} 

\newcommand{\Dspcl}{\tc{darkpastelpurple}{\slashed{\Delta}}}
\newcommand{\Dspop}{\tc{darkpastelpurple}{\Delta}}
\newcommand{\Dspdata}{\tc{darkpastelpurple}{\underline{\Delta}}}

\newcommand{\CM}{\tc{orange}{c_{\homMfd}}}

\newcommand{\LdivestNEW}[1]{\tc{red}{\ell}_{#1}}

\newcommand{\tang}{\tc{cornflowerblue}{T}}

\newcommand{\CLA}{\tc{orange}{b}}

\definecolor{babyblue}{rgb}{0.54, 0.81, 0.94}
\newcommand{\inj}{\tc{babyblue}{\sigma}}

\newcommand{\otimesRR}{\otimes_{\R}}

\newcommand{\cc}{\tc{ballblue}{c}}

\newcommand{\Vconst}{\mathcal{H}}
\newcommand{\epspower}{\tc{brinkpink}{\gamma}}
\newcommand{\CsmallGR}{\tc{junglegreen}{\eps}}
\newcommand{\ClargeGR}{\tc{junglegreen}{C}}
\newcommand{\CinGR}{\tc{junglegreen}{b}}

\newcommand{\HdataNEW}{\tc{bleudefrance}{H}_{\tc{bleudefrance}{\textnormal{data}}}}

\newcommand{\CsmallCpt}{\tc{bleudefrance}{\eps}}
\newcommand{\ClargeCpt}{\tc{bleudefrance}{C}}
\newcommand{\CHigherCpt}{\tc{bleudefrance}{b}}

\newcounter{counterqa}

\newcommand{\Pext}{\tc{junglegreen}{\mathcal{E}}}
\newcommand{\Pextbulk}{\tc{junglegreen}{\mathcal{E}_{\textnormal{bulk}}}}

\newcommand{\zzfix}{\zz_*}
\newcommand{\sfix}{s_*}
\newcommand{\refl}{\mathrm{R}}
\newcommand{\reflg}{\mathrm{R}^{\gx}}

\newcommand{\refldatag}{\underline{\mathrm{R}}^{\gx}}

\newcommand{\qmink}{\tc{flamingopink}{q_{0}}}
\newcommand{\deltamink}{\tc{flamingopink}{\delta_{0}}}

\newcommand{\bbulk}{\tc{flamingopink}{b}}
\newcommand{\Cbulk}{\tc{brinkpink}{C}}
\newcommand{\taufix}{\tc{brinkpink}{\tau_*}}

\newcommand{\PP}{\mathrm{P2}}

\newcommand{\taum}{\tau_{m}}

\newcommand{\taumax}{\tau_{1}}

\newcommand{\degu}{k}
\newcommand{\degv}{k'}
\newcommand{\degz}{k''}
\newcommand{\degomega}{q}
\newcommand{\degomegap}{q'}

\newcommand{\smallconstant}{\tc{orange}{\rho}}

\newcommand{\Cname}[1]{\tc{blue}{#1}}
\newcommand{\Capply}[1]{\mathscr{#1}}
\newcommand{\epsapply}[1]{\tc{blue}{\vareps}}

\newcommand{\CqEE}{\tc{lcolor}{C}}

\newcommand{\CqEEApp}{\Cname{\mathscr{C}}}

\newcommand{\Csmall}{\tc{darkcerulean}{\eps}}
\newcommand{\Clarge}{\tc{darkcerulean}{C}}

\newcommand{\CHigherIn}{\tc{junglegreen}{b}'}
\newcommand{\slashCHigherIn}{\NEW{\tilde{b}}}

\newcommand{\zcut}{{\tc{blue}{\zz_{*}}}}

\newcommand{\ALinM}{\tc{amethyst}{a}}
\newcommand{\amink}{\tc{amethyst}{a}}

\newcommand{\Amink}{\tc{amethyst}{A}}
\newcommand{\Lmink}{\tc{amethyst}{L}}
\newcommand{\Bmink}{\tc{amethyst}{B}}

\newcommand{\SFmink}{\tc{amethyst}{\beta}}

\newcommand{\ginj}{\tc{scolor}{G}}
\newcommand{\conj}{\tc{scolor}{J}}

\newcommand{\ideg}{k}
\newcommand{\outind}{i}

\newcommand{\deltabulk}{\tc{brinkpink}{\delta_0}}

\newcommand{\qq}{\tc{junglegreen}{q}}

\newcommand{\uf}{\underline{\smash{f}}}

\newcommand{\Cauxmain}{\tc{brinkpink}{C}}
\newcommand{\LL}{L}
\newcommand{\VV}{\tc{orange}{H}}

\newcommand{\MdimNEW}{\tc{orange}{m}} 
\newcommand{\mataux}{\alpha}
\newcommand{\qminkaux}{q_0'}

\newcommand{\bb}{\tc{red}{b}'}

\newcommand{\xivec}{|\vec{\smash{\xi}\mathclap{\phantom{i}}}|}

\newtheoremstyle{mytheoremstyle}
{14pt}
{14pt}
{\itshape}
{20pt}
{\bfseries}
{.}
{.5em}
{}

\newtheoremstyle{myremarkstyle}
{10pt}
{10pt}
{}
{20pt}
{\itshape}
{.}
{.5em}
{}

\newtheoremstyle{myproofstyle}
{12pt}
{12pt}
{}
{20pt}
{\bf\itshape}
{.}
{.5em}
{}

\newtheoremstyle{mycommentstyle}
{5pt}
{5pt}
{\footnotesize\sffamily\color{airforceblue}}
{0pt}
{}
{}
{.5em}
{}

\newtheoremstyle{myissuestyle}
{14pt}
{14pt}
{\itshape\color{magenta}}
{20pt}
{\bfseries}
{.}
{.5em}
{}

\newtheoremstyle{myquestionstyle}
{14pt}
{14pt}
{\color{airforceblue}}
{20pt}
{\bfseries}
{.}
{.5em}
{}

\theoremstyle{mytheoremstyle}
\newtheorem{theorem}{Theorem}
\newtheorem*{theorem*}{Theorem}
\newtheorem{definition}{Definition}
\newtheorem{lemma}{Lemma}
\newtheorem{prop}[theorem]{Proposition}
\newtheorem{corollary}[lemma]{Corollary}

\newtheorem{convention}{Convention}

\theoremstyle{myissuestyle}

\theoremstyle{myquestionstyle}

\theoremstyle{myremarkstyle}
\newtheorem{remark}{Remark}

\theoremstyle{myproofstyle}
\newtheorem*{PROOF}{Proof}

\newcommand{\coverfigureintro}{
\begin{tikzpicture}[inner sep=0pt,scale=1]
\def\coordleft{
(-3.25133,0.748668) (-3.25133,0.748668) (-3.15133,0.68591) (-3.05133,0.622794) (-2.95133,0.559559) (-2.85133,0.496437) (-2.75133,0.433666) (-2.65133,0.371487) (-2.55133,0.310152) (-2.45133,0.249926) (-2.35133,0.191088) (-2.25133,0.133942) (-2.15133,0.0788104) (-2.05133,0.0260449) (-2.,0.)
}
\def\coordright
{(3.25133,0.748668) (3.25133,0.748668) (3.15133,0.68591) (3.05133,0.622794) (2.95133,0.559559) (2.85133,0.496437) (2.75133,0.433666) (2.65133,0.371487) (2.55133,0.310152) (2.45133,0.249926) (2.35133,0.191088) (2.25133,0.133942) (2.15133,0.0788104) (2.05133,0.0260449) (2.,0.)}
\def\coordleftDs{(-3.78973,0.210274) (-3.78973,0.210274) (-3.68973,0.159174) (-3.58973,0.10807) (-3.48973,0.0571956) (-3.38973,0.00678628) (-3.37617,0.)}
\def\coordrightDs{(3.37617,0.) (3.38973,0.00678628) (3.48973,0.0571956) (3.58973,0.10807) (3.68973,0.159174) (3.78973,0.210274) (3.78973,0.210274)}
\node (tip) at (0,4) {}; %
\node (l0) at (-4,0) {}; %
\node (r0) at (4,0) {}; %
\fill[gray!60] plot coordinates {\coordleftDs \coordrightDs (0,4)};
\draw [thick,gray] plot [smooth] coordinates {\coordleft}; 
\draw [thick,gray] plot [smooth] coordinates {\coordright}; 

\draw[color=black,thick, fill=white] (tip) circle (.05);
\draw[color=black,thick, fill=white] (r0) circle (.05);
\node[anchor=south,yshift=1mm] at (tip) {$\timeinf$};
\node[anchor=south west,yshift=1mm] at (2,2) {$\fnullinf$};

\draw[color=black,thick, fill=white] (l0) circle (.05);
\node[anchor=south west,xshift=-3mm,yshift=1mm] at (l0) {$\spaceinf$};
\node[anchor=south east,xshift=3mm,yshift=1mm] at (r0) {$\spaceinf$};
\draw[->] (-4.5,-0.2) -- (-4.5,4.5) node[anchor=south east,yshift=1mm] {$\tau$};
  \foreach \y in {0,4}
    \draw (-4.6,\y) -- (-4.4,\y);
	\node[left,thick] at (-4.65,0) {$0$};
	\node[left,thick] at (-4.65,4) {$\pi$};
\draw[->] (-4.65,0) -- (5,0) 
	node[anchor=north,xshift=2mm,yshift=-1.1mm] {$\arccos\xi^4$};
 \draw[line width=1.5pt, black] (-3.95,0) -- (3.95,0);
  \foreach \x in {-4,0,4}
    \draw (\x,-0.1) -- (\x,0.1);
	\node[below,thick] at (-4,-0.2) {$0$};
	\node[below,thick] at (0,-0.2) {$\pi$};
	\node[below,thick] at (4,-0.2) {$0$};
\fill[pattern={Dots[distance=1.7pt]}, pattern color=black]
  plot[smooth] coordinates {(-4,0) \coordleft};
\fill[pattern={Dots[distance=1.7pt]}, pattern color=black]
  plot[smooth] coordinates {(4,0) \coordright};
\draw[line width=1.5pt,ufogreen] (-3.95,0.05)--(-0.05,3.95);
\draw[line width=1.5pt,ufogreen] (3.95,0.05)--(0.05,3.95);
\end{tikzpicture}}

\newcommand{\IntroLevelsetsS}{
\begin{tikzpicture}[inner sep=0pt]
\node (i0) at (0,0) {};
\node (idown) at (4.4,0) {};
\node (iup) at (2.,2.) {};

\node (s1up) at (1.4,1.4) {} ;
\node (s1down) at (4.2,0) {} ;

\node (s2up) at (0.7,0.7) {} ;
\node (s2down) at (2.1,0) {} ;

\node (s3up) at (0.35,0.35) {} ;
\node (s3down) at (1.05,0) {} ;

\node (s4up) at (0.175,0.175) {} ;
\node (s4down) at (0.525,0) {} ;

\node (t0) at ({3., 0.6}) {} ;
\node (t1) at (2.33333, 0.933333) {} ;
\node (t2) at (1.90909, 1.14545) {} ;
\node (t3) at (1.61538,1.29231) {} ;

\node(chartup) at (0,2.45) {};
\node(chartright) at (4.6,0) {};

\path[draw,line width=1 pt] (iup.center)--(i0.center)--(idown.center);
\fill[pattern={Dots[distance=2.6pt,radius=0.6pt]},pattern color=black] (i0.center)--(s1up.center)--(s1down.center)--(i0.center);

\path[draw,->] (i0.center)--(chartup.center) node [anchor=south east] {$y^0$};
\path[draw,->] (i0.center)--(chartright.center) node [anchor=north west] {$|\vec{y}|$};

\path[draw,line width=1 pt,ceruleanblue] (s1up.center)--(s1down.center) 
	node [midway, above,yshift=1.5mm,xshift=1.5mm,black] {\footnotesize $\s=1$};
\path[draw,line width=1 pt,ceruleanblue] (s2up.center)--(s2down.center);
\path[draw,line width=1 pt,ceruleanblue] (s3up.center)--(s3down.center);
\path[draw,line width=1 pt,ceruleanblue] (s4up.center)--(s4down.center);

\draw[line width=1.pt,ufogreen] (i0)--(iup.center);

\draw[color=black, fill=white, thick] (i0.center) circle (.03) 
	node[anchor=north east] (i0) {$\spaceinf$};

\node[anchor=south east,yshift=1mm] at (1,1) {$\fnullinf$};
\end{tikzpicture}
}
\newcommand{\IntroLevelsetsTau}{
\begin{tikzpicture}[inner sep=0pt,scale=0.55]
\def\coordleftDs{(-3.78973,0.210274) (-3.78973,0.210274) (-3.68973,0.159174) (-3.58973,0.10807) (-3.48973,0.0571956) (-3.38973,0.00678628) (-3.37617,0.)}
\def\coordrightDs{(3.37617,0.) (3.38973,0.00678628) (3.48973,0.0571956) (3.58973,0.10807) (3.68973,0.159174) (3.78973,0.210274) (3.78973,0.210274)}
\node (tip) at (0,4) {}; %
\node (l0) at (-4,0) {}; %
\node (r0) at (4,0) {}; %
\draw[->] (-4.5,-0.2) -- (-4.5,4.5) node[anchor=south east,yshift=0.5mm] {$\tau$};
  \foreach \y in {0,4}
    \draw (-4.6,\y) -- (-4.4,\y);
	\node[left,thick] at (-4.65,0) {$0$};
	\node[left,thick] at (-4.65,4) {$\pi$};
	
\draw[->] (-4.65,0) -- (5,0) node[anchor=north,yshift=0mm,xshift=-4mm] {$\arccos\xi^4$};
  \foreach \x in {-4,0,4}
    \draw (\x,-0.1) -- (\x,0.1);
	\node[below] at (-4,-0.2) {$0$};
	\node[below] at (0,-0.2) {$\pi$};
\fill[color=gray!60]
  plot coordinates {\coordleftDs \coordrightDs (0,4)};
\draw[thick,color=gray!80]
  plot coordinates {\coordleftDs \coordrightDs (0,4) \coordleftDs};

\draw[color=black,thick, fill=white] (tip) circle (.05);
\draw[color=black,thick, fill=white] (r0) circle (.05);
\node[anchor=south,yshift=1mm] at (tip) {$\timeinf$};

\draw[color=black,thick, fill=white] (l0) circle (.05);
\node[anchor=south west,xshift=-2.3mm,yshift=1mm] at (l0) {$\spaceinf$};
\node[anchor=south east,xshift=3mm,yshift=1mm] at (r0) {$\spaceinf$};
\foreach \x in {0.5,1,1.5,2,2.5,3,3.5}
    \path[draw,line width=0.75 pt,ceruleanblue] (-4+\x,\x) -- (4-\x,\x);
    
    \path[draw,line width=0.75 pt,ceruleanblue] (-3.37617,0) -- (3.37617,0);
    \path[draw,line width=0.75 pt,ceruleanblue] (-3.62288,0.125) -- (3.62288,0.125);
    \path[draw,line width=0.75 pt,ceruleanblue] (-4+0.25,0.25) -- (4-0.25,0.25);
    \path[draw,line width=0.75 pt,ceruleanblue] (-4+0.375,0.375) -- (4-0.375,0.375);
   \draw[line width=1.pt,ufogreen] (-3.95,0.05)--(-0.05,3.95);
   \draw[line width=1.pt,ufogreen] (3.95,0.05)--(0.05,3.95);
   \draw[line width=1.pt,black] (-3.95,0)--(3.95,0);
\node[anchor=south west,yshift=1mm] at (2,2) {$\fnullinf$};
\end{tikzpicture}
}

\title{Perturbations of Minkowski spacetime \\ 
with regular conformal compactification}
\author{Andrea N\"utzi}
\date{}

\begin{document}
\maketitle
\begin{abstract}
We construct perturbations of Minkowski spacetime in general relativity, when given initial data that decays inverse polynomially to initial data of a Kerr spacetime towards spacelike infinity. We show that the perturbations admit a regular conformal compactification at null and timelike infinity, where the degree of regularity increases linearly with the rate of decay of the initial data to Kerr initial data. In particular, the compactification is smooth if the initial data decays rapidly to Kerr initial data. This generalizes results of Friedrich, who constructed spacetimes with a smooth conformal compactification in the case when the initial data is identical to Kerr initial data on the complement of a compact set. Our results rely on a novel formulation of the Einstein equations about Minkowski spacetime introduced by the author, that allows one to formulate the dynamic problem as a quasilinear, symmetric hyperbolic PDE that is regular at null infinity and with null infinity being at a fixed locus. It is not regular at spacelike infinity, due to the asymptotics of Kerr. Thus the main technical task is the construction of solutions near spacelike infinity, using tailored energy estimates. To accomplish this, we organize the equations according to homogeneity with respect to scaling about spacelike infinity, which identifies terms that are leading, respectively lower order, near spacelike infinity, with contributions from Kerr being lower order.
\end{abstract}

\tableofcontents

\section{Introduction}
\label{sec:Introduction}

We study small perturbations of Minkowski spacetime,
as solutions of the vacuum Einstein equations in general relativity.
Minkowski spacetime is stable \cite{ChristodoulouKlainermanStability},
that is, small perturbations of the initial data
yield solutions of the Einstein equations that are globally close to Minkowski spacetime.
For such perturbations it is interesting to understand  
the null and timelike asymptotics, 
which carry information
about the scattering of gravitational waves.
In particular, it is a long standing question if,
and under what conditions, 
the perturbations admit (like Minkowski itself) 
a smooth conformal compactification \cite{PenroseAsymptoticSimplicity}.

A simple class of initial data for the Einstein equations
is given by solutions of the constraint equations
that are identical to the initial data of a Kerr spacetime 
on the complement a compact set,
and everywhere close to Minkowski initial data 
\cite{ChruscielDelayGluing,CorvinoGluing,CorvinoSchoenGluing}.
(Due to the positive mass theorem \cite{SchoenYauPositiveMass,WittenPositiveMass}, there is no nontrivial initial data
that is identical to Minkowski on the complement of a compact set.)
Solutions of the Einstein equations with such initial data do indeed 
admit a smooth conformal compactification at null and timelike infinity \cite{ChruscielDelayAsymptoticallySimpleSpacetimes,
CorvinoStabPenrose,
Friedrichgeodesicallycomplete} 
(see \cite{FriedrichConformalEinsteinEvolution,FriedrichPeelingQuestion} for review articles).
This is obtained as follows:
By finite speed of propagation, the metric
is identical to a Kerr spacetime in a neighborhood of spacelike infinity,
and Kerr itself admits a smooth conformal compactification at null infinity (not at spacelike infinity).
Away from spacelike infinity the metric is then constructed
using Friedrich's conformal field equations \cite{FriedrichCFE1,FriedrichCFE2,Friedrichgeodesicallycomplete}. 
In this formulation of the Einstein equations,
the dynamic problem is hyperbolic including 
along null and timelike infinity,
and thus one must only solve 
a hyperbolic PDE with small initial data on a compact domain.
The conformal field equations are formulated in terms of
a conformally rescaled smooth metric and the conformal factor,
and thus readily imply that the physical metric
admits a smooth conformal compactification.

More general asymptotically flat
initial data was considered in the stability results 
\cite{BieriZipserStability,
ChristodoulouKlainermanStability,
HintzVasyGlobalStability,
KlainermanNicoloEvolution,
KlainermanNicoloPeeling,
LindbladRodnianskiGlobalStability}.
Under these more general assumptions, 
sharp decay rates and precise asymptotics of the solutions towards null infinity
are obtained, in different kinds of gauges.
However, the solutions are not shown to admit a regular
conformal compactification.
In \cite{KlainermanNicoloPeeling},
using a double null gauge,
it was shown that for a large class of
initial data, the corresponding metrics satisfy peeling \cite{PenroseAsymptoticSimplicity}, 
which is necessary but not sufficient for existence
of a regular conformal compactification.
In \cite{LindbladRodnianskiGlobalStability} the solutions are constructed in a harmonic gauge,
with sharp decay estimates given in \cite{LindbladAsymptotics}.
In \cite{HintzVasyGlobalStability}, using a generalized harmonic gauge, 
it was shown in particular that polyhomogeneous initial data 
yield metrics that are polyhomogeneous at null infinity,
and that a leading logarithmic term at null infinity 
is nonzero, 
by expressing its spherical average in terms of the Bondi mass.
However, this logarithmic term is a consequence of the gauge,
namely it is not there when using an 
improved gauge \cite{HintzImprovedGauge}.

Here we consider initial data 
that asymptote to Kerr initial data towards spacelike infinity
(but that are not necessarily equal to Kerr near spacelike infinity),
and show that also this class of initial data yields solutions of the Einstein equations 
that admit a regular conformal compactification at null and timelike infinity.
Informally, the decay of the initial data and the regularity 
of the conformal compactification are related as follows:
There exists a positive integer $a$ such that for every sufficiently 
large positive integer $\ell$:
\begin{quote}
If the initial data decays to Kerr inverse polynomially with rate $\ell$,
then the corresponding Ricci-flat metric admits a 
regular conformal compactification 
at null and timelike infinity with $\ell-a$ derivatives.
If the data decays to Kerr rapidly, then the 
compactification is smooth.
\end{quote}
This is made precise in Theorem \ref{thm:main} 
at the end of this introduction,
and also in Theorem \ref{thm:mainpointwise} below,
which is a simplified version of Theorem \ref{thm:main}.

The results in this paper are 
conditional on the existence of solutions of the
constraint equations with specific asymptotics at infinity.
We plan to construct these solutions in a subsequent article
based on \cite{HomotopyPaper} (which was motivated by this
particular application), see Remark \ref{rem:Plinearized}.

The results in this paper rely in particular on:
\begin{itemize}
\item 
A new formulation of the Einstein equations about Minkowski spacetime
in which the dynamic problem is quasilinear symmetric hyperbolic including 
at null and timelike infinity, with null infinity being at a fixed locus \cite{Thesis}.
It is independent of, but inspired by,
the conformal field equations \cite{FriedrichCFE1,FriedrichCFE2}.
\item 
Energy estimates near spacelike infinity for the dynamic problem, 
which is mildly singular due to the asymptotics of Kerr
(we cannot appeal to finite speed of propagation since
our initial data is not equal to Kerr near spacelike infinity).
This is the main novelty of this paper (Section \ref{sec:Abstract}, \ref{sec:SpaceinfConstruction}).
\end{itemize}

In the remainder of this introduction we 
review the formulation of the Einstein equations from \cite{Thesis};
state the simplified Theorem \ref{thm:mainpointwise}
which only uses pointwise estimates;
state the main Theorem \ref{thm:main};
and outline the strategy of the proof.
\step
\textbf{Geometric conformal compactification.}
The conformal compactification of Minkowski spacetime is given by the Einstein cylinder. 
This is the manifold $\cyl=\R\times S^3$ together with the 
conformal Lorentzian metric $[\gcyl]$, where $\gcyl=-d\tau^{\otimes2} +\gS$. 
Here $\tau$ is the standard coordinate on $\R$,
and $\gS$ is the round metric on the three-sphere $S^3$.
We view $S^3$ as the unit sphere in $\R^4$
and denote the standard coordinates on $\R^4$
by $\xi=(\xi^1,\xi^2,\xi^3,\xi^4)$.
Define
\begin{equation}\label{eq:defh}
\nullgen\;=\;\cos(\tau)-\xi^4 \;\in\; C^\infty(\cyl)
\end{equation}
Then Minkowski spacetime is isometric to $(\diamond,\eta)$ defined as follows:
\begin{equation}\label{eq:defdiamondIntro}
\diamond 
\;=\;
\{ (\tau,\xi) \in \cyl \mid 
-\pi<\tau<\pi\;,\;
0  < \nullgen(\tau,\xi) \}
\qquad
\gmink = \nullgen^{-2} \gcyl|_{\diamond}
\end{equation}
We refer to $\diamond$ as the Minkowski diamond. 
Its boundary has five components
given by 
spacelike infinity $\spaceinf=(0,(0,0,0,1))$,
future and past timelike infinity $\fptimeinf=(\pm\pi,(0,0,0,-1))$,
and future and past null infinity $\fpnullinf$,
see Figure \ref{fig:cylinderintro}. 

Let $\diamondy\subset\cyl$ be the image of $\diamond$ under the 
spacial reflection $(\tau,\xi)\mapsto(\tau,-\xi)$, which is an open neighborhood
of $\spaceinf$.
We will use two sets of coordinates:
\begin{align}\label{eq:xyintro}
\text{Cartesian coordinates $x$ on $\diamond$}
\qquad
\text{Cartesian coordinates $y$ on $\diamondy$}
\end{align}
defined in \eqref{eq:xxcoords} respectively \eqref{eq:yycoords}. 
On $\diamond$ one has
$\eta = \eta_{\mu\nu}dx^\mu\otimes dx^\nu$ where $\eta_{\mu\nu}$
are the components of the matrix $\diag(-1,1,1,1)$
(we use the convention that 
repeated Greek indices are implicitly summed over $0,1,2,3$).
On the common domain of definition the coordinates 
are related by Kelvin inversion,
$y = \frac{x}{\eta_{\mu\nu}x^\mu x^\nu}$.
The coordinates $y$ are in particular regular near $\spaceinf$,
and $\spaceinf$ is the origin $y=0$.

\begin{figure}
\centering 
\includegraphics[width=0.28\linewidth]{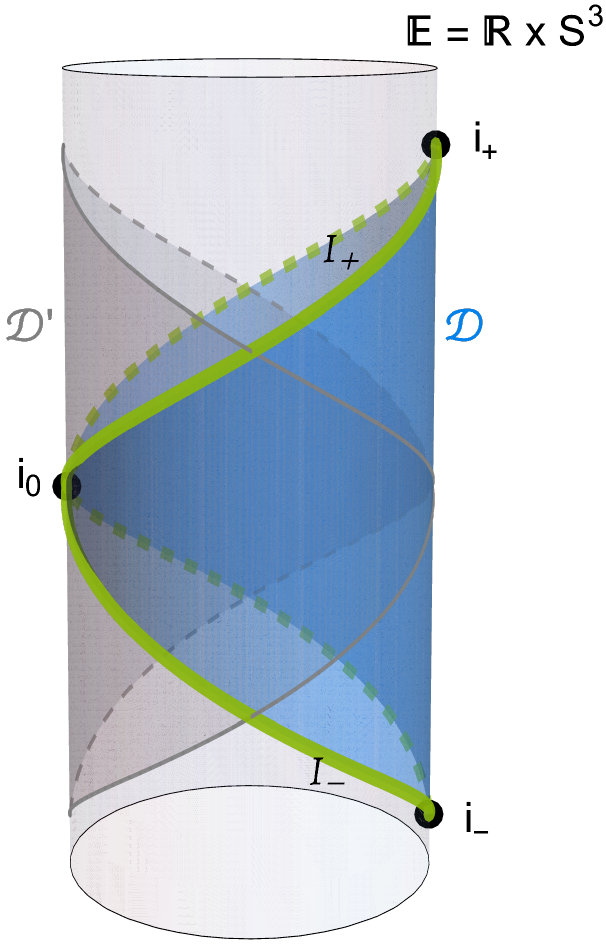}
\captionsetup{width=115mm}
\caption{
The figure illustrates the Minkowski diamond $\diamond\subset\cyl$,
where the sphere $S^3$ is drawn as an $S^1$.
The boundary of $\diamond$ is given by spacelike infinity $\spaceinf$,
future and past timelike infinity $\fptimeinf$, and
future and past null infinity $\fpnullinf$.
The set $\diamondy\subset\cyl$ is given by the image of $\diamond$
under the reflection $(\tau,\xi)\mapsto(\tau,-\xi)$,
and is an open neighborhood of $\spaceinf$.}
\label{fig:cylinderintro}
\end{figure}
\step
\textbf{Equivalent reformulation of the Einstein equations.}
We recall the formulation of the Einstein equations introduced in \cite{Thesis},
defined on the conformal compactification of Minkowski spacetime.
This formulation is effectively a repackaging of 
the Newman Penrose orthonormal frame formalism \cite{NewmanPenroseSpinCoeff},
respectively of the differential graded Lie algebra formalism \cite{RTgLa1,RTgLa2},
but written in a way which makes the regularity properties at the boundary clear.
See Remark \ref{rem:Friedrich} for a comparison to 
Friedrich's  conformal field equations.

The Einstein equations take the form
\begin{align}\label{eq:MC}
\dg u + \tfrac12[u,u] = 0
\end{align}
The unknown $u$ describes an orthonormal frame,
a connection, and a Weyl curvature on $\diamond$.
The operators $\dg$ and $[\cdot,\cdot]$ are linear respectively
bilinear first order differential operators, where $\dg$ describes
linearized gravity about Minkowski, and $[\cdot,\cdot]$ describes
the gravitational interaction. Informally one has
\begin{align}\label{eq:11withmetrics}
\frac{\text{small $u$ on $\diamond$ that solve \eqref{eq:MC}}}{\sim}
\;\;\overset{1:1}{=}\;\;
\frac{\text{Ricci-flat metrics on $\diamond$}}{\text{diffeomorphisms}}
\end{align}
where the denominator $\sim$ stands for the gauge group,
given by diffeomorphisms and orthonormal frame transformations,
see also \cite{RTgLa2}.
Via this correspondence, the trivial solution $u=0$ of \eqref{eq:MC}
yields the Minkowski metric.

The equation \eqref{eq:MC} has the form of a Maurer-Cartan (MC) equation.

The unknown $u$ consists of two components, $u=\uo\oplus\uI$,
corresponding to orthonormal frame and connection
respectively to Weyl curvature, 
which are elements in the first respectively second direct summand of a space
\begin{align}\label{eq:gx1Dintro}
\gx^1(\diamond) \;=\; (\Omega^1(\diamond)\otimesRR\Kil) \oplus \I^2(\diamond)
\end{align}
This space is the module of sections, over $\diamond$, 
of a (trivial) rank 50 vector bundle 
$\gx^1$ 
defined globally on $\cyl$.
We denote the space of sections of $\gx^1$ over $\cyl$ by
\begin{align*}
\gx^1(\cyl) \;=\; (\Omega^1(\cyl)\otimesRR\Kil) \oplus \I^2(\cyl)
\end{align*}
We will see that $\dg$, $[\cdot,\cdot]$ 
are smooth differential operators on $\cyl$.
Here:
\begin{itemize}
\item 
\textit{Weyl curvature.}
$\I^2(\cyl)$ is the rank 10 submodule 
of $S^2(\Omega^2(\cyl))[-1]$ of all sections that
satisfy the algebraic symmetries and traceless condition 
(relative to the conformally flat cylinder metric $[\gcyl]$) of Weyl tensors.
Here $S^2$ is the symmetric tensor product over $C^\infty$,
and $\Omega^2(\cyl)$ are the smooth differential two-forms on $\cyl$,
and $[-1]$ denotes tensoring with a density bundle
(Definition \ref{def:It}).
The Weyl curvature $\uI$ is an element of $\I^2(\diamond)$.
\item 
\textit{Orthonormal frame and connection.}
$\Omega^1(\cyl)$ are the smooth differential one-forms,
a module of rank $4$.
Further $\Kil$ is isomorphic to the Lie  algebra of the Poincar\'e group,
and explicitly given by the 10-dimensional vector space of 
conformal Killing fields on $\cyl$ that restrict to 
ordinary Killing fields with respect to the Minkowski metric $\eta$ on $\diamond$.
The element $\uo$ lies in the rank 40 module $\Omega^1(\diamond)\otimesRR\Kil$,
where 16 degrees of freedom describe 
an orthonormal frame and hence a metric,
and where 24 degrees of freedom describe an affine connection that is
compatible with that metric.
Concretely, expand
\[ 
u_0 = \tsum_{i=1}^{10}\omega_i\otimes \KilBasis_i
\]
with $\omega_i\in\Omega^1(\diamond)$
and $\KilBasis_1,\dots,\KilBasis_{10}$ a basis of $\Kil$.
Define the $C^\infty$-linear map
\begin{align}\label{eq:Fuodef}
\frame{\uo}: \Gamma(T\diamond)\to\Gamma(T\diamond)
&&
\frame{\uo}(X) =  \tsum_{i=1}^{10}\omega_i(X)\KilBasis_i
\end{align}
If the endomorphism $\one+\frame{\uo}$ is invertible at every point on $\diamond$, then 
\begin{align}\label{eq:metricdefIntro}
g^{-1} 
= 
(\one + \frame{\uo})^{\otimes2} \eta^{-1}
\end{align}
defines a smooth Lorentzian metric $g$ on $\diamond$,
where $\eta$ is the Minkowski metric \eqref{eq:defdiamondIntro}.
This formula is to be understood as follows:
$\eta^{-1}$ is an element in the second tensor power of $\Gamma(T\diamond)$, 
and the endomorphism $\one + \frame{\uo}$ is applied to each factor.
With this definition of the metric $g$, 
the map $\one + \frame{\uo}$ is an orthonormal frame for $g$, 
in the sense that
\begin{equation}\label{eq:1Fframe}%
g((\one+\frame{\uo})X,(\one+\frame{\uo})Y) = \eta(X,Y)
\end{equation}
for $X,Y\in\Gamma(T\diamond)$.
See \eqref{eq:h^2gcyl} and Proposition \ref{prop:metricregularity}
for the regularity properties of $g$ at the boundary.
See Remark \ref{rem:ginbasis,connection} 
for the definition of the connection. 
\end{itemize}

\begin{remark}\label{rem:ginbasis,connection}
One can define \eqref{eq:metricdefIntro} equivalently using a basis,
which will also allow us to specify the associated metric compatible affine connection.
We use the coordinates $x$ in \eqref{eq:xyintro},
and the basis of $\Kil$ given by boosts $B^{\mu\nu}=-B^{\nu\mu}$ and translations $T_\mu$
in \eqref{eq:TB}. 
On the open Minkowski diamond $\diamond$, expand 
\begin{align}\label{eq:uobasis}
\begin{aligned}
\uo 
	&= (E_\mu^\nu-\delta_\mu^\nu) dx^\mu \otimes T_\nu\\
	&\quad
	-\tfrac12 \eta_{\alpha\beta}\Gamma^{\beta}_{\mu\nu} 
	\big(dx^\mu\otimes B^{\nu\alpha} 
	- 
	( x^\nu dx^\mu \otimes (\eta^{\alpha\kappa} T_\kappa)
	-
	x^\alpha dx^\mu  \otimes (\eta^{\nu\kappa} T_\kappa))\big) 
\end{aligned}
\end{align}
for unique functions $E_\mu^\nu$ and $\Gamma^{\beta}_{\mu\nu}$
such that $\eta_{\alpha\beta}\Gamma^{\beta}_{\mu\nu}$ is
antisymmetric in $\nu\alpha$,
and with $\delta_{\mu}^\nu$ the Kronecker delta.
Define the four vector fields $E_\mu = E_\mu^\nu \p_{x^\nu}$ 
(one has $E_\mu = (\one+\frame{\uo})\p_{x^\mu}$).
These are pointwise linearly independent iff 
$\one+\frame{\uo}$ in \eqref{eq:metricdefIntro} is pointwise invertible.
In this case, the metric $g$ and the metric compatible affine connection $\nabla$ 
associated to $\uo$ are
$g^{-1} = \eta^{\mu\nu} E_\mu \otimes E_\nu$, 
$\nabla_{E_{\mu}} E_\nu =\Gamma_{\mu\nu}^{\beta} E_\beta$.
\end{remark}

The left hand side of the Einstein equations \eqref{eq:MC} takes values in the space
\begin{equation}\label{eq:g2D}
\gx^2(\diamond) \;=\; (\Omega^2(\diamond)\otimesRR\Kil) \oplus \I^3(\diamond)
\end{equation}
where $\I^3(\diamond)\subset (\Omega^3(\diamond)\otimesCinf\Omega^2(\diamond))[-1]$ 
is a submodule of rank 16, specified in Definition \ref{def:It}.
The components of the equation in the first direct summand 
of \eqref{eq:g2D} are the conditions
that the connection is torsion-free, and that the Riemann curvature is
equal to the Weyl curvature; the components in the second direct
summand of \eqref{eq:g2D} are the equations of motion for the Weyl curvature.
The space $\gx^2(\diamond)$ is again the module of sections, over $\diamond$, 
of a (trivial) rank 76 vector bundle $\gx^2$ defined on $\cyl$.
We denote the space of sections of $\gx^2$ over $\cyl$ by
\begin{align*}
\gx^2(\cyl) \;=\; (\Omega^2(\cyl)\otimesRR\Kil) \oplus \I^3(\cyl)
\end{align*}
The operators $\dg$, $[\cdot,\cdot]$ are 
smooth differential operators on $\cyl$, 
given as follows:
\begin{itemize}
\item 
$\dg:\gx^1(\cyl)\to\gx^2(\cyl)$ is a 
smooth first order linear differential operator
\[ 
\dg u = \big((\ddR\otimes\one) \uo - \inj\uI\big) \oplus (\dI\uI)
\]
where $\ddR$ is the de Rham differential;
$\dI$ is a first order differential operator
that is conformally invariant
(i.e.~commutes with conformal isometries of the Einstein cylinder,
which act on the modules $\I^2(\cyl)$ and $\I^3(\cyl)$);
and $\inj:\I^2(\cyl)\to\Omega^1(\cyl)\otimesRR\Kil$ 
is a $C^\infty$-linear map that is only Poincar\'e invariant.
See Definition \ref{def:dI} and \ref{def:inj}.
\item 
$[\cdot,\cdot]:\gx^1(\cyl)\times \gx^1(\cyl)\to\gx^2(\cyl)$
is a smooth first order bilinear differential operator, 
see \eqref{eq:[]gdef}.
Here we only define its restriction to $\Omega^1(\cyl)\otimesRR\Kil$
in both inputs.
This takes values in $\Omega^2(\cyl)\otimesRR\Kil$, and is given by 
\begin{align*}
[\omega\otimes \KilEl,\omega'\otimes \KilEl']
	&= \omega\wedge \omega' \otimes [\KilEl,\KilEl']
	+ \omega\wedge (\Lie_{\KilEl}\omega')\otimes \KilEl'
	- (\Lie_{\KilEl'}\omega) \wedge \omega'\otimes \KilEl
\end{align*}
where $\Lie$ denotes the Lie derivative.
\end{itemize}

The correspondence \eqref{eq:11withmetrics} is given as follows.
If $u=\uo\oplus\uI \in\gx^1(\diamond)$ solves \eqref{eq:MC},
and if $\one+\frame{\uo}$ is pointwise invertible on $\diamond$, 
then the metric $g$ in \eqref{eq:metricdefIntro} 
is Ricci-flat:
\[ 
\Ric(g) = 0
\]
Conversely, given a Ricci-flat metric, 
after choosing an
orthonormal frame one obtains a solution $u=\uo\oplus\uI$ of \eqref{eq:MC}
by defining $\uo$ as in \eqref{eq:uobasis}, 
with \smash{$E_\mu^\nu$} the components of the orthonormal frame relative to 
$\p_{x^\nu}$
and with \smash{$\Gamma^\beta_{\mu\nu}$}
the coefficients of the Levi-Civita connection of $g$
relative to the orthonormal frame,
and by defining $\uI$ to be the Weyl curvature.
See Section \ref{sec:recoveryofmetric}.
\step
The Einstein equations \eqref{eq:MC} are regular including along
the boundary of $\diamond$.
In particular, under appropriate gauge fixing conditions, 
\eqref{eq:MC} contains a necessary square subsystem that
is quasilinear symmetric hyperbolic on $\overline{\diamond}$,
with a principal symbol that does not degenerate along $\p\diamond$.
The remaining equations are the constraints, which themselves solve a linear homogeneous symmetric hyperbolic system, i.e.~the constraints propagate.
See Section \ref{sec:gauge_i0} and \ref{sec:Gauge_cyl}.

We will construct solutions $u=\uo\oplus\uI$ of \eqref{eq:MC} on $\diamond$
that extend $C^k$-regularly to null and timelike infinity
(where $k$ depends on the decay of the initial data to Kerr),
and such that $\frame{\uo}$ is uniformly small.
It then follows that the associated metric $g$ 
defined in \eqref{eq:metricdefIntro} admits a 
$C^k$-regular conformal compactification
at null and timelike infinity, specifically,
$\nullgen^2g$ extends as a $C^k$-metric, where
$\nullgen$ is the conformal factor \eqref{eq:defh}.
This can easily be seen using 
(by \eqref{eq:metricdefIntro} and \eqref{eq:defdiamondIntro})
\begin{equation}\label{eq:h^2gcyl}
(\nullgen^2g)^{-1} 
= 
(\one + \frame{\uo})^{\otimes2} \gcyl^{-1} 
\end{equation}

To understand the causal structure of $g$, note that 
every vector field in $\Kil$ is tangential to $\p\diamond$,
which implies that the one-form 
$\dualframe{\uo}(d\nullgen)/\nullgen$
is regular along $\p\diamond\setminus\spaceinf$
(where $\dualframe{\uo}$ is the linear dual of  $\frame{\uo}$), 
and for our solutions $u$ it will be 
uniformly small (using, near spacelike infinity, 
a homogeneous basis, see below).
This will imply that the metric $g$ on $\diamond$ is null geodesically complete,
and that the locus of future and past null infinity of $g$
is equal to $\fnullinf$ respectively $\pnullinf$,
i.e., to the locus of 
future respectively past 
null infinity of the Minkowski metric.

See Proposition \ref{prop:metricregularity} for a detailed
account about how properties of the endomorphism $\frame{\uo}$ 
imply properties of the metric $g$.
\begin{remark}\label{rem:Friedrich}
We compare the formulation of the Einstein equations \eqref{eq:MC}
to Friedrich's conformal field equations \cite{FriedrichCFE1,FriedrichCFE2}.
The conformal field equations are formulated in terms
of a smooth conformal factor $C$ and a smooth metric $\tilde{g}$.
The physical spacetime is then given by the domain 
$C>0$ and the metric $C^{-2}\tilde{g}$.
The vanishing locus $C=0$, $\ddR C\neq0$ is the null infinity
locus of the physical spacetime.
The Einstein equation are then
\begin{equation}\label{eq:RicCg}
\Ric(C^{-2}\tilde{g})=0
\end{equation}
and are apparently singular when $C=0$.
Using an orthonormal frame formalism similar to Newman Penrose,
and gauge fixing, Friedrich reduces \eqref{eq:RicCg} 
to a first order quasilinear symmetric hyperbolic system,
and shows that the constraints propagate.
Our approach \eqref{eq:MC} differs from Friedrich's approach
\cite[Remark 1]{Thesis}:
\begin{itemize}
\item 
The characteristic feature of Friedrich's approach is that
the conformal factor $C$ is used as an unknown, 
and determined by the Einstein equations. 
In particular the locus of null infinity, determined by $C=0$, $\ddR C\neq0$,
depends on the unknown. 
In our approach the conformal factor is not used as an unknown,
and null infinity of the physical spacetime is always equal to 
the null infinity locus $\pnullinf\cup\fnullinf$ of Minkowski spacetime.
\item 
Friedrich's approach is background independent. 
Our approach is background dependent,
designed for perturbation theory about Minkowski.
This allows us for example to construct solutions on the fixed manifold $\diamond$.
\end{itemize}
\end{remark}
It is useful to observe that the spaces $\gx^1(\cyl)$ and $\gx^2(\cyl)$
are the degree one respectively two components of a 
differential graded Lie algebra (dgLa)
$$
\gx(\cyl) = \oplus_{k=0}^{4}\gx^k(\cyl)
$$
with differential and Lie bracket given by $\dg$ respectively $[\cdot,\cdot]$,
see Section \ref{sec:dgLa}.
Then \eqref{eq:MC} is called a Maurer-Cartan equation.
This perspective is useful to keep track of identities, 
for gauge fixing, 
and to implement constraint propagation.
We note that, on the open Minkowski diamond $\diamond$,
this differential graded Lie algebra coincides with the
construction in \cite{RTgLa1,RTgLa2},
see also \cite[Remark 2]{Thesis}.
\step
\textbf{Homogeneous basis.}
There is a natural $\R_+$-action on $\gx(\cyl)$
that acts on the base manifold
$\cyl$ by scaling 
$x\mapsto \lambda x$ on $\diamond$, equivalently $y\mapsto \lambda^{-1}y$ on $\diamondy$,  
for each $\lambda>0$ 
(where we use the coordinates in \eqref{eq:xyintro})
and that commutes with the operators $\dg$ and $[\cdot,\cdot]$,
see Section \ref{sec:R+action}.
For the analysis near $\spaceinf$ we will use a basis
of the space of sections $\gx^1(\diamondy\setminus\spaceinf)$
that is homogeneous,
in the sense that it is 
given by sections that are homogeneous of degree zero
under this $\R_+$-action.
For example, a homogeneous basis of 
$\Omega^1(\diamondy\setminus\spaceinf)\otimesRR\Kil$ is given by 
\begin{align}\label{eq:hombasis|y|}
\frac{dy^{\mu}}{|y|} \otimes B^{\alpha\beta}
\qquad
\frac{1}{|y|}\frac{dy^{\mu}}{|y|} \otimes T_\nu
\qquad\qquad
\begin{aligned}
&\mu,\alpha,\beta,\nu=0\dots3\\
&\alpha<\beta
\end{aligned}
\end{align}
where $B^{\alpha\beta}, T_\nu\in\Kil$
are the boosts and translations \eqref{eq:TB} 
and $|y|^2=\sum_{i=0}^3|y^i|^2$,
similarly for $\I^2(\diamondy\setminus\spaceinf)$.
Another homogeneous basis, which we will use in the initial value problem below, 
is given by replacing $|y|$ in \eqref{eq:hombasis|y|} by $2y^0+|\vec{y}|$
where 
$|\vec{y}|^2=\sum_{i=1}^3|y^i|^2$.
This basis is regular when $y^0\ge0$ and $|\vec{y}|>0$,
and it is defined in Section \ref{sec:basis_i0}.
\step
\textbf{Kerr near spacelike infinity.} 
Via the correspondence \eqref{eq:11withmetrics}, the family of Kerr spacetimes
may be viewed as a family of solutions $\kerr(m,\vec{a})$ of \eqref{eq:MC}.
More precisely \cite[Lemma 113, Remark 59, Theorem 20]{Thesis}:
For every $m\in[0,2^{-10}]$ and $\vec{a}\in\R^3$ with $|\vec{a}|\le\frac12$,
there exists a smooth section $\kerr(m,\vec{a})$ of $\gx^1$ on 
the portion of $\diamond$ where $|y|\le\frac{1}{100}$
(an open neighborhood of $\spaceinf$ in $\diamond$), 
that satisfies\footnote{%
Explicit formulas for $\kerr(m,\vec{a})$ are in \cite[Section 4.8]{Thesis}.}%
\textsuperscript{,}%
\footnote{%
By applying boosts and translations to $\kerr(m,\vec{a})$, 
one obtains a ten-parameter family of Kerr elements.
The properties in the three items also
hold for this 10-parameter family.}:
\begin{itemize}
\item 
$\kerr(m,\vec{a})$ solves \eqref{eq:MC},
and the associated metric \eqref{eq:metricdefIntro} 
on $\diamond\cap\{|y|\le\frac{1}{100}\}$
is the standard Kerr metric with mass $m$ and angular momentum vector $\vec{a}$.
\item 
$\kerr(m,\vec{a})$ extends smoothly to future and 
past null infinity (not to $\spaceinf$).
\item 
For every integer $\ell\ge0$ one has the pointwise estimate
\begin{equation}\label{eq:KerrDecay}
|(|y|\p_y)^{\le\ell}\kerr(m,\vec{a})|
\;\lesssim_{\ell}\;
m |y|(1+|\log|y||)
\end{equation}
on $\diamond\cap\{|y|\le\frac{1}{100}\}$ (the notation $\lesssim_{\ell}$
is explained in Remark \ref{rem:lesssim}).
This estimate is understood as follows:
It holds for each component of $\kerr(m,\vec{a})$
relative to a homogeneous basis;
the notation $(|y|\p_y)^{\le\ell}$ means that
the components are differentiated at most $\ell$ times
with respect to the homogeneous of degree zero
vector fields $|y|\p_{y^0},\dots,|y|\p_{y^3}$.
\end{itemize}
Obtaining Kerr elements with these properties requires choosing
an appropriate gauge, that is, appropriate coordinates and orthonormal frame.
In \cite[Section 4.8]{Thesis} 
they are constructed by using an interpolation of 
Kerr in Kerr-Schild coordinates and 
Kerr in time reflected Kerr-Schild coordinates.

Theorem \ref{thm:mainpointwise} and \ref{thm:main} below
do not make explicit reference to a Kerr spacetime.
Instead, they only require a smooth section $\kerr$ of $\gx^1$ on 
a neighborhood of $\spaceinf$ in $\diamond$, 
that must satisfy certain assumptions.
One may take $\kerr$ to be equal to $\kerr(m,\vec{a})$ with small $m$,
which will satisfy the assumptions.

\step
\textbf{Initial data.}
Initial data for the Einstein equations \eqref{eq:MC} will be given 
on the $\tau=0$ time slice on $\diamond$,
that is, on
$$
\diamonddata \;=\; \diamond\cap (\{0\}\times S^3)
$$
This is equivalently the $x^0=0$ time slice, 
in particular $\diamonddata\simeq\R^3$ via the coordinates $x^1,x^2,x^3$.
Denote by $\gxdata^1$ the (trivial) vector bundle
on $\{0\}\times S^3$ whose fiber at $p\in \{0\}\times S^3$
is given by the fiber of $\gx^1$ at $p$, i.e.~$\gxdata^1$
is the pullback bundle of $\gx^1$ 
under the inclusion map $\{0\}\times S^3\hookrightarrow\cyl$.
Denote by $\gxdata^1(\diamonddata)$ the space of smooth 
sections over $\diamonddata$.
Initial data for \eqref{eq:MC} is given by a section
$\udata\in\gxdata^1(\diamonddata)$ that solves the constraint equations,
which are the necessary and sufficient conditions 
on $\udata$ for the local existence of a solution to \eqref{eq:MC}
that restricts to $\udata$ along $\diamonddata$.
The constraints are a nonlinear
first order PDE along $\diamonddata$, that we denote by 
$$
\Pconstraints(\udata) \;=\; 0
$$
The operator $\Pconstraints$ is in Definition \ref{def:Pconstraints}.
\step
\textbf{Main theorem (simplified version).}
We construct solutions in the future of the initial hypersurface
$\diamonddata$, that is, on the subset of $\diamond$ where $\tau\ge0$:
\[ 
\diamond_+ \;=\; \diamond\cap\{\tau\ge0\} 
\]
c.f.~Appendix \ref{ap:ConstructionOnD}.
Define 
\begin{equation*}
\Dspcl \;=\; (\diamondy\cap \overline{\diamond}_+) \setminus\spaceinf
\end{equation*}
This contains a portion of $\fnullinf$, but not $\spaceinf$.
On $\Dspcl$ define the smooth function 
\begin{equation}\label{eq:sdefintro}
\s = 2y^0 + |\vec{y}|
\end{equation}
The factor $2$ is chosen so that its level sets are 
safely spacelike for the Minkowski metric. 
\begin{figure}
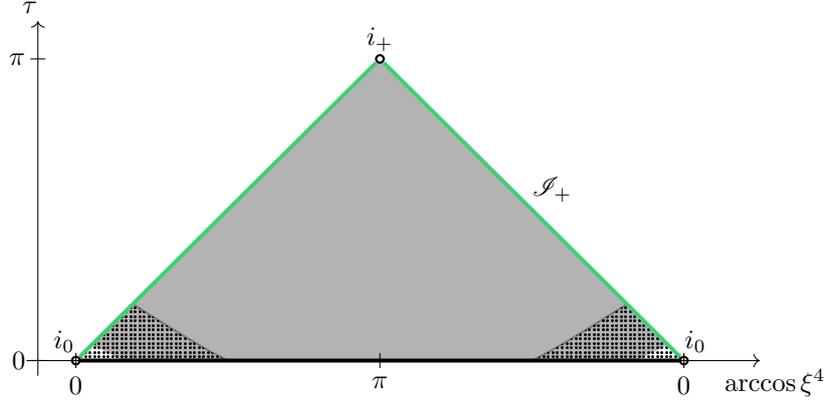

\centering
\coverfigureintro
\captionsetup{width=115mm}
\caption{The figure shows a cross-section of $\diamond_+$,
using $\tau$ and $\arccos\xi^4$ as coordinates.
Recall that $\diamond_+$ is the intersection 
of the open Minkowski diamond $\diamond$ with $\tau\ge0$.
It may be covered by $\Dspop_{\le1}$, which is the dotted region, 
and by $\diamond_+\setminus\Dspop_{<1/6}$, which is the gray shaded region.
The intersection of $\diamond_+$ with the $\tau=0$ time slice
is the initial hypersurface $\diamonddata$, which 
has a corresponding covering by 
$\Dspdata_{\le1}$ and by
$\diamonddata \setminus \Dspdata_{<1/6}$.}
\label{fig:DspCl}
\end{figure}
For $\sfix>0$ define (see Figure \ref{fig:DspCl})
\begin{align}\label{eq:Dspopcl}
\begin{aligned}
\Dspcl_{\le\sfix} &= \{ p\in\Dspcl \mid \s(p)\le\sfix \}\\ 
\Dspop_{\le\sfix} &= 
		\Dspcl_{\le\sfix} \cap\diamond_+\\
\Dspdata_{\le\sfix} 
				 &= \Dspcl_{\le\sfix} \cap \diamonddata
\end{aligned}
\end{align}
Analogously for $<$.
The sets $\Dspcl_{\le\sfix}$ and $\Dspop_{\le\sfix}$ have
nonempty respectively empty intersection with $\fnullinf$.
The set $\Dspdata_{\le\sfix}$ is the portion of 
$\diamonddata$ where $|\vec{y}|\le\sfix$.

These definitions yield a cover of
$\diamond_+$ respectively $\diamonddata$
consisting of a neighborhood of spacelike infinity,
and a set away from spacelike infinity: 
\begin{subequations}
\begin{align}
\diamond_+ 
	&= \Dspop_{\le\sfix} \cup (\diamond_+ \setminus \Dspop_{<\frac{\sfix}{6}})
		\label{eq:coverD+Thm1}\\
\diamonddata 
	&= \Dspdata_{\le\sfix} \cup (\diamonddata \setminus \Dspdata_{<\frac{\sfix}{6}})
		\label{eq:coverDataThm1}
\end{align}
\end{subequations}
In Theorem \ref{thm:mainpointwise} below, all inequalities are understood componentwise.
For the inequalities near spacelike infinity 
(concretely \ref{item:KPts}, \ref{item:K-uPts}, \eqref{eq:uconclNptwsp}, Part 2)
we use the components relative to the 
homogeneous basis \new{in Section \ref{sec:basis_i0}}.
For the inequalities away from spacelike infinity
(concretely \ref{item:uxptw}, \eqref{eq:uconclNptwCN})
we use the components relative to the 
basis in Section \ref{sec:basis_bulk},
which is regular on $\overline\diamond$.
In the norm 
\smash{$\|{\cdot}\|_{C^{\NN-3}(\diamond_+\setminus\Dspop_{<\sfix/6})}$},
the components are differentiated 
at most $\NN-3$ times 
with respect to the vector fields \eqref{eq:defV}, 
which are a regular frame on $\overline\diamond$.

Denote by $\gx^1(\Dspcl_{\le\sfix})$ the space of 
smooth sections of $\gx^1$ over $\Dspcl_{\le\sfix}$,
and by $\gx^1(\diamond_+)$ the space of smooth sections over $\diamond_+$.

\begin{theorem}[Simplified version of Theorem \ref{thm:main}]\label{thm:mainpointwise}
For all $\NN\in\Z_{\ge7}$ and $\sfix\in(0,1]$ there exist
$C>0$ and $\eps_0\in(0,1]$ such that for all
$\eps\in(0,\eps_0]$ and all 
\begin{align}\label{eq:kerrandudata}
\kerr \in \gx^1(\Dspcl_{\le\sfix})
&&
\udata \in \gxdata^1(\diamonddata)
\end{align}
the following holds. 
Abbreviate $\kerrdata = \kerr|_{\tau=0}$.
If
\begin{enumerate}[label=(a\arabic{enumi}),leftmargin=10mm]
\item \label{item:kerrmc}
$\dg\kerr + \frac12[\kerr,\kerr]=0$
\item \label{item:ucons}
$\udata$ solves the constraints $\Pconstraints(\udata)=0$, 
see Definition \ref{def:Pconstraints}
\item \label{item:KPts}
$|(|y|\p_{y})^{\le\NN+3}\kerr|\le\eps |y|(1+|\log|y||)$ on $\Dspop_{\le\sfix}$
\item \label{item:uxptw}
$|\p_{\vec x}^{\le \NN+1}\udata| \le \eps$ on $\diamonddata\setminus\Dspdata_{<\frac{\sfix}{\new{6}}}$,
where $\vec{x}=(x^1,x^2,x^3)$
\item \label{item:K-uPts}
$|(|\vec{y}|\p_{\vec{y}})^{\le\NN+3}(\udata-\kerrdata)|
\le \eps |\vec{y}|^{\NN+5}$ on $\Dspdata_{\le\sfix}$,
where $\vec{y}=(y^1,y^2,y^3)$
\setcounter{counterqa}{\value{enumi}}
\end{enumerate}
then there exists $u\in\gx^1(\diamond_+)$ that satisfies
\begin{equation}\label{eq:MCT1}
\dg u +\tfrac12[u,u] = 0
\;,\qquad
u|_{\tau=0} = \udata
\end{equation}
and:
\begin{itemize}
\item 
\textbf{Part 1 (decay and regularity).}
$u$ extends in $C^{\NN-3}$ to $\overline{\diamond}_+\setminus\spaceinf$
and
\begin{subequations}\label{eq:uconclNptw}
\begin{align}
|(|y|\p_y)^{\le\NN}(u-\kerr)| 
	&\le C\eps |y|^{\NN+4}
	\qquad\quad\text{on $\new{\Dspop}_{\le\frac\sfix2}$}\label{eq:uconclNptwsp}
	\\
\|u\|_{C^{\NN-3}(\diamond_+\setminus\Dspop_{\new{<}\frac\sfix6})} 
	&\le C \eps \label{eq:uconclNptwCN}
\end{align}
\end{subequations}
\item 
\textbf{Part 2 (higher decay and regularity).}
For all $k\in\Z_{\ge\NN}$, if
\begin{subequations}\label{eq:SimpHigher}
\begin{align}
\|\tfrac{(|y|\p_{y})^{\le k+3}\kerr}{|y|(1+|\log|y||)}\|_{L^\infty(\Dspop_{\le\sfix})} 
&\;<\; \infty\label{eq:KPtsHigher}\\
\|\tfrac{(|\vec{y}|\p_{\vec{y}})^{\le k+3}(\udata-\kerrdata)}{|\vec{y}|^{k+5}}\|_{L^\infty(\Dspdata_{\le\sfix})}
&\;<\;\infty\label{eq:K-uPtsHigher}
\end{align}
\end{subequations}
then $u$ extends in $C^{k-3}$ to $\overline{\diamond}_+\setminus\spaceinf$
and 
$\|\tfrac{(|y|\p_y)^{\le k}(u-\kerr)}{|y|^{k+4}}\|_{L^\infty(\Dspop_{\le\frac\sfix2})} <\infty$.
\item 
\textbf{Part 3 (metric).}
Decompose $u=\uo\oplus\uI$ using \eqref{eq:gx1Dintro}.
The frame $\one+\frame{\uo}$ is invertible at every point on $\diamond_+$.
The smooth metric $g$ on $\diamond_+$ defined by
\begin{equation}\label{eq:gthm}
g^{-1} 
= 
(\one + \frame{\uo})^{\otimes2} \eta^{-1}
\end{equation}
is Ricci-flat, future null geodesically complete, 
and the future null infinity locus of $g$ equals $\fnullinf$.
Moreover, $\nullgen^2g$ extends to an everywhere nondegenerate Lorentzian $C^{\NN-3}$-metric 
(respectively $C^{k-3}$ under the assumptions of Part 2)
on $\overline\diamond_+\setminus\spaceinf$.
More generally, the assumptions and conclusions of Proposition \ref{prop:metricregularity}
below hold with parameters \eqref{eq:ksu0} given by $\NN-3$, $\sfix$, $\uo$.
\end{itemize}
\end{theorem}
We prove Theorem \ref{thm:mainpointwise} in Section \ref{sec:ProofTheoremMain}
as a corollary of Theorem \ref{thm:main}.

Note that \eqref{eq:kerrandudata} requires 
that $\kerr$ is smooth including along future null infinity,
not at $\spaceinf$.
One may choose $\kerr$ to be equal to a Kerr element $\kerr(m,\vec{a})$,
in fact for every choice of $\NN$, $\sfix\le\frac{1}{100}$ 
the assumptions \eqref{eq:kerrandudata}, \ref{item:kerrmc}, \ref{item:KPts}
are satisfied for $\kerr=\kerr(m,\vec{a})$ provided that $m$ is sufficiently small,
using \eqref{eq:KerrDecay}.
Also note that if $\kerr=\kerr(m,\vec{a})$ then \eqref{eq:KPtsHigher} 
is satisfied for all $k$, using \eqref{eq:KerrDecay}.

The assumptions \ref{item:KPts}, \ref{item:uxptw}, \ref{item:K-uPts}
require in particular that the initial data $\udata$ is 
small, as dictated by $\eps$, on $\diamonddata$.
Further \ref{item:K-uPts} requires that
$\udata$ decays to $\kerrdata$ inverse polynomially 
in $x$-coordinates (one has $|\vec{y}| = |\vec{x}|^{-1}$ on $\diamonddata$), at a sufficiently fast rate.
Part 1 states that one then obtains regularity of
the solution $u$, and hence of the conformally rescaled metric
$\nullgen^2g$, at null and timelike infinity.
Part 2 is the statement that faster decay of the initial data 
$\udata$ to $\kerrdata$ implies higher regularity of the solution $u$,
and hence of $\nullgen^2g$, at null and timelike infinity.

In particular, we obtain:
\begin{corollary}
In Theorem \ref{thm:mainpointwise}, 
if 
\ref{item:kerrmc}-%
\ref{item:K-uPts} hold, and if:
\begin{itemize}
\item The components of $\udata-\kerrdata$
decay rapidly towards $\spaceinf$.
\item 
For all $k\in\Z_{\ge0}$ one has
$\|\frac{(|y|\p_{y})^{\le k}\kerr}{|y|(1+|\log|y||)}\|_{L^\infty(\Dspop_{\le\sfix})}<\infty$.\footnote{%
This is automatic when $\kerr=\kerr(m,\vec{a})$.}
\end{itemize}
then $u$ extends smoothly to $\overline{\diamond}_+\setminus\spaceinf$, 
and $\nullgen^2 g$ extends as a smooth metric to
$\overline{\diamond}_+\setminus\spaceinf$, 
and the components of $u-\kerr$ vanish to infinite order at $\spaceinf$.
\end{corollary}
In the following proposition we 
state how properties of the endomorphism $\frame{\uo}$
in \eqref{eq:Fuodef} imply 
properties of the associated metric $g$ in \eqref{eq:metricdefIntro}.
The assumptions and conclusions 
hold in particular for the solution in 
Theorem \ref{thm:mainpointwise} (see Part 3).
\begin{prop}\label{prop:metricregularity}
Let 
\begin{equation}\label{eq:ksu0}
k\in\Z_{\ge2} 
\qquad
s\in (0,1]
\qquad
\uo\in\Omega^1(\diamond_+)\otimesRR\Kil
\end{equation}
Assume that $\uo$ extends in $C^k$ to $\overline\diamond_+\setminus\spaceinf$,
and assume that:
\begin{enumerate}[label=(b\arabic{enumi}),leftmargin=10mm]
\item \label{item:F1/16assp}
The $\ell^2$-matrix norm of the endomorphism $\frame{\uo}$
satisfies:
At every point on $\Dspop_{\le s}$ it is bounded by 
$\frac{1}{16}$, using the basis $\s\p_{y^0},\dots,\s\p_{y^3}$;
at every point on $\diamond_{+}$ it is bounded by 
$\frac{1}{16}$, using the basis \eqref{eq:defV} which is regular on $\overline\diamond_+$.
\item \label{item:Fdh1/16assp}
The $\ell^2$-vector norm of the one-form ${\dualframe{\uo}(d\nullgen)}/{\nullgen}$
satisfies: 
At every point on $\Dspop_{\le s}$ it is bounded by 
$\frac{1}{16}$, using the basis ${dy^0}/{\s},\dots,{dy^3}/{\s}$;
at every point on $\diamond_{+}\setminus\Dspop_{<\frac{s}{6}}$ 
it is bounded by 
$\frac{1}{16}$, using the basis dual to \eqref{eq:defV}.
\end{enumerate} 
Then 
$\one+\frame{\uo}$ is invertible at every point on $\diamond_+$,
and the metric $g$ on $\diamond_+$ defined by \eqref{eq:metricdefIntro} 
has the following properties:
\begin{enumerate}[label=(c\arabic{enumi}),leftmargin=10mm]
\item \label{item:metric_regularextension}
	$\nullgen^2g$ extends to a Lorentzian $C^{k}$-metric on $\overline\diamond_+\setminus\spaceinf$.
	In particular, the extension is everywhere nondegenerate,
	including along $\fnullinf\cup\timeinf$.
\item \label{item:metric_eikonal}
	$(\nullgen^2g)^{-1}(d\nullgen,d\nullgen) = \nullgen f$
	for a function $f\in C^{k}(\overline\diamond_+\setminus\spaceinf)$,\\
	in particular $(\nullgen^2g)^{-1}(d\nullgen,d\nullgen)=0$ on $\fnullinf$.
\item\label{item:metric_nullcompleteness}
	The metric $g$ on $\diamond_+$ is
	future null geodesically complete, 
	and the null infinity locus is $\fnullinf$.
	More precisely, for every $p_0\in\diamonddata$,
	and every $v_*\in T_{p_0}\diamond_+$ 
	that is null with respect to $g$ and 
	normalized such that $d\tau(v_*)=1$,
	there exist
	$\taumax>0$ and 
	$\gamma\in C^{\new{\infty}}([0,\taumax),\diamond_+)$
	of the form
	\begin{equation}\label{eq:gammapartau_intro}
	\gamma:\; \tau\mapsto (\tau,\xi(\tau))
	\end{equation}
	with $\xi(\tau)\in S^3$, 
	that satisfies the null geodesic initial value problem
	\begin{align}\label{eq:nullgeodeq_intro}
	\nabla^g_{\dot\gamma}\dot\gamma  &\ \propto\  \dot\gamma
	&
	\gamma(0) &= p_0
	&
	\dot\gamma(0) &= v_*
	\end{align}
	that further extends in $C^k$ to $[0,\taumax]$, and that satisfies:
	\begin{enumerate}[label=(\roman*)]
	\item \label{item:null_intro}
	$\dot\gamma(\tau)$ is null for all $\tau\in[0,\taumax]$.
	\item \label{item:gammaD_intro}
	$\gamma(\taumax)\in\fnullinf$ and $\dot\gamma(\taumax)$ is transversal to null infinity: $d\nullgen(\dot\gamma(\taumax))\neq 0$.
	\item \label{item:gammaAff_intro}
	The affine parameter (relative to $g$) goes to infinity along $\gamma$ as $\tau\uparrow\tau_1$.
	\end{enumerate}
	Moreover, every maximal null geodesic of $g$ in $\diamond_+$
	is given by such a $\gamma$,
	and every point in $\fnullinf$ is reached by such a $\gamma$.
\item \label{item:metric_nullinfinity}
	Let $\LL$ be the field of lines on $\fnullinf$ 
	spanned at each $p\in \fnullinf$ 
	by the nonzero vector $(\nullgen^2g)^{-1}(d\nullgen,\,\cdot\,)|_{p}$,
	which is tangential to $\fnullinf$ and null with respect to $\nullgen^2g$.
	For every $v_*\in T_{\timeinf}\cyl$
	that is null with respect to $\nullgen^2g$
	and normalized such that $d\tau(v_*)=1$,
	there exists $\gamma\in C^{\new{k}}((0,\pi],\fnullinf\cup\timeinf)$
	of the form\footnote{%
	\new{By $C^k$ we mean that $\gamma$ is 
	$C^{k}$ as a map $(0,\pi]\to \cyl$.}
	}
	\begin{equation}\label{eq:gammatau_intro}
	\gamma:\; \tau\mapsto(\tau,\xi(\tau))
	\end{equation}
	with $\xi(\tau)\in S^3$,
	that is an integral curve of $\LL$ when $\tau\in(0,\pi)$,
	and satisfies
	\begin{equation}\label{eq:gammaini_intro}
	\gamma(\pi) = \timeinf
	\qquad\qquad
	\dot\gamma(\pi)=v_*
	\end{equation}
	The union of these integral curves is $\fnullinf$.
	Every such $\gamma$ is a null geodesic for $\nullgen^2g$
	(not affinely parametrized in general).
\end{enumerate}
\end{prop}
The proof of Proposition \ref{prop:metricregularity}
is at the end of Section \ref{sec:recoveryofmetric},
where for \ref{item:metric_nullcompleteness}
and \ref{item:metric_nullinfinity} we will only give a detailed sketch
(since it is somewhat off topic for this paper,
we plan to write out the full proofs in an upcoming paper).

Note that if $g$ is a metric on $\diamond$
that satisfies \ref{item:metric_nullcompleteness} on $\diamond_+$,
and such that the pullback of $g$ along the time reflection
$(\tau,\xi)\mapsto(-\tau,\xi)$ also satisfies \ref{item:metric_nullcompleteness} on $\diamond_+$, then the metric $g$ on $\diamond$ is null geodesically complete,
and the locus of future and past null infinity is $\fnullinf$
respectively $\pnullinf$ (c.f.~Appendix \ref{ap:ConstructionOnD}).
\step
\textbf{Main theorem.}
We now foliate the two sets of the cover \eqref{eq:coverD+Thm1}
by level sets of $\s$ respectively $\tau$, 
see Figure \ref{fig:s_tau_foliations}.
For $s\in (0,\sfix]$
denote by $\Dspop_{s}$ the portion of  $\Dspop_{\le\sfix}$
where $\s=s$. 
For $\tau\in [0,\pi)$
denote by $\diamond_{\tau,\sfix}$ the intersection
of $\diamond_+\setminus \Dspop_{<\sfix/6}$ 
and $\{\tau\}\times S^3$.
We use the following norms for sections $u$ of $\gx^1$
(some of them are actually seminorms, but we refer to them
as norms for simplicity):
\begin{itemize}
\item 
\textit{Homogeneous norms near $\spaceinf$ (Definition \ref{def:norms_i0}).}
For $k\in\Z_{\ge0}$ define
\begin{equation}\label{eq:bnorms_intro}
\|u\|_{\nosCb^k(\Dspop_{\le\sfix})}
\qquad
\|u\|_{\nosHb^k(\Dspop_{\le\sfix})}
\qquad 
\|u\|_{\sCb^k(\Dspop_{s})}
\qquad
\|u\|_{\sHb^k(\Dspop_{s})}
\end{equation}
as follows:
The components of $u$ with respect to the homogeneous basis
\new{in Section \ref{sec:basis_i0}} are differentiated at most $k$ times with respect to
\[ 
\s \p_{y^0},\ \s \p_{y^1},\ \s \p_{y^2},\ \s \p_{y^3}
\]
For $\nosCb^k$ and $\sCb^k$ we then take the supremum over 
$\Dspop_{\le\sfix}$ respectively $\Dspop_{s}$. 
For $\nosHb^k$ and $\sHb^k$ we take the $L^2$-norm 
over $\Dspop_{\le\sfix}$ respectively $\Dspop_{s}$
with respect to measures that are homogeneous of degree zero.
\item 
\textit{Norms away from $\spaceinf$ (Definition \ref{def:bulknorms}).}
For $k\in\Z_{\ge0}$ define 
\begin{equation}\label{eq:bulknorms_intro} 
\|u\|_{\sC^{k}(\diamond_{\tau,\sfix})}
\qquad\qquad
\|u\|_{\sH^{k}(\diamond_{\tau,\sfix})}
\end{equation}
using the basis in Section \ref{sec:basis_bulk}
which is regular on $\overline\diamond$;
derivatives are with respect to the vector fields
\eqref{eq:defV} which are a regular frame on $\overline\diamond$;
and for $\sH^k$ the $L^2$-norm with respect to the 
standard measure on $S^3$ is used.
\end{itemize}
Beware that the slashed norms over $\Dspop_{s}$ in \eqref{eq:bnorms_intro},
and over $\diamond_{\tau,\sfix}$ in 
\eqref{eq:bulknorms_intro}, are not determined by the
restriction of $u$ to $\Dspop_{s}$ respectively to $\diamond_{\tau,\sfix}$.

\begin{figure}
\begin{subfigure}[b]{.5\linewidth}
\centering
\IntroLevelsetsS
\captionsetup{width=53mm}
\caption{Depicted is $\Dspop_{\le1}$ in $y$-coordinates,
spherical directions suppressed.
Spacelike infinity $\spaceinf$ is the origin $y=0$.
The four tilted lines are the sets $\Dspop_{0.125},\Dspop_{0.25},\Dspop_{0.5},\Dspop_{1}$,
given by the level sets 
$\s=0.125,0.25,0.5,1$.}
\label{fig:s}
\end{subfigure}
\begin{subfigure}[b]{.47\linewidth}
\centering
\IntroLevelsetsTau
\captionsetup{width=53mm}
\caption{Depicted is $\diamond_+\setminus\Dspop_{<1/6}$.
The horizontal lines are the sets $\diamond_{\tau,1}$
with $\tau=0,0.125,0.25,0.375,0.5,1,\dots,3.5$,
given by the intersection of 
$\diamond_+\setminus\Dspop_{<1/6}$ and $\{\tau\}\times S^3$.
\newline $ $}
\end{subfigure}
\captionsetup{width=113mm}
\caption{
The two figures depict $\Dspop_{\le1}$ and $\diamond_+\setminus\Dspop_{<1/6}$,
and the respective foliations by $\s$-level sets respectively $\tau$-level sets.
Their union is $\diamond_+$, see Figure \ref{fig:DspCl}.}
\label{fig:s_tau_foliations}
\end{figure}

The cover \eqref{eq:coverDataThm1} of the initial hypersurface 
$\diamonddata$
is equivalently given by 
$\diamonddata = \Dspdata_{\le\sfix} \cup \diamond_{0,\sfix}$.
We use the following norms for sections $\udata$ of $\gxdata^1$:
\begin{itemize}
\item 
\textit{Homogeneous norms for data near $\spaceinf$ (Definition \ref{def:normsdatasp}).}
For $k\in\Z_{\ge0}$ the norm 
$\|\udata\|_{\Cb^k(\Dspdata_{\le \sfix})}$
is defined using the homogeneous basis \new{in Section \ref{sec:basis_i0}}
and derivatives are with respect to the homogeneous vector fields 
\begin{equation}\label{eq:spVF}
|\vec{y}| \p_{y^1},\ |\vec{y}| \p_{y^2},\ |\vec{y}| \p_{y^3}
\end{equation}
For $k\in\Z_{\ge1}$ and $a\ge0$ define
\begin{align*}
\|\udata\|_{\HdataNEW^{a,k}(\Dspdata_{\le\sfix})}
&=\textstyle
\int_{0}^{\sfix} \big(\frac{\sfix}{s}\big)^{a+(k-1)} 
\left(1+|\log(\tfrac{\sfix}{s})|\right)^{k-1}
\|\udata\|_{\Hb^{k}(\Dspdata_{\frac{s}{\new{3}},s})} \frac{ds}{s}
\end{align*}
where $\smash{\Dspdata_{\frac{s}{\new{3}},s}}$ is the set given by $\smash{\frac{s}{\new{3}}}\le|\vec{y}|\le s$,
and $\|\udata\|_{\Hb^{k}(\Dspdata_{\frac{s}{\new{3}},s})}$ 
is defined using the homogeneous basis 
in Section \ref{sec:basis_i0},
the vector fields \eqref{eq:spVF},
and the $L^2$-norm with respect to a homogeneous of degree zero measure.
\item 
\textit{Norms for data away from $\spaceinf$ (Definition \ref{def:normdataawayfromi0}).}
For $k\in\Z_{\ge0}$ the norm 
$\|\udata\|_{\H^{k}(\diamond_{0,\sfix})}$
is defined
using the basis in Section \ref{sec:basis_bulk},
the frame of vector fields $\V_1,\V_2,\V_3$ on $S^3$ defined in \eqref{eq:defV},
and the standard measure on $S^3$.
\end{itemize}
\begin{theorem}\label{thm:main}
For all 
\begin{equation}\label{eq:mainconst}
\NN\in\Z_{\ge\new{7}}
\qquad
\epspower\in(0,1]
\qquad
\sfix\in(0,1]
\qquad
\Cinmain>0
\end{equation}
there exist $\Clargemain>0$ and $\Csmallmain\in(0,1]$ 
such that for all 
\begin{align}\label{eq:KerrMain}
\kerr\in\gx^1(\Dspcl_{\le\sfix}) && \udata\in\gxdata^1(\diamonddata)
\end{align}
the following holds.
Abbreviate $\Dspop = \Dspop_{\le\sfix}$ and 
$\Dspdata=\Dspdata_{\le\sfix}$
and $\kerrdata = \kerr|_{\tau=0}$.
If
\begin{multicols}{2}
\begin{enumerate}[label=(d\arabic{enumi}),leftmargin=10mm]
\item \label{item:KMCmain}
	$\dg\kerr+\frac12[\kerr,\kerr]=0$
\item \label{item:KBoundMain}
	$\|\kerr\|_{\nosCb^{\NN+3}(\Dspop)}\le\Cinmain$
\item \label{item:KL1Main}
	$\int_{0}^{\sfix}\|\kerr\|_{\sCb^{1}(\Dspop_{s})} \frac{ds}{s}\le\Cinmain$
\item \label{item:KsmallMain}
	$\|\kerr\|_{\nosCb^{\NN+1}(\Dspop)}\le\Csmallmain$
\item \label{item:uconstrMain}
	$\Pconstraints(\udata)=0$, see Definition \ref{def:Pconstraints}
\item \label{item:uBoundMain}
	$\|\udata\|_{\Cb^{\NN+3}(\Dspdata)}\le\Cinmain$ 
\item \label{item:usmallMain}
	$\|\udata\|_{\H^{\NN+1}(\diamond_{0,\sfix})} \le\Csmallmain$ 
\item \label{item:diffsmallmain}
$\|\udata-\kerrdata\|_{\HdataNEW^{\frac52+\epspower,\NN+3}(\Dspdata)}
	\le\Csmallmain$ 
\setcounter{counterqa}{\value{enumi}}
\end{enumerate}
\end{multicols}
\noindent 
then there exists $u\in \gx^1(\diamond_+)$ that satisfies
\begin{align}
\dg u +\tfrac12[u,u] 
	=0 
\;,\qquad
u|_{\tau=0} 
	= \udata \label{eq:MCandudatamain} 
\end{align}
and:
\begin{itemize}
\item 
\textbf{Part 1 (decay and regularity).}
The solution $u$ extends in $C^{\NN-3}$ to $\overline{\diamond}_+\setminus\spaceinf$,
and the following estimates hold: For all $s\in(0,\tfrac\sfix2]$,
\begin{subequations}\label{eq:umain}
\begin{align}\label{eq:uspBoundMain}
\begin{aligned}
&\|u-\kerr\|_{\sHb^{\NN+2}(\Dspop_{s})}+\|u-\kerr\|_{\sCb^{\NN}(\Dspop_{s})} \\
	&\hspace{2.5cm}
	\le \Clargemain\,
	\big(\tfrac{s}{\sfix}\big)^{\frac92+\epspower+\NN}
	\|\udata-\kerrdata\|_{\HdataNEW^{\frac52+\epspower,\NN+3}(\Dspdata)}
\end{aligned}
\end{align}
and
\begin{align}\label{eq:uCptBoundMain}
\hspace{-2mm}
\begin{aligned}
&\sup_{\tau\in[0,\pi)}\|u\|_{\sH^{\NN}(\diamond_{\tau,\sfix})}
+
\sup_{\tau\in[0,\pi)}\|u\|_{\sC^{\NN-3}(\diamond_{\tau,\sfix})}
	\\
	&
	\quad\le 
	\Clargemain\, 
	\big(\|\udata-\kerrdata\|_{\HdataNEW^{\frac52+\epspower,\NN+3}(\Dspdata)}
	+
	\|\kerr\|_{\nosCb^{\NN+1}(\Dspop)}
	+
	\|\udata\|_{\H^{\NN+1}(\diamond_{0,\sfix})}\big)
\end{aligned}
\end{align}
\end{subequations}
\item 
\textbf{Part 2 (higher decay and regularity).}
For all $k\in\Z_{\ge\NN}$ and $\CinmainHigher\new{>}0$,
if 
\begin{multicols}{2}
\begin{enumerate}[label=(d\arabic{enumi}),leftmargin=10mm]
\setcounter{enumi}{\value{counterqa}}
\item \label{item:KHigherBoundMain}
	$\|\kerr\|_{\nosCb^{k+3}(\Dspop)}\le\CinmainHigher$
\item 
\label{item:uCkHigherBoundMain}
	$\|\udata\|_{\Cb^{k+3}(\Dspdata)}\le\CinmainHigher$
\item 
\label{item:uHkHigherBoundMain}
$\|\udata\|_{\H^{k+1}(\diamond_{0,\sfix})}\le\CinmainHigher$
\item 
\label{item:diffHighersmallmain}
$\smash{\|\udata-\kerrdata\|_{\HdataNEW^{\frac52+\epspower,k+3}(\Dspdata)}
	\le\CinmainHigher}$
\end{enumerate}
\end{multicols}
then $u$ extends in $C^{k-3}$ to $\overline{\diamond}_+\setminus\spaceinf$,
and there exists a constant $\ClargemainHigher>0$ that depends only on 
$k,\epspower,\sfix,\Cinmain,\CinmainHigher$, such that 
the estimates \eqref{eq:umain}
hold verbatim with $\NN$ and $\Clargemain$ 
replaced by $k$ and $\ClargemainHigher$, respectively.
\item 
\textbf{Part 3 (metric).}
Part 3 of Theorem \ref{thm:mainpointwise} holds verbatim.
\end{itemize}
\end{theorem}

The proof of Theorem \ref{thm:main} is in Section \ref{sec:ProofTheoremMain}.

One may choose $\kerr$ to be equal to a Kerr element $\kerr(m,\vec{a})$,
in fact for every choice of 
$\NN$, $\epspower$, $\sfix\le\frac{1}{100}$, $\Cinmain$
the assumptions \eqref{eq:KerrMain}, \ref{item:KMCmain}, \ref{item:KBoundMain}, \ref{item:KL1Main}, \ref{item:KsmallMain}
are satisfied for $\kerr=\kerr(m,\vec{a})$ provided that $m$ is sufficiently small,
using \eqref{eq:KerrDecay}.
Also, if $\kerr=\kerr(m,\vec{a})$ then \ref{item:KHigherBoundMain}
is satisfied for all $k$ and sufficiently large $\CinmainHigher$, 
using \eqref{eq:KerrDecay}.

Theorem \ref{thm:main} is not sharp in terms of differentiability,
for example the loss of three derivatives with respect to Sobolev
norms in \eqref{eq:uCptBoundMain} 
is for technical convenience and can certainly be improved.

Regularity along null infinity has also been studied 
for certain scattering problems
on a Minkowski or Schwarzschild background,
see e.g.~\cite{KadarKehrberger,ACaseAgainst5,Kroon}.
We note that our result does not exclude that there
exist spacetimes that
do not admit, in a gauge invariant sense,
a regular conformal compactification.
\begin{remark}\label{rem:Plinearized}
Theorem \ref{thm:mainpointwise} and \ref{thm:main} are conditional
on the existence of solutions of the constraints 
$\Pconstraints(\udata)=0$ with specific asymptotics towards $\spaceinf$.
The zero initial data $\udata=0$ solves $\Pconstraints(0)=0$
and corresponds to Minkowski initial data
(recall that under the correspondence \eqref{eq:11withmetrics}
the zero solution is the Minkowski metric).
The linearization of $\Pconstraints$ at the zero solution
is, in a basis, a map $C^\infty(\R^{3},\R^{50})\to C^\infty(\R^{3},\R^{36})$
given by the matrix differential operator (c.f.~Remark \ref{rem:Constraints36})
\begin{align}\label{eq:ConstraintsLinearized}
\begin{pmatrix}
\tc{gray}{0}	& \curl^{\oplus10} & \ast  & \ast \\
\tc{gray}{0}	& \tc{gray}{0} & \DIV & \tc{gray}{0} \\
\tc{gray}{0}	& \tc{gray}{0} & \tc{gray}{0}& \DIV
\end{pmatrix}
\end{align}
The block $\curl^{\oplus10}$ has size $30\times30$,
and is a $10\times10$ block diagonal matrix where
each diagonal block is given by $\curl$, 
a $3\times3$ matrix differential operator.
The block $\DIV$ has size $3\times 5$, and, 
identifying $C^\infty(\R^3,\R^5)$ with the space of symmetric traceless matrices
whose entries are smooth functions, it is given by 
applying the divergence to each column.
%
The blocks $\ast$ are $C^\infty$-linear, not constant coefficient.
In \cite{HomotopyPaper} we constructed right inverses of 
the operators $\curl$ and $\DIV$,
up to necessary integrability conditions,
that have optimal asymptotic properties at infinity.
Via back-substitution one then obtains 
a right inverse of \eqref{eq:ConstraintsLinearized},
again up to integrability conditions.
We expect that using this right inverse one can construct
solutions of the constraints $\Pconstraints(\udata)=0$,
with the asymptotics required by Theorem \ref{thm:main}.
Like in \cite{MaoOhTao}, the construction will use renormalization 
of charges using the Kerr parameters, and a Banach fixed point argument.
\end{remark}

\textbf{Proof outline.}
The construction of $u$ in Theorem \ref{thm:main} has three parts:
\begin{itemize}
\item 
Construction on $\Dspop_{\le\sfix}$
(Section \ref{sec:Abstract} and \ref{sec:SpaceinfConstruction}).
We set 
\begin{equation}\label{eq:u=v+c}
u \;=\; v + c 
\qquad \text{where} \qquad v = \kerr + \Pext(\udata-\kerrdata)
\end{equation}
Here $\Pext$ is an extension operator (Definition \ref{def:extensionspinf}),
and the correction $c$ is the new unknown.
Using (bi)linearity of $\dg$ and $[\cdot,\cdot]$, 
the equation for $c$ reads:
\begin{equation}\label{eq:MCforc}
(\dg + [v,\cdot\,])c 
+ \tfrac12[c,c] + 
(\dg v + \tfrac12[v,v]) = 0\ , 
\qquad c|_{\tau=0}=0
\end{equation}
We impose ten pointwise gauge fixing conditions on $c$, 
which means that we require that $c$ lie
in a $C^\infty$-submodule of $\gx^1(\Dspop_{\le\sfix})$ with corank ten
(Section \ref{sec:gauge_i0}).
One then considers a necessary (and sufficient up to constraints) 
square subsystem of \eqref{eq:MCforc}
that is quasilinear symmetric hyperbolic, including along future null infinity $\fnullinf\subset\cyl$.
The linear part of the symmetric hyperbolic system 
is given by the operator $\dg + [v,\cdot\,] $.
When written in a homogeneous basis, 
the Minkowski differential $\dg$
is a matrix differential operator that is
homogeneous of degree zero,
\new{explicitly given} in Section \ref{sec:reformSHS_i0}.
The operator $[v,\cdot\,]$ is lower order in terms of homogeneity,
by \ref{item:KL1Main} and \ref{item:diffsmallmain}.
The causal structure is qualitatively the same as that of Minkowski
spacetime.
Existence of $c$ is shown in 
Proposition \ref{prop:ApplySHS} (Section \ref{sec:SpaceinfConstruction}).
This is in turn proven as an application of Theorem \ref{thm:nonlinEE} 
(Section \ref{sec:Abstract})
where we study a more general class of symmetric hyperbolic systems.
\item
Construction on $\diamond_+\setminus \Dspop_{<\frac\sfix6}$ 
(Section \ref{sec:bulk}).
Here the solution on $\Dspop_{\le\sfix}$ is 
extended to a solution on $\diamond_+$.
Gauge fixing (similar to the first item) yields
a symmetric hyperbolic system that is regular including 
at null and timelike infinity (Section \ref{sec:Gauge_cyl}).
By finite speed of propagation, 
the problem studied here is causally separated from spacelike infinity.
Thus one must only solve a symmetric hyperbolic system with small data on a compact domain, which is more routine.
See Proposition \ref{prop:MainCpt}.
\item 
The constructions are combined in 
the proof of Theorem \ref{thm:main} in
Section \ref{sec:ProofTheoremMain}.
\end{itemize}

Theorem \ref{thm:main} is stated as an initial value problem,
and the solution $u$ is only constructed 
for positive time, $\tau\ge0$.
However,
if $\kerr$ is defined also for negative time
(which is the case for $\kerr=\kerr(m,\vec{a})$), 
then applying the construction of $u$ 
in Theorem \ref{thm:main} once to $\udata$ and $\kerr$,
and once to the time reflection of $\udata$ and $\kerr$,
one obtains a smooth solution on $\diamond$, 
with control over the regularity also at past null infinity.
The associated metric on $\diamond$
is asymptotically simple and satisfies peeling
\cite[Section 9, 14]{PenroseAsymptoticSimplicity}.
The metric is null geodesically complete, 
with future and past null infinity locus
given by $\fnullinf$ respectively $\pnullinf$.
See Appendix \ref{ap:ConstructionOnD}.
\step
\textbf{Acknowledgement.}
I thank Rafe Mazzeo and Michael Reiterer for discussions about this project.
During this research, the author was supported by the 
Swiss National Science Foundation, project number P500PT-214470.
\section[A dgLa for general relativity about Minkowski spacetime]{A dgLa for general relativity about Minkowski}
\label{sec:dgLa}

We recall the formulation of the Einstein equations introduced in \cite{Thesis}.

\subsection{Geometric conformal compactification}
\label{sec:geomconfcpt}
The Einstein cylinder is the oriented conformally flat manifold
\[ 
\Big(
\cyl = \R\times S^3 \;,\;\; \big[\gcyl 
= -d\tau^{\otimes2} + \gS \big]
\Big)
\]
where $\tau$ is the standard coordinate on $\R$ and
$\gS$ is the round metric on $S^3$.
We view $S^3$ as the unit sphere in $\R^4$
and denote the standard coordinates on $\R^4$
by $\xi=(\xi^1,\xi^2,\xi^3,\xi^4)$.
We fix the following global frame of vector fields on $\cyl$:
\begin{align}\label{eq:defV}
\begin{aligned}
\V_0 &= \p_{\tau}\\
\V_1 &= 
	(\xi^1\p_{\xi^4}-\xi^4\p_{\xi^1})-(\xi^2\p_{\xi^3}-\xi^3\p_{\xi^2})\\
\V_2 &= 
	(\xi^2\p_{\xi^4}-\xi^4\p_{\xi^2})-(\xi^3\p_{\xi^1}-\xi^1\p_{\xi^3})\\
\V_3 &= 
	(\xi^3\p_{\xi^4}-\xi^4\p_{\xi^3})-(\xi^1\p_{\xi^2}-\xi^2\p_{\xi^1})
\end{aligned}
\end{align}
which is orthonormal with respect to $\gcyl$.
The orientation on $\cyl$ is fixed so that this frame is positive.
Note that $\V_1,\V_2,\V_3$ are vector fields on $S^3$.
We denote by $\Vd^0,\dots,\Vd^3$ the frame of one-forms that is dual to \eqref{eq:defV}.

Define the smooth functions
\begin{equation}\label{eq:nullgendef}
\nullgen = \cos(\tau)-\xi^4 
\qquad\qquad
\ynullgen = \cos(\tau)+\xi^4 
\end{equation}
Then Minkowski spacetime is isometric to $(\diamond,\eta)$ defined by
\begin{equation}\label{eq:defdiamond}
\diamond 
\;=\;
\big\{ (\tau,\xi) \in \cyl \mid 
-\pi<\tau<\pi\;,\;
0  < \nullgen(\tau,\xi) \big\}
\qquad
\gmink = \nullgen^{-2} \gcyl|_{\diamond}
\end{equation}
Note that $\diamond$ is equivalently given by the set of all points 
$(\tau,\xi)\in\cyl$ for which $|\tau|$ is strictly smaller then the 
$S^3$-distance from $\xi$ to $(0,0,0,1)\in S^3$.
We refer to $\diamond$ as the Minkowski diamond. Its boundary
has five components:
\begin{itemize}
\item 
Future/past null infinity $\fpnullinf$, given by 
the set of all $(\tau,\xi)$ with $\nullgen(\tau,\xi)=0$ and
$\tau\in (0,\pi)$ respectively $\tau\in(-\pi,0)$.
Observe $d\nullgen\neq0$ pointwise on $\fpnullinf$.
\item 
Future/past timelike infinity 
$\fptimeinf= (\pm\pi,(0,0,0,-1))$. 
Observe $d\nullgen|_{\fptimeinf}=0$.
\item 
Spacelike infinity 
$\spaceinf = (0,(0,0,0,1))$.
Observe $d\nullgen|_{\spaceinf}=0$.
\end{itemize}

Define the smooth functions 
\begin{align}\label{eq:xxcoords}
x
	=(x^0,\dots,x^3) 
	&= \left( \tfrac{\sin\tau}{\nullgen},
	\tfrac{\xi^1}{\nullgen},
	\tfrac{\xi^2}{\nullgen},
	\tfrac{\xi^3}{\nullgen} \right) 
	\quad \text{on}\; \quad \cyl\setminus\{\nullgen=0\}
\end{align}
They satisfy
\begin{align}
\nullgen^{-2}\gcyl 
&= -(dx^0)^{\otimes2} + (dx^1)^{\otimes2}+(dx^2)^{\otimes2}+(dx^3)^{\otimes2}
\label{eq:minkmetric}
\end{align}
and restrict to smooth coordinates on every connected component of $\cyl\setminus\{\nullgen=0\}$.
In particular, they restrict to coordinates on $\diamond$
that establish the isometry of \eqref{eq:defdiamond} with Minkowski spacetime.

We introduce coordinates that are regular near spacelike infinity.
Define
\begin{equation}\label{eq:diamondy}
\diamondy
\;=\;
\{ (\tau,\xi) \in \cyl \mid 
-\pi<\tau<\pi\;,\;
0  < \ynullgen(\tau,\xi) \}
\qquad
\etay = \ynullgen^{-2}\gcyl|_{\diamondy}
\end{equation}
The set $\diamondy$
is equivalently given by the image of $\diamond$ under $(\tau,\xi)\mapsto(\tau,-\xi)$, and is an 
open neighborhood of $\spaceinf$ (see Figure \ref{fig:cylinderintro}).
Define the smooth functions
\begin{align}\label{eq:yycoords}
y
	=(y^0,\dots,y^3) 
	&=\left(\tfrac{\sin\tau}{\ynullgen},
	\tfrac{\xi^1}{\ynullgen},
	\tfrac{\xi^2}{\ynullgen},
	\tfrac{\xi^3}{\ynullgen} \right) 
	\quad \text{on}\; \quad\cyl\setminus\{\ynullgen=0\}
\end{align}
They satisfy $\ynullgen^{-2}\gcyl 
	=\eta_{\mu\nu}dy^\mu\otimes dy^\nu$
and $\nullgen = \ynullgen (\eta_{\mu\nu}y^\mu y^\nu)$,
and they restrict to smooth coordinates on every connected component
of $\cyl\setminus\{\ynullgen=0\}$.
In particular, 
they restrict to coordinates on $\diamondy$,
where $\spaceinf$ is the origin $y=0$,
and where $\diamondy\cap\diamond$ and $\diamondy \cap \p\diamond$
are given by 
$\eta_{\mu\nu}y^\mu y^\nu>0$ respectively $\eta_{\mu\nu}y^\mu y^\nu=0$.

On their common domain of definition,
the functions $x$ and $y$ are related by Kelvin inversion,
\begin{equation}\label{eq:Kelvin}
x=\tfrac{y}{\eta_{\mu\nu}y^\mu y^\nu}
\qquad
y = \tfrac{x}{\eta_{\mu\nu}x^\mu x^\nu}
\end{equation}
and, on 
$\diamond\cap\diamondy$, the representatives $\eta$ and $\etay$
of $[\gcyl]$ satisfy $\eta = (\eta_{\mu\nu}y^\mu y^\nu)^{-2} \etay$.
\begin{remark}[Orientation]\label{rem:orientation}
Recall that the orientation on $\cyl$ is fixed so that 
the frame \eqref{eq:defV} is positive. 
Relative to this orientation,
the frame $\p_{x^0},\dots,\p_{x^3}$ is positively oriented,
and the frame $\p_{y^0},\dots,\p_{y^3}$ is negatively oriented.
\end{remark}
The space of conformal Killing fields on the Einstein cylinder
is given by 
\[ 
\ConfKil
\;=\;
\{ X\in\Gamma(T\cyl) \mid \exists f\in C^\infty(\cyl) :\; \Lie_X\gcyl=f\gcyl \}
\]
This is a 15-dimensional real Lie algebra, isomorphic to $\so(2,4)$.
Define 
\[ 
\Kil 
\;=\;
\{X\in\ConfKil \mid \Lie_{X|_{\diamond}}\gmink = 0\} 
\]
which is the set of all conformal Killing fields on the Einstein cylinder
that restrict to ordinary Killing fields for the Minkowski metric on $\diamond$.
This is a real Lie algebra of dimension 10, isomorphic to the Lie algebra
of the Poincar\'e group. 

A basis of $\Kil$ is given by the boosts and translations
\begin{equation}\label{eq:Kilbas}
B^{01},B^{02},B^{03},B^{12},B^{23},B^{31},T^0,T^1,T^2,T^3
\end{equation}
which on the dense subset $\cyl\setminus\{\nullgen=0\}$ 
are given by
\begin{subequations}
\begin{align}
B^{\mu \nu} &=  
	(x^\mu \eta^{\nu\sigma}-x^\nu \eta^{\mu \sigma})\p_{x^\sigma}\;, 
&
	T_\mu &= \p_{x^\mu} 
\label{eq:TB}
\intertext{%
On $\cyl\setminus\{\ynullgen=0\}$ one has, using \eqref{eq:Kelvin},}
B^{\mu \nu} &=  (y^\mu \eta^{\nu\sigma} 
- y^\nu \eta^{\mu \sigma}) \p_{y^\sigma}\;, &
T_\mu &= 
y^{\nu} y^{\sigma}(\eta_{\nu\sigma} \p_{y^\mu}-2 \eta_{\mu\nu} \p_{y^\sigma})
\label{eq:TBy}
\end{align}
\end{subequations}
\begin{remark}\label{rem:Kil(h)}
For all $i,j=1,2,3$ one has,
by direct calculation using \eqref{eq:xxcoords},
\begin{align*}
\tfrac{B^{0i}(\nullgen)}{\nullgen} &= -\xi^i\sin\tau &
\tfrac{B^{ij}(\nullgen)}{\nullgen} &= 0 &
\tfrac{T_0(\nullgen)}{\nullgen} &= \xi^4\sin\tau &
\tfrac{T_i(\nullgen)}{\nullgen} &= -\xi^i\cos\tau
\end{align*}
%
\end{remark}

\subsection{Definition of the dgLa}
\label{sec:dgLa_def}

We state the definition of the differential graded Lie algebra $\gx(\cyl)$,
see Theorem \ref{thm:gaxioms}.
This is a summary of \cite[Section 3.3]{Thesis},
where one can find more details.

Let $\Omega(\cyl)$ be the real smooth differential forms on $\cyl$.
Recall that the three-sphere $S^3$ and hence $\cyl$ are parallelizable,
hence $\Omega(\cyl)$ is a free $C^\infty$-module.

Let $\Der^k(\Omega(\cyl))$ be
the space of derivations of $\Omega(\cyl)$ with degree $k$,
given by all $\R$-linear maps $\Omega(\cyl)\to \Omega(\cyl)$
that restrict to $\Omega^{i}(\cyl)\to \Omega^{i+k}(\cyl)$ for all $i$,
and that satisfy the Leibniz rule with signs.
\begin{definition}\label{def:Ldef}
Define
\[ 
\lx(\cyl) = \Omega(\cyl)\otimesRR \Kil
\]
This carries a grading given by 
$\lx(\cyl)=\oplus_{k=0}^4\lx^k(\cyl)$ with 
$\lx^k(\cyl) = \Omega^k(\cyl)\otimesRR\Kil$.
Further it is a graded $\Omega(\cyl)$-module
where the module multiplication is given by 
\begin{subequations}
\begin{equation}\label{eq:lmod}
\omega(\omega'\otimes \KilEl) =(\omega\wedge\omega')\otimes \KilEl
\end{equation}
for all $\omega,\omega'\in\Omega(\cyl)$ and $\KilEl\in\Kil$.
Define the operations
\begin{align*}
\dl:\;&\lx^k(\cyl)\to\lx^{k+1}(\cyl)\\
[\cdot,\cdot]:\;&\lx^{k}(\cyl)\times\lx^{k'}(\cyl)\to\lx^{k+k'}(\cyl)\\
\anchorL:\;&\lx^k(\cyl)\to \Der^k(\Omega(\cyl))
\end{align*}
where $\dl$ is $\R$-linear,
$[\cdot,\cdot]$ is $\R$-bilinear,
and $\anchorL$ is $\R$-linear, by
\begin{align}
\dl(\omega\otimes \KilEl) 
	&= (\ddR\omega)\otimes \KilEl \\
[\omega\otimes \KilEl,\omega'\otimes \KilEl']
	&= \omega\wedge\omega' \otimes [\KilEl,\KilEl']
	+ \omega\wedge(\Lie_{\KilEl}\omega')\otimes \KilEl'
	- (\Lie_{\KilEl'}\omega)\wedge\omega'\otimes \KilEl
	\label{eq:llbracket}\\
\anchorL(\omega\otimes \KilEl)(\omega') 
	&= \omega \wedge(\Lie_{\KilEl}\omega')\label{eq:anchorLdef}
\end{align}
\end{subequations}
for all $\omega,\omega'\in\Omega(\cyl)$ and $\KilEl,\KilEl'\in\Kil$.
Here $\ddR$ denotes the de Rham differential
and $\Lie$ denotes the Lie derivative.
\end{definition}
The operations $\dl$, $[\cdot,\cdot]$, $\anchorL$
are called differential, bracket and anchor, respectively,
and satisfy various algebraic identities (see \cite[Lemma 19]{Thesis}). 
The anchor $\anchorL$ is important because
it encodes the principal part of the bracket,
which can be seen from the Leibniz rule
\[ 
[\ell,\omega \ell'] = \anchorL(\ell)(\omega) \ell' +(-1)^{qk} \omega[\ell,\ell']
\]
which holds for all 
$\ell\in\lx^k(\cyl)$, $\ell'\in\lx^{k'}(\cyl)$, $\omega\in\Omega^{q}(\cyl)$,
where juxtaposition stands for the module multiplication \eqref{eq:lmod}.
\step
For $s\in\R$ let $\dens{s}(\cyl)$ be the module of 
sections of the $s$-density bundle on $\cyl$, where we use
the convention detailed in Remark \ref{rem:densityconv}.
\begin{remark}\label{rem:densityconv}
The fiber of $\dens{s}(\cyl)$ at $p\in\cyl$ is given by 
all $\mu:\Lambda^4(T_p\cyl)\setminus\{0\}\to\R$
such that for all $\lambda\in\R\setminus\{0\}$ and $X\in \Lambda^4(T_p\cyl)\setminus\{0\}$
one has 
$$\mu(\lambda X) = |\lambda|^{\frac{s}{4}}\mu(X)$$
Note that in this convention one integrates $4$-densities.
For a representative $g$ of $[\gcyl]$ we denote by
$\mu_{g}^{s}$ the associated density in $\dens{s}(\cyl)$,
whose value at $p\in\cyl$ is given by 
$\mu_{g}^s|_{p}(e_0\wedge\dots\wedge e_3)=1$ for any 
$g$-orthonormal basis $e_0,\dots,e_3$ of $T_p\cyl$.
If $f\in C^\infty(\cyl)$ is nowhere vanishing then 
$\mu^s_{f^2g} = |f|^s \mu_{g}^s$.
\end{remark}
Let $\OmegaC(\cyl)$ be the complex smooth differential forms.
Canonically 
$\OmegaC^2(\cyl)=\Omega^2_+(\cyl)\oplus\Omega^2_-(\cyl)$ where 
$\Omega^2_\pm(\cyl) = \{\omega\in\OmegaC^2(\cyl) \mid 
\star_{[\gcyl]}\omega = \pm i\omega\}$.
Here $\star_{[\gcyl]}$ is the Hodge dual for two-forms
that is associated to the conformal metric $[\gcyl]$,
using the orientation on $\cyl$, see
Remark \ref{rem:orientation}.
Recall the $\gcyl$-orthonormal frame $\V_\mu$ in \eqref{eq:defV}.
\begin{definition}\label{def:It}
Define the following $C^\infty(\cyl,\C)$-modules:
\begin{itemize}
\item 
$\It^2_{\pm}(\cyl) \subset S^2(\Omega^2_{\pm}(\cyl))$ is given by 
all $u$ that satisfy
$$
\eta^{\alpha\beta}\eta^{\mu\nu}u(\V_{\alpha},\V_{\mu},\V_{\beta},\V_{\nu}) = 0
$$
where $S^2$ is the symmetric tensor product over $C^\infty$.
\item 
$\It^3_{\pm}(\cyl)\subset \OmegaC^3(\cyl)\otimesCinf\Omega^2_{\pm}(\cyl)$
is given by all $u$ that satisfy
$$
\eta^{\alpha\beta}\eta^{\mu\nu}u(\,\cdot\,,\V_{\alpha},\V_{\mu},\V_{\beta},\V_{\nu})= 0
$$
\item 
$\It^4_{\pm}(\cyl) = \OmegaC^4(\cyl)\otimesCinf\Omega^2_{\pm}(\cyl)$
\end{itemize}
For $k=2,3,4$
define $\I^k_{\pm}(\cyl) = \dens{-1}(\cyl)\otimesCinf\It^k_{\pm}(\cyl)$, 
and define the real subspace 
\[ 
\I^k(\cyl)= \left(\I^k_+(\cyl)\oplus \I^k_-(\cyl)\right)_{\R}
\]
given by all $\up\oplus\upbar$ with $\up\in\I^k_+(\cyl)$.
Here the bar denotes complex conjugation,
which maps $\I^k_\pm(\cyl)\to \I^k_\mp(\cyl)$.
Define
$\I(\cyl) = \I^2(\cyl)\oplus\I^3(\cyl)\oplus\I^4(\cyl)$.
\end{definition}
Each $\I^k(\cyl)$ is the module of sections of a trivial vector bundle
on $\cyl$, of rank $10,16,6$ when $k=2,3,4$,
that we denote by $\I^k$.
Elements in $\I^2(\cyl)$ satisfy the symmetry and 
traceless conditions (relative to $[\gcyl]$) of Weyl curvatures.
\begin{remark}[Sweedler notation]\label{rem:Sweedler}
Every element $u\in \I^k_{\pm}(\cyl)$ can be written as a finite
sum of product elements, that is, 
$$u=\tsum_{i=1}^n
\mu_i\otimes\omega_{i}\otimes\omega_i'$$ for some $n\in\N$
and elements $\mu_i\in\dens{-1}(\cyl)$, 
$\omega_i\in\OmegaC^k(\cyl)$, $\omega_i'\in\Omega_{\pm}^2(\cyl)$.
We will abbreviate this sum by $u=\mu\otimes\omega\otimes\omega'$,
which is known as Sweedler's notation.
\end{remark}
\begin{definition}\label{def:Imod}
Define the multiplication $\OmegaC^q(\cyl)\times\I_{\pm}^{k}(\cyl)\to \I_{\pm}^{q+k}(\cyl)$ by 
\begin{equation}
\label{eq:multIP}
\omega u_{\pm} = \mu\otimes(\nu\wedge\omega)\otimes\omega'
\end{equation}
where we write $u_\pm=\mu\otimes\omega\otimes\omega'$ using
Sweedler's notation in Remark \ref{rem:Sweedler}.
Define the multiplication $\Omega^q(\cyl)\times\I^{k}(\cyl)\to \I^{q+k}(\cyl)$ by 
\begin{equation*}
\omega (u_+\oplus u_-) = \omega u_+\oplus \omega u_-
\end{equation*}
\end{definition}
It is easy to see from Definition \ref{def:It} that 
the right hand side of \eqref{eq:multIP}
indeed is in $\I_{\pm}^{q+k}(\cyl)$.
Definition \ref{def:Imod} equips $\I(\cyl)$
with the structure of a graded $\Omega(\cyl)$-module.
\begin{definition}\label{def:dI}
Define the $\C$-linear map $\dI: \I^{k}_{\pm}(\cyl)\to\I^{k+1}_{\pm}(\cyl)$ by 
\begin{equation}\label{eq:dIpm}
\dI(u_\pm) 
	= \Vd^\alpha \left( \mu_{\gcyl}^{-1} \otimes 
	 \nabla^{\gcyl}_{\V_\alpha} (\mu_{\gcyl}^1\otimes u_{\pm} )\right)
\end{equation}
where $\mu_{\gcyl}^{s}$ is the $s$-density associated to $\gcyl$ (see Remark \ref{rem:densityconv}), where $\nabla^{\gcyl}$ is the Levi-Civita connection 
of $\gcyl$, and where $\Vd^0,\dots,\Vd^3$ is the frame of one-forms 
dual to $\V_0,\dots,\V_3$.
This formula is to be understood as follows:
One has $\mu_{\gcyl}^1\otimes u_{\pm}\in \It_{\pm}^k(\cyl)$
using the canonical isomorphism 
$\dens{1}(\cyl)\otimesCinf\dens{-1}(\cyl)\simeq C^\infty(\cyl)$;
then the covariant derivative produces an element in 
$\It_{\pm}^k(\cyl)$ (by \cite[Lemma 22]{Thesis});
then tensoring with $\mu_{\gcyl}^{-1}$ produces an element in $\I_{\pm}^{k}(\cyl)$;
then multiplication with the one-form $\Vd^\alpha$, using \eqref{eq:multIP},
yields an element in $\I_{\pm}^{k+1}(\cyl)$.

Define the $\R$-linear map $\dI: \I^{k}(\cyl)\to\I^{k+1}(\cyl)$ by 
\[ 
\dI(\up \oplus \um) = \dI\up \oplus \dI\um
\]
\end{definition}
\begin{remark}\label{rem:dIprop}
The map $\dI$ has the following properties \cite[Lemma 24]{Thesis}:
It is independent of the chosen representative metric $\gcyl$ of $[\gcyl]$, that is, in \eqref{eq:dIpm}
one can replace $\gcyl$ by any other representative of $[\gcyl]$.
It is independent of the chosen frame $\V_\alpha$,
that is, one can replace $\V_\alpha$, $\Vd^\alpha$ by any other frame
and associated dual frame of one-forms.
Furthermore it is a differential, $\dI \circ \dI = 0$.
It is compatible with the $\Omega$-module structure in Definition \ref{def:Imod},
in the sense that it satisfies the Leibniz rule 
$\dI(\omega u) = (\ddR\omega)u + (-1)^q\omega(\dI u)$
for all $\omega\in\Omega^q(\cyl)$.
\end{remark}
If $\KilEl\in\Kil$ is a Killing field then the Lie derivative
with respect to $\KilEl$ 
is a map $\Lie_{\KilEl}:\I_{\pm}(\cyl)\to\I_{\pm}(\cyl)$ 
\cite[Lemma 25]{Thesis} 
(this also holds more generally when $\KilEl$ is a conformal Killing field). 
This allows the following definition.
\begin{definition}\label{def:bracketIL}
Define the $\R$-bilinear map
$[\cdot,\cdot]_{\I}:\lx^{q}(\cyl)\times\I^{k}(\cyl)\to \I^{q+k}(\cyl)$ by 
\begin{equation}
\label{eq:LIbracket}
[\ell,u]_{\I} = \omega(\Lie_{\KilEl} \up) \oplus \omega(\Lie_{\KilEl} \um)
\end{equation}
for all $\ell=\omega\otimes \KilEl\in\lx^q(\cyl)$ and 
$u=\up\oplus\um\in \I^{k}(\cyl)$, 
where we use the module multiplication in Definition \ref{def:Imod}.
Further define
$[u,\ell]_{\I} = 
-(-1)^{q k} [\ell,u]_{\I}
$.
\end{definition}
By \cite[Lemma 27]{Thesis}, 
there exists a unique $C^\infty(\cyl)$-linear map
\begin{equation}\label{eq:injaux}
\injaux:\; \dens{-1}(\cyl)\otimesCinf \Omega^2(\cyl) \;\to\; \lx^0(\cyl)
\end{equation}
whose restriction to the dense subset $\cyl\setminus\{\nullgen=0\}$
is given by
\begin{align}\label{eq:injauxformula}
\nullgen \mu_{\gcyl}^{-1}\otimes(dx^\mu\wedge dx^\nu) 
\mapsto 1\otimes B^{\mu\nu} - 
(x^\mu\otimes (\eta^{\nu\alpha}T_\alpha)-x^\nu\otimes (\eta^{\mu\alpha}T_\alpha))
\end{align}
where we use $x$ in \eqref{eq:xxcoords},
which are coordinates on every connected component of $\cyl\setminus\{\nullgen=0\}$,
the boosts and translations \eqref{eq:Kilbas},
and where $\mu_{\gcyl}^{-1}$ is the $-1$-density associated
to $\gcyl$, see Remark \ref{rem:densityconv}
(on $\diamond$ one has
$\nullgen \mu_{\gcyl}^{-1}=\mu_{\gmink}^{-1}$ using \eqref{eq:defdiamond}).

Let $\injaux_\C$ be the $\C$-linear extension of \eqref{eq:injaux} 
and let $\lx_\C(\cyl) = \OmegaC(\cyl)\otimesRR\Kil$.
\begin{definition}\label{def:inj}
Define the $C^\infty(\cyl,\C)$-linear maps $\inj_{\pm}:\I_{\pm}^k(\cyl)\to \lx_{\C}^k(\cyl)$ by 
\[ 
\upm
\;\mapsto\;
\omega \injaux_{\C} (\mu \otimes \omega')
\]
where we write $\upm=\mu \otimes \omega \otimes \omega'$
using Sweedler's notation in Remark \ref{rem:Sweedler}.
Define the $C^\infty(\cyl)$-linear map $\inj:\I^k(\cyl)\to\lx^k(\cyl)$ by 
\[ 
\up\oplus\um \;\mapsto\; \inj_+(\up) + \inj_-(\um)
\]
\end{definition}
Note that $\inj$ is not only linear over $C^\infty(\cyl)$ 
but also linear over $\Omega(\cyl)$.

We now define the dgLa $\gx(\cyl)$,
using Definition \ref{def:Ldef}, \ref{def:It}, \ref{def:Imod}, \ref{def:dI},
\ref{def:bracketIL}, \ref{def:inj}.
\begin{theorem}[{\cite[Theorem 9]{Thesis}}]\label{thm:gaxioms}
Define $\gx(\cyl) = \oplus_{k=0}^4\gx^k(\cyl)$ where
$$
\gx^k(\cyl) \;=\; \lx^k(\cyl)\oplus\I^{k+1}(\cyl)\Ieps
$$ 
This is a graded $\Omega(\cyl)$-module where the multiplication
$\Omega^{q}(\cyl)\times\gx^{k}(\cyl)\to\gx^{q+k}(\cyl)$ is given by,
using \eqref{eq:lmod} respectively Definition \ref{def:Imod},
\begin{subequations}\label{eq:gop} 
\begin{align}\label{eq:gmod}
\omega (\uo\oplus\uI) = (\omega \uo) \oplus (\omega\uI) 
\end{align}
for all $\omega\in\Omega^q(\cyl)$
and $\uo\oplus\uI\in\gx^k(\cyl)$.
Define the operations
\begin{align*}
\dg:&\;\gx^k(\cyl)\to\gx^{k+1}(\cyl)\\
[\cdot,\cdot]:&\;\gx^k(\cyl)\times\gx^{k'}(\cyl)\to\gx^{k+k'}(\cyl)\\
\anchorg:&\; \gx^k(\cyl) \to \Der^k(\Omega(\cyl))
\end{align*}
by
\begin{align}
\dg u
&= 
\big(\dl \uo -(-1)^{\new{k+1}} \inj(\uI) \big)\oplus (\dI \uI)\Ieps
\label{eq:dgdef}\\
[u,u']
&=
[\uo,\uo']_{\lx} \oplus \big([\uo,\uI']_{\I} + (-1)^{\new{k'}}[\uI,\uo']_{\I}\big)\Ieps
\label{eq:[]gdef}\\
\anchorg(u) &= \anchorL(\uo)\label{eq:anchorgdef}
\end{align}
\end{subequations}
for all
$u=\uo\oplus \uI\Ieps\in \gx^{k}(\cyl)$ 
and $u'=\uo'\oplus \uI'\Ieps\in \gx^{k'}(\cyl)$.
The operation $\dg$ is $\R$-linear and called differential,
$[\cdot,\cdot]$ is $\R$-bilinear and called bracket,
and $\anchorg$ is $\R$-linear and called anchor.
They satisfy:
\begin{subequations}\label{eq:dgLaax}
\begin{align}
&\dg\circ\dg=0\label{eq:dgdifferential}\\
&\dg(\omega u) = (\ddR\omega)u + (-1)^{\degomega} \omega \dg u\label{eq:dgmodule}\\
&\dg[u,v] = [\dg u,v] + (-1)^{\degu}[u,\dg v]\label{eq:dgbracket}\\
&[u,\omega v] = \anchorg(u)(\omega)v + (-1)^{\degomega \degu}\omega[u,v]\label{eq:anchorbracket}\\
&\anchorg(u)(\omega\wedge\omega') 
= \anchorg(u)(\omega)\wedge\omega' + (-1)^{\degomega\degu} \omega\wedge \anchorg(u)(\omega') \label{eq:anchorleib}\\
&\anchorg(\omega u) = \omega\anchorg(u) \label{eq:anchorlin}\\
&\anchorg(\dg u) = \ddR\circ \anchorg(u)-(-1)^{\degu} \anchorg(u)\circ\ddR\\
&\anchorg([u,v]) = \anchorg(u)\circ\anchorg(v)-(-1)^{\degu\degv}\anchorg(v)\circ\anchorg(u)\\
&[u,v] = -(-1)^{\degu\degv}[v,u]\label{eq:bracketas}\\
&[u,[v,z]] +(-1)^{\degz(\degu+\degv)}[z,[u,v]] + (-1)^{\degu(\degv+\degz)}[v,[z,u]] = 0\label{eq:bracketjacobi}
\end{align}
\end{subequations}
for all 
$u\in\gx^{\degu}(\cyl)$, 
$v\in\gx^{\degv}(\cyl)$, 
$z\in\gx^{\degz}(\cyl)$ and
$\omega\in \Omega^{\degomega}(\cyl)$,
$\omega'\in \Omega^{\degomegap}(\cyl)$.
Algebraically, this means that $\gx(\cyl)$ is an
$\Omega(\cyl)$-differential graded Lie algebroid.
\end{theorem}
\begin{remark}\label{rem:gdiamond}
The module $\gx(\cyl)$ is the module of smooth sections of a
trivial vector bundle on $\cyl$, that we denote by $\gx$.
We denote by $\gx(\diamond)$ the module of smooth sections of $\gx$ over $\diamond$.
Since all operations \eqref{eq:gop} are local, they restrict to maps on $\gx(\diamond)$.
Analogously for other sufficiently nice subsets of $\cyl$.
\end{remark}
\subsection{$\R_+$-action on the dgLa}
\label{sec:R+action}

In this section we introduce a natural $\R_+$-action on $\gx(\cyl)$,
that on the base manifold $\cyl$ is given by the flow 
of a conformal Killing vector field that scales about $\spaceinf$,
and that commutes with the operations \eqref{eq:gop}.

Concretely, for $\lambda>0$ let 
\begin{equation}\label{eq:scall}
\Scal_\lambda:\cyl\to\cyl
\end{equation}
be the diffeomorphism 
that on the dense subset $\cyl\setminus\{\nullgen=0\}$ is given by 
\begin{equation*}
\Scal_\lambda^{\new{\ast}}(x) = \lambda x
\end{equation*}
where we use the functions $x$ in \eqref{eq:xxcoords},
which are coordinates on every connected component of $\cyl\setminus\{\nullgen=0\}$.
One then has 
$\Scal_\lambda^{\new{\ast}}(y) = \lambda^{-1} y$
using \eqref{eq:Kelvin}.
In particular, $\Scal_{\lambda}$ restricts to a diffeomorphism
on each of $\diamond$, $\diamond_+$, $\diamondy$.
\begin{definition}\label{def:scalg}
For $\lambda>0$ define the $\R$-linear map
$\Scalg_\lambda:\gx(\cyl)\to\gx(\cyl)$ by 
\begin{align}\label{eq:scalg}
\Scalg_{\lambda}\left((\omega\otimes \KilEl)\oplus (\up\oplus\um)\Ieps\right)
&=
(\Scal^*_{\lambda}\omega\otimes \Scal^*_{\lambda}\KilEl)\oplus 
(\lambda^{-1}\Scal^*_{\lambda}\up \oplus \lambda^{-1}\Scal^*_{\lambda}\um)\Ieps
\end{align}
where $\omega\otimes \KilEl\in\lx(\cyl)$ and $\up\oplus\um\in\I(\cyl)$.
Here $\Scal^*_{\lambda}$ is the pullback along $\Scal_{\lambda}$,
and in the case of $\up$ and $\um$ this also involves the pullback of densities.
For the translations and boosts \eqref{eq:TB} one has
$\Scal^*_\lambda T_{\mu} = \lambda^{-1}T_\mu$
and 
$\Scal^*_\lambda B^{\mu\nu} = B^{\mu\nu}$.
\end{definition}
Note that the $\R_+$-action in Definition \ref{def:scalg} restricts
both to $\lx(\cyl)$ and to $\I(\cyl)$.
\begin{lemma}\label{lem:homogeneity}
One has:
\begin{subequations}
\begin{align}
\Scal_{\lambda}^* \ddR 
	&= 
	\ddR \Scal_\lambda^*  
	&
\Scalg_{\lambda} \dg 
	&= 
	\dg \Scalg_{\lambda}
	&
\Scalg_{\lambda}[\,\cdot\,,\,\cdot\,]
	&= 
	[\Scalg_{\lambda}\,\cdot\,,\Scalg_{\lambda}\,\cdot\,] 
	\label{eq:ddbr_hom}
\end{align}
Furthermore, for all $u\in\gx(\cyl)$ and $\omega\in\Omega(\cyl)$:
\begin{align}
\Scalg_{\lambda}(\omega u) 
	&= (\Scal_{\lambda}^*\omega) (\Scalg_{\lambda}u)
	\label{eq:modmult_hom}\\
\Scal^*_{\lambda}\big(\anchorg(u)(\omega)\big)
	&= 
	\anchorg(\Scalg_{\lambda}u)(\Scal_{\lambda}^*\omega) 
	\label{eq:anchor_hom}
\end{align}
\end{subequations}
\end{lemma}
\begin{proof}
First of \eqref{eq:ddbr_hom}: 
Clear.
Second of \eqref{eq:ddbr_hom}:
Denote the restrictions of $\Scalg_\lambda$
to $\lx(\cyl)$ and $\I(\cyl)$ by 
$\ScalL_\lambda$ and $\ScalI_\lambda$, respectively.
It suffices to show:
\begin{align*}
\ScalL_\lambda \dl &= \dl \ScalL_\lambda
&
\ScalI_\lambda \dI &= \dI \ScalI_\lambda
&
\ScalL_\lambda \inj &= \inj \ScalI_\lambda
\end{align*}
The first is clear.
The second holds by \cite[Remark 27]{Thesis} and the fact that
$\Scal^*_\lambda$ is a conformal isometry 
(more explicitly, use \eqref{eq:dIpm}, the fact that 
$\Scal_\lambda^*\gcyl$ is again a representative of $[\gcyl]$,
and Remark \ref{rem:dIprop}).
The third follows from \eqref{eq:injauxformula}, using the
explicit $\lambda^{-1}$ factor in \eqref{eq:scalg}.
Third of \eqref{eq:ddbr_hom}:
Use \eqref{eq:llbracket}, \eqref{eq:LIbracket}, compatibility
of the pullback with wedge product and Lie derivative, 
and $\Scal_{\lambda}^*[\KilEl,\KilEl']=[\Scal_{\lambda}^*\KilEl,\Scal_{\lambda}^*\KilEl']$
for $\KilEl,\KilEl'\in\Kil$.
\eqref{eq:modmult_hom}:
Use \eqref{eq:lmod} and \eqref{eq:multIP}. 
\eqref{eq:anchor_hom}:
Use \eqref{eq:anchorgdef} and \eqref{eq:anchorLdef}. 
\qed
\end{proof}

\subsection{Relation to Ricci-flat metrics,
proof of Proposition \ref{prop:metricregularity}}
	\label{sec:recoveryofmetric}
We recall the relation between 
solutions of \eqref{eq:MC}
and Ricci-flat metrics on $\diamond$.
For simplicity, in the following definitions and statements
we use $\diamond$ as the underlying domain.
It is understood that the same definitions and statements
can be made for other domains, particularly $\diamond_+$,
since all operations are local.
\begin{definition}\label{def:Fu}
Let $\uo\in\Omega^1(\diamond)\otimesRR\Kil$. Expand 
\[ 
\uo = \tsum_{i=1}^{10}\omega_i\otimes \KilBasis_i
\]
where $\KilBasis_1,\dots,\KilBasis_{10}$ is a basis of $\Kil$.
Define the $C^\infty$-linear maps 
\begin{align*}
\frame{\uo}: 
\Gamma(T\diamond)&\to\Gamma(T\diamond)
&
\frame{\uo}(X) &=  \tsum_{i=1}^{10}\omega_i(X)\KilBasis_i\\
\dualframe{\uo}: 
\Omega^1(\diamond)&\to\Omega^1(\diamond)
&
\dualframe{\uo}(\theta) &=  \tsum_{i=1}^{10}\omega_i \theta(\KilBasis_i)
\end{align*}
They are independent of the chosen basis $\KilBasis_1,\dots,\KilBasis_{10}$,
and $\dualframe{\uo}$ is the dual of $\frame{\uo}$,
in the sense that 
$\theta(\frame{\uo}(X)) = (\dualframe{\uo}(\theta))(X)$ for all $X,\theta$.
We say that $\uo$ is nondegenerate if and only if
$\one+\frame{\uo}$ is invertible at every point on $\diamond$,
and that $u=\uo\oplus \uI \Ieps\in \gx^1(\diamond)$ is
nondegenerate if and only if $\uo$ is nondegenerate.

Note that $\frame{\uo}$ and $\dualframe{\uo}$ are $C^\infty$-linear in $\uo$.
\end{definition}
\begin{prop}[{\cite[Prop.~10]{Thesis}}]\label{prop:metricdef}
Let $u=\uo\oplus\uI\in\gx^1(\diamond)$ be nondegenerate.
If 
$\dg u+\frac12[u,u]=0$
then the smooth Lorentzian metric $g$ on $\diamond$ defined by 
\begin{align}\label{eq:metricdef}
g^{-1} 
= 
(\one + \frame{\uo})^{\otimes2} \eta^{-1}
\end{align}
is Ricci-flat, $\Ric(g)=0$.
The formula \eqref{eq:metricdef} is to be understood as follows:
$\eta^{-1}$ is a section of the second tensor power of $T\diamond$, 
and we apply $\one + \frame{\uo}$ to each factor
(explicitly 
$g^{-1} = \eta^{\mu\nu} (\one + \frame{\uo})(\p_{x^\mu})\otimes(\one + \frame{\uo})(\p_{x^\nu})$).
\end{prop}
The map $\one+\frame{\uo}$ is an orthonormal frame
for the metric \eqref{eq:metricdef}, see also \eqref{eq:1Fframe}.

We refer to \cite[Proposition 11]{Thesis} for the converse statement, 
i.e., that every Ricci-flat metric on $\diamond$ defines,
up to the choice of an orthonormal frame,
a nondegenerate solution of \eqref{eq:MC}.
\begin{remark}\label{rem:metricdefinitions}
Using \eqref{eq:minkmetric} we can equivalently rewrite \eqref{eq:metricdef} as
\begin{align*}
(\nullgen^2g)^{-1} 
= 
(\one + \frame{\uo})^{\otimes2} \gcyl^{-1}
\end{align*}
In particular, for all 
one-forms $\theta,\theta'$
and all vector fields $X,X'$:
\begin{align*}
(\nullgen^2g)^{-1}(\theta,\theta') 
	&=
	\gcyl^{-1}\big( (\one+\dualframe{\uo})\theta,(\one+\dualframe{\uo})\theta' \big)\\
(\nullgen^2g)(X,X') 
	&=
	\gcyl\big( (\one+\frame{\uo})^{-1}X,(\one+\frame{\uo})^{-1}X' \big)
\end{align*}
\end{remark}
In the following we state basic properties of $\frame{\uo}$
(see also Section \ref{sec:estframe} for more quantitative statements),
and prove Proposition \ref{prop:metricregularity}.
\begin{lemma}\label{lem:frameglobalbasis}
Let $(e_i)_{i=1\dots40}$ be the basis of $\Omega^1(\diamond)\otimesRR\Kil$
given by the elements
\begin{align}\label{eq:auxglobbasis}
\Vd^\mu \otimes B^{\alpha\beta}
\qquad
\Vd^\mu \otimes T_\nu
\qquad\quad 
\begin{aligned}
\mu,\alpha,\beta,\nu=0\dots3,\ 
\alpha<\beta
\end{aligned}
\end{align}
where we recall that $\Vd^{0},\dots,\Vd^3$
is the basis of one-forms dual to \eqref{eq:defV}.
For each $i$, the components of $\frame{e_i}:\Gamma(T\diamond)\to\Gamma(T\diamond)$ with respect 
to the basis $\V_0,\dots,\V_3$ are smooth on $\overline\diamond$.
Moreover, the components of the one-form
\begin{align}\label{eq:frameedh}
\tfrac{1}{\nullgen}\dualframe{e_i}(d\nullgen) 
\end{align}
with respect to the basis $\Vd^{0},\dots,\Vd^3$ are smooth on $\overline\diamond$.
\end{lemma}
\begin{proof}
The first statement follows from the definition of $\frame{e_i}$
and the fact that $B^{\alpha\beta}$, $T_\nu$
are smooth vector fields on $\cyl$.
For smoothness of \eqref{eq:frameedh}, note that
\begin{align*}
\tfrac{1}{\nullgen}\dualframe{\Vd^\mu \otimes B^{\alpha\beta}}(d\nullgen) 
	&=
	\tfrac{B^{\alpha\beta}(\nullgen)}{\nullgen}\Vd^\mu &
\tfrac{1}{\nullgen}\dualframe{\Vd^\mu \otimes T_\nu}(d\nullgen) 
	&=
	\tfrac{T_\nu(\nullgen)}{\nullgen}\Vd^\mu
\end{align*}
and use Remark \ref{rem:Kil(h)}.
\qed
\end{proof}
\begin{lemma}\label{lem:Fglobbasic}
For all $\uo\in \Omega^1(\diamond)\otimesRR\Kil$ one has:
\begin{itemize}
\item 
For every $k\in\Z_{\ge0}$, 
if $\uo$ extends in $C^k$ to $\overline{\diamond}\setminus\spaceinf$ then 
the components of $\frame{\uo}$ with respect to the basis $\V_0,\dots,\V_3$
extend in $C^k$ to $\overline{\diamond}\setminus\spaceinf$.
\item 
At every point on $\diamond$:
Denoting by $\|\frame{\uo}\|$ the $\ell^2$-matrix
norm with respect to the basis $\V_0,\dots,\V_3$,
and by $\|\uo\|$ the $\ell^2$-vector norm 
with respect to the basis \eqref{eq:auxglobbasis}, 
then one has $\|\frame{\uo}\| \lesssim \|\uo\|$
(the notation $\lesssim$ is in Remark \ref{rem:lesssim}).
\end{itemize}
\end{lemma}
\begin{proof}
Expand 
$\uo = \sum_{i=1}^{40} u_{0i}e_i$ 
where $(e_i)_{i=1\dots40}$ is the basis \eqref{eq:auxglobbasis}. 
Then
$\frame{\uo} = \tsum_{i=1}^{40} u_{0i} \frame{e_i}$.
By Lemma \ref{lem:frameglobalbasis} the components of $\frame{e_i}$
are smooth on $\overline{\diamond}$, hence the first item follows.
Since the components of $\frame{e_i}$ are smooth on $\overline{\diamond}$, 
they are also uniformly bounded, hence the second item follows.
\qed
\end{proof}
\begin{proof}[of Proposition \ref{prop:metricregularity},
proofs of 
\ref{item:metric_nullcompleteness}, \ref{item:metric_nullinfinity}
are only sketched]
\proofheader{Proof of \ref{item:metric_regularextension}.}
By the first item of Lemma \ref{lem:Fglobbasic} and the regularity assumption on $\uo$,
the components of $\frame{\uo}$ with respect to $\V_\mu$
extend in $C^{k}$ to $\overline{\diamond}_{\new{+}}\setminus\spaceinf$.
By \ref{item:F1/16assp}, $\one+\frame{\uo}$ is invertible,
and the components of the inverse with respect to $\V_\mu$
also extend in $C^{k}$ to $\overline{\diamond}_{\new{+}}\setminus\spaceinf$.
The claim now follows from the formula (see Remark \ref{rem:metricdefinitions})
\begin{equation}\label{eq:metricaltNEW}
(\nullgen^2g)(\V_\mu,\V_\nu)
= 
\gcyl\big( (\one + \frame{\uo})^{-1}\V_\mu , (\one + \frame{\uo})^{-1}\V_\nu  \big)
\end{equation}

\proofheader{Proof of \ref{item:metric_eikonal}.}
Using Remark \ref{rem:metricdefinitions},
\begin{align}
&\tfrac{1}{\nullgen}(\nullgen^2g)^{-1}(d\nullgen,d\nullgen) 
\label{eq:gnullgen}\\
&\qquad\qquad= 
\tfrac{1}{\nullgen}\gcyl^{-1}\big( (\one + \dualframe{\uo})d\nullgen , (\one + \dualframe{\uo})d\nullgen\big) \nonumber\\
&\qquad\qquad=
\tfrac{1}{\nullgen}\gcyl^{-1}(d\nullgen,d\nullgen)
+
2\gcyl^{-1}(\tfrac{1}{\nullgen}\dualframe{\uo}(d\nullgen),d\nullgen)
+
\gcyl^{-1}(\tfrac{1}{\nullgen}\dualframe{\uo}(d\nullgen),\dualframe{\uo}(d\nullgen))
\nonumber
\end{align}
We show that this function extends in $C^{k}$ to $\overline{\diamond}_{\new{+}}\setminus\spaceinf$.
By direct calculation $\tfrac{1}{\nullgen}\gcyl^{-1}(d\nullgen,d\nullgen)=\cos(\tau)+\xi^4$,
which is smooth on $\cyl$. 
The components of $\tfrac{1}{\nullgen}\dualframe{\uo}(d\nullgen)$
with respect to $\Vd^\mu$ extend in $C^{k}$ to 
$\overline{\diamond}_{\new{+}}\setminus\spaceinf$
by Lemma \ref{lem:frameglobalbasis}, 
$C^\infty$-linearity of $\dualframe{\uo}$ in $\uo$,
and the assumption that $\uo$ extends in $C^{k}$.
Thus the claim follows.

\proofheader{Proof sketch of \ref{item:metric_nullinfinity}.}
By Remark \ref{rem:metricdefinitions} we have
\begin{align}\label{eq:vvf}
(\nullgen^2g)^{-1}(d\nullgen,\,\cdot\,) 
= \gcyl^{-1} ((\one+\dualframe{\uo})d\nullgen,(\one+\dualframe{\uo})\,\cdot\,)
\end{align}
By \ref{item:metric_regularextension} this is a $C^k$-vector field
on $\overline\diamond_+\setminus\spaceinf$,
and by \ref{item:metric_eikonal} this is tangential to $\fnullinf$,
and null for the metric $\nullgen^2g$ along $\fnullinf$.
The point $\timeinf$ is a critical point for \eqref{eq:vvf} (meaning that
it vanishes there), because $d\nullgen|_{\timeinf}=0$.
A computation shows that its Jacobian at $\timeinf$ equals $-\one$,
using $\dualframe{\uo}|_{\timeinf}=0$ (because
every vector field in $\Kil$ vanishes at $\timeinf$).
Together with \ref{item:F1/16assp}, this implies
that \eqref{eq:vvf} is future directed along $\fnullinf$.
In summary, denoting the vector field \eqref{eq:vvf} by $\VV$:
\begin{align*}
\VV|_{\timeinf}=0
\qquad
(\p_{y^\nu}\VV(y^\mu))|_{\timeinf} = -\delta_{\nu}^{\mu}
\qquad
\VV(\tau)>0\;\text{along $\fnullinf$}
\end{align*}
From these properties one can conclude that the
rescaled vector field $\VV/\VV(\tau)$, viewed as a vector field
along $\fnullinf$, lifts to a $C^{k-1}$-vector field on the blowup
of $\timeinf$ in $\fnullinf\cup\timeinf$, given 
by $(0,\pi]\times S^2$.
Then the statement easily follows.

\proofheader{Proof sketch of \ref{item:metric_nullcompleteness}.}
Recall that $g$ and $\nullgen^2g$ have the same null geodesics.
The metric $\nullgen^2g$ is $C^k$-regular on $\overline\diamond_+\setminus\spaceinf$
by \ref{item:metric_regularextension}.
Every future directed null geodesic can be non-affinely parametrized by $\tau$, 
that is, in the form \eqref{eq:gammapartau_intro}
(the level sets of $\tau$ are spacelike
by \ref{item:F1/16assp}; future directed means that
$\tau$ is increasing).
To construct $\gamma$ in \eqref{eq:gammapartau_intro}, 
we equivalently reformulate the null geodesic equation for
$\nullgen^2 g$ as a first order ODE for $\tau\mapsto (\gamma(\tau),\dot\xi(\tau))$, where $\dot\xi(\tau)$
is viewed as an element in $\R^3$ using the vector fields 
$\V_1,\V_2,\V_3$.
This is an autonomous ODE on the 7-dimensional manifold with boundary 
$(\overline\diamond_+\setminus\spaceinf) \times \R^3$.
The ODE is given by a $C^{k-1}$ vector field
(smooth in the interior) by \ref{item:metric_regularextension},
and $\dot\xi(\tau)$ stays in the ball $|\dot\xi(\tau)|\le2$
by \ref{item:F1/16assp}.
Then the maximal integral curve with initial condition 
given by \eqref{eq:nullgeodeq_intro}
yields 
$\tau_1>0$ and $\gamma\in C^\infty([0,\tau_1),\diamond_+)$ 
that satisfies \eqref{eq:nullgeodeq_intro}.
One shows that $\gamma$ uniquely extends in $C^{\new{k}}$ to $\tau_1$,
and $\gamma(\tau_1)\in \fnullinf\cup\timeinf$.
By construction \ref{item:null_intro} holds.
By uniqueness of null geodesics and \ref{item:metric_nullinfinity}
one obtains \ref{item:gammaD_intro}.
One obtains  \ref{item:gammaAff_intro}, using the fact that
$(\nullgen\circ\gamma)(\tau)$ vanishes first order as $\tau\uparrow\tau_1$,
by \ref{item:gammaD_intro}.

For the last statement in \ref{item:metric_nullcompleteness}, 
it suffices to check that for every 
$p_0=(\tau_0,\xi_0)\in\diamond_+$ and every 
$v_*\in T_{p_0}\diamond_+$ that is null with respect to $g$ and 
normalized such that $d\tau(v_*)=1$,
the null geodesic that at $p_0$ has velocity $v_*$
reaches $\diamonddata$
when going in the negative $\tau$ direction,
in particular it 
does not go to either $\fnullinf$ or $\spaceinf$.
The null geodesic may be constructed using the same ODE as above,
to see that it reaches $\diamonddata$ use
\ref{item:F1/16assp}, \ref{item:Fdh1/16assp}
and \eqref{eq:eqcone}, \eqref{eq:phidef}.
To see that every point $p\in\fnullinf$ is reached 
by a null geodesic \eqref{eq:gammapartau_intro},
choose $v\in T_p\cyl$ that is null with respect to $\nullgen^2g$,
transversal to $\fnullinf$, and normalized such that $d\tau(v)=1$, 
and solve the ODE with initial data given by $p$ and $v$, in the negative $\tau$ direction.
\qed
\end{proof}

\subsection{Initial data and constraint equations}
\label{sec:constraints}

Initial data for \eqref{eq:MC} is given by a section on the 
initial hypersurface 
\begin{equation}\label{eq:diamonddatadef}
\diamonddata = \diamond\cap (\{0\}\times S^3)
\end{equation}
In this section we formulate the constraint equations
for the initial data, that is, 
the necessary and sufficient 
conditions for local solvability of \eqref{eq:MC}.
\begin{definition}\label{def:gxdata}
For $k=0\dots4$ let $\gxdata^k$ be the trivial vector bundle on $\{0\}\times S^3$
that is given by the pullback of the bundle $\gx^k$ under  
the inclusion $\{0\}\times S^3 \hookrightarrow \cyl$.
\end{definition}
We denote by $\gxdata^k(\diamonddata)$ the space of smooth sections 
of $\gxdata^k$ over $\diamonddata$, 
analogously for other sufficiently nice subsets of $\{0\}\times S^3$. 
Note that $\gxdata^k(\diamonddata)=\frac{\gx^{k}(\diamond)}{\tau\gx^{k}(\diamond)}$.
\begin{definition}\label{def:Pconstraints}
Define the non-linear first order differential operator 
\[ 
\Pconstraints: \;\;
\gxdata^1(\diamonddata) \;\to\;
\gxdata^3(\diamonddata) 
\]
by 
\begin{align}\label{eq:Pconstraintsdef}
\Pconstraints(\udata) \;=\; 
\left(d\tau + \anchorg(u)(\tau)\right)\left( \dg u + \tfrac12[u,u] \right)\big|_{\tau=0}
\end{align}
where $\udata\in\gxdata^1(\diamonddata)$,
and where $u\in\gx^1(\diamond_+)$ is any element that 
satisfies $u|_{\tau=0}=\udata$
(see Lemma \ref{lem:pext} for independence of the choice of $u$).
This formula is to be understood as follows:
The elements $d\tau + \anchorg(u)(\tau) \in \Omega^1(\diamond_+)$
and $\dg u + \tfrac12[u,u]\in \gx^2(\diamond_+)$ are multiplied
using 
\eqref{eq:gmod},
which is $C^\infty$-bilinear;
then the product is restricted to $\tau=0$, which
gives an element in $\gxdata^{3}(\diamonddata)$.
\end{definition}
\begin{lemma}\label{lem:pext}
The operator $\Pconstraints$ is independent of:
\begin{itemize}
\item  
The choice of extension $u$:
If $u,u'\in\gx^1(\diamond_+)$ satisfy
$u|_{\tau=0}=u'|_{\tau=0}$ then
\begin{align*}
(d\tau + \anchorg(u)(\tau)) ( \dg u + \tfrac12[u,u] )\big|_{\tau=0}
=
(d\tau + \anchorg(u')(\tau))( \dg u' + \tfrac12[u',u'] )\big|_{\tau=0}
\end{align*}
\item 
The choice of time function $\tau$, in the following sense:
If $f\in C^\infty(\diamond_+)$ satisfies $f|_{\tau=0}=0$ 
and if $df\neq0$ at every point on $\diamonddata$ then
\[ 
\Pconstraints(\udata) = 
(\tfrac{\tau}{f})|_{\tau=0} \left(df + \anchorg(u)(f)\right)
( \dg u + \tfrac12[u,u] )\big|_{\tau=0}
\]
where $(\tfrac{\tau}{f})|_{\tau=0}$
is a nowhere vanishing smooth function on $\diamonddata$.
\end{itemize}
\end{lemma}
\begin{proof}
\proofheader{Proof of first item.}
Since $u-u'$ vanishes along $\diamonddata$
we have $u-u' = \tau v$ for some $v\in\gx^1(\diamond_+)$. 
By linearity and bilinearity of $\dg$ respectively $[\cdot,\cdot]$,
\begin{align*}
\dg u + \tfrac12[u,u]
&=
\dg u' + \tfrac12[u',u']
+ 
\dg(\tau v)
+
[u',\tau v]
+
\tfrac12[\tau v,\tau v]
\intertext{
where we also use \eqref{eq:bracketas}.
Using \eqref{eq:dgmodule}, \eqref{eq:anchorbracket}, \eqref{eq:anchorlin}
we obtain}
\dg u + \tfrac12[u,u]
&=
\dg u' + \tfrac12[u',u']
+ 
(d\tau+\anchorg(u')(\tau))v\\
&\qquad
+
\tau\big( \dg v
+
[u',v]
+
 \tfrac12\anchorg(v)(\tau v)
+
\tfrac12[ v,\tau v]\big)
\end{align*}
Thus, using the fact that 
the module multiplication \eqref{eq:gmod} is $C^\infty$-bilinear,
\begin{align*}
\omega \big(\dg u + \tfrac12[u,u]\big)|_{\tau=0}
&=
\omega \big(\dg u' + \tfrac12[u',u']
+
\omega' v\big)|_{\tau=0}
\end{align*}
where we abbreviate $\omega = d\tau + \anchorg(u)(\tau)$ and
$\omega' = d\tau + \anchorg(u')(\tau)$.
Using
$\omega = \omega' + \tau \anchorg(v)(\tau)$
by \eqref{eq:anchorlin},
we can replace $\omega$ on the right hand side by $\omega'$.
With $\omega'\wedge\omega'=0$ (and associativity of 
the module multiplication), the claim follows.

\proofheader{Proof of second item.}
We have $f = \tau g$ where $g\in C^\infty(\diamond_+)$ is 
nowhere zero on $\diamonddata$. 
Then by the Leibniz rule for the de Rham differential 
and \eqref{eq:anchorleib},
\[ 
df + \anchorg(u)(f)
=
g (d\tau + \anchorg(u)(\tau))
+
\tau (dg + \anchorg(u)(g))
\] 
Thus the claim follows, using 
$g|_{\tau=0} = (f/\tau)|_{\tau=0}$.\qed
\end{proof}
We will refer to
\begin{equation}\label{eq:1constraints}
\Pconstraints(\udata) = 0
\end{equation}
as the constraint equations. 
This is a first order nonlinear partial differential equation along $\diamonddata$,
and the nonlinearity is at most cubic.
Clearly \eqref{eq:1constraints} is necessary for local solvability of
the initial value problem
\begin{equation}\label{eq:MCconstr}
\dg u + \tfrac12[u,u]=0 \qquad u|_{\tau=0}=\udata
\end{equation}
In Lemma \ref{lem:Pnecsuf} below we show that it is also sufficient
for local solvability.

The remainder of this section is not logically 
used to prove Theorem \ref{thm:main}.
We thus allow ourselves to use gauges that will 
only be introduced later on.

So let $\gxG^k(\cyl)\subset\gx^k(\cyl)$ be the gauge submodules,
and $\gBil^k:\gxG^k(\cyl)\times \gx^{k+1}(\cyl)\to C^\infty(\cyl)$ the 
$C^\infty$-bilinear forms, introduced in Definition \ref{def:gauge_bulk}.
Each $\gxG^k(\cyl)$ is the module of sections of a trivial vector bundle 
$\gxG^k$ on $\cyl$.

\begin{remark}\label{rem:Constraints36}
The constraint equations \eqref{eq:1constraints} are 46 equations for
$\udata$, because $\gxdata^3(\diamonddata)$ is a module of rank 46.
However it turns out that already 36 equations are necessary and sufficient.
To see this, 
let $\gxdataG^k$ be the bundle on $\{0\}\times S^3$
that is given by the pullback of $\gxG^k$ along $\{0\}\times S^3\hookrightarrow\cyl$, this has rank $36,10$ for $k=2,3$ respectively.
Let $\gxdataG^k(\diamonddata)$ be the sections over $\diamonddata$.
Define the composition
\[ 
\PconstraintsG:\;\;
\gxdata^1(\diamonddata) 
\;\xrightarrow{\;\Pconstraints\;}\;
\gxdata^3(\diamonddata) 
\;\twoheadrightarrow\;
\frac{\gxdata^3(\diamonddata)}{\gxdataG^3(\diamonddata)}
\]
One has:
\begin{itemize} 
\item 
The codomain $ \gxdata^3(\diamonddata)/\gxdataG^3(\diamonddata)$
is a free module of rank 36. 
Proof: The map $\gxdataG^2(\diamonddata)\to \gxdata^3(\diamonddata)/\gxdataG^3(\diamonddata)$,
$\underline{v}\mapsto (d\tau) \underline{v}$ is an 
isomorphism by Lemma \ref{lem:gauge_mainprop_bulk}.
\item 
For all $\udata\in\gxdata^1(\diamonddata)$ such that
$d\tau + \anchorg(\udata)(\tau)$ is timelike with respect to $\gcyl$
at every point on $\diamonddata$, one has:
\[ 
\PconstraintsG(\udata)=0
\;\;\Leftrightarrow\;\;
\Pconstraints(\udata)=0
\]
Proof: 
Abbreviate $\omega=d\tau+\anchorg(\udata)(\tau)$.
If $\PconstraintsG(\udata)=0$ then $\Pconstraints(\udata)\in\gxdataG^3(\diamonddata)$,
furthermore $\omega\Pconstraints(\udata)=0$ using $\omega\wedge\omega=0$,
and thus $\Pconstraints(\udata)=0$ by Lemma \ref{lem:gauge_mainprop_bulk} using the
fact that $\omega$ is timelike. The converse direction is immediate.
\end{itemize}
See Remark \ref{rem:Plinearized} for the linearization of $\Pconstraints_{G}$
about the zero solution $\udata=0$.
\end{remark}
\begin{lemma}\label{lem:Pnecsuf}
Let $p\in\diamonddata$, let $\Udata_p\subset\diamonddata$
be an open neighborhood of $p$, and 
let $\udata\in\gxdata^1(\Udata_p)$ be such that
\begin{equation} \label{eq:udatatime}
\text{the one-form $d\tau + \anchorg(\udata)(\tau)$ is timelike 
with respect to $\gcyl$ at $p$}
\end{equation}
Then the following are equivalent:
\begin{enumerate}[label=(\roman*)]
\item \label{item:Pu=0}
$\Pconstraints(\udata)=0$ on an open neighborhood of $p$.
\item \label{item:existssol}
There exists an open neighborhood $V_p\subset\diamond_+$ of $p$
and $u\in \gx^1(V_p)$ such that
\begin{align*}
\dg u + \tfrac12[u,u]=0 
\qquad
u|_{\tau=0}=\udata
\end{align*}
\item \label{item:existsgaugedsol}
If  $U_p\subset\diamond_+$ is an open neighborhood of $p$ 
and $v\in \gx^1(U_p)$ with $v|_{\tau=0}=\udata$ then there exists 
an open neighborhood $V_p\subset U_p$ of $p$ and $c \in \gxG^1(V_p)$ with 
\begin{align*}
\dg(v+c) + \tfrac12[v+c,v+c]&=0 \;\;\text{(on $V_p$)}
&
c|_{\tau=0} &= 0
\end{align*}
\end{enumerate}
\end{lemma}
\begin{proof} 
We show:
\begin{itemize}
\item 
\ref{item:existsgaugedsol} $\Rightarrow$ \ref{item:existssol}: 
Use $V_p$ from \ref{item:existsgaugedsol} and set $u=v+c$.
\item 
\ref{item:existssol} $\Rightarrow$ \ref{item:Pu=0}:
On $V_p\cap\diamonddata$ one has
$$\Pconstraints(\udata)
= \Pconstraints(u|_{\tau=0})
=
(\dg\tau + \anchorg(u)(\tau)) (\dg u + \tfrac12[u,u])|_{\tau=0} 
= 0$$
\item 
\ref{item:Pu=0} $\Rightarrow$ \ref{item:existsgaugedsol}:
Consider the necessary subsystem
\begin{align}
\BilG^1\left(\,\cdot\,,d(v+c) + \tfrac12[v+c,v+c]\right) =0 
\qquad
c|_{\tau=0}=0 \label{eq:SHSforc}
\end{align}
This is a quasilinear symmetric hyperbolic system 
(Lemma \ref{lem:translationofeqcyl}, \ref{lem:aAminkcpt}, \ref{lem:deltacausalcpt},
see also \cite[Lemma 42]{Thesis}).
The fact that $\tau=0$ is an admissible initial hypersurface
(i.e.~that the positivity condition for symmetric hyperbolic systems is satisfied)
follows from \eqref{eq:udatatime} (see \eqref{eq:posiffofut_bulk}).
Thus by well-posedness of symmetric
hyperbolic systems \cite[Section 16.1-16.2]{Taylor3}
there exists a solution $c$ on an open neighborhood $V_p\subset U_p$ of $p$.
Define 
\[ 
R = \dg(v+c) + \tfrac12[v+c,v+c] 
\]
This satisfies:
\begin{align*}
R\in\gxG^2(V_p)
&&
\dg R + [v+c,R] = 0
&&
R|_{\tau=0}=0
\end{align*}
The first holds by \eqref{eq:SHSforc} 
and \ref{item:gaugekernelNEW_bulk};
the second holds by \eqref{eq:dgdifferential}, \eqref{eq:dgbracket}, \eqref{eq:bracketas}, \eqref{eq:bracketjacobi};
for the third note that by \ref{item:Pu=0} we have
\[ 
0 = \Pconstraints(\udata) = (d\tau + \anchorg(\udata)(\tau))R |_{\tau=0}
\]
(make $V_p$ smaller if necessary), 
which implies $R|_{\tau=0}=0$ by $R\in\gxG^2(V_p)$, \eqref{eq:udatatime}, 
Lemma \ref{lem:gauge_mainprop_bulk}.
By Lemma \ref{lem:translationofeqcyl}, \ref{lem:aAminkcpt}, \ref{lem:deltacausalcpt}
the linear homogeneous equation 
\[ 
\BilG^2(\,\cdot\,,\dg R + [v+c,R]) = 0
\]
is symmetric hyperbolic, thus $R=0$ (make $V_p$ smaller if necessary).\qed
\end{itemize}
\end{proof}

\section{Abstract semiglobal existence theorem}
	\label{sec:Abstract}

The main task in proving Theorem \ref{thm:main} is to 
control the solution near spacelike infinity $\spaceinf$.
For this it is useful to introduce homogeneous coordinates
$$
\smash{\log(\s),\;\tfrac{y^0}{|\vec{y}|},\;\tfrac{\vec{y}}{|\vec{y}|}}
$$
where $\s=2y^0+|\vec{y}|$ was introduced in \eqref{eq:sdefintro},
which identify (see Figure \ref{fig:HomogeneousCoordinates})
\begin{equation}\label{eq:IIS2}
\Dspop_{\le1}\;\simeq\;(-\infty,0] \times [0,1) \times S^2
\end{equation}
In these coordinates, \eqref{eq:scall} acts by 
translation in the first factor $(-\infty,0]$.

In this section we analyze a class of 
inhomogeneous (i.e.~with source term)
quasilinear symmetric hyperbolic systems on \eqref{eq:IIS2},
with trivial initial data along $(-\infty,0] \times \{0\} \times S^2$,
and where the coefficients satisfy uniformity assumptions 
in the factor $(-\infty,0]$.
We prove existence and uniqueness, and 
relate the asymptotics of the source term
towards $-\infty$ (i.e.~towards $\spaceinf$) to the regularity of the solution 
along $(-\infty,0] \times \{1\} \times S^2$ (i.e.~along $\fnullinf$),
see Theorem \ref{thm:nonlinEE} and \ref{thm:AbstractUniqueness}.

In Section \ref{sec:SpaceinfConstruction},
the result will be applied to the equation \eqref{eq:MCforc}.
For flexibility of the results, 
we will work with a general closed manifold instead of $S^2$.

\begin{figure}
\begin{subfigure}{.5\linewidth}
\centering
\begin{tikzpicture}[inner sep=0pt]
\node (i0) at (0,0) {};
\node (idown) at (4.4,0) {};
\node (iup) at (2.,2.) {};

\node (s1up) at (1.4,1.4) {} ;
\node (s1down) at (4.2,0) {} ;

\node (s2up) at (0.7,0.7) {} ;
\node (s2down) at (2.1,0) {} ;

\node (s3up) at (0.35,0.35) {} ;
\node (s3down) at (1.05,0) {} ;

\node (s4up) at (0.175,0.175) {} ;
\node (s4down) at (0.525,0) {} ;

\node (t0) at ({3., 0.6}) {} ;
\node (t1) at (2.33333, 0.933333) {} ;
\node (t2) at (1.90909, 1.14545) {} ;
\node (t3) at (1.61538,1.29231) {} ;

\node(chartup) at (0,2.3) {};
\node(chartright) at (4.6,0) {};

\fill[gray!50] (i0.center)--(s1up.center)--(s1down.center)--(i0.center);

\path[draw,->] (i0.center)--(chartup.center) node [anchor=south east] {$y^0$};
\path[draw,->] (i0.center)--(chartright.center) node [anchor=north west] {$|\vec{y}|$};

\path[draw,line width=1 pt,blue] (s1up.center)--(s1down.center);
\path[draw,line width=1 pt,blue] (s2up.center)--(s2down.center);
\path[draw,line width=1 pt,blue] (s3up.center)--(s3down.center);
\path[draw,line width=1 pt,blue] (s4up.center)--(s4down.center);

\path[draw,line width=1 pt] (0,0)--(idown.center);
\draw[line width=1.pt,ufogreen] (0,0)--(iup.center);
\node[anchor=south east,yshift=1mm] at (1,1) {$\fnullinf$};

\path[draw,line width=1 pt,red,densely dashed] (i0.center)--(t0.center);
\path[draw,line width=1 pt,red,densely dashed] (i0.center)--(t1.center);
\path[draw,line width=1 pt,red,densely dashed] (i0.center)--(t2.center);
\path[draw,line width=1 pt,red,densely dashed] (i0.center)--(t3.center);

\draw[color=black, fill=white] (i0.center) circle (.05) node[anchor=north east] (i0) {$\spaceinf$};

\end{tikzpicture}
\end{subfigure}
\begin{subfigure}{.5\linewidth}
\centering
\begin{tikzpicture}[inner sep=0pt,scale=0.85]
\node (ldown) at (0,0) {};
\node (0down) at (4.5,0) {};
\node (rdown) at (5,0) {};
\node (lup) at (0,2) {};
\node (0up) at (4.5,2) {};
\node (rup) at (5,2) {};

\node (s1up) at (4.5,2) {} ;
\node (s1down) at (4.5,0) {} ;

\node (s2up) at (3.3,2) {} ;
\node (s2down) at (3.3,0) {} ;

\node (s3up) at (2.1,2) {} ;
\node (s3down) at (2.1,0) {} ;

\node (s4up) at (0.9,2) {} ;
\node (s4down) at (0.9,0) {} ;

\node (t0right) at (4.5, 0.4) {} ;
\node (t0left) at (0, 0.4) {} ;

\node (t1right) at (4.5, 0.8) {} ;
\node (t1left) at (0, 0.8) {} ;

\node (t2right) at (4.5, 1.2) {} ;
\node (t2left) at (0, 1.2) {} ;

\node (t3right) at (4.5, 1.6) {} ;
\node (t3left) at (0, 1.6) {} ;

\path[draw,line width=1 pt] (ldown.center)--(rdown.center);
\path[draw,line width=1 pt,fill=gray!50] (ldown.center)--(0down.center)--(0up.center)--(lup.center);

\node (chartcen) at (0,0) {};
\node (chart0down) at (4.5,0) {};
\node (chart0up) at (4.5,2.5) {};
\node (chartup) at (0,2.3) {};
\node (chartright) at (5.3,0) {};

\path[draw,->] (4.5,-0.1)--(chart0up.center) node [anchor=south east] {$\frac{y^0}{|\vec{y}|}$};
\path[draw,->] (chartcen.center)--(chartright.center) node [anchor=north west] {$\log\s$};

\path[draw,line width=1 pt,blue] (s1up.center)--(s1down.center);
\path[draw,line width=1 pt,blue] (s2up.center)--(s2down.center);
\path[draw,line width=1 pt,blue] (s3up.center)--(s3down.center);
\path[draw,line width=1 pt,blue] (s4up.center)--(s4down.center);

\path[draw,line width=1 pt,ufogreen] (lup.center)--(rup.center);

\path[draw,line width=1 pt,red,densely dashed] (t0right.center)--(t0left.center);
\path[draw,line width=1 pt,red,densely dashed] (t1right.center)--(t1left.center);
\path[draw,line width=1 pt,red,densely dashed] (t2right.center)--(t2left.center);
\path[draw,line width=1 pt,red,densely dashed] (t3right.center)--(t3left.center);
\end{tikzpicture}
\end{subfigure}
\captionsetup{width=115mm}
\caption{The gray shaded domain depicts $\Dspop_{\le1}$, 
with spherical directions suppressed.
On the left it is shown using $y^0$ and $|\vec y|$ as coordinates,
where spacelike infinity $\spaceinf$ is the origin,
on the right it is shown using  $\log\s$ and $y^0/|\vec y|$ as coordinates,
where $\spaceinf$ corresponds to $\log\s=-\infty$.
The four solid lines (blue) are the level sets $\s=0.125,0.25,0.5,1$,
the four dashed lines (red) are the level sets $y^0/|\vec y|=0.2,0.4,0.6,0.8$.}
\label{fig:HomogeneousCoordinates}
\end{figure}
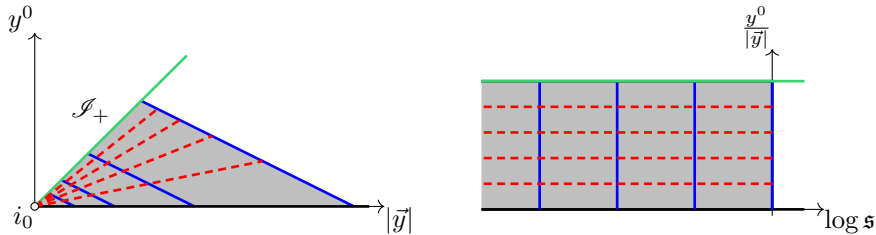

\subsection{Abstract geometric setup}
\label{sec:AbstractGeom}
We introduce the geometric background that 
will be used throughout Section \ref{sec:Abstract}.
Let $\Mcpt$ be a smooth orientable closed manifold
of dimension $\MdimNEW-1\ge1$. 
Define
\begin{equation}\label{eq:homMfddef}
\homMfd =  (-\infty,0]\times [0,1\new{)}\times \Mcpt
\end{equation}
This has dimension $\MdimNEW+1$.
We assume that $\homMfd$ is parallelizable.
We use standard coordinates $(\zzeta,\ttcoord)$ on the first two factors
$(-\infty,0]\times [0,1]$.
For every $\zz\le0$ let 
\begin{align}\label{eq:Mz}
\homMfd_{\zz} \;=\; \zzeta^{-1}(\{\zz\})
\qquad\qquad
\homMfd_{\le\zz} \;=\; \zzeta^{-1}((-\infty,\zz])
\end{align}
be the portion of $\homMfd$ where $\zzeta=\zz$ respectively $\zzeta\le\zz$.
Fix:
\begin{itemize}
\item 
A frame of vector fields 
\begin{equation}\label{eq:abstractX}
X_{0},\dots,X_{\MdimNEW}\in\Gamma(T\bar{\homMfd})
\end{equation}
that are smooth and linearly independent on the closure $\bar{\homMfd}$,
and satisfy:
\begin{subequations}
\label{eq:Xproperties}
\begin{align}
&\text{$X_0,\dots,X_{\MdimNEW}$ are translation invariant in $\zzeta$}
	\label{eq:translinv}\\
&\text{$X_{1},\dots,X_{\MdimNEW}$ are tangential to $\ttcoord^{-1}(\{0\})$}
	\label{eq:X1..Tang}\\
&\text{$X_{1},\dots,X_{\MdimNEW}$ are closed under commutators}
	\label{eq:X1...Comm}
\end{align}
\end{subequations}
where \eqref{eq:X1...Comm} means 
$[X_i,X_j]\in\SPAN_{C^\infty(\bar\homMfd)}\{X_{1},\dots,X_{\MdimNEW}\}$
for $i,j=1\dots\MdimNEW$.
\item 
A density\footnote{
We use the convention that on an $n$-dimensional manifold,
one integrates $n$-densities.
See Remark \ref{rem:densityconv} for the definition when $n=4$.} 
\begin{align}\label{eq:AbstractmuM}
\muM \;\in\; \dens{\MdimNEW+1}(\bar{\homMfd})
\end{align}
that is smooth on the closure $\bar{\homMfd}$
and translation invariant in $\zzeta$. 
Let $\muMform$ be a volume form whose associated density is $\muM$,
unique up to sign. That is, 
$\muMform\in\Omega^{\MdimNEW+1}(\bar{\homMfd})$
and $|\muMform|=\muM$.
On each $\bar{\homMfd}_{\zz}$ we define the density 
\begin{align}\label{eq:AbstractmuM'}
\muMz = |\intermult_{\p_{\zzeta}}\muMform| \in \dens{\MdimNEW}(\bar{\homMfd}_{\zz})
\end{align}
where $\intermult$ denotes the interior multiplication.
\end{itemize}
See Convention \ref{conv:ztXmu_NEW} for a concrete, admissible, choice
in the case of \eqref{eq:IIS2}.

In Section \ref{sec:Abstract}, 
repeated indices $i,j$ will be summed implicitly over $0,\dots,\MdimNEW$,
unless indicated otherwise.

\subsection{Norms}
\label{sec:AbstractNorms}
We introduce the norms that will be used in Section \ref{sec:Abstract}
(some of them are actually seminorms, but we refer to them
as norms for simplicity).

For $\kk\in\Z_{\ge0}$ define 
$$
\indset_{\kk} = \{0,1,\dots,\MdimNEW \}^{\kk}
\qquad
\indset_{\le k} = \cup_{k'\le k}\indset_{k'}$$
Given an index $\indI=(i_1,\dots,i_k)\in\indset_k$ 
and a function $f$ we denote
\begin{align}\label{eq:indexnotation}
X^{\indI} &= X_{i_1}\cdots X_{i_k} & 
|\indI| &= k & 
f_{\indI} = X^{\indI} f = X_{i_1}\cdots X_{i_k}f
\end{align}
One has $\indset_{0}=\{()\}$ where $()$
is the empty tuple, for which $f_{()}=X^{()}f=f$.
We use the convention that if $k<0$ then $\indset_{k}=\{\}$ and
$\indset_{\le k}=\{\}$.

We denote by $\countzero(\indI)$ the number of $i_{1},\dots,i_{k}$
that are equal to $0$ (for example $\countzero((1,3,0,0,1)) = 2$).
For $k_0,k \in\Z_{\ge0}$ define
\begin{align}\label{eq:indsetdef}
\indset_{k_0,k} 
&= 
\{ \indI \in \indset_{k_0+k} \mid \countzero(\indI)=k_0 \}
\end{align}
and $\indset_{k_0,\le k} =\cup_{k'\le k} \indset_{k_0,k'}$ and
$\indset_{\le k_0,\le k} =\cup_{k_0'\le k_0} \indset_{k_0',\le k}$.

\begin{definition}\label{def:sHXDefinition}
For every $k\in\Z_{\ge0}$ 
and $\zz\le0$ and $f\in C^\infty(\homMfd_{\le\zz})$ define:
\begin{align}
\label{eq:tnorms}
\begin{aligned}
\|f\|_{\Ht^k(\homMfd_{\le\zz})}^2 
&= 
\textstyle\sum_{\indI\in\indset_{0,\le k}}\int_{\homMfd_{\le\zz}} |f_{\indI}|^2\,\muM\\
\|f\|_{\Ht^k(\homMfd_{\zz})}^2 
&= 
\textstyle\sum_{\indI\in\indset_{0,\le k}}\int_{\homMfd_{\zz}} |f_{\indI}|^2\,\muMz\\
\|f\|_{\Ct^k(\homMfd_{\le\zz})}
&=
\textstyle\sum_{\indI\in\indset_{0,\le k}} \sup_{p\in\homMfd_{\le\zz}}|f_{\indI}(p)|\\
\|f\|_{\Ct^k(\homMfd_{\zz})}
&=
\textstyle\sum_{\indI\in\indset_{0,\le k}} \sup_{p\in\homMfd_{\zz}}|f_{\indI}(p)|
\end{aligned}
\end{align}
and define
\begin{align}
\label{eq:Xnorms}
\begin{aligned}
\|f\|_{\HX^k(\homMfd_{\le\zz})}^2 
&= 
\textstyle\sum_{\indI\in\indset_{\le k}}\int_{\homMfd_{\le\zz}} |f_{\indI}|^2\, \muM\\
\|f\|_{\sHX^k(\homMfd_{\zz})}^2 
&= 
\textstyle\sum_{\indI\in\indset_{\le k}}\int_{\homMfd_{\zz}} |f_{\indI}|^2\, \muMz\\
\|f\|_{\CX^k(\homMfd_{\le\zz})}
&=
\textstyle\sum_{\indI\in\indset_{\le k}} \sup_{p\in\homMfd_{\le\zz}}|f_{\indI}(p)|\\
\|f\|_{\sCX^k(\homMfd_{\zz})}
&=
\textstyle\sum_{\indI\in\indset_{\le k}} \sup_{p\in\homMfd_{\zz}}|f_{\indI}(p)|
\end{aligned}
\end{align}
using the notation \eqref{eq:indexnotation}.
For the norms on $\homMfd_{\zz}$ we make the same
definition when $f$ is only defined near $\homMfd_{\zz}$.
We make analogous definitions for vector- and matrix-valued functions,
where we apply the norms componentwise 
and then take the $\ell^2$-sum of the components;
and for matrix differential
operators of the form $a^i X_i$,
where we apply the norms to the matrices $a^i$
and then sum over $i$.
For $k\in\Z_{<0}$ we declare \eqref{eq:tnorms}, \eqref{eq:Xnorms} to be zero.
\end{definition}

The norms in \eqref{eq:tnorms}, 
decorated with a $\tang$, 
measure differentiability with respect to
the vector fields $X_1,\dots,X_{\MdimNEW}$
(but not $X_0$), 
which are tangential to $\ttcoord=0$, see \eqref{eq:X1..Tang}.
The norms in \eqref{eq:Xnorms}
measure differentiability with respect
to all vector fields $X_0,\dots,X_{\MdimNEW}$,
which are not all tangential to $\homMfd_{\zz}$.
In particular, the slashed 
norms over the level sets $\homMfd_{\zz}$
are not determined by the restriction of $f$ to $\homMfd_{\zz}$.
\begin{remark}\label{rem:k=0norms}
For $k=0$, the norms in \eqref{eq:tnorms} and the norms
in \eqref{eq:Xnorms} are equal.
Further they are then equal to the standard $L^2$-norms
with respect to the given measures,
respectively the standard $C^0$-norms.
\end{remark}
\begin{remark}\label{rem:lesssim}
We use the standard $\lesssim$ notation from \cite[p.~xiv]{Tao}:
If $X$ and $Y$ are two quantities then the notation $X\lesssim Y$
means that there exists a constant $C>0$ such that $X\le CY$.
If, in addition, $a_1,\dots,a_k$ are parameters then 
$X\lesssim_{a_1,\dots,a_k}Y$ means that there exists a constant
$C(a_1,\dots,a_k)>0$, that depends only on the parameters $a_1,\dots,a_k$,
such that $X\le C(a_1,\dots,a_k)Y$.
\end{remark}
Define the tuple
\begin{equation}\label{eq:CM}
\CM = \left( \Mcpt, X_0,\dots,X_{\MdimNEW}, \muM \right)
\end{equation}
A standard Sobolev estimate, using the fact that $X_0,\dots,X_{\MdimNEW}$ and $\muM$
are translation invariant in $\zzeta$, yields that for all $k\in\Z_{\ge0}$,
and $d\in\Z$ with
$d>\frac{\MdimNEW}{2}$:
\begin{align}
\|f\|_{\sCX^k(\homMfd_{\zz})}
&\lesssim_{\CM,k,d}
\|f\|_{\sHX^{k+d}(\homMfd_{\zz})}
\label{eq:sobolevMz}
\end{align}
where the constant is in particular independent of $\zz$.
\begin{lemma}\label{lem:LinfL2toL1L2}
Let $k\in\Z_{\ge0}$ and $\zz\le0$ and $f\in C^\infty(\bar\homMfd)$.
Then 
\[ 
\|f\|_{\sHX^{k}(\homMfd_{\zz})}
\lesssim_{\CM,k}
\tint_{\zz-1}^{\zz}
\|f\|_{\sHX^{k+1}(\homMfd_{\zz'})}\,d\zz'
\]
If $\zz\le-1$ then the same inequality holds
with $\tint_{\zz-1}^{\zz}$ replaced by $\tint_{\zz}^{\zz+1}$.
\end{lemma}
\begin{proof}
It suffices to prove the inequality for $k=0$.
Let $\phi:\R\to[0,1]$ be a smooth cutoff function that satisfies
$\phi(x)=0$ for $x\le-\frac23$ and $\phi(x)=1$ when $x\ge-\frac13$. 
Then for all $(\zz,\tt,p)\in \homMfd$, where $p\in\Mcpt$,
\begin{align}
f(\zz,\tt,p) 
	&= \tint_{\zz-1}^{\zz} \tfrac{d}{d\zz'}\left(\phi(\zz'-\zz)f(\zz',\tt,p)\right) d\zz'\nonumber
\intertext{and hence}
|f(\zz,\tt,p)| 
	&\lesssim_{\phi} 
	\tint_{\zz-1}^{\zz} (|f(\zz',\tt,p)|+|(\p_{\zzeta}f)(\zz',\tt,p)|) d\zz'\label{eq:hl1}
\intertext{Now Minkowski's inequality for integrals yields}
\|f\|_{L^{2}(\homMfd_{\zz})}
	&\lesssim_{\phi} 
	\tint_{\zz-1}^{\zz}
	(\|f\|_{L^{2}(\homMfd_{\zz'})}+\|\p_{\zzeta}f\|_{L^{2}(\homMfd_{\zz'})}) d\zz'
	\nonumber
\intertext{
where the $L^2$-norm is taken with respect to the measure $\muMz$,
and where we use the fact that $\muMz$ is translation invariant in $\zzeta$.
Thus}
\|f\|_{L^{2}(\homMfd_{\zz})}
	&\lesssim_{\phi,\CM} 
	\tint_{\zz-1}^{\zz}
	\|f\|_{\sHX^{1}(\homMfd_{\zz'})}\,d\zz'
	\nonumber
\end{align}
where we use the fact that 
\eqref{eq:abstractX} span the tangent space at every point,
and where the constant is independent of $\zz$ because 
\eqref{eq:abstractX} are translation invariant in $\zzeta$.

The last statement of the lemma is checked analogously.
\qed
\end{proof}

\subsection{Auxiliary functions associated to linear terms}
\label{sec:AuxQuang}
We introduce linear algebraic quantities 
(Definition \ref{def:ellCcomdef}, \ref{def:defCcom}), 
that will later be used to estimate the propagator
in the energy estimates,
and thus will be important when we relate the 
asymptotics of the source term towards $\zzeta\to-\infty$ 
to the regularity of the solution along the boundary $\ttcoord=1$.

Recall that repeated indices $i,j$ are summed implicitly over $0,\dots,\MdimNEW$, unless indicated otherwise.
For the remainder of this Section \ref{sec:AuxQuang} fix:
\begin{align}\label{eq:fixedauxfct}
\begin{aligned}
\nn&\in\Z_{\ge1}\\
\Ccomp&\in C^\infty(\bar{\homMfd},\R^{\MdimNEW})\\
\LMat&\in C^\infty(\bar{\homMfd},\End(\R^{\nn}))\\
\a^i &\in C^\infty(\bar{\homMfd},S^2\R^{\nn}) \qquad i=0\dots\MdimNEW\\
\a &= \a^i X_i
\end{aligned}
\end{align}
where $S^2\R^{\nn}$ is the space of symmetric $\nn\times\nn$ matrices.
We assume that:
\begin{align}\label{eq:ALassp}
&\text{$d\zzeta(\a) = \a^i d\zzeta(X_i)$ is a positive matrix at every point on $\bar{\homMfd}$.}
\end{align}
The definitions and statements below apply analogously
when \eqref{eq:fixedauxfct} are only defined on $\bar\homMfd_{\le\zzfix}$
for some $\zzfix<0$, this will be indicated.
We use $\bar\homMfd$ for simplicity.

For $\nn\times \nn$ matrices $S$ and $P$, 
with $P$ symmetric and positive definite, 
define
\begin{equation}\label{eq:sigdef}
\sig(S,P) \;=\; \sup_{\substack{v\in\R^{n}\\ v\neq0}} \frac{v^T S v}{v^T P v}
\end{equation}
The matrix $S$ is not assumed to be symmetric,
but $\sig(S,P)$ depends only on the symmetric part of $S$.
\begin{definition}\label{def:ellCcomdef}
Define the function $\LdivestNEW{\LMat,\a}:(-\infty,0]\to\R$ by 
\[ 
\LdivestNEW{\LMat,\a}(\zz) 
= 
\sup_{p\in \homMfd_{\zz}}
\sig\Big(\LMat|_p+\tfrac12\div_{\muM}(\a)|_p\,,\,d\zzeta(\a)|_p\Big)
\]
where $|_p$ is evaluation at $p$. 
Here $\div_{\muM}$ is the divergence with respect to the density $\muM$,
which applied to $\a=\a^i X_i$
yields a symmetric matrix with components 
$(\div_{\muM}(\a))_{jk} = \div_{\muM}(\a^i_{jk}X_i)$.
The function $\LdivestNEW{\LMat,\a}$ is continuous in $\zz$.

We make an analogous definition when \eqref{eq:fixedauxfct} are only defined 
on $\bar\homMfd_{\le\zzfix}$ for some $\zzfix<0$, then 
$\LdivestNEW{\LMat,\a}$ is a function $(-\infty,\zzfix]\to\R$.
\end{definition}
\begin{definition}\label{def:defCcom}
Define the functions
\[
\CcomTang{\a,\Ccomp}:(-\infty,0]\to\R
\qquad\qquad
\CcomTransv{\a,\Ccomp}:(-\infty,0]\to\R
\]
as follows. For each $\zz\le0$:
\begin{itemize}
\item 
$\CcomTang{\a,\Ccomp}(\zz)$ is the smallest real number such that
for all $w\in C^\infty(\new{\bar\homMfd},\R^{\nn})$ the following inequality
holds at every point on $\homMfd_{\zz}$:
\begin{align}\label{eq:CcomTangdef}
-\sum_{i,j=1}^{\MdimNEW} (X_iw)^T \Xd^j([X_i,\a]-\Ccomp_i\a) X_j w
&\;\le\;
\CcomTang{\a,\Ccomp}(\zz) \sum_{i=1}^{\MdimNEW} \|X_iw\|_{\a}^2
\intertext{%
where in $X_iw$ we differentiate componentwise,
where $\Ccomp_i$ are the components of $\Ccomp$,
where $\Xd^0,\dots,\Xd^{\MdimNEW}$ is the frame of 
one-forms that is dual to \eqref{eq:abstractX},
and where we denote $\|v\|_{\a}^2=v^T d\zzeta(\a)v$
for every $v\in C^\infty(\new{\bar\homMfd},\R^{\nn})$.
\item 
$\CcomTransv{\a,\Ccomp}(\zz)$ is the smallest real number such that
for all $w\in C^\infty(\new{\bar\homMfd},\R^{\nn})$ the following inequality
holds at every point on $\homMfd_{\zz}$:
}
\label{eq:CcomTransvdef}
-\sum_{i=1}^{\MdimNEW} (X_iw)^T \Xd^0([X_i,\a]-\Ccomp_i\a) X_0 w
&\;\le\;
\CcomTransv{\a,\Ccomp}(\zz) \sum_{i=1}^{\MdimNEW} \|X_iw\|_{\a}\|X_0w\|_{\a}
\end{align}
\end{itemize}
The expressions \eqref{eq:CcomTangdef} and \eqref{eq:CcomTransvdef}
only depend on $w$ through its first derivatives with respect to the frame
\eqref{eq:abstractX}.
The functions $\CcomTang{\a,\Ccomp},\CcomTransv{\a,\Ccomp}$ 
are continuous in $\zz$.

We make analogous definitions when \eqref{eq:fixedauxfct} are only defined 
on $\bar\homMfd_{\le\zzfix}$ for some $\zzfix<0$, then 
$\CcomTang{\a,\Ccomp}$ and $\CcomTransv{\a,\Ccomp}$ are functions $(-\infty,\zzfix]\to\R$.
\end{definition}
\begin{lemma}\label{lem:kappaupperbound}
Let $\Cpos\ge1$.
For every $\zz\le0$, if 
$\tfrac{1}{\Cpos}\one\le d\zzeta(\a)\le\Cpos\one$ on $\homMfd_{\zz}$ 
then 
\begin{align}\label{eq:comTTest}
|\CcomTransv{\a,\Ccomp}(\zz)|, |\CcomTang{\a,\Ccomp}(\zz)| 
	&\lesssim_{\CM,\nn,\Cpos} 
	(1+\|\Ccomp\|_{\sCX^0(\homMfd_{\zz})})\|\a\|_{\Ct^1(\homMfd_{\zz})}
\end{align}
using $\CM=( \Mcpt, X_0,\dots,X_{\MdimNEW}, \muM)$, see \eqref{eq:CM}, 
and the norms in Definition \ref{def:sHXDefinition}.
Recall that the $\sCX^0$-norm and the $\Ct^0$-norm coincide, 
see Remark \ref{rem:k=0norms}.

The lemma holds analogously when \eqref{eq:fixedauxfct} are only defined 
on $\bar\homMfd_{\le\zzfix}$ for some $\zzfix<0$,
then the constant in \eqref{eq:comTTest} is independent of $\zzfix$.
\end{lemma}
\begin{proof}
By direct inspection.\qed
\end{proof}
With this definition, at every point on $\homMfd_{\zz}$ one has
\begin{align}\label{eq:kappaall}
\begin{aligned}
-\tsum_{i=1}^{\MdimNEW} (X_iw)^T [X_i,\a] w
\le
&\textstyle\sum_{i=1}^{\MdimNEW} |\Ccomp_i| \|X_iw\| \|\a w\|\\
&\textstyle+\CcomTang{\a,\Ccomp}(\zz) \sum_{i=1}^{\MdimNEW} \|X_iw\|_{\a}^2\\
&\textstyle+|\CcomTransv{\a,\Ccomp}(\zz)| \sum_{i=1}^{\MdimNEW} \|X_iw\|_{\a}\|X_0w\|_{\a}
\end{aligned}
\end{align}
This will be useful to control commutators when we derive energy estimates,
essentially $\CcomTransv{\a,\Ccomp}$ and $\CcomTang{\a,\Ccomp}$ extract 
the leading terms with respect to the number of derivatives,
when $w$ solves an equation with principal part $\a w$.
(We note that for Minkowski, $\CcomTransv{\a,\Ccomp}(\zz)$ will be seen to vanish,
and $\CcomTang{\a,\Ccomp}(\zz)$ equals one.)

We define the propagator that will appear in the energy estimates.
\begin{definition}\label{def:PGammaKDef}
Associated to 
$\LMat,\a,\Ccomp$ in \eqref{eq:fixedauxfct},
an integer $\kk\in\Z_{\ge0}$, a real number $C>0$,
and functions $u,\Fvec\in C^\infty(\bar\homMfd,\R^{\nn})$,
for every $\zz_1,\zz_0\le0$ define
\begin{align}\label{eq:PropDef}
\begin{aligned}
&\kprop^{\LMat,\a,\Fvec,\Ccomp}_{\kk,u,C}(\zz_1,\zz_0)
=
\exp\big(\int_{\zz_0}^{\zz_1}
\big(\LdivestNEW{\LMat,\a}(\zz)
+\kk\max\{0,\CcomTang{\a,\Ccomp}(\zz)\}\big)\,d\zz\big)\\
&\quad\times\exp\big(
C\int_{\zz_0}^{\zz_1}
\big(
|\CcomTransv{\a,\Ccomp}(\zz)|
+
\|u\|_{\Ct^{{\lfloor \frac{\kk+1}{2} \rfloor}}(\homMfd_{\zz})}
+
\|\Fvec\|_{\sCX^{\lfloor \frac{\kk+1}{2}\rfloor-1}(\homMfd_{\zz})}
\big)\, d\zz\big)
\end{aligned}
\end{align}
where we use the convention that $\int_{\zz_0}^{\zz_1} =-\int_{\zz_1}^{\zz_0}$
when $\zz_0\ge\zz_1$.

We make an analogous definition when \eqref{eq:fixedauxfct} and $u,\Fvec$ are only defined on $\bar\homMfd_{\le\zzfix}$ for some $\zzfix<0$, then 
\eqref{eq:PropDef} is defined for $\zz_1,\zz_0\le\zzfix$.
\end{definition}
\begin{lemma}\label{lem:ellkappaLip}
Let $\tilde\LMat\in C^\infty(\bar\homMfd,\End(\R^{\nn}))$ and 
$\tilde\a^i\in C^\infty(\bar\homMfd,S^2\R^{\nn})$ and $\Cpos\ge1$.
Denote $\tilde\a=\tilde\a^i X_i$ and assume that this satisfies \eqref{eq:ALassp}.
For all $\zz\le0$, if
\begin{align}\label{eq:posasspaatilde}
\tfrac{1}{\Cpos}\one\le d\zzeta(\a),d\zzeta(\tilde{\a})\le\Cpos\one
\qquad
\text{on $\homMfd_{\zz}$}
\end{align}
then
\begin{subequations}\label{eq:lkappalip}
\begin{align}
|\LdivestNEW{\LMat,\a}(\zz)-\LdivestNEW{\tilde\LMat,\tilde\a}(\zz)|
&\lesssim_{\CM,\nn,\Cpos}
\|\LMat-\tilde\LMat\|_{\sCX^0(\homMfd_{\zz})}\nonumber \\
&\quad +
\big(1+\|\LMat\|_{\sCX^0(\homMfd_{\zz})}+
\|\a\|_{\sCX^1(\homMfd_{\zz})}\big)
\|\a-\tilde\a\|_{\sCX^1(\homMfd_{\zz})} 
\label{eq:LdivestEstimate}\\
|\CcomTang{\a,\Ccomp}(\zz)-\CcomTang{\tilde\a,\Ccomp}(\zz)|
&\lesssim_{\CM,\nn,\Cpos,\|\Ccomp\|_{\sCX^0(\homMfd_{\zz})}}
\big(1 + \|\a\|_{\Ct^1(\homMfd_{\zz})} \big)
\|\a-\tilde\a\|_{\Ct^1(\homMfd_{\zz})} \label{eq:CcomTangEstimate}\\
|\CcomTransv{\a,\Ccomp}(\zz)-\CcomTransv{\tilde\a,\Ccomp}(\zz)|
&\lesssim_{\CM,\nn,\Cpos,\|\Ccomp\|_{\sCX^0(\homMfd_{\zz})}}
\big(1 + \|\a\|_{\Ct^1(\homMfd_{\zz})} \big)
\|\a-\tilde\a\|_{\Ct^1(\homMfd_{\zz})} \label{eq:CcomTransvEstimate}
\end{align}
\end{subequations}
Recall that the $\sCX^0$-norm and the $\Ct^0$-norm coincide, 
see Remark \ref{rem:k=0norms}.

The lemma holds analogously when \eqref{eq:fixedauxfct} and
$\tilde\LMat$, $\tilde\a^i$ are only defined 
on $\bar\homMfd_{\le\zzfix}$ for some $\zzfix<0$,
then the constants in \eqref{eq:lkappalip} are independent of $\zzfix$.
\end{lemma}
\begin{proof}
We show \eqref{eq:LdivestEstimate}:
For $\nn\times\nn$ matrices $S,P$ and $S',P'$ 
with $P,P'$ symmetric and $1/\Cpos\le P,P'\le\Cpos$ one has
\begin{align}\label{eq:sigmadiff}
|\sig(S,P)-\sig(S',P')| 
\le
\Cpos\|S-S'\|+\Cpos^2\|\new{S}\|\|P-P'\|
\end{align}
where $\|{\cdot}\|$ denotes the $\ell^2$-matrix norm.
To check this use \eqref{eq:sigdef},
the inequality $|\sup f-\sup g|\le\sup|f-g|$
for the supremum of two functions $f,g$,
and then add and subtract appropriately.
%
Using $|\sup f-\sup g|\le\sup|f-g|$ again, 
\begin{align*}
&|\LdivestNEW{\LMat,\a}(\zz)
-
\LdivestNEW{\tilde\LMat,\tilde\a}(\zz)|\\
&\le
\sup_{p\in \homMfd_{\zz}}
\Big|\sig\Big(\LMat|_p+\tfrac12\div_{\muM}(\a)|_p\,,\,d\zzeta(\a)|_p\Big)
-
\sig\Big(\tilde\LMat|_p+\tfrac12\div_{\muM}(\tilde\a)|_p\,,\,d\zzeta(\tilde\a)|_p\Big)
\Big|
\end{align*}
Now \eqref{eq:sigmadiff}, applicable by \eqref{eq:posasspaatilde},  yields 
\begin{align*}
|\LdivestNEW{\LMat,\a}(\zz)-\LdivestNEW{\tilde\LMat,\tilde\a}(\zz)|
&\le
\Cpos\sup_{p\in\homMfd_{\zz}}
\Big(\|
(\LMat|_p-\tilde\LMat|_p)
+
\tfrac12\div_{\muM}(\a-\tilde\a)|_p\| \Big)\\
&\qquad+
\Cpos^2\sup_{p\in\homMfd_{\zz}}
\Big(
\|\LMat|_p+\tfrac12\div_{\muM}(\a)|_p\|\|d\zzeta(\a-\tilde\a)|_p\|\Big)
\end{align*}
Clearly both terms are bounded by the right hand side of \eqref{eq:LdivestEstimate}.

We show \eqref{eq:CcomTangEstimate}: 
Let $w\in C^\infty(\bar\homMfd,\R^{\nn})$ such that 
$\sum_{i=1}^{\MdimNEW} \|X_iw\|^2$
is nowhere zero on $\homMfd_{\zz}$. 
Abbreviate the two terms in \eqref{eq:CcomTangdef} by
\begin{align*}
f_{\a}(w) &= - \tsum_{i,j=1}^{\MdimNEW} (X_iw)^T \Xd^j([X_i,\a]-\Ccomp_i\a) X_j w\\
h_{\a}(w) &= \tsum_{i=1}^{\MdimNEW} \|X_iw\|_{\a}^2
\end{align*}
and analogously for $f_{\tilde\a}(w)$, $h_{\tilde\a}(w)$.
At every point on $\homMfd_{\zz}$, using \eqref{eq:posasspaatilde},
\begin{align*}
	|f_{\a}(w)-f_{\tilde\a}(w)|,
	|h_{\a}(w)-h_{\tilde\a}(w)|
	\lesssim_{\CM,\nn,\Cpos,\|\Ccomp\|_{\sCX^0(\homMfd_{\zz})}}
	\|\a-\tilde\a\|_{\Ct^1(\homMfd_{\zz})}  h_{\a}(w)
\end{align*}
Set
$R_{\a}(w) = \tfrac{f_{\a}(w)}{h_{\a}(w)}$.
Using $R_{\a}(w)\le \CcomTang{\a,\Ccomp}$, 
and analogously $R_{\tilde\a}(w)\le \CcomTang{\tilde\a,\Ccomp}$,
\begin{align}\label{eq:Rtilde}
\begin{aligned}
R_{\tilde\a}(w) 
&
\le 
|R_{\tilde\a}(w) - R_{\a}(w)| + \CcomTang{\a,\Ccomp}\\
R_{\a}(w) 
&\le
|R_{\tilde\a}(w) - R_{\a}(w)| + \CcomTang{\tilde\a,\Ccomp}
\end{aligned}
\end{align}
Adding and subtracting yields
\begin{align*}
R_{\tilde\a}(w) - R_{\a}(w)
&=\textstyle
\frac{f_{\tilde\a}(w)-f_{\a}(w)}{h_{\tilde\a}(w)}
+
\frac{f_{\a}(w)}{h_{\a}(w)}
\frac{h_{\a}(w)-h_{\tilde\a}(w)}{h_{\tilde\a}(w)}
\end{align*}
Hence, with \eqref{eq:posasspaatilde}, we obtain
\begin{align}
|R_{\a}(w)-R_{\tilde\a}(w)|
&\lesssim_{\CM,\nn,\Cpos,\|\Ccomp\|_{\sCX^0(\homMfd_{\zz})}}
\big(
\tfrac{|h_{\a}(w)|}{|h_{\tilde\a}(w)|} + \tfrac{|f_{\a}(w)|}{|h_{\tilde\a}(w)|} 
\big)
\|\a-\tilde\a\|_{\Ct^1(\homMfd_{\zz})} \nonumber\\
&\lesssim_{\CM,\nn,\Cpos,\|\Ccomp\|_{\sCX^0(\homMfd_{\zz})}}
(1 + \|\a\|_{\Ct^1(\homMfd_{\zz})} )
\|\a-\tilde\a\|_{\Ct^1(\homMfd_{\zz})} \label{eq:boundkkp}
\end{align}
The right hand side is independent of $w$.
Together with \eqref{eq:Rtilde} this implies 
that $\CcomTang{\tilde\a,\Ccomp}-\CcomTang{\a,\Ccomp}$
respectively $\CcomTang{\a,\Ccomp}-\CcomTang{\tilde\a,\Ccomp}$
are bounded by \eqref{eq:boundkkp}, 
thus implying \eqref{eq:CcomTangEstimate}.

The inequality \eqref{eq:CcomTransvEstimate} is checked similarly.
\qed
\end{proof}

\subsection{Quasilinear energy estimate}
\label{sec:EnergyEstimate}

We derive the a priori energy estimates (Lemma \ref{lem:Nonlindiffenergyineq}), 
that will be used in the proof of Theorem \ref{thm:nonlinEE}.
The energy estimates control the $\Ht^{k}(\homMfd_{\zz})$ norms, 
which only differentiate with respect to the 
vector fields $X_1,\dots,X_{\MdimNEW}$, 
which are tangential to the initial hypersurface $\ttcoord=0$
(see Definition \ref{def:sHXDefinition}).
Derivatives with respect to $X_0$ are then controlled using the equation.
The results in this section are under compact support assumptions
for the source term and solution.
This will be sufficient for Theorem \ref{thm:nonlinEE}
by a finite speed of propagation argument.

In the following lemma we derive auxiliary estimates for a linear system.
Note that $f\in C^\infty_c(\bar\homMfd)$ iff $f$ is smooth on $\bar\homMfd$
and $f|_{\homMfd_{\le\zz}}=0$ for some $\zz\le0$.
Recall that repeated indices $i,j$ are 
summed implicitly over $0\dots\MdimNEW$.
\begin{lemma}[Linear energy estimate]\label{lem:Lindiffenergyineq}
Let $\Mcpt, X_0,\dots,X_{\MdimNEW}, \muM$
be as in Section \ref{sec:AbstractGeom}.
For all 
$\nn\in\Z_{\ge1}$,
all 
$\Cpos\ge1$ 
and all
\begin{align*}
u &\in C^\infty_c(\bar\homMfd,\R^{\nn})\\
\a^i &\in C^\infty(\bar\homMfd,S^2\R^{\nn}) \qquad i=0\dots\MdimNEW \\
\LMat&\in C^\infty(\bar\homMfd,\End(\R^{\nn})) \\
\Fvec&\in C^\infty_c(\bar\homMfd,\R^{\nn})\\
\Ccomp&\in C^\infty(\bar\homMfd,\R^{\MdimNEW})
\end{align*}
the following holds.
If 
\begin{subequations}\label{eq:LinAssp}
\begin{align}
\a^i X_i u &= \LMat u + \Fvec \label{eq:LinearSHS}
\\
\qquad\qquad u|_{\ttcoord=0} &=0\label{eq:LinearUdata}\\
\qquad\qquad\qquad\Cpos^{-1}\one &\le d\zzeta(\a^i X_i) \le \Cpos \one && 
\text{on $\bar\homMfd$}\qquad\label{eq:A(zeta)pos}\\
\qquad\qquad\qquad\Cpos^{-1}\one &\le \a^0 \le \Cpos \one && 
\text{on $\bar\homMfd$}\qquad\label{eq:A(0)pos}\\
0&\le d\ttcoord(\a^i X_i) && \text{on $\ttcoord^{-1}(\{1\})$}
\label{eq:A(t)pos}
\end{align}
\end{subequations}
then:
\begin{itemize}
\item 
\textbf{Part 1.}
For all $k_0\in\Z_{\ge1}$, $k\in\Z_{\ge0}$, $\indI \in \indset_{k_0,k}$,
at every point on $\bar\homMfd$:
\begin{align}\label{eq:indX0Lin}
\begin{aligned}
\|u_{\indI}\|
&\lesssim_{\CM,\nn,k_0,k,\Cpos}
\tsum_{\substack{\indJ \in \indset_{\le k_0,\le k}\\
	|\indJ|\le k_0+k-1}} \|u_{\indJ}\|
+
\tsum_{\indJ\in\indset_{\le k_0-1,\le k}}\|\Fvec_{\indJ}\|
\\
&\qquad+
\big(\tsum_{i=0}^{\MdimNEW}\|\a^i\|\big)
\big(
\sum_{\indJ \in \indset_{\le k_0-1,\le k+1}} \|u_{\indJ}\|\big)\\
&\qquad
+\big(
\tsum_{\indJ\in \indset_{\le k_0-1,\le k}} \|\LMat_{\indJ}\|
\big)
\big(
\tsum_{\indJ\in\indset_{\le k_0-1,\le k}} \|u_{\indJ}\|
\big)\\
&\qquad
+
\tsum_{\substack{
\indJ,\indK\in\indset_{\le k_0,\le k+1}\\
|\indJ|+|\indK|\le k_0+k\\
|\indJ|,|\indK|\le k_0+k-1\\
\countzero(\indJ)+\countzero(\indK) \le k_0
}}
(\tsum_{i=0}^{\MdimNEW}\|\a^i_{\indJ}\|)\|u_{\indK}\|
\end{aligned}
\end{align}
where $\|{\cdot}\|$ is the $\ell^2$-vector respectively matrix norm,
where we use the notation \eqref{eq:indexnotation},
where in 
$u_{\indI} = X^{\indI}u$, 
$\a^{i}_{\indJ}=X^{\indJ}\a^i$, 
$\LMat_{\indJ}=X^{\indJ}\LMat$
we differentiate componentwise,
and where $\CM=(\Mcpt, X_0,\dots,X_{\MdimNEW}, \muM)$, see \eqref{eq:CM}.
\item 
\textbf{Part 2.}
For all $k_0\in\{0,1\}$, $k\in\Z_{\ge0}$ and $\zz\le0$ 
define the energies
$$
E_{k_0,k}(\zz) = 
\sum_{\indI\in\indset_{k_0,k}}\int_{\homMfd_{\zz}}
u_{\indI}^T d\zzeta(\a^i X_i) u_{\indI}\,\muMz
\qquad
E_{k_0,\le k}(\zz) = \sum_{k'\le k} E_{k_0,k'}(\zz)
$$
using the density $\muMz$ in \eqref{eq:AbstractmuM'}.
Then 
\begin{align}
&\tfrac{d}{d\zz}E_{0,\kk}(\zz)
-
2\Big(\LdivestNEW{\LMat,\a}(\zz)+\kk\CcomTang{\a,\Ccomp}(\zz)\Big) E_{0,\kk}(\zz)
	\label{eq:LinearEE}\\
	&\;\;\lesssim_{\CM,\nn,\kk,\Cpos,\|\Ccomp\|_{\sCX^0(\homMfd_{\zz})}}
	\sqrt{E_{0,\kk}}
	\Big(
	|\CcomTransv{\a,\Ccomp}| \sqrt{E\smash{_{1,k-1}}}
	+
	\|\LMat\|_{\Ct^{\kk}(\homMfd_{\zz})} \sqrt{E_{0,\le \kk-1}}\Big)\nonumber \\
	&\phantom{\;\;\lesssim_{\CM,\nn,\kk,\Cpos,\Ccomp}}
	+\sqrt{E_{0,\kk}} \big(\tsum_{\indI\in\indset_{0,\le\kk}}
	\tint_{\homMfd_{\zz}}\|\Fvec_{\indI}\|^2 \muMz\big)^{\frac12}\nonumber\\
	&\phantom{\;\;\lesssim_{\CM,\kk,\nn,\Cpos,\Ccomp}}
	+\tsum_{\substack{
	\indI,\indJ\in \indset_{0,\le k}\\
	\indK\in \indset_{0,\le k}\cup \indset_{1,\le k-1}\nonumber\\
	|\indJ|+|\indK| \le k+1 \\
	|\indI|+|\indK| \le 2k-1
	}}\sum_{i=0}^{\MdimNEW}
	\tint_{\homMfd_{\zz}}\|\a^{i}_{\indJ}\| \|u_{\indI}\| \|u_{\indK}\|\,\muMz\nonumber
\end{align}
where, on the right hand side, the evaluation at $\zz$ is implicit,
and where we abbreviate $\a=\a^i X_i$.
Note that $E_{k_0,-1}=0$ and $E_{k_0,\le -1}=0$.
\end{itemize}
\end{lemma}
\begin{proof}
\proofheader{Proof of Part 1.}
Since $k_0\ge1$ we can write 
\[ 
\|u_{\indI}\|
\lesssim_{\CM,k_0,k}
\textstyle
\sum_{\indJ\in\indset_{k_0-1,k}}\|X_0 u_{\indJ}\|
+
\sum_{\substack{\indJ\in\indset_{\le k_0,\le k}\\
	|\indJ|\le k_0+k-1}}
	\| u_{\indJ}\|
\]
We show that for every $\indJ\in\indset_{k_0-1,k}$
the term $\|X_0 u_{\indJ}\|$ is bounded by the 
right hand side of \eqref{eq:indX0Lin}.
Differentiating the equation \eqref{eq:LinearSHS} with respect to 
$X^{\indJ}$ yields
\begin{align*}
\a u_{\indJ}  
	&= 
	- 
	[X^{\indJ},\a] u
	+
	X^{\indJ} \LMat u
	+ 
	\Fvec_{\indJ}
\intertext{
Write $\a = \a^0X_0 + \sum_{i=1}^{\MdimNEW}\a^i X_i$, 
and put the second term on the right hand side.
By \eqref{eq:A(0)pos} the matrix $\a^0$ is invertible, thus}
X_0 u_{\indJ}  
	&= 
	(\a^0)^{-1}
	\big(
	- 
	[X^{\indJ},\a] u
	-
	\tsum_{i=1}^{\MdimNEW}\a^i X_i u_{\indJ}  
	+
	X^{\indJ} \LMat u
	+ 
	\Fvec_{\indJ}
	\big)
\intertext{By \eqref{eq:A(0)pos} we have $(\a^0)^{-1}\le\Cpos\one$, hence }
\|X_0 u_{\indJ}\|  
	&\lesssim_{\Cpos} 
	\|[X^{\indJ},\a] u\|
	+
	\tsum_{i=1}^{\MdimNEW}\|\a^i\| \|X_i u_{\indJ}\|
	+
	\|X^{\indJ}\LMat u\|
	+ 
	\|\Fvec_{\indJ}\|
\end{align*}
It is now easy to see that 
each term is bounded by the right hand side of \eqref{eq:indX0Lin}.

\proofheader{Proof of Part 2.}
For $\kk\in\Z_{\ge0}$, $\indI\in\indset_{0,\kk}$ define\footnote{%
Beware that the index $\indI$ is used in two different ways,
in $u_{\indI}$ it stands for the derivative of $u$ (see \eqref{eq:indexnotation}),
while in $\current_{\indI}$ and $E_{\indI}$ it is part of the name.} 
\begin{align*}
\current_{\indI} &= (u_{\indI}^T \a^i u_{\indI}) X_i
&
E_{\indI}(\zz)
=
\tint_{\homMfd_{\zz}}
u_{\indI}^T d\zzeta(\a) u_{\indI}\,\muMz
\end{align*}
The vector field $\current_{\indI}$ is the current.
Note that $E_{0,\kk} = \tsum_{\indI\in\indset_{0,\kk}}E_{\indI}$.

In the following we will use Stokes' theorem.
For this we fix a volume form $\muMform\in\Omega^{\MdimNEW+1}(\homMfd)$
such that the density associated to $\muMform$ is $\muM$, that is, 
\[ 
|\muMform| = \muM
\]
We fix an orientation on $\homMfd$ such that $\muMform$ is positive.
Note that $|\intermult_{\p_{\zzeta}}\muMform| = \muMz$,
see \eqref{eq:AbstractmuM'}.
Then, by definition of the divergence,
\[ 
\div_{\muM}(\current_{\indI})\muMform = \ddR\left( \intermult_{\current_{\indI}}\muMform \right)
\]
where $\intermult_{\current_{\indI}}$ is the interior multiplication
with $\current_{\indI}$.
The current $\current_{\indI}$ has compact support, because $u$ has.
Thus integrating over $\homMfd_{\le\zz}$ and using Stokes' theorem,
\begin{align}
\tint_{\homMfd_{\le\zz}} \div_{\muM}(\current_{\indI})\muMform
&=
\tint_{\homMfd_{\zz}} \intermult_{\current_{\indI}}\muMform \nonumber \\
&\quad+
	\tint_{\bar\homMfd_{\le\zz}\cap\ttcoord^{-1}(\{0\})}
	\intermult_{\current_{\indI}}\muMform \label{eq:t=0term}\\
&\quad+
	\tint_{\bar\homMfd_{\le\zz}\cap\ttcoord^{-1}(\{1\})} 
	\intermult_{\current_{\indI}}\muMform \label{eq:t=1term}
\end{align}
where, on the right hand side, we use the induced orientation.
Observe:
\begin{itemize}
\item 
The left hand side equals 
$\tint_{\homMfd_{\le\zz}} \div_{\muM}(\current_{\indI})\muM$,
an integral relative to the density $\muM$ over the 
unoriented $\homMfd_{\le\zz}$.
\item 
$\tint_{\homMfd_{\zz}} \intermult_{\current_{\indI}}\muMform
=
\tint_{\homMfd_{\zz}} (u_{\indI}^T d\zzeta(\a) u_{\indI})
\intermult_{\p_{\zzeta}}\muMform
=E_{\indI}(\zz)$,
where we use $|\intermult_{\p_{\zzeta}}\muMform|=\muMz$
and the fact that $\intermult_{\p_{\zzeta}}\muMform$ 
is positive with respect to the induced orientation.
\item 
The term \eqref{eq:t=0term} vanishes vanishes
by \eqref{eq:LinearUdata} and \eqref{eq:X1..Tang}. 
\item 
The term \eqref{eq:t=1term} is increasing in $\zz$.
To see this, note that it is equal to
\[ 
\tint_{\homMfd_{\le\zz}\cap\ttcoord^{-1}(\{1\})} 
u_{\indI}^T d\ttcoord(\a) u_{\indI}\, \intermult_{\p_{\ttcoord}}\muMform
\]
The form $\intermult_{\p_{\ttcoord}}\muMform$ is positive
with respect to the induced orientation.
By \eqref{eq:A(t)pos} we have $u_{\indI}^T d\ttcoord(\a) u_{\indI}\ge0$.
Hence \eqref{eq:t=1term} is increasing in $\zz$ as claimed.
\end{itemize}
Thus differentiating in $\zz$, and using Fubini and $\muM=|d\zzeta|\muMz$, yields
\begin{align}\label{eq:diffE}
\tfrac{d}{d\zz}E_{\indI}(\zz)
&\le\textstyle
\int_{\homMfd_{\zz}} \div_{\muM}(\current_{\indI})\,\muMz
\end{align}
By the symmetry assumption $(\a^i)^T=\a^i$, and the 
Leibniz rule for the divergence,
\[ 
\div_{\muM}(\current_{\indI})
=
2 u_{\indI}^T \a u_{\indI} + u_{\indI}^T \div_{\muM}(\a) u_{\indI}
\]
Differentiating \eqref{eq:LinearSHS} with respect to $X^{\indI}$ yields
\begin{align}\label{eq:diffeq}
\a u_{\indI}  
	&= 
	\LMat u_{\indI} - [X^{\indI},\a] u + [X^{\indI},\LMat]u + \Fvec_{\indI}
\end{align}
Thus
\begin{align*}
\div_{\muM}(\current_{\indI})
&=
u_{\indI}^T (2\LMat+\div_{\muM}(\a)) u_{\indI} 
- 2u_{\indI}^T [X^{\indI},\a] u 
+ 2u_{\indI}^T [X^{\indI},\LMat]u 
+ 2u_{\indI}^T \Fvec_{\indI}
\end{align*}
Plugging this into \eqref{eq:diffE} and summing over $\indI\in\indset_{0,k}$ yields
\begin{subequations}
\begin{align}
\tfrac{d}{d\zz}E_{k}(\zz)
&\le
\tsum_{\indI\in\indset_{0,k}}\tint_{\homMfd_{\zz}}
u_{\indI}^T (2\LMat+\div_{\muM}(\a)) u_{\indI}\,\muMz\label{eq:LinEE1} \\
&\;\;-
\tsum_{\indI\in\indset_{0,k}}\tint_{\homMfd_{\zz}}
2u_{\indI}^T [X^{\indI},\a] u\,\muMz \label{eq:LinEE2}\\
&\;\;+ 
\tsum_{\indI\in\indset_{0,k}}\tint_{\homMfd_{\zz}}
2u_{\indI}^T [X^{\indI},\LMat]u \,\muMz\label{eq:LinEE3}\\
&\;\;+ 
\tsum_{\indI\in\indset_{0,k}}\tint_{\homMfd_{\zz}}
2u_{\indI}^T \Fvec_{\indI}\,\muMz\label{eq:LinEE4}
\end{align}
\end{subequations}
We estimate the four terms:
\begin{itemize}
\item 
\eqref{eq:LinEE1}:
By Definition \ref{def:ellCcomdef}, 
for each $\indI$ and 
at every point on $\homMfd_{\zz}$ we have
$$u_{\indI}^T (\LMat+\tfrac12\div_{\muM}(\a)) u_{\indI} 
\le
\LdivestNEW{\LMat,\a}(\zz) u_{\indI}^Td\zzeta(\a)u_{\indI}$$
Thus \eqref{eq:LinEE1} is bounded by 
$\le2\LdivestNEW{\LMat,\a}(\zz) E_{0,k}(\zz)$.

\item
\eqref{eq:LinEE4}:
By Cauchy Schwarz, for each $\indI$ we have 
\begin{align*}
\tint_{\homMfd_{\zz}} 2u_{\indI}^T \Fvec_{\indI}\,\muMz
&\le
2 
\big(\tint_{\homMfd_{\zz}} \|u_{\indI}\|^2\,\muMz\big)^{\frac12}
\big(\tint_{\homMfd_{\zz}} \|\Fvec_{\indI}\|^2 \,\muMz\big)^{\frac12}
\intertext{By \eqref{eq:A(zeta)pos} we have 
$\|u_{\indI}\|^2\lesssim_{\Cpos} u_{\indI}^Td\zzeta(\a)u_{\indI}$,
hence this is bounded by}
&\lesssim_{\Cpos}
\sqrt{E_{\indI}(\zz)}
\big(\tint_{\homMfd_{\zz}} \|\Fvec_{\indI}\|^2 \,\muMz\big)^{\frac12}
\end{align*}
Thus \eqref{eq:LinEE4} is bounded by 
$\lesssim_{k,\Cpos} 
\sqrt{E_{0,k}(\zz)} 
\big(
\tsum_{\indI\in\indset_{0,k}}
\tint_{\homMfd_{\zz}}\|\Fvec_{\indI}\|^2 \muMz\big)^{\frac12}$.

\item 
\eqref{eq:LinEE3}:
By Cauchy Schwarz, for each $\indI$ we have
\[ 
\tint_{\homMfd_{\zz}}
2u_{\indI}^T [X^{\indI},\LMat]u\, \muMz
\le
2
\big(\tint_{\homMfd_{\zz}}\|u_{\indI}\|^2 \,\muMz\big)^{\frac12}
\big(\tint_{\homMfd_{\zz}}\|[X^{\indI},\LMat]u\|^2\,\muMz\big)^{\frac12}
\]
The commutator $[X^{\indI},\LMat]$ is lower order, in the sense that
\[ 
\|[X^{\indI},\LMat]u\|
\lesssim_{\CM,\nn,k}
\|\LMat\|_{\Ct^{k}(\homMfd_{\zz})}
\tsum_{\indJ\in\indset_{0,\le k-1}} \|u_{\indJ}\|
\]
Using \eqref{eq:A(zeta)pos} we obtain that 
\eqref{eq:LinEE3} is bounded by 
$$
\lesssim_{\CM,\nn,k,\Cpos}
\|\LMat\|_{\Ct^{k}(\homMfd_{\zz})}
\sqrt{E\smash{_{0,k}}(\zz)} \sqrt{E\smash{_{0,\le k-1}}(\zz)}
$$

\item 
\eqref{eq:LinEE2}:
We first estimate
$- 2 \tsum_{\indI\in\indset_{0,k}} u_{\indI}^T [X^{\indI},\a] u$
pointwise on $\homMfd_{\zz}$,
using Definition \ref{def:defCcom}.
Define the $\R$-trilinear differential operators
\begin{align*}
B_k(\underline{\mataux},v,w)
&=
-2k \tsum_{i=1}^{\MdimNEW}\sum_{\indI\in\indset_{0,k-1}}
(X_{i}v_{\indI})^T[X_{i},\mataux^j X_j]w_{\indI}\\
\tilde{B}_k(\underline{\mataux},v,w)
&=
- 2 \tsum_{\indI\in\indset_{0,k}} v_{\indI}^T [X^{\indI},\mataux^j X_j] w 
- B_k(\underline{\mataux},v,w)
\end{align*}
where $\underline{\mataux}=(\mataux^{j})_{j=0\dots\MdimNEW}\in C^\infty(\bar\homMfd,(\R^{\nn\times\nn})^{\MdimNEW+1})$,
$v,w\in C^\infty(\bar\homMfd,\R^{\nn})$. 
Then
\begin{equation}\label{eq:bbtilde}
- 2 \tsum_{\indI\in\indset_{0,k}} u_{\indI}^T [X^{\indI},\a] u
=
B_k(\underline{\a},u,u)
+
\tilde{B}_k(\underline{\a},u,u)
\end{equation}
with $\underline{\a}=(\a^{i})_{i=0\dots\MdimNEW}$.
We estimate the two terms on the right hand side separately.
The term $\tilde{B}_k$ is lower order, more precisely one has
\begin{align*}
\tilde{B}_k(\underline{\mataux},v,w)
=\textstyle
\sum_{\substack{
\indI,\indJ\in\indset_{0,\le k}\\
\indK\in \indset_{0,\le k}\cup \indset_{1,\le k-1}\\
|\indJ|+|\indK| \le k+1 \\
|\indI|+|\indK| \le 2k-1
}}
\tilde{B}_{k,\indI\indJ\indK}(\underline{\mataux}_{\indJ},v_{\indI},w_{\indK})
\end{align*}
using \eqref{eq:X1...Comm}, where 
each $\tilde{B}_{k,\indI\indJ\indK}$ is a $C^\infty$-trilinear form
that satisfies
$$
|\tilde{B}_{k,\indI\indJ\indK}(\underline{\mataux},v,w)|
\lesssim_{\CM,\nn,k}
\|\underline{\mataux}\| \|v\| \|w\|
$$
Thus at every point on $\homMfd_{\zz}$,
\begin{equation} \label{eq:Btildeest}
|\tilde{B}_k(\underline{\a},u,u)|
\lesssim_{\CM,\nn,k}\textstyle
\sum_{\substack{
\indI,\indJ\in \indset_{0,\le k}\\
\indK\in \indset_{0,\le k}\cup \indset_{1,\le k-1}\\
|\indJ|+|\indK| \le k+1 \\
|\indI|+|\indK| \le 2k-1
}}
\|\underline{\a}_{\indJ}\| \|u_{\indI}\| \|u_{\indK}\|
\end{equation}
Consider $B_k$. By Definition \ref{def:defCcom},
see also \eqref{eq:kappaall},
for each $\indI\in\indset_{0,k-1}$ and 
at every point on $\homMfd_{\zz}$,
\begin{align*}
-\tsum_{i=1}^{\MdimNEW} (X_i u_{\indI})^T [X_i,\a] u_{\indI}
	&\le\textstyle
	\sum_{i=1}^{\MdimNEW} |\Ccomp_i| \|X_iu_{\indI}\| \|\a u_{\indI}\|\\
	&\textstyle\quad+
	\CcomTang{\a,\Ccomp}(\zz) \sum_{i=1}^{\MdimNEW} \|X_iu_{\indI}\|_{\a}^2\\
	&\textstyle\quad+
	|\CcomTransv{\a,\Ccomp}(\zz)| \sum_{i=1}^{\MdimNEW} \|X_iu_{\indI}\|_{\a}\|X_0u_{\indI}\|_{\a}
\end{align*}
We multiply this inequality with $2k$, 
and sum over $\indI\in\indset_{0,k-1}$.
Then the left hand side yields $B_k(\underline{\a},u,u)$,
and the second term on the right hand side yields
$2k\CcomTang{\a,\Ccomp}(\zz) \sum_{\indI\in\indset_{0,k}}\|u_{\indI}\|_{\a}^2$.
Thus 
\begin{align*}
&B_{k}(\underline{\a},u,u) - 2k\CcomTang{\a,\Ccomp}(\zz) 
\tsum_{\indI\in\indset_{0,k}} \|u_{\indI}\|_{\a}^2\\
	&\qquad\le
	2k\tsum_{\indI\in\indset_{0,k-1}}\sum_{i=1}^{\MdimNEW} |\Ccomp_i| \|X_iu_{\indI}\| \|\a u_{\indI}\|\\
	&\qquad\textstyle\quad+
	2k|\CcomTransv{\a,\Ccomp}(\zz)| \tsum_{\indI\in\indset_{0,k-1}}\sum_{i=1}^{\MdimNEW} \|X_iu_{\indI}\|_{\a}\|X_0u_{\indI}\|_{\a}\\
	&\qquad\lesssim_{k,\Cpos}
	\|\Ccomp\|
	(\tsum_{\indI\in\indset_{0,k}} \|u_{\indI}\|) 
	(\tsum_{\indI\in\indset_{0,k-1}}\|\a u_{\indI}\|)\\
	&\qquad\quad+
	|\CcomTransv{\a,\Ccomp}(\zz)|
	(\tsum_{\indI\in\indset_{0,k}} \|u_{\indI}\|)
	(\tsum_{\indI\in\indset_{1,k-1}} \|u_{\indI}\|)
\end{align*}
For all $\indI\in\indset_{0,k-1}$ (see \eqref{eq:diffeq}),
$$\|\a u_{\indI}\|
\le
\|X^{\indI}\LMat u\|
+
\|[X^{\indI},\a] u\| + \|\Fvec_{\indI}\|$$
where, at every point on $\homMfd_{\zz}$,
\begin{align*}
\|X^{\indI}\LMat u\|
&\lesssim_{\CM,\nn,k}
\|\LMat\|_{\Ct^{k-1}(\homMfd_{\zz})}
\tsum_{\indK\in\indset_{0,\le k-1}} \|u_{\indK}\|\\
\|[X^{\indI},\a] u\|
&\lesssim_{\CM,\nn,k}
\tsum_{\substack{\indJ\in\indset_{0,\le k-1}\\
\indK \in\indset_{0,\le k-1}\cup\indset_{1,\le k-2} \\ |\indJ|+|\indK|\le k}} 
\|\underline{\a}_{\indJ}\| \|u_{\indK}\|
\end{align*}
Thus at every point on $\homMfd_{\zz}$,
\begin{align*}
&B_{k}(\underline{\a},u,u)
-
2k \CcomTang{\a,\Ccomp}(\zz) \tsum_{\indI\in\indset_{0,k}} \|u_{\indI}\|_{\a}^2\\
&\;\;\lesssim_{\CM,\nn,k}
	(\tsum_{\indI\in\indset_{0,k}} \|u_{\indI}\|)
	\Big[
	|\CcomTransv{\a,\Ccomp}(\zz)|
	(\tsum_{\indI\in\indset_{1,k-1}} \|u_{\indI}\|)\\
&\qquad\qquad	
	+
	\|\Ccomp\|
	\Big(
	\|\LMat\|_{\Ct^{k-1}(\homMfd_{\zz})} \tsum_{\indI\in\indset_{0,\le k-1}} \|u_{\indI}\| +
	\sum_{\indI\in\indset_{0,k-1}}\|\Fvec_{\indI}\|\Big)
\\
&\qquad\qquad+
	\|\Ccomp\|
	\Big(\tsum_{\substack{\indJ\in\indset_{0,\le k-1}\\
	\indK \in\indset_{0,\le k-1}\cup\indset_{1,\le k-2} \\ |\indJ|+|\indK|\le k}} 
	\|\underline{\a}_{\indJ}\| 
	\|u_{\indK}\|\Big)
	\Big]
\end{align*}
With \eqref{eq:Btildeest}
this gives an estimate for \eqref{eq:bbtilde}.
Integrating over $\homMfd_{\zz}$ and using Cauchy Schwarz and \eqref{eq:A(zeta)pos},
we obtain the following bound for \eqref{eq:LinEE2}:
\begin{align*}
&-   
\tsum_{\indI\in\indset_k} 
\tint_{\homMfd_{\zz}}
2u_{\indI}^T [X^{\indI},\a] u\,\muMz
-
2k \CcomTang{\a,\Ccomp}(\zz)E_{k}(\zz)\\
&\quad\lesssim_{\CM,\nn,k,\Cpos,\|\Ccomp\|_{\sCX^0(\homMfd_{\zz})}}
\sqrt{E\smash{_{0,k}}(\zz)}
|\CcomTransv{\a,\Ccomp}(\zz)| \sqrt{E\smash{_{1,k-1}}}\\
&\quad\phantom{\lesssim_{\CM,k,\nn,\new{\|\Ccomp\|_{\sCX^0(\homMfd_{\zz})}},\Cpos}}+
\sqrt{E\smash{_{0,k}}(\zz)}
\|\LMat\|_{\Ct^{k-1}(\homMfd_{\zz})} \sqrt{E\smash{_{0,\le k-1}}(\zz)}\\
&\quad\phantom{\lesssim_{k,\nn,\new{\|\Ccomp\|_{\sCX^0(\homMfd_{\zz})}},\Cpos,\CM}}+
\sqrt{E\smash{_{0,k}}(\zz)}\big(\tsum_{\indI\in\indset_{0,k-1}}
\tint_{\homMfd_{\zz}}\|\Fvec_{\indI}\|^2\,\muMz\big)^{\frac12}
\\
&\quad\phantom{\lesssim_{k,\nn,\new{\|\Ccomp\|_{\sCX^0(\homMfd_{\zz})}},\Cpos,\CM}}
+\tsum_{\substack{
\indI,\indJ\in \indset_{0,\le k}\\
\indK\in \indset_{0,\le k}\cup \indset_{1,\le k-1}\\
|\indJ|+|\indK| \le k+1 \\
|\indI|+|\indK| \le 2k-1
}}
\tint_{\homMfd_{\zz}}\|\underline{\a}_{\indJ}\| \|u_{\indI}\| \|u_{\indK}\|\,\muMz
\end{align*}
\end{itemize}
Collecting terms yields \eqref{eq:LinearEE}.\qed
\end{proof}
In the following lemma we use the norms in Definition \ref{def:sHXDefinition}.
Recall that, by definition,
$\HX^{-1}$, $\sHX^{-1}$, $\CX^{-1}$, $\sCX^{-1}$ are zero.
\begin{lemma}[Quasilinear energy estimate]\label{lem:Nonlindiffenergyineq}
Let $\Mcpt, X_0,\dots,X_{\MdimNEW}, \muM$
be as in Section \ref{sec:AbstractGeom}.
For all 
\begin{equation}\label{eq:constEE}
\nn\in\Z_{\ge1}
\qquad
\NN\in\Z_{\ge1}
\qquad
\Cpos\ge1
\qquad
\CLA>0
\end{equation}
there exists a real number $\CqEE>0$ such that for all 
$\zzfix\le0$ and all 
\begin{align}\label{eq:inputQEE}
\begin{aligned}
u &\in  C^\infty_c(\bar\homMfd_{\le\zzfix},\R^{\nn})\\
\AmatLin^i &\in C^\infty(\bar\homMfd_{\le\zzfix},S^2\R^{\nn}) 
	&& i=0\dots\MdimNEW \\
\qquad\qquad
\AmatBil^i &\in C^\infty(\bar\homMfd_{\le\zzfix},\Hom(\R^{\nn}, S^2\R^{\nn}) ) 
	&& i=0\dots\MdimNEW \qquad\\
\LMat&\in C^\infty(\bar\homMfd_{\le\zzfix},\End(\R^{\nn})) \\
\BSHS&\in C^\infty(\bar\homMfd_{\le\zzfix},\Hom(\R^{\nn}\otimes\R^{\nn},\R^{\nn})) \\
\Fvec&\in C^\infty_c(\bar\homMfd_{\le\zzfix}, \R^{\nn})\\
\Ccomp&\in C^\infty(\bar\homMfd_{\le\zzfix}, \R^{\MdimNEW}) 
\end{aligned}
\end{align}
the following holds.
If%
\begin{subequations}\label{eq:nonlinassp}
\begin{align}
&(\AmatLin^{i} + \AmatBil^i(u))X_i u = \LMat u + \BSHS(u,u)+\Fvec \label{eq:SHSBil}\\
& u|_{\ttcoord=0} = 0 \label{eq:DataBil}\\
&\Cpos^{-1}\one \le d\zzeta((\AmatLin^{i} + \AmatBil^i(u))X_i)\le \Cpos \one
\quad\text{on $\bar\homMfd_{\le\zzfix}$}
\label{eq:A(zeta)posBIL}\\
&\Cpos^{-1}\one \le \AmatLin^0 + \AmatBil^0(u)\le \Cpos \one
\qquad\quad\;\quad\text{on $\bar\homMfd_{\le\zzfix}$}
\label{eq:A(0)posBIL}\\
&0 \le d\ttcoord((\AmatLin^{i} + \AmatBil^i(u))X_i) \qquad\qquad\quad \text{on $\ttcoord^{-1}(\{1\})\cap \bar\homMfd_{\le\zzfix}$}\label{eq:A(t)posBIL}\\
&
\new{\|u\|_{\Ct^{\lfloor\frac{\NN+1}{2}\rfloor}(\homMfd_{\le\zzfix})}} \le \CLA 
\label{eq:Cubound}\\
&
\new{\|\Fvec\|_{\CX^{\lfloor\frac{\NN}{2}\rfloor-1}}(\homMfd_{\le\zzfix})}
 \le \CLA 
\label{eq:CFbound}\\
&
\begin{aligned}
&
\|\AmatLin^i\|_{\CX^{\NN}(\homMfd_{\le\zzfix})},
\|\AmatBil^i\|_{\CX^{\NN}(\homMfd_{\le\zzfix})},\\
&\quad
\|\LMat\|_{\CX^{\NN}(\homMfd_{\le\zzfix})},
\|\BSHS\|_{\CX^{\NN}(\homMfd_{\le\zzfix})},
\|\Ccomp\|_{\CX^0(\homMfd_{\le\zzfix})}
\le
\CLA
\end{aligned}
\label{eq:matbounds}
\end{align}
\end{subequations}
then:
\begin{itemize}
\item \textbf{Part 1.}
For all $k\in\Z_{\ge0}$ with $k\le\NN$ and all $\zz\le\zzfix$,
\begin{subequations}
\begin{align}
\|u\|_{\sCX^k(\homMfd_{\zz})} 
	&\le \CqEE \big(
	\|u\|_{\Ct^k(\homMfd_{\zz})}
	+
	\|\Fvec\|_{\sCX^{k-1}(\homMfd_{\zz})}
	\big)
	\label{eq:CallCt}\\
\|u\|_{\sHX^k(\homMfd_{\zz})} 
	&\le \CqEE \big(
	\|u\|_{\Ht^k(\homMfd_{\zz})}
	+
	\|\Fvec\|_{\sHX^{k-1}(\homMfd_{\zz})}
	\big)
	\label{eq:HallHt}
\intertext{
Furthermore, if 
$k+\lfloor\frac{\MdimNEW}{2}\rfloor+1\le\NN$ then}
\|u\|_{\sCX^k(\homMfd_{\zz})} 
	&\le \CqEE \big( \|u\|_{\Ht^{k+\lfloor \frac{\MdimNEW}{2}\rfloor+1}(\homMfd_{\zz})} 
	+ 
	\|\Fvec\|_{\sHX^{k+\lfloor \frac{\MdimNEW}{2}\rfloor}(\homMfd_{\zz})} 
	\big)
	\label{eq:CtSobolev}
\end{align}
\end{subequations}

\item \textbf{Part 2.}
For all $\zz\le\zzfix$,
\begin{align*}
\|u\|_{\Ht^{\NN}(\homMfd_{\zz})}
\le
\CqEE \int_{-\infty}^{\zz} 
\kprop^{\LMat,\Amat(u),\Fvec,\Ccomp}_{\NN,u,\CqEE}(\zz,\zz')
(1 + |\zz-\zz'|)^{\NN}
\|\Fvec\|_{\sHX^{\NN}(\homMfd_{\zz'})}
\,d\zz'
\end{align*}
where we abbreviate $\Amat(u)=(\AmatLin^{i} + \AmatBil^i(u))X_i$,
and use Definition \ref{def:PGammaKDef}.
\end{itemize}
\end{lemma}
Beware that in the estimate in Part 2, 
the solution $u$ still appears in the propagator on the right hand side.
In the proof of Theorem \ref{thm:nonlinEE}, under a priori assumptions, 
we estimate this propagator by a term that is independent of $u$.
\begin{proof}
It suffices to prove the lemma for $\zzfix=0$
(for general $\zzfix\le0$ apply the $\zzfix=0$ statement to the translation
of \eqref{eq:inputQEE} by $\zzfix$).
Instead of specifying $\CqEE$ up front, 
we will make finitely many admissible largeness assumptions on 
$\CqEE$ during the proof, where admissible means that they depend only on 
\eqref{eq:constEE} and on $\CM$, see \eqref{eq:CM}.
We will abbreviate $\lesssim_{\CM,\nn,\NN,\Cpos,\CLA}$ by $\lesssim_*$.

We will use Lemma \ref{lem:Lindiffenergyineq} 
with the parameters in Table \ref{tab:ApplyLinearEE}.
The assumptions \eqref{eq:LinAssp} of 
Lemma \ref{lem:Lindiffenergyineq} hold by 
\eqref{eq:SHSBil},
\eqref{eq:DataBil},
\eqref{eq:A(zeta)posBIL},
\eqref{eq:A(0)posBIL},
\eqref{eq:A(t)posBIL}.

\begin{table}
\centering
\begin{tabular}{c|c}
\begin{tabular}{@{}c@{}}
Parameters \\ 
in Lemma \ref{lem:Lindiffenergyineq}
\end{tabular}
	& 
	\begin{tabular}{@{}c@{}}
	Parameters  \\ used to invoke Lemma \ref{lem:Lindiffenergyineq}
	\end{tabular}\\
\hline
$\Mcpt$, $X_0,\dots,X_{\MdimNEW}$, $\muM$
	& $\Mcpt$, $X_0,\dots,X_{\MdimNEW}$, $\muM$ \\
$\nn$, $\Cpos$
	& $\nn$, $\Cpos$\\
$u$
	& $u$\\
$\a^i$
	& $\AmatLin^i + \AmatBil^i(u)$\\
$\LMat$
	& $\LMat$\\
$\Fvec$
	& $\BSHS(u,u) + \Fvec$\\
$\Ccomp$ 
	& $\Ccomp$ 
\end{tabular}
\captionsetup{width=115mm}
\caption{
The first column lists the input parameters of Lemma \ref{lem:Lindiffenergyineq}.
The second column specifies the choice of these parameters
when invoking Lemma \ref{lem:Lindiffenergyineq} 
in the proof of Lemma \ref{lem:Nonlindiffenergyineq},
in terms of 
the input parameters of Lemma \ref{lem:Nonlindiffenergyineq}.
For example, the input parameter $\a^i$ in Lemma \ref{lem:Lindiffenergyineq}
is chosen to be $\AmatLin^i + \AmatBil^i(u)$, 
using $\AmatLin^i$, $\AmatBil^i$, $u$ from \eqref{eq:inputQEE}.}
\label{tab:ApplyLinearEE}
\end{table}

\proofheader{Proof of Part 1.}
For all $k_0\in\Z_{\ge1}$ and $j\in\Z_{\ge0}$ with $k_0+j\le \NN$,
and all $\indI\in\indset_{k_0,j}$, at every point on $\bar\homMfd$
one has
\begin{align}\label{eq:uptw}
\begin{aligned}
\|u_{\indI}\|
&\lesssim_{*}
\tsum_{\indJ \in \indset_{\le k_0-1,\le j+1}} \|u_{\indJ}\|
+
\tsum_{\indJ\in\indset_{k_0,\le j-1}} \|u_{\indJ}\| 
\\
&\qquad+
\tsum_{\substack{
\indJ,\indK\in\indset_{\le k_0,\le j+1}\\
|\indJ|+|\indK|\le k_0+j\\
|\indJ|,|\indK|\le k_0+j-1\\
\countzero(\indJ)+\countzero(\indK) \le k_0
}}
\|u_{\indJ}\|\|u_{\indK}\|
+
\tsum_{\indJ\in\indset_{\le k_0-1,\le j}}\|\Fvec_{\indJ}\|
\end{aligned}
\end{align}
This follows from Part 1 of Lemma \ref{lem:Lindiffenergyineq} 
(see Table \ref{tab:ApplyLinearEE}),
where we replace the letter $k$ by the letter $j$,
and using \eqref{eq:Cubound} and \eqref{eq:matbounds}.

Proof of \eqref{eq:CallCt}:
It suffices to show that for every element in the set of tuples
\begin{align}\label{eq:tups}
\{(k_0,j)\in\Z_{\ge0}\times\Z_{\ge0} \mid k_0+j\le\NN\}
\end{align}
the following statement $S_{(k_0,j)}$ is true:
\begin{align}
\label{eq:Pk0kt}
S_{(k_0,j)}:\quad
\begin{aligned} 
&\text{For all $\indI\in\indset_{k_0,j}$ and all $\zz\le0$:}\\
&\|u_{\indI}\|_{\sCX^0(\homMfd_{\zz})}\lesssim_{*} \|u\|_{\Ct^{k_0+j}(\homMfd_{\zz})} + \|\Fvec\|_{\sCX^{k_0+j-1}(\homMfd_{\zz})}
\end{aligned}
\end{align}
We prove this by induction, where 
we order the tuples \eqref{eq:tups} lexicographically:
\[ 
(k_0,j) \le (k_0',j')
\quad\Leftrightarrow\quad
k_0<k_0' \; \text{or}\; (k_0=k_0'\;\text{and}\;j\le j')
\]
Clearly $S_{(0,j)}$ holds for all $j\le\NN$.
Now let $(k_0,j)$ be a tuple in \eqref{eq:tups} with $k_0\ge1$,
and assume by induction that $S_{(k_0',j')}$ holds for all $(k_0',j')<(k_0,j)$. Let $\indI\in \indset_{k_0,j}$.
We use \eqref{eq:uptw} for the index $\indI$. 
Observe that on the right hand side of \eqref{eq:uptw}, each $\indJ,\indK$ 
is an element in some $\indset_{(k_0',j')}$ with $(k_0',j')<(k_0,j)$.
Thus, using the induction hypothesis,
\begin{align*}
&\|u_{\indI}\|_{\sCX^0(\homMfd_{\zz})}
\lesssim_{*}
\|u\|_{\Ct^{k_0+j}(\homMfd_{\zz})} + \|\Fvec\|_{\sCX^{k_0+j-1}(\homMfd_{\zz})}\\
&\qquad+
\tsum_{\substack{i,i'\le k_0+j-1\\ 
i+i'\le k_0+j}}
\big(\|u\|_{\Ct^{i}(\homMfd_{\zz})} + \|\Fvec\|_{\sCX^{i-1}(\homMfd_{\zz})}\big)
\big(\|u\|_{\Ct^{i'}(\homMfd_{\zz})} + \|\Fvec\|_{\sCX^{i'-1}(\homMfd_{\zz})}\big)
\end{align*}
The term in the second line is bounded by 
\begin{align*}
&\lesssim_{\NN}
\big(\|u\|_{\Ct^{\lfloor\frac{\NN}{2}\rfloor}(\homMfd_{\zz})} + \|\Fvec\|_{\sCX^{\lfloor\frac{\NN}{2}\rfloor-1}(\homMfd_{\zz})}\big)
\big(\|u\|_{\Ct^{k_0+j-1}(\homMfd_{\zz})} + \|\Fvec\|_{\sCX^{k_0+j-2}(\homMfd_{\zz})}\big)\\
&\lesssim_{\CLA}
\|u\|_{\Ct^{k_0+j-1}(\homMfd_{\zz})} + \|\Fvec\|_{\sCX^{k_0+j-2}(\homMfd_{\zz})}
\end{align*}
using \eqref{eq:Cubound} and \eqref{eq:CFbound} in the last inequality.
This concludes the induction step. 
Now \eqref{eq:CallCt} follows
under an admissible largeness assumption on $\CqEE$.

Proof of \eqref{eq:HallHt}: 
This checked similarly to \eqref{eq:CallCt}.
Here one shows by induction that for each tuple in \eqref{eq:tups},
the following statement is true:
\begin{align}\label{eq:inductionHk}
S_{(k_0,j)}:\quad
\begin{aligned} 
&\text{For all $\indI\in\indset_{k_0,j}$ and all $\zz\le0$:}\\
&\|u_{\indI}\|_{L^2(\homMfd_{\zz})}\lesssim_{*} \|u\|_{\Ht^{k_0+j}(\homMfd_{\zz})} + \|\Fvec\|_{\sHX^{k_0+j-1}(\homMfd_{\zz})}
\end{aligned}
\end{align}
where the $L^2$-norm is defined with respect to the density $\muMz$.
We sketch the induction step.
Let $(k_0,j)$ be in \eqref{eq:tups} with $k_0\ge1$
and assume that $S_{(k_0',j')}$ holds for all $(k_0',j')<(k_0,j)$.
Using \eqref{eq:uptw} and the induction hypothesis,
\begin{align*}
\|u_{\indI}\|_{L^2(\homMfd_{\zz})}
&\lesssim_{*}
\|u\|_{\Ht^{k_0+j}(\homMfd_{\zz})} + \|\Fvec\|_{\sHX^{k_0+j-1}(\homMfd_{\zz})}\\
&\qquad\qquad\qquad+
\tsum_{\substack{
\indJ,\indK\in\indset_{\le k_0,\le j+1}\\
|\indJ|+|\indK|\le k_0+j\\
|\indJ|,|\indK|\le k_0+j-1\\
\countzero(\indJ)+\countzero(\indK) \le k_0
}}
\|\|u_{\indJ}\|\|u_{\indK}\|\|_{L^2(\homMfd_{\zz})}
\end{align*}
The term in the second line is bounded by 
\begin{align*}
&\lesssim_{*}
\|u\|_{\sCX^{\lfloor\frac{\NN}{2}\rfloor}(\homMfd_{\zz})}
\tsum_{\substack{
\indK\in\indset_{\le k_0,\le j+1}\\
|\indK|\le k_0+j-1
}}
\|u_{\indK}\|_{L^2(\homMfd_{\zz})}\\
&\lesssim_{*}
\big(\|u\|_{\Ct^{\lfloor\frac{\NN}{2}\rfloor}(\homMfd_{\zz})}
	+
	\|\Fvec\|_{\sCX^{\lfloor\frac{\NN}{2}\rfloor-1}(\homMfd_{\zz})}\big)
\big(\|u\|_{\Ht^{k_0+j-1}(\homMfd_{\zz})}
+\|F\|_{\sHX^{k_0+j-2}(\homMfd_{\zz})}\big)\\
&\lesssim_{\CLA}
\|u\|_{\Ht^{k_0+j-1}(\homMfd_{\zz})}
+\|F\|_{\Ht^{k_0+j-2}(\homMfd_{\zz})}
\end{align*}
where to bound the first factor we use \eqref{eq:CallCt}
and then \eqref{eq:Cubound} and \eqref{eq:CFbound};
and to bound the second factor we use the induction hypothesis.
This proves \eqref{eq:HallHt}.

We check \eqref{eq:CtSobolev}:
By the Sobolev inequality \eqref{eq:sobolevMz}
and $\lfloor\frac{\MdimNEW}{2}\rfloor+1>\frac{\MdimNEW}{2}$,
\begin{align*}
\|u\|_{\sCX^k(\homMfd_{\zz})} 
\lesssim_{\CM,\nn,\NN}
\|u\|_{\sHX^{k+\lfloor\frac{\MdimNEW}{2}\rfloor+1}(\homMfd_{\zz})} 
\end{align*}
Now \eqref{eq:HallHt} (applicable by the additional assumption on $k$)
yields \eqref{eq:CtSobolev},
under an admissible largeness condition on $\CqEE$.

\proofheader{Proof of Part 2.}
%
For $k_0\in\{0,1\}$ and $\kk\le\NN$ and $\zz\le0$ define
\begin{align*}
E_{k_0,\kk}(\zz) &= \textstyle
\sum_{\indI\in\indset_{k_0,\kk}}\int_{\homMfd_{\zz}}
u_{\indI}^T d\zzeta(\Amat(u)) u_{\indI}\,\muMz
\\
E_{k_0,\le\kk}(\zz) &= \textstyle\sum_{\kk'\le\kk} E_{k_0,\kk'}(\zz)
\\
e_{\le\kk}(\zz) &= \sqrt{E\smash{_{0,\le\kk}}(\zz)}
\end{align*}
Recall that we abbreviate 
$\lesssim_{\CM,\nn,\NN,\Cpos,\CLA}$ by $\lesssim_*$.
By \eqref{eq:A(zeta)posBIL} and \eqref{eq:HallHt},
\begin{align}\label{eq:E1k-1}
\begin{aligned}
\sqrt{E\smash{_{1,\le\kk-1}}(\zz)}
&\lesssim_*
\|u\|_{\sHX^{\kk}(\homMfd_{\zz})}
\lesssim_*
\|u\|_{\Ht^{\kk}(\homMfd_{\zz})}+\|\Fvec\|_{\sHX^{\kk-1}(\homMfd_{\zz})}\\
&\lesssim_*
e_{\le\kk}(\zz)+\|\Fvec\|_{\sHX^{\kk-1}(\homMfd_{\zz})}\\
\sqrt{E\smash{_{1,\le\kk-2}}(\zz)}
&\lesssim_*
e_{\le\kk-1}(\zz)+\|\Fvec\|_{\sHX^{\kk-2}(\homMfd_{\zz})}
\end{aligned}
\end{align}

\claimheader{Claim:}
There exists a real number $\CqEE_0>0$ that depends only on 
\eqref{eq:constEE} and on $\CM$,
such that for all $\kk\le\NN$
and $\zz\le0$:
\begin{align}\label{eq:EnergyEstimateENEW}
\tfrac{d}{d\zz}E_{0,\le\kk}(\zz)
&\le
2\Big(
\LdivestNEW{\LMat,\Amat(u)}(\zz)
+\NN \max\{0, \CcomTang{\Amat(u),\Ccomp}(\zz)\}\Big) E_{0,\le\kk}(\zz)\nonumber\\
&\ +
2\CqEE_0 
\Big(
|\CcomTransv{\Amat(u),\Ccomp}(\zz)|
+
\|u\|_{\Ct^{{\lfloor \frac{\NN+1}{2} \rfloor}}(\homMfd_{\zz})}
+
\|\Fvec\|_{\sCX^{\lfloor \frac{\NN+1}{2}\rfloor-1}(\homMfd_{\zz})}
\Big) E_{0,\le\kk}(\zz)\nonumber\\
&\ +
2\CqEE_0 \Big(e_{\le\kk-1}(\zz)
+
2\CqEE_0 \|\Fvec\|_{\sHX^{\NN}(\homMfd_{\zz})}\Big)e_{\le\kk}(\zz)
\end{align}

\claimheader{Proof of claim:}
By Part 2 of Lemma \ref{lem:Lindiffenergyineq} (see Table \ref{tab:ApplyLinearEE}),
for all $\kk\le\NN$, $\zz\le0$:
\begin{align*}
\tfrac{d}{d\zz}E_{0,\kk}(\zz)
-
&2\left(\LdivestNEW{\LMat,\Amat(u)}(\zz)+\kk\CcomTang{\Amat(u),\Ccomp}(\zz)\right) E_{0,\kk}(\zz)
\;\lesssim_{*}\;
W_0+\dots+W_5
\end{align*}
where, suppressing evaluation at $\zz$,
\begin{align*}
W_0 
&=
\sqrt{E_{0,\kk}} \sqrt{E_{1,\kk-1}} |\CcomTransv{\Amat(u),\Ccomp}|\\
W_1
&=
\|\LMat\|_{\Ct^{\kk}(\homMfd_{\zz})} \sqrt{E_{0,\kk}}  e_{\le \kk-1} 
\le \CLA \sqrt{E_{0,\kk}}  e_{\le \kk-1} \\
W_2
&=
\sqrt{E_{0,\kk}} \big(\tsum_{\indI\in\indset_{0,\le\kk}}
\tint_{\homMfd_{\zz}}\|X^{\indI}(\BSHS(u,u))\|^2 \muMz\big)^{\frac12}\\
W_3
&=
\sqrt{E_{0,\kk}} \big(\tsum_{\indI\in\indset_{0,\le\kk}}
\tint_{\homMfd_{\zz}}\|\Fvec_{\indI}\|^2 \muMz\big)^{\frac12}
\lesssim_{\nn,\NN}
\sqrt{E_{0,\kk}} \|\Fvec\|_{\Ht^{\kk}(\homMfd_{\zz})}
\\
W_4
&=
\tsum_{\substack{
\indI,\indJ\in \indset_{0,\le \kk} \\
\indK\in\indset_{0,\le k}\cup\indset_{1,\le k-1}\\
|\indJ|+|\indK| \le \kk+1 \\
|\indI|+|\indK| \le 2\kk-1
}}\sum_{i=0}^{\MdimNEW}
\tint_{\homMfd_{\zz}}\|X^{\indJ}(\AmatLin^{i})\| \|u_{\indI}\| \|u_{\indK}\|\,\muMz\\
W_5
&=
\tsum_{\substack{
\indI,\indJ\in \indset_{0,\le \kk} \\
\indK \in \indset_{0,\le\kk}\cup\indset_{1,\le \kk-1}\\
|\indJ|+|\indK| \le \kk+1 \\
|\indI|+|\indK| \le 2\kk-1
}}\sum_{i=0}^{\MdimNEW}
\tint_{\homMfd_{\zz}}\|X^{\indJ}(\AmatBil^i(u))\| \|u_{\indI}\| \|u_{\indK}\|\,\muMz
\end{align*}
where we use \eqref{eq:matbounds} to absorb the dependency
of the constant on $\Ccomp$ into $\CLA$, 
and where in the estimate for $W_1$ we use \eqref{eq:matbounds}.
We estimate the terms separately:
\begin{itemize}
\item $W_0$:
Using \eqref{eq:E1k-1},
\begin{align*}
W_{0}
&\lesssim_*
e_{\le\kk}|\CcomTransv{\Amat(u),\Ccomp}|(e_{\le\kk}
+
\|\Fvec\|_{\sHX^{\kk-1}(\homMfd_{\zz})})\\
&\lesssim
e_{\le\kk}^2 |\CcomTransv{\Amat(u),\Ccomp}|
+
e_{\le\kk}\|\Fvec\|_{\sHX^{\kk-1}(\homMfd_{\zz})}
\end{align*}
where in the last step we use 
\[ 
|\CcomTransv{\Amat(u),\Ccomp}|
\lesssim_*
(1+\|\Ccomp\|_{\sCX^0(\homMfd_{\zz})})\|\Amat(u)\|_{\Ct^1(\homMfd_{\zz})}
\lesssim_*
1
\]
by Lemma \ref{lem:kappaupperbound}
and \eqref{eq:matbounds} (using $\NN\ge1$) and \eqref{eq:Cubound}.

\item $W_2$: 
For each $\indI\in\indset_{0,\le k}$ and 
at every point on $\homMfd_{\zz}$ we have
\begin{align*}
\|X^{\indI}\BSHS(u,u)\|
&\lesssim_*
\tsum_{\substack{\indJ,\indK\in\indset_{0,\le k} \\ 
|\indJ|+|\indK|\le k}} \|u_{\indJ}\|\|u_{\indK}\|
\lesssim_*
\|u\|_{\Ct^{{\lfloor \frac{\NN}{2} \rfloor}}(\homMfd_{\zz})}
\tsum_{\indJ\in\indset_{0,\le k}} \|u_{\indJ}\|
\end{align*}
using \eqref{eq:matbounds}.
Using \eqref{eq:A(zeta)posBIL} we obtain
$
W_2 
\lesssim_*
\|u\|_{\Ct^{{\lfloor \frac{\NN}{2} \rfloor}}(\homMfd_{\zz})}
e_{\le\kk}^2
$.

\item $W_4$:
By \eqref{eq:matbounds}, Cauchy Schwarz and \eqref{eq:A(zeta)posBIL},
\begin{align*}
W_4
&\lesssim_*
\tsum_{\substack{
\indI,\indK\in \indset_{0,\le \kk} \\
|\indI|+|\indK| \le 2\kk-1
}} 
(\tint_{\homMfd_{\zz}} \|u_{\indI}\|^2\,\muMz)^{\frac12}
(\tint_{\homMfd_{\zz}} \|u_{\indK}\|^2\,\muMz)^{\frac12}\\
&\quad+
\tsum_{\substack{
\indI\in\indset_{0,\le \kk}\\
\indK\in \indset_{1,\le \kk-1} \\
|\indI|+|\indK| \le 2\kk-1
}} 
(\tint_{\homMfd_{\zz}} \|u_{\indI}\|^2\,\muMz)^{\frac12}
(\tint_{\homMfd_{\zz}} \|u_{\indK}\|^2\,\muMz)^{\frac12}\\
&\lesssim_*
e_{\le\kk} e_{\le\kk-1}
+
\sqrt{E_{1,\le\kk-1}} e_{\le\kk-1}
+
\sqrt{E_{1,\le\kk-2}} e_{\le\kk}
\end{align*}
By \eqref{eq:E1k-1},
\begin{align*}
\sqrt{E_{1,\le\kk-1}} e_{\le\kk-1}
&\lesssim_*
e_{\le\kk}e_{\le\kk-1}+\|\Fvec\|_{\sHX^{\kk-1}(\homMfd_{\zz})}e_{\le\kk-1}\\
\sqrt{E_{1,\le\kk-2}} e_{\le\kk}
&\lesssim_*
e_{\le\kk-1}e_{\le\kk}+\|\Fvec\|_{\sHX^{\kk-2}(\homMfd_{\zz})}e_{\le\kk}
\end{align*}
Thus
$
W_4
\lesssim_*
e_{\le\kk} e_{\le\kk-1}
+\|\Fvec\|_{\sHX^{\kk}(\homMfd_{\zz})}e_{\le\kk}
$.

\item $W_5$:
Using \eqref{eq:matbounds} we obtain $W_5\lesssim_* W_{5,1}+W_{5,2}$ where
\begin{align*}
W_{5,1}
&=
\tsum_{\substack{
\indI,\indJ,\indK\in \indset_{0,\le \kk} \\
|\indJ|+|\indK| \le \kk+1 
}} 
\tint_{\homMfd_{\zz}}
\|u_{\indJ}\| \|u_{\indI}\| \|u_{\indK}\|\,\muMz\\
W_{5,2}
&=
\tsum_{\substack{
\indI,\indJ\in \indset_{0,\le \kk} \\
\indK\in\indset_{1,\le\kk-1}\\
|\indJ|+|\indK| \le \kk+1 
}} 
\tint_{\homMfd_{\zz}}
\|u_{\indJ}\| \|u_{\indI}\| \|u_{\indK}\|\,\muMz
\end{align*}
Using Cauchy Schwarz and \eqref{eq:A(zeta)posBIL},
\begin{align*}
W_{5,1}&\lesssim_*
e_{\le\kk}
\tsum_{\substack{
\indJ,\indK\in \indset_{0,\le \kk} \\
|\indJ|+|\indK| \le \kk+1
}} 
\big(\tint_{\homMfd_{\zz}}\|u_{\indJ}\|^2  \|u_{\indK}\|^2\,\muMz\big)^{\frac12}\\
&\lesssim_*
e_{\le\kk}^2\|u\|_{\Ct^{\lfloor \frac{\NN+1}{2}\rfloor}(\homMfd_{\zz})}\\
W_{5,2}
	&\lesssim_*
	e_{\le\kk}
	\tsum_{\substack{
	\indJ\in \indset_{0,\le \kk} \\
	\indK\in\indset_{1,\le\kk-1}\\
	|\indJ|+|\indK| \le \kk+1 
	}} 
	\big(\tint_{\homMfd_{\zz}}
	\|u_{\indJ}\|^2 \|u_{\indK}\|^2\,\muMz\big)^{\frac12}\\
	&\lesssim_*
	e_{\le\kk}
	\big(
	\|u\|_{\Ct^{\lfloor \frac{\NN+1}{2}\rfloor}(\homMfd_{\zz})}
	\sqrt{E_{1,\le\kk-1}}
	+
	\|u\|_{\sCX^{\lfloor \frac{\NN+1}{2}\rfloor}(\homMfd_{\zz})}
	e_{\le\kk}
	\big)
\intertext{
Using \eqref{eq:E1k-1}, \eqref{eq:CallCt} and \eqref{eq:Cubound},}
W_{5,2}
	&\lesssim_*
	e_{\le\kk}
	\|u\|_{\Ct^{\lfloor \frac{\NN+1}{2}\rfloor}(\homMfd_{\zz})}
	\big(e_{\le\kk}+\|\Fvec\|_{\sHX^{\kk-1}(\homMfd_{\zz})}\big)
	\\
	&\quad+
	e_{\le\kk}^2
	\big(\|u\|_{\Ct^{\lfloor \frac{\NN+1}{2}\rfloor}(\homMfd_{\zz})}
	+
	\|\Fvec\|_{\sCX^{\lfloor \frac{\NN+1}{2}\rfloor-1}(\homMfd_{\zz})}
	\big)\\
	&\lesssim_*
	e_{\le\kk}^2
	\big(\|u\|_{\Ct^{\lfloor \frac{\NN+1}{2}\rfloor}(\homMfd_{\zz})}
	+
	\|\Fvec\|_{\sCX^{\lfloor \frac{\NN+1}{2}\rfloor-1}(\homMfd_{\zz})}
	\big)
	+
	e_{\le\kk}
	\|\Fvec\|_{\sHX^{\kk-1}(\homMfd_{\zz})}
\end{align*}
Thus
\[ 
W_5
\lesssim_*
e_{\le\kk}^2
	\big(\|u\|_{\Ct^{\lfloor \frac{\NN+1}{2}\rfloor}(\homMfd_{\zz})}
	+
	\|\Fvec\|_{\sCX^{\lfloor \frac{\NN+1}{2}\rfloor-1}(\homMfd_{\zz})}
	\big)
	+
	e_{\le\kk}
	\|\Fvec\|_{\sHX^{\kk-1}(\homMfd_{\zz})}
\]
\end{itemize}
Collecting terms, we obtain that for each $\kk\le\NN$,
\begin{align*}
&\tfrac{d}{d\zz} E_{0,\kk}(\zz)
-
2\big(\LdivestNEW{\LMat,\Amat(u)}(\zz)+\kk\CcomTang{\Amat(u),\Ccomp}(\zz)\big) E_{0,\kk}(\zz)\\
&\qquad\lesssim_*
\big(
|\CcomTransv{\Amat(u),\Ccomp}(\zz)|
+
\|u\|_{\Ct^{\lfloor \frac{\NN+1}{2}\rfloor}(\homMfd_{\zz})}
+
\|\Fvec\|_{\sCX^{\lfloor \frac{\NN+1}{2}\rfloor-1}(\homMfd_{\zz})}
\big)E_{0,\le\kk}(\zz)\\
&\qquad\quad+
\big(e_{\le\kk-1}(\zz) + \|\Fvec\|_{\sHX^{\kk}(\homMfd_{\zz})}\big)
e_{\le\kk}(\zz)
\end{align*}
On the left hand side replace $\kk\CcomTang{\Amat(u),\Ccomp}(\zz)$
by $\NN \max\{0,\CcomTang{\Amat(u),\Ccomp}(\zz)\}$.
Then replace $\kk$ by $\kk'$ and take the sum $\sum_{\kk'=0}^{\kk}$. 
This yields 
\begin{align*}
&\tfrac{d}{d\zz} E_{0,\le\kk}(\zz)
-
2\big(\LdivestNEW{\LMat,\Amat(u)}(\zz)+\NN \max\{0,\CcomTang{\Amat(u),\Ccomp}(\zz)\}\big) E_{0,\le\kk}(\zz)\\
&\qquad\lesssim_*
\big(
|\CcomTransv{\Amat(u),\Ccomp}(\zz)|
+
\|u\|_{\Ct^{\lfloor \frac{\NN+1}{2}\rfloor}(\homMfd_{\zz})}
+
\|\Fvec\|_{\sCX^{\lfloor \frac{\NN+1}{2}\rfloor-1}(\homMfd_{\zz})}
\big)E_{0,\le\kk}(\zz)\\
&\qquad\quad+
\big(e_{\le\kk-1}(\zz) + \|\Fvec\|_{\sHX^{\kk}(\homMfd_{\zz})}\big)e_{\le\kk}(\zz)
\end{align*} 
Using  $\|\Fvec\|_{\sHX^{\kk}(\homMfd_{\zz})}\le\|\Fvec\|_{\sHX^{\NN}(\homMfd_{\zz})} $
for each $\kk\le\NN$, 
the claim \eqref{eq:EnergyEstimateENEW} follows.

The inequality \eqref{eq:EnergyEstimateENEW} implies\footnote{%
In \eqref{eq:EnergyEstimateEsquareroot}, which is used to
derive \eqref{eq:eNNest}, beware that 
$e_{\le\kk}(\zz)$ may not be differentiable in $\zz$ 
when $E_{0,\le\kk}(\zz)=0$.
To make the derivation rigorous,
one can use the regularized 
$\tilde{e}_{\le\kk}(z) = \sqrt{E_{0,\le\kk}(\zz)+\eps f(\zz)}$
where $\eps\in(0,1]$ and 
where $f(\zz)>0$ is the solution of 
$\frac{d}{d\zz} f(\zz) = 2\big(\LdivestNEW{\LMat,\Amat(u)}(\zz)+\NN \max\{0,\CcomTang{\Amat(u),\Ccomp}(\zz)\}\big)f(\zz)$
with $f(0)=1$.
Note that $\tilde{e}_{\le\kk}(z)$ is differentiable, 
and one can check that it satisfies the same inequality \eqref{eq:EnergyEstimateEsquareroot},
independent of $\eps$.
One then obtains \eqref{eq:eNNest} for $\tilde{e}_{\le\NN}(z)$,
and then takes $\eps\downarrow0$.}
\begin{align}\label{eq:EnergyEstimateEsquareroot}
\tfrac{d}{d\zz}e_{\le\kk}(\zz)
&\le
\Big(
\LdivestNEW{\LMat,\Amat(u)}(\zz)
+\NN \max\{0, \CcomTang{\Amat(u),\Ccomp}(\zz)\}\Big) e_{\le\kk}(\zz)\nonumber\\
&\ +
\CqEE_0 
\Big(
|\CcomTransv{\Amat(u),\Ccomp}(\zz)|
+
\|u\|_{\Ct^{{\lfloor \frac{\NN+1}{2} \rfloor}}(\homMfd_{\zz})}
+
\|\Fvec\|_{\sCX^{\lfloor \frac{\NN+1}{2}\rfloor-1}(\homMfd_{\zz})}
\Big) e_{\le\kk}(\zz)\nonumber\\
&\ +
\CqEE_0 e_{\le\kk-1}(\zz)
+
\CqEE_0 \|\Fvec\|_{\sHX^{\NN}(\homMfd_{\zz})}
\end{align}
Write the system of inequalities 
\eqref{eq:EnergyEstimateEsquareroot} for $k=0,\dots,\NN$ as follows:
\begin{align} \label{eq:ineqsystem}
\tfrac{d}{d\zz}\vec{e}(\zz) 
\le 
	\left(g(\zz)\one + \CqEE_0 Q \right)\vec{e}(\zz)
+ \CqEE_0 \|\Fvec\|_{\sHX^{\NN}(\homMfd_{\zz})} V
\end{align}
where 
\begin{align*}
\vec{e}(\zz)&=\left(e_{\le\NN}(\zz),e_{\le\NN-1}(\zz),\dots,e_{\le 0}(\zz)\right)^T\\
g(\zz)&=\LdivestNEW{\LMat,\Amat(u)}(\zz)
+\NN \max\{0, \CcomTang{\Amat(u),\Ccomp}(\zz)\}\\
&\quad+
\CqEE_0 
\Big(
|\CcomTransv{\Amat(u),\Ccomp}(\zz)|
+
\|u\|_{\Ct^{{\lfloor \frac{\NN+1}{2} \rfloor}}(\homMfd_{\zz})}
+
\|\Fvec\|_{\sCX^{\lfloor \frac{\NN+1}{2}\rfloor-1}(\homMfd_{\zz})}
\Big)\\
V&=(1,1,\dots,1)^T
\end{align*}
and where $Q$ is the $(\NN+1)\times(\NN+1)$-matrix
given by $Q_{i,i+1}=1$ and $Q_{i,j}=0$ if $j\neq i+1$.
Note that $Q^{\NN+1}=0$, and that by Definition \ref{def:PGammaKDef},
\[ 
\kprop^{\LMat,\Amat(u),\Fvec,\Ccomp}_{\NN,u,\CqEE_0}(\zz_1,\zz_0)
=
\exp(\tint_{\zz_0}^{\zz_1}g(\zz')d\zz')
\]
Thus the propagator of the linear system of \eqref{eq:ineqsystem}
is given, for all $\zz_0,\zz_1\le0$, by
\[ 
P_0(\zz_1,\zz_0) 
= 
\kprop^{\LMat,\Amat(u),\Fvec,\Ccomp}_{\NN,u,\CqEE_0}(\zz_1,\zz_0)
\exp((\zz_1-\zz_0) \CqEE_0 Q )
\]

\claimheader{Claim:}
For all $\zz\le0$ the following estimate holds componentwise:
\begin{equation}\label{eq:veest}
\vec{e}(\zz) \le 
\CqEE_0\tint_{-\infty}^{\zz} P_0(\zz,\zz')V\|\Fvec\|_{\sHX^{\NN}(\homMfd_{\zz'})} \,d\zz'
\end{equation}

\claimheader{Proof of claim:}
Abbreviate $G(\zz)=g(\zz)\one + \CqEE_0 Q$.
For all $\zz'\le \zz\le 0$ we have
\begin{align*}
\tfrac{d}{d\zz'} (P_0(\zz,\zz') \vec{e}(\zz'))
	&=
	-P_0(\zz,\zz')G(\zz') \vec{e}(\zz')
	+P_0(\zz,\zz') \tfrac{d}{d\zz'}\vec{e}(\zz')
\end{align*}
All entries of $P_0(\zz,\zz')$ are non-negative using $\zz'\le \zz\le 0$
(this would fail for $\zz'>\zz$).
Thus we can use \eqref{eq:ineqsystem} in the second term on the right, which yields
\begin{align*}
\tfrac{d}{d\zz'} (P_0(\zz,\zz') \vec{e}(\zz'))
	&\le
	\CqEE_0 P_0(\zz,\zz') V \|\Fvec\|_{\sHX^{\NN}(\homMfd_{\zz'})}
\end{align*}
Using compact support of $u$ and $\Fvec$,
integrating over $\int_{-\infty}^{\zz} d\zz'$ yields \eqref{eq:veest}.

Since $Q^{\NN+1}=0$ we have 
$\exp((\zz_1-\zz_0) \CqEE_0 Q ) =
\tsum_{i=0}^{\NN} \frac{1}{i!} ((\zz_1-\zz_0)\CqEE_0 Q)^{i} $, 
thus 
\begin{align}
|P_0(\zz_1,\zz_0)|
&\lesssim_{\NN,\CqEE_0}
\kprop^{\LMat,\Amat(u),\Fvec,\Ccomp}_{\NN,u,\CqEE_0}(\zz_1,\zz_0)
(1+|\zz_1-\zz_0|)^{\NN}
\label{eq:pest}
\end{align}
where the estimate is understood component-wise.
With \eqref{eq:veest} this yields
\begin{align}\label{eq:eNNest}
e_{\le\NN}(\zz) 
	\lesssim_{\NN,\CqEE_0} 
	\tint_{-\infty}^{\zz}
	\kprop^{\LMat,\Amat(u),\Fvec,\Ccomp}_{\NN,u,\CqEE_0}(\zz,\zz') (1+|\zz-\zz'|)^{\NN}
	\|\Fvec\|_{\sHX^{\NN}(\homMfd_{\zz'})} \,d\zz'
\end{align}
Together with \eqref{eq:A(zeta)posBIL} we obtain
\[ 
\|u\|_{\Ht^{\NN}(\homMfd_{\zz})}
	\le
	\CqEE_1 
	\tint_{-\infty}^{\zz}
	\kprop^{\LMat,\Amat(u),\Fvec,\Ccomp}_{\NN,u,\CqEE_0}(\zz,\zz') (1+|\zz-\zz'|)^{\NN}
	\|\Fvec\|_{\sHX^{\NN}(\homMfd_{\zz'})} \,d\zz'
\]
for a constant $\CqEE_1>0$ depending only on \eqref{eq:constEE} and $\CM$.
This proves Part 2 under the 
admissible largeness assumption $\CqEE\ge \max\{\CqEE_0,\CqEE_1\}$.
\qed
\end{proof}

\subsection{A semiglobal existence and uniqueness theorem}

We state and prove the main result of Section \ref{sec:Abstract},
that is, existence and uniqueness for a 
class of quasilinear symmetric hyperbolic systems on $\homMfd$
(Theorem \ref{thm:nonlinEE}, \ref{thm:AbstractUniqueness}).
The proof is based on the a priori energy estimates in Section \ref{sec:EnergyEstimate}.
The result will be applied to the Einstein equations in 
Section \ref{sec:SpaceinfConstruction}.

Recall that repeated indices $i,j$
are implicitly summed over $0\dots\MdimNEW$.
\begin{theorem}\label{thm:nonlinEE}
Let $\Mcpt, X_0,\dots,X_{\MdimNEW}, \muM$
be as in Section \ref{sec:AbstractGeom}.
For all 
\begin{align}\label{eq:AbstrExConstants}
\nn\in\Z_{\ge1}
\qquad
\NN\in\Z_{\ge\MdimNEW+3}
\qquad
\Cpos \ge1 
\qquad
\CLA >0 
\end{align}
there exists $\Clarge>0$
such that for all $\delta>0$ there exists
$\Csmall\in(0,1]$ such that for all
\begin{align}\label{eq:ATHMMats}
\begin{aligned} 
\AmatLin^i 
	&\in C^\infty(\bar\homMfd,S^2\R^{\nn}) 
	&& i=0\dots\MdimNEW \\
\qquad\qquad
\AmatBil^i 
	&\in C^\infty(\bar\homMfd,\Hom(\R^{\nn}, S^2\R^{\nn}) ) 
	&& i=0\dots\MdimNEW \\
\LMat&\in C^\infty(\bar\homMfd,\End(\R^{\nn})) \\
\BSHS&\in C^\infty(\bar\homMfd,\Hom(\R^{\nn}\otimes\R^{\nn},\R^{\nn})) \\
\Fvec&\in C^\infty(\bar\homMfd, \R^{\nn})\\
\Ccomp
	&\in
	C^\infty(\bar\homMfd,\R^{\MdimNEW})
\end{aligned}
\end{align}
the following holds.
For $\kk\in\Z_{\ge0}$ and $\zz_0,\zz_1\le0$ and $\zz\le0$ define 
\begin{align}
\kprop_{\kk}(\zz_1,\zz_0) 
	&= \textstyle
	\exp\Big(\int_{\zz_0}^{\zz_1} 
	\big(\LdivestNEW{\LMat,\AmatLin}(\zz')
	+\kk\max\{0,\CcomTang{\AmatLin,\Ccomp}(\zz')\}\big)\,d\zz'\Big)
	\label{eq:kpropdef}\\
\Fconst_{\kk}(\zz)
	&= \textstyle
	\int_{-\infty}^{\zz}\kprop_{\kk}(\zz,\zz') 
	(1+|\zz-\zz'|)^{\kk}
	\|\Fvec\|_{\sHX^{\kk}(\homMfd_{\zz'})}
	\,d\zz'
	\label{eq:Fconstdef}
	\;\in\; [0,\infty]
\end{align}
where $\AmatLin=\AmatLin^i X_i$,
and using Definitions \ref{def:ellCcomdef} and \ref{def:defCcom}.
If
\begin{enumerate}[({e}1),leftmargin=10mm]
\item \label{item:Fsmall}
$\sup_{\zz\in(-\infty,0]}\Fconst_{\NN}(\zz)\le\Csmall$
and
$\tint_{-\infty}^{0}\Fconst_{\NN}(\zz) d\zz \le \Csmall$
\item \label{item:InhSmall}
$\sup_{\zz\in(-\infty,0]}\|\Fvec\|_{\sHX^{\NN-1}(\homMfd_{\zz})}\le\Csmall$
and
$\tint_{-\infty}^{0}\|\Fvec\|_{\sHX^{\NN-1}(\homMfd_{\zz})} d\zz \le \CLA$
\item \label{item:KappaTransv}
$\tint_{-\infty}^{0}|\CcomTransv{\AmatLin,\Ccomp}(\zz)| d\zz \le \CLA$
\item \label{item:Apos}
For all $w\in\R^{\nn}$, if $\sqrt{w^Tw}\le\delta$ then 
at every point on $\homMfd$:
\begin{subequations}\label{eq:Apos}
\begin{align}
\Cpos^{-1}\one 
&\le d\zzeta\big((\AmatLin^i+\AmatBil^i(w))X_i\big)
\le \Cpos\one 
\label{eq:Azetapos}\\
\Cpos^{-1}\one 
&\le \AmatLin^0+\AmatBil^0(w)
\le \Cpos\one 
\label{eq:A0pos}\\
(1-\ttcoord)\Cpos^{-1} \one 
&\le 
d\ttcoord\big((\AmatLin^i+\AmatBil^i(w))X_i\big)
\le\Cpos\one
\label{eq:Atpos}
\end{align}
\end{subequations}
\item 
\label{item:CNnormsSHS}
$
\|\AmatLin^i\|_{\CX^{\NN}(\homMfd)},
\|\AmatBil^i\|_{\CX^{\NN}(\homMfd)},
\|\LMat\|_{\CX^{\NN}(\homMfd)},
\|\BSHS\|_{\CX^{\NN}(\homMfd)},
\|\Ccomp\|_{\CX^0(\homMfd)}
\le
\CLA$
\end{enumerate}
Then
there exists $u\in C^\infty(\homMfd,\R^{\nn})$ such that
\begin{subequations}\label{eq:ubasicNEW}
\begin{align}
(\AmatLin^i+\AmatBil^i(u)) X_i u 
	&\;=\; \LMat u + \BSHS(u,u) + \Fvec\label{eq:ueq}\\
u|_{\ttcoord=0} 
	&\;=\; 0\label{eq:udata}\\
\sqrt{u^Tu} 
	&\;\le\; \delta \qquad
\text{on $\homMfd$} \label{eq:uinf}
\end{align}
\end{subequations}
Furthermore:
\begin{itemize}
\item \textbf{Part 0.}
\new{$u$ is unique, in the sense that
for every $u'\in C^\infty(\homMfd,\R^{\nn})$ 
that satisfies \eqref{eq:ueq} and \eqref{eq:udata}
with $u$ replaced by $u'$ one has $u=u'$.}

\item \textbf{Part 1.}
For all $\zz\le0$:
\begin{subequations}\label{eq:uconclNEW}
\begin{align}
\|u\|_{\Ht^{\NN}(\homMfd_{\zz})}
	&\;\le\;
	\Clarge\Fconst_{\NN}(\zz)
	\label{eq:uETangDer}\\
\|u\|_{\sHX^{\NN}(\homMfd_{\zz})}
	&\;\le\;
	\Clarge(\Fconst_{\NN}(\zz) + \|\Fvec\|_{\sHX^{\NN-1}(\homMfd_{\zz})})
	\label{eq:uEAllDer}
\end{align}
\end{subequations}
Furthermore $\Fconst_{\NN}(\zz)\le \kprop_{\NN}(\zz,0)\Fconst_{\NN}(0)$.
\item 
\textbf{Part 2.}
For every $\kk\in\Z_{\ge\NN}$ and every $\CHigherIn\new{>}0$, if 
\begin{enumerate}[({e}1),leftmargin=10mm,resume]
\item 
\label{item:Higherassp12}
\new{$\Fconst_{\kk}(0)<\infty$} 
and 
$\sup_{\zz\in(-\infty,0]}\Fconst_{\kk-1}(\zz)\le\CHigherIn$ 
and 
$\int_{-\infty}^{0}\Fconst_{\kk-1}(\zz)\, d\zz
\le \CHigherIn$ 
\item \label{item:InhSmallHigher}
$\sup_{\zz\in(-\infty,0]}\|\Fvec\|_{\sHX^{\kk-2}(\homMfd_{\zz})}\le \CHigherIn$
and
$\tint_{-\infty}^{0}\|\Fvec\|_{\sHX^{\kk-2}(\homMfd_{\zz})} d\zz \le \CHigherIn$

\item \label{item:Higherassp3}
$
\|\AmatLin^i\|_{\CX^{\kk}(\homMfd)},
\|\AmatBil^i\|_{\CX^{\kk}(\homMfd)},
\|\LMat\|_{\CX^{\kk}(\homMfd)},
\|\BSHS\|_{\CX^{\kk}(\homMfd)}
\le
\CHigherIn$
\end{enumerate}
then for all $\zz\le0$:
\begin{subequations}\label{eq:HigherAbstract}
\begin{align}
\|u\|_{\Ht^{\kk}(\homMfd_{\zz})}
\;&\lesssim_{\CM,\nn,\kk,\Cpos,\CLA,\CHigherIn}\;
\Fconst_{\kk}(\zz)
\label{eq:Ekkest}\\
\|u\|_{\sHX^{\kk}(\homMfd_{\zz})}
\;&\lesssim_{\CM,\nn,\kk,\Cpos,\CLA,\CHigherIn}\;
\Fconst_{\kk}(\zz) + \|\Fvec\|_{\sHX^{\kk-1}(\homMfd_{\zz})}
\label{eq:EkkestX}
\end{align}
\end{subequations}
with $\CM$ as in \eqref{eq:CM}.
Furthermore $\Fconst_{\kk}(\zz)\le \kprop_{\kk}(\zz,0)\Fconst_{\kk}(0)$.
\end{itemize}
\end{theorem}
Before we prove this, we provide a stronger, localized, uniqueness statement.
For each 
$\Cpos\ge1$ and  
$(\zz_0,\tt_0)\in(-\infty,0]\times(0,1)$ define the cone
(see Figure \ref{fig:conepic})
\begin{equation}\label{eq:eqcone}
\cone^{\Cpos}_{\zz_0,\tt_0}
\;=\;
\Big\{(\zz,\tt,p)\in \homMfd\mid
\zz\le \zz_0\,,\;
\tt  \le \tt_0+ \tfrac{1-\tt_0}{2\Cpos^2}(\zz-\zz_0)
\Big\}
\end{equation}
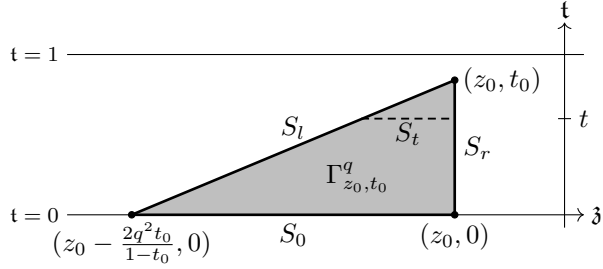
\begin{figure}%
\centering%
\begin{tikzpicture}[inner sep=0pt,scale=0.85]
\node (tip) at (6,2.1) {}; %
\node (l0) at (1,0) {}; %
\node (r0) at (6,0) {}; %
\fill[fill=gray!50] (tip.center)--(l0.center)--(r0.center);
   \path[draw,line width=1 pt] (tip)--(l0) node[midway,anchor=south,yshift=1mm] {$S_l$};
   \path[draw,line width=1 pt] (l0)--(r0) node[midway,anchor=north,yshift=-1mm] {$S_0$};
   \path[draw,line width=1 pt] (r0)--(tip) node[midway,anchor=west,xshift=1mm] {$S_r$};;
\draw[color=black, fill=black] (tip) circle (.05);
   \draw[color=black, fill=black]  (r0) circle (.05);
   \draw[color=black, fill=black]  (l0) circle (.05);
\draw[] (0,2.5) -- (8,2.5);
\draw[->] (0,0) -- (8,0) node[anchor=west,xshift=1mm] {$\zzeta$};
\draw[->] (7.7,-0.2) -- (7.7,3) node[anchor=south,yshift=0.5mm] {$\ttcoord$};
\draw[line width=0.7 pt,densely dashed] (4.55,1.5) -- (6,1.5);
\node[anchor=north,yshift=-0.5mm] at (5.275,1.5) {$S_{\tt}$};
%
\draw[] (7.6,1.5) -- (7.8,1.5) 
	node[anchor=west,xshift=1mm] {$\tt$};
\node[yshift=-1mm] (c) at (4.5,0.7) {$\cone^{\Cpos}_{\zz_0,\tt_0}$};
\node[anchor=west,xshift=1mm] at (tip) {$(\zz_0,\tt_0)$};
\node[anchor=north,yshift=-1mm] at (l0) {$(\zz_0-\tfrac{2\Cpos^2t_0}{1-t_0},0)$};
\node[anchor=north,yshift=-1mm] at (r0) {$(\zz_0,0)$};
\node[anchor=east,xshift=-1mm] at (0,0) {\footnotesize$\ttcoord=0$};
\node[anchor=east,xshift=-1mm] at (0,2.5) {\footnotesize$\ttcoord=1$};
\end{tikzpicture}
\captionsetup{width=115mm}
\caption{%
The gray domain depicts  
$\cone^{\Cpos}_{\zz_0,\tt_0}$,
with the factor $\Mcpt$ suppressed.
The sets $S_l$, $S_r$, $S_0$ are the boundary components.
The dashed line depicts $S_t=\cone^{\Cpos}_{\zz_0,\tt_0}\cap\ttcoord^{-1}(\{\tt\})$.
}%
\label{fig:conepic}%
\end{figure}%
Note that for every fixed $\Cpos$, 
the union of all such cones is $\homMfd$, in fact
$$
\textstyle\bigcup_{\tt_0\in(0,1)}\cone^{\Cpos}_{0,\tt_0}
	=\homMfd
$$%
\begin{theorem}[Uniqueness on cones]\label{thm:AbstractUniqueness}
Let $\Mcpt, X_0,\dots,X_{\MdimNEW}, \muM$
be as in Section \ref{sec:AbstractGeom}.
For all $\nn\in\Z_{\ge1}$, all $\Cpos\ge1$,
all $(\zz_0,\tt_0) \in (-\infty,0]\times(0,1)$,
all $\AmatLin^i,\AmatBil^i,\LMat,\BSHS,\Fvec$ as in \eqref{eq:ATHMMats}
that are defined on $\cone^{\Cpos}_{\zz_0,\tt_0}$, and all
$$
u_1,u_2\in C^\infty(\cone^{\Cpos}_{\zz_0,\tt_0},\R^{\nn})
$$
if on $\cone^{\Cpos}_{\zz_0,\tt_0}$ one has
\begin{subequations}\label{eq:asspUniq}
\begin{align}
(\AmatLin^i+\AmatBil^i(u_\ell)) X_i u_\ell
	&= \LMat u_\ell + \BSHS(u_\ell,u_\ell) + \Fvec &&\text{for $\ell=1,2$}\label{eq:ueqUniq}\\
u_\ell|_{\ttcoord=0} 
	&= 0 &&\text{for $\ell=1,2$}\label{eq:udataUniq}\\
\Cpos^{-1}\one 
	&\le d\zzeta\big((\AmatLin^i+\AmatBil^i(u_1))X_i\big)
	\le \Cpos\one 
	\label{eq:AzetaposUniq}\\
(1-\ttcoord)\Cpos^{-1} \one 
	&\le d\ttcoord((\AmatLin^i+\AmatBil^i(u_1))X_i)
	\le \Cpos\one 
	\label{eq:AtposUniq}
\end{align}
\end{subequations}
then $u_1=u_2$.
\end{theorem}
This uniqueness theorem can be applied in particular 
with $u_1$ equal to the solution $u$ from Theorem \ref{thm:nonlinEE},
restricted to $\cone^{\Cpos}_{\zz_0,\tt_0}$.
This satisfies \eqref{eq:AzetaposUniq}, \eqref{eq:AtposUniq}
by \eqref{eq:uinf}, \eqref{eq:Azetapos}, \eqref{eq:Atpos}.
Hence the uniqueness statement in Theorem \ref{thm:nonlinEE} follows.
Further it follows that the restriction of $u$ to the cone $\cone^{\Cpos}_{\zz_0,\tt_0}$
only depends on the restriction of the data \eqref{eq:ATHMMats} to that cone.
This will be used to reduce the proof of 
Theorem \ref{thm:nonlinEE} to the case where $\Fvec$ has compact support.
\begin{proof}[of Theorem \ref{thm:AbstractUniqueness}]
By \eqref{eq:ueqUniq} and \eqref{eq:udataUniq}, 
the difference $U=u_1-u_2$ satisfies 
the linear homogeneous symmetric hyperbolic system
\begin{align}\label{eq:linEqUUniq}
(\AmatLin^i + \AmatBil^i(u_1)) X_i U &= \tilde \LMat U
&
U|_{\ttcoord=0} &= 0
\end{align}
with the $C^\infty$-linear term 
\begin{align*}
\tilde \LMat U &= \LMat U - \AmatBil^i(U) X_i u_2 + \BSHS(u_1,U)+\BSHS(U,u_2)
\end{align*}
The boundary of the cone $\cone^{\Cpos}_{\zz_0,\tt_0}$ has three components
that we denote by $\p\cone^{\Cpos}_{\zz_0,\tt_0} = S_l\cup S_r \cup S_0$
as indicated in Figure \ref{fig:conepic}.
For $\tt\in[0,\tt_0)$ we define 
$S_{\tt}=\cone^{\Cpos}_{\zz_0,\tt_0}\cap \ttcoord^{-1}(\{\tt\})$,
which coincides with the boundary component $S_0$ when $\tt=0$.
We claim that $S_l,S_r,S_{\tt}$ are spacelike, with
\begin{subequations}\label{eq:conecausal}
\begin{align}
d\zzeta\big((\AmatLin^{i} + \AmatBil^{i}(u_1)) X_i\big) 
	&>0 \qquad \text{on $S_r$}
	\label{eq:conecausal_Sr}\\
\nu\big((\AmatLin^{i} + \AmatBil^{i}(u_1))X_i\big) &>0 \qquad \text{on $S_l$}
	\label{eq:conecausal_Sl}\\
d\ttcoord\big((\AmatLin^{i} + \AmatBil^{i}(u_1))X_i\big) 
	&>0 \qquad \text{on $S_{\tt}$ for each $\tt\in[0,t_0)$}
	\label{eq:conecausal_St}
\end{align}
\end{subequations}
where $\nu=d\ttcoord-\frac{1-\tt_0}{2\Cpos^2}d\zzeta$ 
is an outward pointing normal one-form on $S_l$.

Proof of \eqref{eq:conecausal}: 
For \eqref{eq:conecausal_Sr} use \eqref{eq:AzetaposUniq};
for \eqref{eq:conecausal_St} use \eqref{eq:AtposUniq} and $\tt_0<1$;
for \eqref{eq:conecausal_Sl} note that on $S_{l}$ we have
\begin{align*}
\nu\big((\AmatLin^{i} + \AmatBil^{i}(u_1))X_i\big)
	&=\textstyle
	d\ttcoord\big((\AmatLin^{i} + \AmatBil^{i}(u_1))X_i\big)
	-
	\tfrac{1-\tt_0}{2\Cpos^2} 
	d\zzeta\big((\AmatLin^{i} + \AmatBil^{i}(u_1))X_i\big)\\
	&\ge\textstyle
	\Cpos^{-1}((1-\ttcoord)-\frac{1}{2} (1-\tt_0) )
	\ge
	\frac12\Cpos^{-1}(1-\tt_0)
	>0
\end{align*}
where we use \eqref{eq:AzetaposUniq} and \eqref{eq:AtposUniq},
the fact that $\ttcoord\le\tt_0$ on $S_l$,
and $\tt_0<1$.

\new{Given \eqref{eq:conecausal}, 
one obtains $U=0$ using 
standard energy estimates for the linear homogeneous symmetric hyperbolic system \eqref{eq:linEqUUniq}, with energies over $S_{\tt}$.}\qed
\end{proof}
\begin{proof}[of Theorem \ref{thm:nonlinEE} in compact support case]
We first prove Theorem \ref{thm:nonlinEE} under the additional assumption 
that $\Fvec$ has compact support,
\begin{equation}\label{eq:Fcomp}
\Fvec\in C^\infty_c(\bar\homMfd,\R^{\nn})
\end{equation}
which means that $\Fvec$ vanishes for all large negative $\zzeta$.
(Theorem \ref{thm:nonlinEE} in the general case,
that is, without the assumption \eqref{eq:Fcomp}, 
will be proved below, by reducing it to the theorem in the
compact support case.)

\proofheader{Proof of Part 0.}
It suffices to show $u=u'$ on every cone \eqref{eq:eqcone}.
We use Theorem \ref{thm:AbstractUniqueness} with $u_1=u$ and $u_2=u'$. 
Clearly the assumptions \eqref{eq:ueqUniq}, \eqref{eq:udataUniq} hold,
and \eqref{eq:AzetaposUniq}, \eqref{eq:AtposUniq}
hold by \eqref{eq:uinf}, \eqref{eq:Azetapos}, \eqref{eq:Atpos}.
Thus $u=u'$ on \eqref{eq:eqcone}.

\proofheader{Proof of existence and Part 1.}
We show that there exists
\begin{equation}\label{eq:ucpt}
u\in C^\infty_c(\bar\homMfd,\R^{\nn})
\end{equation}
that satisfies \eqref{eq:ubasicNEW} and Part 1.
Note that $u$ is unique by Part 0.

We will specify the constant $\Clarge$ during the proof.
We will not specify $\Csmall$, but 
make finitely many admissible smallness assumptions on $\Csmall$,
where admissible means that they depend only on 
\eqref{eq:AbstrExConstants}, $\delta$, $\CM$, see \eqref{eq:CM}.

We will use Lemma \ref{lem:Nonlindiffenergyineq} with the parameters in 
the second column of Table \ref{tab:NonlinAppN}. 
Let $\CqEEApp>0$ be the constant produced by 
Lemma \ref{lem:Nonlindiffenergyineq} (called $\CqEE$ there),
which depends only on \eqref{eq:AbstrExConstants} and $\CM$.
We need the following preliminaries:
\begin{itemize}
\item 
There exists $\CFsob>0$ that depends only on $\CM,\nn,\NN$,
such that for all $\zz\le0$:
\begin{equation}\label{eq:FSobolev}
\|\Fvec\|_{\sCX^{\lfloor\frac{\NN+1}{2}\rfloor-1}(\homMfd_{\zz})}
\le
\CFsob \|\Fvec\|_{\sHX^{\NN-1}(\homMfd_{\zz})}
\end{equation}
This holds by \eqref{eq:sobolevMz},
since $\lfloor\frac{\NN+1}{2}\rfloor-1+\frac{\MdimNEW}{2} < \NN-1$
using $\NN\ge\MdimNEW\new{+2}$.
\item 
There exists $\Clk>0$ that depends only on \eqref{eq:AbstrExConstants} and $\CM$,
such that for all $\zzmax\le0$, all 
$u\in C^\infty_c(\bar\homMfd_{\le\zzmax},\R^{\nn})$ 
with $\sqrt{u^Tu}\le\delta$, and all $\zz\le\zzmax$,
\begin{align}\label{eq:lkdiff}
\begin{aligned}
|\LdivestNEW{\LMat,\AmatLin^iX_i}(\zz)
-\LdivestNEW{\LMat,\Amat(u)}(\zz)|
	&\le
	\Clk\|u\|_{\sCX^1(\homMfd_{\zz})}\\
|\max\{0,\CcomTang{\AmatLin^iX_i,\Ccomp}(\zz)\}
-\max\{0,\CcomTang{\Amat(u),\Ccomp}(\zz)\}|
	&\le
	\Clk\|u\|_{\sCX^1(\homMfd_{\zz})}\\
|\CcomTransv{\AmatLin^iX_i,\Ccomp}(\zz)
	-\CcomTransv{\Amat(u),\Ccomp}(\zz)|
	&\le
	\Clk\|u\|_{\sCX^1(\homMfd_{\zz})}
\end{aligned}
\end{align}
where we abbreviate $\Amat(u)=(\AmatLin^i+\AmatBil^i(u))X_i$.
This holds using Lemma \ref{lem:ellkappaLip} 
(applicable by \eqref{eq:Azetapos}) and
$\Amat(u)-\AmatLin^i X_i = \AmatBil^i(u)X_i$ and \ref{item:CNnormsSHS} (use $\NN\ge1$).
\end{itemize}
Define the constant
\begin{equation}\label{eq:CaCap}
\Ca \;=\; 
e^{
\CLA\CqEEApp(
1
+
\CFsob
+
\Caprime )
}
\qquad
\text{where} 
\qquad
\Caprime=\Clk+\NN\Clk+\CqEEApp\Clk+\CqEEApp
\end{equation}

\begin{table}
\centering
\begin{tabular}{cc|c|c}
	&
	Parameters 
	&
	\multicolumn{2}{c}{Parameters used to invoke Lemma \ref{lem:Nonlindiffenergyineq}}
	\\
	&
	in Lemma \ref{lem:Nonlindiffenergyineq}
	&
	\textit{Existence and Part 1}
	&
	\textit{Part 2}
	\\
	\hline
Input
&$\Mcpt$, $X_0,\dots,X_{\MdimNEW}$, $\muM$
	&$\Mcpt$, $X_0,\dots,X_{\MdimNEW}$, $\muM$
	&$\Mcpt$, $X_0,\dots,X_{\MdimNEW}$, $\muM$\\
&$\nn$, $\NN$, $\Cpos$, $\CLA$ 
	& $\nn$, $\NN$, $\Cpos$, $\CLA$
	& $\nn$, $\kk$, $\Cpos$, \new{$\slashCHigherIn$ in \eqref{eq:slbchoice}}\\
&$\zz_*$ 
	& $\zzmax$
	& $0$\\
&$u$ 
	& $u$ in \eqref{eq:uapriori}
	& $u$ in \eqref{eq:ucpt}\\
&$\AmatLin^i$, $\AmatBil^i$, $\LMat$, $\BSHS$, $\Fvec$
	& $\AmatLin^i$, $\AmatBil^i$, $\LMat$, $\BSHS$, $\Fvec$
	& $\AmatLin^i$, $\AmatBil^i$, $\LMat$, $\BSHS$, $\Fvec$\\
&$\Ccomp$ 
	& $\Ccomp$
	& $\Ccomp$\\
\hline
Output
&$\CqEE$
	& $\CqEEApp$
	& $\CqEEApp_{\kk}$
\end{tabular}
\captionsetup{width=115mm}
\caption{
The first column lists the input and output 
parameters of Lemma \ref{lem:Nonlindiffenergyineq}.
The second column specifies the choice of input parameters used
to invoke Lemma \ref{lem:Nonlindiffenergyineq} in the proof
of existence and Part 1 of Theorem \ref{thm:nonlinEE} in the compact support case,
in terms of the input parameters of Theorem \ref{thm:nonlinEE}
and the parameters introduced in this proof.
The output parameter produced by this invocation of Lemma \ref{lem:Nonlindiffenergyineq}
is denoted $\CqEEApp$, and it depends only on the parameters in the first two rows
of the second column.
Analogously for the third column, used to invoke Lemma \ref{lem:Nonlindiffenergyineq} 
in the proof of Part 2 in the compact support case.}
\label{tab:NonlinAppN}
\end{table}

Will make an open-closed (bootstrap) argument, based on the next claim.

\claimheader{Claim:}
Let $\zzmax\le0$ and let
\begin{equation}\label{eq:uapriori}
u\in C^\infty_c(\bar\homMfd_{\le\zzmax},\R^{\nn}) 
\end{equation}
such that
\begin{subequations}\label{eq:apriori}
\begin{align}
(\AmatLin^{i}+\AmatBil^i(u)) X_i u &= \LMat u + \BSHS(u,u) + \Fvec
\label{eq:aprioriSHS}\\
u|_{\ttcoord=0} &= 0
\label{eq:aprioriData}\\
\|u\|_{\Ct^{\lfloor\frac{\NN+1}{2}\rfloor}(\homMfd_{\le\zzmax})}
&\le \min\{\delta,\CLA\}
\label{eq:apriori3NEW} \\
\|u\|_{\Ht^{\NN}(\homMfd_{\zz})} &\le 6 \Ca\CqEEApp \Fconst_{\NN}(\zz)
&&\text{for $\zz\le\zzmax$}
\label{eq:apriori1}\\
\qquad\qquad
\tfrac{\kprop^{\LMat,\Amat(u),\Fvec,\Ccomp}_{\NN,u,\CqEEApp}(\zz_1,\zz_0)}{\kprop_{\NN}(\zz_1,\zz_0)} &\le 3\Ca
&&\text{for $\zz_0\le\zz_1\le\zzmax$} 
\label{eq:apriori2}
\end{align}
\end{subequations}
where in \eqref{eq:apriori2} we use Definition \ref{def:PGammaKDef} 
and write $\Amat(u)=(\AmatLin^{i}+\AmatBil^i(u)) X_i$.
Then, under admissible smallness assumptions on $\Csmall$,
the inequalities hold with a gap:
\begin{subequations}\label{eq:apriorigap}
\begin{align}
\|u\|_{\Ct^{\lfloor\frac{\NN+1}{2}\rfloor}(\homMfd_{\le\zzmax})}
&\le \tfrac12 \min\{\delta,\CLA\} 
\label{eq:apriori3gapNEW}\\
\|u\|_{\Ht^{\NN}(\homMfd_{\zz})} &\le 3\Ca \CqEEApp \Fconst_{\NN}(\zz)
&&\text{for $\zz\le\zzmax$}
\label{eq:apriori1gap}\\
\qquad\tfrac{\kprop^{\LMat,\Amat(u),\Fvec,\Ccomp}_{\NN,u,\CqEEApp}(\zz_1,\zz_0)}{\kprop_{\NN}(\zz_1,\zz_0)}&\le 2\Ca
&&\text{for $\zz_0\le\zz_1\le\zzmax$}\qquad
\label{eq:apriori2gap}
\end{align}
\end{subequations}
Furthermore one has
\begin{align}\label{eq:uallder}
\|u\|_{\sHX^{\NN}(\homMfd_{\zz})} 
	&\le 
	\CqEEApp(3\Ca\CqEEApp \Fconst_{\NN}(\zz)+ \|\Fvec\|_{\sHX^{\NN-1}(\homMfd_{\zz})} )
	&&\text{for $\zz\le\zzmax$}
\end{align}

Note that
\eqref{eq:aprioriSHS}, \eqref{eq:aprioriData}, \eqref{eq:apriori3NEW}
determine $u$ uniquely (c.f.~proof of Part 0).

\claimheader{Proof of claim:}
We use Lemma \ref{lem:Nonlindiffenergyineq}
with the parameters in the second column of Table \ref{tab:NonlinAppN}.
We check that the assumptions of Lemma \ref{lem:Nonlindiffenergyineq} hold:
$\Fvec$ and $u$ have compact support 
by the assumptions \eqref{eq:Fcomp} and \eqref{eq:uapriori};
$\AmatLin^i$, $\AmatBil^i$ are symmetric;
\eqref{eq:SHSBil} holds by \eqref{eq:aprioriSHS};
\eqref{eq:DataBil} holds by \eqref{eq:aprioriData};
\eqref{eq:A(zeta)posBIL} holds by \eqref{eq:apriori3NEW} and \eqref{eq:Azetapos};
\eqref{eq:A(0)posBIL} holds by \eqref{eq:apriori3NEW} and \eqref{eq:A0pos};
\eqref{eq:A(t)posBIL} holds by \eqref{eq:apriori3NEW} and \eqref{eq:Atpos};
\eqref{eq:Cubound} holds by \eqref{eq:apriori3NEW};
for \eqref{eq:CFbound} note that by 
\eqref{eq:FSobolev} and \ref{item:InhSmall},
for all $\zz\le0$:
$$
\|\Fvec\|_{\sCX^{\lfloor\frac{\NN+1}{2}\rfloor-1}(\homMfd_{\zz})}
\le
\CFsob \Csmall
\le
\CLA
$$
where for the last step we make 
the admissible smallness assumption $\Csmall\le\CLA/\CFsob$;
and \eqref{eq:matbounds} holds by \ref{item:CNnormsSHS}.
Thus the assumptions of Lemma \ref{lem:Nonlindiffenergyineq} hold.

Using \eqref{eq:CtSobolev}
with $k=\lfloor\frac{\NN+1}{2}\rfloor$ 
(applicable by $\lfloor\frac{\NN+1}{2}\rfloor + \lfloor\frac{\MdimNEW}{2}\rfloor+1 \le \NN$, use $\NN\ge\MdimNEW+2$),
and then using \eqref{eq:apriori1},
for all $\zz\le\zzmax$ we have
\begin{align}\label{eq:uSobolev}
\|u\|_{\sCX^{\lfloor\frac{\NN+1}{2}\rfloor}(\homMfd_{\zz})} 
	&\le
	\CqEEApp(\|u\|_{\Ht^{\NN}(\homMfd_{\zz})} 
	+ 
	\|\Fvec\|_{\sHX^{\NN-1}(\homMfd_{\zz})})\nonumber\\
	&\le
	\CqEEApp(6\Ca\CqEEApp \Fconst_{\NN}(\zz)
	+ 
	\|\Fvec\|_{\sHX^{\NN-1}(\homMfd_{\zz})})
\end{align}

We can now conclude \eqref{eq:apriorigap}:
\begin{itemize}
\item 
\eqref{eq:apriori3gapNEW}:
By \eqref{eq:uSobolev} and \ref{item:Fsmall} and \ref{item:InhSmall},
for all $\zz\le\zzmax$,
$$\|u\|_{\Ct^{\lfloor\frac{\NN+1}{2}\rfloor}(\homMfd_{\zz})}
\le
\CqEEApp(6\Ca\CqEEApp +1)\Csmall
\le \tfrac12\min\{\delta,\CLA\}
$$
where the last inequality holds under an admissible
smallness assumption on $\Csmall$,
using the fact that $\CqEEApp$ and $\Ca$ depend only on 
\eqref{eq:AbstrExConstants} and $\CM$.

\item 
\eqref{eq:apriori1gap}:
By Part 2 of Lemma \ref{lem:Nonlindiffenergyineq},
for all $\zz\le \zzmax$,
\begin{align*}
\|u\|_{\Ht^{\NN}(\homMfd_{\zz})}
&\le
\CqEEApp \tint_{-\infty}^{\zz} 
\kprop^{\LMat,\Amat(u),\Fvec,\Ccomp}_{\NN,u,\CqEEApp}(\zz,\zz')
(1 + |\zz-\zz'|)^{\NN}
\|\Fvec\|_{\sHX^{\NN}(\homMfd_{\zz'})}
d\zz'\\
\intertext{
which by \eqref{eq:apriori2} is bounded by
}
&\le
3\Ca\CqEEApp \tint_{-\infty}^{\zz} 
\kprop_{\NN}(\zz,\zz')
(1 + |\zz-\zz'|)^{\NN}
\|\Fvec\|_{\sHX^{\NN}(\homMfd_{\zz'})}
d\zz'\\
&=
3\Ca\CqEEApp \Fconst_{\NN}(\zz)
\end{align*}

\item 
\eqref{eq:apriori2gap}:
Write 
\begin{align}\label{eq:WW}
\tfrac{\kprop^{\LMat,\Amat(u),\Fvec,\Ccomp}_{\NN,u,\CqEEApp}(\zz_1,\zz_0)}{\kprop_{\NN}(\zz_1,\zz_0)}
=
\exp\big(\tint_{\zz_0}^{\zz_1}(
W_1(\zz)+W_2(\zz)+\dots+W_5(\zz))d\zz\big)
\end{align}
where
\begin{align*}
W_1(\zz)
	&=
	\LdivestNEW{\LMat,\Amat(u)}(\zz)-\LdivestNEW{\LMat,\AmatLin}(\zz)
	\\
W_2(\zz)
	&=
	\NN
	(\max\{0,\CcomTang{\Amat(u),\Ccomp}(\zz)\}-\max\{0,\CcomTang{\AmatLin,\Ccomp}(\zz)\})\\
W_3(\zz)
	&=
	\CqEEApp
	|\CcomTransv{\Amat(u),\Ccomp}(\zz)|
	\\
W_4(\zz)
	&=
	\CqEEApp
	\|u\|_{\Ct^{\lfloor \frac{\NN+1}{2} \rfloor}(\homMfd_{\zz})}
	\\
W_5(\zz)
	&=
	\CqEEApp
	\|\Fvec\|_{\sCX^{\lfloor \frac{\NN+1}{2}\rfloor-1}(\homMfd_{\zz})}
\end{align*}
By \eqref{eq:lkdiff},
\begin{align*}
|W_1(\zz)| &\le \Clk \|u\|_{\sCX^{1}(\homMfd_{\zz})}\\
|W_2(\zz)| &\le \NN\Clk\|u\|_{\sCX^{1}(\homMfd_{\zz})}\\
|W_3(\zz)| 
	&\le 
	\CqEEApp
	|\CcomTransv{\AmatLin,\Ccomp}(\zz)-\CcomTransv{\Amat(u),\Ccomp}(\zz)|
	+
	\CqEEApp
	|\CcomTransv{\AmatLin,\Ccomp}(\zz)|\\
	&\le
	\CqEEApp\Clk \|u\|_{\sCX^{1}(\homMfd_{\zz})}
	+
	\CqEEApp
	|\CcomTransv{\AmatLin,\Ccomp}(\zz)|
\intertext{
where for $W_3$ we also use the triangle inequality.
By \eqref{eq:FSobolev}, 
}
|W_5(\zz)|
	&\le
	\CqEEApp\CFsob \|\Fvec\|_{\sHX^{\NN-1}(\homMfd_{\zz})}
\end{align*}
Thus, using \smash{$\lfloor\frac{\NN+1}{2}\rfloor\ge1$}, 
\begin{align*}
\tsum_{i=1}^5|W_i(\zz)|
&\le
\CqEEApp
|\CcomTransv{\AmatLin,\Ccomp}(\zz)|
+
\Caprime\|u\|_{\sCX^{\lfloor\frac{\NN+1}{2}\rfloor}(\homMfd_{\zz})}
+
\CqEEApp\CFsob \|\Fvec\|_{\sHX^{\NN-1}(\homMfd_{\zz})}
\intertext{
with $\Caprime$ defined in \eqref{eq:CaCap}.
Together with \eqref{eq:uSobolev}, we obtain}
\tsum_{i=1}^5|W_i(\zz)|
&\le
6\Caprime\Ca\CqEEApp^2 \Fconst_{\NN}(\zz)
+
\CqEEApp|\CcomTransv{\AmatLin,\Ccomp}(\zz)|
+
\CqEEApp\big(\CFsob+\Caprime\big)\|\Fvec\|_{\sHX^{\NN-1}(\homMfd_{\zz})}
\end{align*}
Integrating over \smash{$\int_{\zz_0}^{\zz_1} d\zz$} we obtain, using
\ref{item:Fsmall}, \ref{item:InhSmall}, \ref{item:KappaTransv},
and $\zz_0\le\zz_1$:
\begin{align*}
\tint_{\zz_0}^{\zz_1}\tsum_{i=1}^5|W_i(\zz)|\,d\zz
\le
6\Caprime\Ca\CqEEApp^2 \Csmall
+
\CqEEApp \CLA
+
\CqEEApp\CLA\big(\CFsob+\Caprime\big)
\end{align*}
Thus \eqref{eq:WW} is bounded by 
$e^{6\Caprime\Ca\CqEEApp^2 \Csmall} \Ca$, see \eqref{eq:CaCap}.
Thus \eqref{eq:apriori2gap} holds
under the admissible smallness assumption
$\smash{(6\Caprime\Ca\CqEEApp^2)\Csmall \le \log(2)}$.
\end{itemize}
This proves \eqref{eq:apriorigap}.
Now \eqref{eq:uallder} follows from \eqref{eq:HallHt}, \eqref{eq:apriori1gap}.
Thus the claim holds.

For $\zzmax\le0$ let $Z(\zzmax)$ be the statement 
$$
Z(\zzmax) 
:\;
\text{
There exists $u\in C^\infty_c(\bar\homMfd_{\le\zzmax},\R^{\nn})$
that satisfies \eqref{eq:apriori}.}
$$

\claimheader{Claim:}
The statement $Z(0)$ is true.

\claimheader{Proof of claim:}
Define 
$$
I = 
\{\zzmax\in (-\infty,0] \mid \text{$Z(\zzmax)$ is true} \}
$$
Note that $\zzmax\in I$ implies 
$(-\infty,\zzmax]\subset I$.
We make an open-closed argument to show that $I=(-\infty,0]$:
\begin{itemize}
\item 
$I$ is nonempty:
Since $\Fvec$ has compact support, see \eqref{eq:Fcomp}, 
there exists $\zzmax\le0$ with $\Fvec|_{\bar\homMfd_{\le\zzmax}}=0$.
Then $u=0$ satisfies \eqref{eq:apriori} on $\bar\homMfd_{\le\zzmax}$,
where \eqref{eq:aprioriSHS}, \eqref{eq:aprioriData},
\eqref{eq:apriori3NEW}, \eqref{eq:apriori1} are immediate and
where \eqref{eq:apriori2} holds because
\[ 
\tfrac{\kprop^{\LMat,\Amat(0),0,\Ccomp}_{\NN,u,\CqEEApp}(\zz_1,\zz_0)}{\kprop_{\NN}(\zz_1,\zz_0)}
=
e^{\CqEEApp \int_{\zz_0}^{\zz_1} |\CcomTransv{\AmatLin,\Ccomp}(\zz)|d\zz}
\le
e^{\CqEEApp \CLA}
\le
\Ca
\]
by \ref{item:KappaTransv} and \eqref{eq:CaCap}.
Hence $\zzmax\in I$.
\item 
$I$ is open in $(-\infty,0]$: Let $\zzmax\in I$ with $\zzmax<0$
and let $u$ be the solution that satisfies \eqref{eq:apriori} on $\bar\homMfd_{\le\zzmax}$.
Then $u$ also satisfies \eqref{eq:apriorigap} on $\bar\homMfd_{\le\zzmax}$. 

\claimheader{Claim:} There exists $\zzmax'\in(\zzmax,0]$ and 
$u'\in C^\infty_c(\bar\homMfd_{\le\zzmax'},\R^{\nn})$
such that $u'=u$ on $\bar\homMfd_{\le\zzmax}$ and $u'$ satisfies
\eqref{eq:aprioriSHS}, \eqref{eq:aprioriData}, 
\eqref{eq:apriori3NEW} on $\bar\homMfd_{\le\zzmax'}$.

\claimheader{Proof of claim (sketch):}
This essentially follows from local well-posedness of symmetric
hyperbolic systems \cite[Section 16.1-16.2]{Taylor3}.
Here we indicate in particular how to deal with the
boundaries at $\ttcoord=0$ and $\ttcoord=1$.
We say that a one-form $\theta$ is positive at a point $p$, 
if for all $w\in\R^{\nn}$ with $\sqrt{w^Tw}\le\delta$ one has
$\theta((\AmatLin^i+\AmatBil^i(w))X_i)>0$ at $p$,
analogously for nonnegative.
We use a configuration of triangles $T_0,T_1,T_2\subset\bar\homMfd$
as indicated in Figure \ref{fig:triangles}.
The triangles are closed in $\bar\homMfd$,
and chosen sufficiently flat so that:
\begin{align}\label{eq:triangleflat}
\begin{aligned}
&\text{At every point on the boundary components indicated by the}\\[-1mm]
&\text{dashed lines, the outward pointing normal one-form is positive.}
\end{aligned}
\end{align}
This can be achieved by \ref{item:Apos}.
We construct $u'$ separately on each triangle, 
where we can use \cite[Section 16.1-16.2]{Taylor3}\footnote{%
The results in the reference are stated on an interval times the torus,
in our applications one can always reduce to this case 
by using partitions of unity and finite speed of propagation.}:
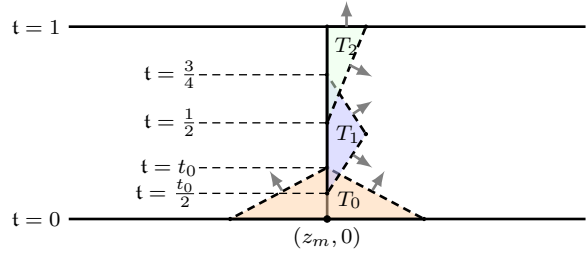
\begin{figure}
\centering
\begin{tikzpicture}[inner sep=0pt,scale=0.85]
\node (tip0) at (4,0.8) {}; %
\node (l0) at (2.5,0) {}; %
\node (r0) at (5.5,0) {}; %
\node (cen0) at (4.15,0.3) {};
\node (tip1) at (4.6,1.325) {}; %
\node (u1) at (4,2.25) {}; %
\node (l1) at (4,0.4) {}; %
\node (cen1) at (4.13,1.325) {};
\node (tip2) at (4.6,3) {}; %
\node (u2) at (4,3) {}; %
\node (l2) at (4,1.5) {}; %
\node (cen2) at (4.11,2.7) {};
\foreach \x in {0,3}
\draw[line width=1 pt] (0,\x) -- (8,\x);
\draw[line width=1 pt] (4,0) -- (4,3);
\node[anchor=east,xshift=-1mm] at (0,0) {\footnotesize$\ttcoord=0$};
\node[anchor=east,xshift=-1mm] at (0,3) {\footnotesize$\ttcoord=1$};
\fill[fill=orange!35,opacity=0.5] (tip0.center)--(l0.center)--(r0.center);
   \path[draw,line width=1 pt,densely dashed] (tip0)--(l0);
   \path[draw,line width=1 pt] (l0)--(r0);
   \path[draw,line width=1 pt,densely dashed] (r0)--(tip0);
\draw[color=black, fill=black] (tip0) circle (.025);
   \draw[color=black, fill=black]  (r0) circle (.025);
   \draw[color=black, fill=black]  (l0) circle (.025);
\draw[-latex,line width=1 pt,gray] (3.3,0.445) -- (3.11,0.75);
\draw[-latex,line width=1 pt,gray] (4.7,0.445) -- (4.89,0.75);
\fill[fill=blue!25,opacity=0.5] (tip1.center)--(u1.center)--(l1.center);
   \path[draw,line width=1 pt,densely dashed] (tip1)--(u1);
   \path[draw,line width=1 pt] (u1)--(l1);
   \path[draw,line width=1 pt,densely dashed] (l1)--(tip1);
\draw[color=black, fill=black] (tip1) circle (.025);
   \draw[color=black, fill=black]  (l1) circle (.025);
   \draw[color=black, fill=black]  (u1) circle (.025);
\draw[-latex,line width=1 pt,gray] (4.4,1.64) -- (4.7,1.84);
\draw[-latex,line width=1 pt,gray] (4.4,1.) -- (4.7,0.8);
\fill[fill=green!10,opacity=0.5] (tip2.center)--(u2.center)--(l2.center);
   \path[draw,line width=1 pt] (tip2)--(u2);
   \path[draw,line width=1 pt] (u2)--(l2);
   \path[draw,line width=1 pt,densely dashed] (l2)--(tip2);
\draw[color=black, fill=black] (tip2) circle (.025);
   \draw[color=black, fill=black]  (l2) circle (.025);
   \draw[color=black, fill=black]  (u2) circle (.025);
\draw[-latex,line width=1 pt,gray] (4.35,2.4) -- (4.7,2.22);
\draw[-latex,line width=1 pt,gray] (4.3,3) -- (4.3,3.4);
\draw[line width=0.5 pt,densely dashed] (4,0.8) -- (2,0.8)
	node[anchor=east] {\footnotesize $\ttcoord=\tt_0$};
\draw[line width=0.5 pt,densely dashed] (4,0.4) -- (2,0.4)
	node[anchor=east] {\footnotesize $\ttcoord=\frac{\tt_0}{2}$};
\draw[line width=0.5 pt,densely dashed] (4,1.5) -- (2,1.5)
	node[anchor=east] {\footnotesize $\ttcoord=\frac12$};
\draw[line width=0.5 pt,densely dashed] (4,2.25) -- (2,2.25)
	node[anchor=east] {\footnotesize $\ttcoord=\frac34$};
\node[anchor=west] at (cen0) {\footnotesize$T_0$};
\node[anchor=west] at (cen1) {\footnotesize$T_1$};
\node[anchor=west] at (cen2) {\footnotesize$T_2$};
\draw[color=black, fill=black] (4,0) circle (.05);
\node[anchor=north,yshift=-1mm] at (4,0) {\footnotesize{$(\zzmax,0)$}};
\end{tikzpicture}
\captionsetup{width=115mm}
\caption{Depicted are the first two factors of $\bar\homMfd=(-\infty,0]\times[0,1]\times\Mcpt$,
with the (closed) triangles $T_0,T_1,T_2\subset\bar\homMfd$.
The arrows indicate outward pointing normal one-forms.
The triangles are sufficiently flat so that, 
on the dashed boundary components, the contraction of these
one-forms with $(\AmatLin^i+\AmatBil^i(w))X_i$
is positive whenever $\sqrt{w^Tw}\le\delta$.
On the boundary component of $T_2$ that intersects $\ttcoord=1$
this contraction is nonnegative by \eqref{eq:Atpos}.}
\label{fig:triangles}
\end{figure}%
\begin{itemize}
\item 
$T_0$:
Since $\AmatLin^i, \AmatBil^i$ are symmetric
and $d\ttcoord$ is positive along $\ttcoord=0$, 
there exists a triangle $T_0$ as in Figure \ref{fig:triangles} 
and $u'_0\in C^\infty(T_0,\R^{\nn})$
that satisfies \eqref{eq:aprioriSHS}, \eqref{eq:aprioriData},
and whose $\Ct^{\lfloor(\NN+1)/2\rfloor}$-norm is bounded by $\min\{\delta,\CLA\}$ 
(by continuity, \eqref{eq:aprioriData}, \eqref{eq:X1..Tang}).
We choose $T_0$ sufficiently flat, meaning that
$\tt_0>0$ in Figure \ref{fig:triangles} is sufficiently small,
so that it satisfies \eqref{eq:triangleflat}.
Then $u'_0=u$ on the overlap $T_0\cap\homMfd_{\le\zzmax}$,
by a finite speed of propagation argument using the fact that 
$d\zzeta$ is positive along $\homMfd_{\zzmax}$ (c.f.~Theorem \ref{thm:AbstractUniqueness}).
In particular, $u'_0$ extends $u$ smoothly.
\item
$T_1$: 
Since $d\zzeta$ is positive along $\homMfd_{\zzmax}$
there exists 
$u'_1\in C^\infty(T_1,\R^{\nn})$
that satisfies \eqref{eq:aprioriSHS}, $u'_1=u$ along $\homMfd_{\zzmax}$,
and whose $\Ct^{\lfloor(\NN+1)/2\rfloor}$-norm is bounded by $\min\{\delta,\CLA\}$ 
(by continuity and \eqref{eq:apriori3gapNEW}).
The triangle $T_1$ is chosen such that
it overlaps with $T_0$ (see Figure \ref{fig:triangles})
and such that it satisfies \eqref{eq:triangleflat}.
Then $u'_1=u'_0$ on $T_0\cap T_1$ by finite speed of propagation.
Clearly $u'_1$ extends $u$ smoothly.

\item 
$T_2$: 
Since $d\zzeta$ is positive along $\bar\homMfd_{\zzmax}$,
and $d\ttcoord$ is nonnegative along $\ttcoord=1$,
there exists a triangle $T_2$ (closed in $\bar\homMfd$) 
and $u'_2\in C^\infty(T_2,\R^{\nn})$
(in particular $u'_2$ is smooth up to $\ttcoord=1$)
that satisfies \eqref{eq:aprioriSHS}, $u'_2=u$ along $\homMfd_{\zzmax}$,
and whose $\Ct^{\lfloor(\NN+1)/2\rfloor}$-norm is bounded by $\min\{\delta,\CLA\}$ 
(by continuity and \eqref{eq:apriori3gapNEW})%
\footnote{%
To apply the standard local existence results in \cite{Taylor3},
smoothly extend \eqref{eq:aprioriSHS} across $\ttcoord=1$ 
such that $d\zzeta$ is positive on the extension,
and smoothly extend the initial data \new{$u|_{\homMfd_{\zzmax}}$}.
Since the pointwise $\ell^2$-norm of $u|_{{\homMfd}_{\zzmax}}$ is less than $\delta/2$,
we may assume that the pointwise $\ell^2$-norm of the extension is less than $\delta$.
The restriction of the solution to $T_2$ is independent of the
extension by finite speed of propagation, using the fact that
$d\ttcoord$ is nonnegative along $\ttcoord=1$.}.
We choose $T_2$ such that it overlaps with $T_1$
and satisfies \eqref{eq:triangleflat}.
Then $u'_2=u'_1$ on $T_1\cap T_2$ by finite speed of propagation.
Further $u'_2$ extends $u$ smoothly:
Clearly the extension is continuous, and smooth
away from the intersection of $\bar\homMfd_{\zzmax}$ with $\ttcoord=1$.
This implies that the extension is in fact smooth also at this intersection,
using the fact that both $u'_2$ and $u$ are smooth there.
\end{itemize}
Now the claim follows by choosing $\zzmax'\in(\zzmax,0]$ sufficiently small.

We check that $u'$ also satisfies 
\eqref{eq:apriori1}, \eqref{eq:apriori2},
where we make $\zzmax'$ smaller if necessary.
\eqref{eq:apriori1}: By \eqref{eq:apriori1gap}, continuity,
and making $\zzmax'$ smaller if necessary.
\eqref{eq:apriori2}:
This estimate depends on two parameters $\zz_0,\zz_1$,
thus we cannot just argue by continuity. Abbreviate
\[
Q(\zz_1,\zz_0) 
= \tfrac{\kprop^{\LMat,\Amat(u'),\Fvec,\Ccomp}_{\NN,u',\CqEEApp}(\zz_1,\zz_0)}{\kprop_{\NN}(\zz_1,\zz_0)}
\]
Since $u'$ agrees with $u$ on $\bar\homMfd_{\le\zzmax}$,
\eqref{eq:apriori2gap} yields
\begin{align}
	\text{For all $\zz_0\le\zz_1\le\zzmax$}:
	&&
	Q(\zz_1,\zz_0) &\le 2\Ca
	\label{eq:Qcase1}
\intertext{
We have 
$Q(\zz_1,\zz_0) = \exp(\int_{\zz_0}^{\zz_1} W(\zz)d\zz)$
where $W(\zz)$ is continuous, c.f.~\eqref{eq:WW}.
Thus by making $\zzmax'$ smaller if necessary we obtain that}
	\text{For all $\zzmax\le \zz_0\le \zz_1\le \zzmax'$}:
	&&
	Q(\zz_1,\zz_0) &\le 1+\tfrac{1}{10} \le 3\Ca
	\label{eq:110}
\intertext{
where the last inequality uses $\Ca\ge1$.
Further we obtain that
} 
	\text{For all $\zz_0\le \zzmax\le \zz_1\le \zzmax'$}:
	&&
	Q(\zz_1,\zz_0)
	&=
	Q(\zz_1,\zzmax)Q(\zzmax,\zz_0)\nonumber\\
	&&&\le
	( 1+\tfrac{1}{10})2\Ca 
	\le 
	3\Ca\nonumber
\end{align}
where we use \eqref{eq:110} to bound $Q(\zz_1,\zzmax)$
and \eqref{eq:Qcase1} to bound $Q(\zzmax,\zz_0)$.
This shows that $u'$ also satisfies \eqref{eq:apriori2}.
Thus $u'$ satisfies \eqref{eq:apriori}, which shows 
$\zzmax'\in I$. Then $(-\infty,\zzmax']\subset I$,
which shows that $I$ is open in $(-\infty,0]$.
\item 
$I$ is closed in $(-\infty,0]$:
Let $\zzmax\in\bar{I}$.
Then there exists a smooth $u$ on 
$(-\infty,\zzmax)\times[0,1]\times \Mcpt$
that satisfies \eqref{eq:apriori} and vanishes for large negative $\zzeta$
(using standard uniqueness, c.f.~proof of Part 0).
Then a persistence of regularity argument 
(essentially the energy estimate \eqref{eq:EkkestX} of Part 2
restricted to $\zz<\zzmax$)
shows that $u$ extends smoothly to $\zz=\zzmax$.
Then \eqref{eq:apriori} holds up to $\zz=\zzmax$ by continuity.
Thus $\zzmax\in I$. 
\end{itemize}
Thus $I = (-\infty,0]$. 
Thus $0\in I$, and thus $Z(0)$ is true, which proves the claim.

Since $Z(0)$ is true, there exists $u$ as in \eqref{eq:ucpt}
that satisfies \eqref{eq:apriori} with $\zzmax=0$,
and thus also satisfies \eqref{eq:apriorigap} and \eqref{eq:uallder}
with $\zzmax=0$.
Thus $u$ satisfies \eqref{eq:ubasicNEW}, 
which concludes the proof of existence.
Further it satisfies \eqref{eq:uconclNEW} of Part 1 with 
$$
\Clarge = \max\{ 3\Ca\CqEEApp,  \CqEEApp(1+3\Ca\CqEEApp) \}
$$
It remains to check the last statement of Part 1.
Using the propagator property 
$\kprop_{\NN}(\zz,\zz')=\kprop_{\NN}(\zz,0)\kprop_{\NN}(0,\zz')$,
we indeed obtain 
\begin{align}\label{eq:fconstbd}
\begin{aligned}
\Fconst_{\NN}(\zz)
&=
\textstyle
\kprop_{\NN}(\zz,0) 
\int_{-\infty}^{\zz}\kprop_{\NN}(0,\zz') 
(1+|\zz-\zz'|)^{\NN}
\|\Fvec\|_{\sHX^{\NN}(\homMfd_{\zz'})}
\,d\zz'\\
&\le
\textstyle
\kprop_{\NN}(\zz,0) 
\int_{-\infty}^{0}\kprop_{\NN}(0,\zz') 
(1+|\zz'|)^{\NN}
\|\Fvec\|_{\sHX^{\NN}(\homMfd_{\zz'})}
\,d\zz'\\
&=
\kprop_{\NN}(\zz,0)\Fconst_{\NN}(0)
\end{aligned}
\end{align}

\proofheader{Proof of Part 2.}
We prove that $u$ in \eqref{eq:ucpt} satisfies Part 2 of the theorem. 
We proceed by induction in $\kk\ge\NN$.
For $\kk\ge\NN$ let $\PP_{\kk}$ be the statement
\begin{align*}
\PP_{\kk}:\;\;
\text{
For all $\CHigherIn\ge0$, if 
\ref{item:Higherassp12}$_{k,\CHigherIn}$,
\ref{item:InhSmallHigher}$_{k,\CHigherIn}$,
\ref{item:Higherassp3}$_{k,\CHigherIn}$
then \eqref{eq:HigherAbstract}$_{k,\CHigherIn}$
}
\end{align*}
where, for example, \ref{item:Higherassp12}$_{k,\CHigherIn}$
means \ref{item:Higherassp12} with parameters $k$ and $\CHigherIn$.
The base case $\PP_{\NN}$ holds by Part 1.
For the induction step we fix $\kk>\NN$, 
and show that $\PP_{\kk-1}$ implies $\PP_{\kk}$.
Let $\CHigherIn\ge0$ and assume that 
\ref{item:Higherassp12}$_{k,\CHigherIn}$,
\ref{item:InhSmallHigher}$_{k,\CHigherIn}$, 
\ref{item:Higherassp3}$_{k,\CHigherIn}$ hold.
Then also 
\ref{item:Higherassp12}$_{k-1,\CHigherIn}$,
\ref{item:InhSmallHigher}$_{k-1,\CHigherIn}$, 
\ref{item:Higherassp3}$_{k-1,\CHigherIn}$ hold,
using the fact that $\Fconst_{k}(\zz)$ is increasing in $k$.
Hence by the induction hypothesis \eqref{eq:HigherAbstract}$_{k-1,\CHigherIn}$ holds.

By \eqref{eq:sobolevMz} and 
$\lfloor\frac{k+1}{2}\rfloor+\lfloor\frac{\MdimNEW}{2}\rfloor < k-1$
(use $k-1\ge\NN\ge\MdimNEW+3$),
for all $\zz\le0$:
\begin{align}
\|u\|_{\sCX^{\lfloor\frac{k+1}{2}\rfloor}(\homMfd_{\zz})}
&\lesssim_{\CM,\nn,k}
\|u\|_{\sHX^{k-1}(\homMfd_{\zz})}\nonumber
\intertext{
Together with \eqref{eq:EkkestX}$_{k-1,\CHigherIn}$ and 
\ref{item:Higherassp12}$_{k,\CHigherIn}$,
\ref{item:InhSmallHigher}$_{k,\CHigherIn}$,
we obtain} 
\|u\|_{\sCX^{\lfloor\frac{k+1}{2}\rfloor}(\homMfd_{\zz})}
&\lesssim_{\CM,\nn,\kk,\Cpos,\CLA,\CHigherIn}
\Fconst_{\kk-1}(\zz) + \|\Fvec\|_{\sHX^{\kk-2}(\homMfd_{\zz})}
\le
2\CHigherIn\label{eq:uCtkF} 
\intertext{
Also by \eqref{eq:sobolevMz}, and then using 
\ref{item:InhSmallHigher}$_{k,\CHigherIn}$,}
\|\Fvec\|_{\sCX^{\lfloor\frac{k+1}{2}\rfloor-1}(\homMfd_{\zz})}
&\lesssim_{\CM,\nn,k}
\|\Fvec\|_{\sHX^{k-2}(\homMfd_{\zz})}
\le 
\CHigherIn\label{eq:Fsobk}
\end{align}
Thus there exists 
\begin{equation}\label{eq:slbchoice}
\slashCHigherIn
\ge
\max\{\CHigherIn,\CLA\}
\end{equation}
that depends only on $\CM,\nn,\kk,\Cpos,\CLA,\CHigherIn$ such that
for all $\zz\le0$:
\begin{equation}\label{eq:uCtk}
\|u\|_{\sCX^{\lfloor\frac{k+1}{2}\rfloor}(\homMfd_{\zz})},\ 
\|\Fvec\|_{\sCX^{\lfloor\frac{k+1}{2}\rfloor-1}(\homMfd_{\zz})}
\;\le\; \slashCHigherIn
\end{equation}

We use Lemma \ref{lem:Nonlindiffenergyineq} with the parameters 
in the third column of Table \new{\ref{tab:NonlinAppN}}.
Let $\CqEEApp_{\kk}$ be the constant produced by Lemma \ref{lem:Nonlindiffenergyineq} (called $\CqEE$ there), 
which depends only on $\CM,\nn,\kk,\Cpos,\CLA,\CHigherIn$.
We check that the assumptions \eqref{eq:nonlinassp} hold:
\eqref{eq:SHSBil} holds by \eqref{eq:ueq};
\eqref{eq:DataBil} holds by \eqref{eq:udata};
\eqref{eq:A(zeta)posBIL}, 
\eqref{eq:A(0)posBIL},
\eqref{eq:A(t)posBIL}
hold by \eqref{eq:uinf} and \ref{item:Apos};
\eqref{eq:Cubound}, \eqref{eq:CFbound} hold by \eqref{eq:uCtk};
\eqref{eq:matbounds} holds by \ref{item:Higherassp3}$_{k,\CHigherIn}$
and \ref{item:CNnormsSHS} (for $\Ccomp$) and \eqref{eq:slbchoice}.

By Part 2 of Lemma \ref{lem:Nonlindiffenergyineq}, for all $\zz\le0$:
\begin{equation}\label{eq:induHest}
\|u\|_{\Ht^{\kk}(\homMfd_{\zz})}
\le
\CqEEApp_{\kk} \tint_{-\infty}^{\zz} 
\kprop^{\LMat,\Amat(u),\Fvec,\Ccomp}_{\kk,u,\CqEEApp_{\kk}}(\zz,\zz')
(1 + |\zz-\zz'|)^{\kk}
\|\Fvec\|_{\sHX^{\kk}(\homMfd_{\zz'})}
d\zz'
\end{equation}
In the following we abbreviate 
$\lesssim_{\CM,\nn,\kk,\Cpos,\CLA,\CHigherIn}$
by $\lesssim_{*}$.

Similarly to \eqref{eq:WW},
one obtains that for all $\zz_0\le\zz_1\le0$:
\begin{align}\label{eq:logpk}
&\log\Big(\tfrac{\kprop^{\LMat,\Amat(u),\Fvec,\Ccomp}_{\kk,u,\CqEEApp_{\kk}}(\zz_1,\zz_0)}{\kprop_{\kk}(\zz_1,\zz_0)}\Big)\\
&\;\;\;\lesssim_{*}
\tint_{\zz_0}^{\zz_1}
\Big(
|\CcomTransv{\AmatLin,\Ccomp}(\zz')|
+
\|u\|_{\sCX^{\lfloor \frac{k+1}{2} \rfloor}(\homMfd_{\zz'})}
+
\|\Fvec\|_{\sCX^{\lfloor \frac{k+1}{2}\rfloor-1}(\homMfd_{\zz'})}\Big)\,d\zz'
\nonumber
\end{align}
By \ref{item:KappaTransv} we have
$\int_{\zz_0}^{\zz_1}|\CcomTransv{\AmatLin,\Ccomp}(\zz')|d\zz'\le\CLA$.
By \eqref{eq:uCtkF}, \ref{item:Higherassp12}$_{k,\CHigherIn}$, \ref{item:InhSmallHigher}$_{k,\CHigherIn}$,
\begin{equation}
\tint_{\zz_0}^{\zz_1}
\|u\|_{\sCX^{\lfloor \frac{k+1}{2} \rfloor}(\homMfd_{\zz'})}
\;d\zz'
\lesssim_*
\tint_{\zz_0}^{\zz_1}
(\Fconst_{\kk-1}(\zz') + \|\Fvec\|_{\sHX^{\kk-2}(\homMfd_{\zz'})})
\;d\zz'
\le
2\CHigherIn
\end{equation}
By \eqref{eq:Fsobk} 
and \ref{item:InhSmallHigher}$_{k,\CHigherIn}$,
\begin{align*}
\tint_{\zz_0}^{\zz_1}
\|\Fvec\|_{\sCX^{\lfloor \frac{k+1}{2}\rfloor-1}(\homMfd_{\zz'})}\,d\zz'
\lesssim_*
\tint_{\zz_0}^{\zz_1}
\|\Fvec\|_{\sHX^{k-2}(\homMfd_{\zz'})}\,d\zz'
\le \CHigherIn
\end{align*}
Thus for all $\zz_0\le\zz_1\le0$:
\[ 
\log\Big(\tfrac{\kprop^{\LMat,\Amat(u),\Fvec,\Ccomp}_{\kk,u,\CqEEApp_{\kk}}(\zz_1,\zz_0)}{\kprop_{\kk}(\zz_1,\zz_0)}\Big)
\lesssim_* 1
\]
Together with \eqref{eq:induHest}, this implies that for all $\zz\le0$:
\begin{equation*}
\|u\|_{\Ht^{\kk}(\homMfd_{\zz})}
\lesssim_*
\tint_{-\infty}^{\zz} 
\kprop_{\kk}(\zz,\zz')
(1 + |\zz-\zz'|)^{\kk}
\|\Fvec\|_{\sHX^{\kk}(\homMfd_{\zz'})}
d\zz'
=
\Fconst_{\kk}(\zz)
\end{equation*}
This proves \eqref{eq:Ekkest}$_{k,\CHigherIn}$.
With \eqref{eq:HallHt} in Lemma \ref{lem:Nonlindiffenergyineq}
(see Table \new{\ref{tab:NonlinAppN}}) we obtain
\begin{align*}
\|u\|_{\sHX^{\kk}(\homMfd_{\zz})}
\lesssim_*
\Fconst_{\kk}(\zz) + \|\Fvec\|_{\sHX^{k-1}(\homMfd_{\zz})}
\end{align*}
which proves \eqref{eq:EkkestX}$_{k,\CHigherIn}$.
This concludes the induction step.

Analogously to \eqref{eq:fconstbd} one checks that 
\begin{equation}\label{eq:fconstbdk}
\Fconst_{\kk}(\zz)\le \kprop_{\kk}(\zz,0)\Fconst_{\kk}(0)
\end{equation}
which concludes the proof of Part 2.\qed
\end{proof}

Below we prove Theorem \ref{thm:nonlinEE} in the general case, by reducing 
it to Theorem \ref{thm:nonlinEE} in the compact support case,
that we have just proven.
To do this we cut off the source term $\Fvec$ at large negative $\zzeta$,
use the fact that the restriction of the solution $u$ to each cone in Figure \ref{fig:conepic} 
depends only on the restriction of $\Fvec$ to that cone,
and use the fact that the estimates for $u$ in the compact support case
are uniform in, i.e.~do not depend on, the size of the support of $\Fvec$.
\begin{proof}[of Theorem \ref{thm:nonlinEE}]
\proofheader{Proof of Part 0.} 
Same as the proof of Part 0 in the compact support case.

\proofheader{Proof of existence and Part 1.} 
Fix a smooth cutoff $\theta:\R \to [0,1]$ with
$\theta(x)=1$ when $x \ge\frac12$
and $\theta(x)=0$ when $x\le 0$.
Then for all $\zcut\le0$, all $\jj\in\Z_{\ge0}$, all $\zz\le0$
and all $f\in C^\infty(\bar\homMfd,\R^{\nn})$:
\begin{align}\label{eq:theta*f}
\|\theta(\zzeta-\zcut) f\|_{\sHX^{\jj}(\homMfd_{\zz})} 
	\lesssim_{\CM,\nn,\jj,\theta} \|f\|_{\sHX^{\jj}(\homMfd_{\zz})}
\end{align}
where we use 
$\|\theta(\zzeta-\zcut)\|_{\sCX^{\jj}(\homMfd_{\zz})}
	\lesssim_{\CM,\jj,\theta}1$ 
and the Leibniz rule.
For each $k_0\in\Z_{\ge0}$ fix a constant 
$C_{k_0,\theta}\ge1$ such that
\eqref{eq:theta*f} holds with inequality for all $\jj\le k_0$:
\begin{equation}\label{eq:CthetaN}
\|\theta(\zzeta-\zcut) f\|_{\sHX^{\jj}(\homMfd_{\zz})} 
	\le C_{k_0,\theta} \|f\|_{\sHX^{\jj}(\homMfd_{\zz})}
\qquad
\text{for all $\jj\le k_0$}
\end{equation}
The constant $C_{k_0,\theta}$ depends only on $\CM,\nn,k_0,\theta$.

For every $\zcut\le 0$ define 
$$
\Fvec_{\zcut} = \theta(\zzeta-\zcut) \Fvec \;\; \in\;\;  C^\infty_c(\bar\homMfd,\R^{\nn})
$$
which vanishes on $\bar\homMfd_{\le \zcut}$;
and define $\Fconst_{\zcut,\kk}(\zz)$ analogously to \eqref{eq:Fconstdef} but
with $\Fvec$ replaced by $\Fvec_{\zcut}$.
By \eqref{eq:CthetaN}, for all $\zz\le0$ and all $k_0\in\Z_{\ge0}$:
\begin{subequations}\label{eq:FvecFconstCut}
\begin{align}
\|\Fvec_{\zcut}\|_{\sHX^{\jj}(\homMfd_{\zz})} 
	\;&\le\; C_{k_0,\theta}\,\|\Fvec\|_{\sHX^{\jj}(\homMfd_{\zz})}
&& \text{for all $\jj\le k_0$}
\label{eq:datasnosFvec}\\
\Fconst_{\zcut,\jj}(\zz) \;&\le\; C_{k_0,\theta}\,\Fconst_{\jj}(\zz)
&& \text{for all $\jj\le k_0$}
\label{eq:datasnosFconst}
\end{align}
\end{subequations}

We will use Theorem \ref{thm:nonlinEE} in the compact support case
with the parameters in Table \ref{tab:reduction},
where the input $\Fvec_{\zcut}$ depends parametrically on $\zcut\le0$.
Let $\Clargecpt$ and $\Csmallcpt$ be the constants
produced by the theorem in the compact support case, 
where $\Clargecpt$ depends only on $\CM,\nn,\NN,\Cpos,\CLA$,
and $\Csmallcpt$ depends only on $\CM,\nn,\NN,\Cpos,\CLA,\delta$.
It is important that they do not depend on $\zcut$.
\begin{table}
\centering
\begin{tabular}{cc|c}
&
	\begin{tabular}{@{}c@{}}
	Parameters in Theorem \ref{thm:nonlinEE} \\ 
	in compact support case
	\end{tabular}
	&
	\begin{tabular}{@{}c@{}}
	Parameters used to invoke\\ 
	Theorem \ref{thm:nonlinEE} in compact support case
	\end{tabular}
	\\
\hline
Input
&$\Mcpt$, $X_0,\dots,X_{\MdimNEW}$, $\muM$
	&$\Mcpt$, $X_0,\dots,X_{\MdimNEW}$, $\muM$\\
&$\nn$, $\NN$, $\Cpos$, $\CLA$ 
	& $\nn$, $\NN$, $\Cpos$, $C_{\NN,\theta} \CLA$\\
&$\delta$ 
	& $\delta$ \\
&$\AmatLin^i$, $\AmatBil^i$, $\LMat$, $\BSHS$, $\Fvec$, $\Ccomp$
	& $\AmatLin^i$, $\AmatBil^i$, $\LMat$, $\BSHS$, $\Fvec_{\zcut}$, $\Ccomp$\\
	&$k$, $\CHigherIn$ (Part 2 only)
	& $k$, $C_{k,\theta}\CHigherIn$\\
\hline
Output
	&$\Clarge$, $\Csmall$
	& $\Clargecpt$, $\Csmallcpt$
\end{tabular}
\captionsetup{width=115mm}
\caption{%
The first column lists the input and output parameters of Theorem \ref{thm:nonlinEE}
in the compact support case. 
The second column specifies the choice of the input parameters used to invoke
Theorem \ref{thm:nonlinEE} in the compact support case,
in terms of the input parameters of Theorem \ref{thm:nonlinEE} \new{in the general case}
and the parameters introduced in this proof.
The output parameters produced by this invocation of
Theorem \ref{thm:nonlinEE} in the compact support case
are denoted $\Clargecpt$, $\Csmallcpt$,
where $\Clargecpt$ depends only on the parameters in the first two rows,
and $\Csmallcpt$ only on those in the first three rows,
in particular they do not depend on $\zcut$.}
\label{tab:reduction}
\end{table}
%
We show that Theorem \ref{thm:nonlinEE} \new{in the general case} holds with 
\begin{equation}\label{eq:Clargesmallgeneral}
\Clarge = C_{\NN,\theta}\Clargecpt
\qquad
\Csmall = \frac{\Csmallcpt}{C_{\NN,\theta}}
\end{equation}
This is an admissible choice because 
$C_{\NN,\theta}$ depends only on $\CM,\nn,\NN,\theta$.

We first check that the assumptions of Theorem \ref{thm:nonlinEE} \new{in the general case}
imply the assumptions of the theorem in the compact support case
for every $\zcut\le0$
(see Table \ref{tab:reduction}).
\ref{item:Fsmall}: use \eqref{eq:datasnosFconst} with $k_0=\NN$ and \eqref{eq:Clargesmallgeneral};
\ref{item:InhSmall}: use \eqref{eq:datasnosFvec} with $k_0=\NN$,
\eqref{eq:Clargesmallgeneral} and the choice $C_{\NN,\theta}\CLA$ for $\CLA$
in the compact support case;
\ref{item:KappaTransv}: use $\CLA\le C_{\NN,\theta}\CLA$;
\ref{item:Apos}: clear;
\ref{item:CNnormsSHS}: use $\CLA\le C_{\NN,\theta}\CLA$.

Thus by Theorem \ref{thm:nonlinEE} in the compact support case
(existence, Part 0, Part 1), 
for every $\zcut\le0$ there exists a unique
\begin{equation}\label{eq:solution}
u_{\zcut} \in C^\infty_c(\bar\homMfd,\R^{\nn})
\end{equation}
that satisfies
\begin{subequations}\label{eq:propusall}
\begin{align}
(\AmatLin^i+\AmatBil^i(u_{\zcut})) X_i u_{\zcut} &= \LMat u_{\zcut} + \BSHS(u_\zcut,u_\zcut) + \Fvec_\zcut
\label{eq:useq}\\
u_\zcut|_{\ttcoord=0} &= 0 
\label{eq:usdata}\\
\sqrt{u\smash{_\zcut^T}u\smash{_{\zcut}}^{\mathclap{\phantom{T}}}} &\le \delta
&&\text{on $\homMfd$}
\label{eq:usinf}\\
\|u_\zcut\|_{\Ht^{\NN}(\homMfd_{\zz})} &\le
\Clarge \Fconst_{\NN}(\zz)
&&\text{for $\zz\le0$}
\label{eq:usEE}\\
\|u_{\zcut}\|_{\sHX^{\NN}(\homMfd_{\zz})}
&\le
\Clarge(\Fconst_{\NN}(\zz) + \|\Fvec\|_{\sHX^{\NN-1}(\homMfd_{\zz})})
&& \text{for $\zz\le0$}
\label{eq:usE}
\end{align}
\end{subequations}
where for \eqref{eq:usEE} and \eqref{eq:usE} we use \eqref{eq:FvecFconstCut}
with $k_0=\NN$ and \eqref{eq:Clargesmallgeneral}.

Let 
$u\in C^\infty(\homMfd,\R^{\nn})$
be the unique function such that for every $\tt\in(0,1)$:
\begin{align}\label{eq:urestr}
u|_{\cone_{0,\tt}} = u_{\zcut(\tt)}|_{\cone_{0,\tt}} 
\qquad \text{where $\zcut(\tt) = -\tfrac{2\Cpos^2\tt}{1-\tt}-1$}
\end{align}
where we abbreviate $\cone_{0,\tt}^{\Cpos} = \cone_{0,\tt}$,
and where $u_{\zcut(\tt)}$ is the solution \eqref{eq:solution} with $\zcut=\zcut(\tt)$.
The number $\zcut(\tt)$ is chosen so that 
$\Fvec_{\zcut(\tt)}=\Fvec$ on $\cone_{0,\tt}$, see Figure \ref{fig:conepic}.
The function $u$ exists because if $\tt_0,\tt_1\in(0,1)$ with $\tt_0\le\tt_1$
then $u_{\zcut(\tt_0)}=u_{\zcut(\tt_1)}$ on the overlap 
$\cone_{0,\tt_0}\cap\cone_{0,\tt_1}=\cone_{0,\tt_0}$,
by Theorem \ref{thm:AbstractUniqueness} and by 
$\Fvec_{\zcut(\tt_0)}=\Fvec_{\zcut(\tt_1)}=\Fvec$ on 
$\cone_{0,\tt_0}$.
Further \eqref{eq:urestr} defines $u$ uniquely because $\cup_{\tt\in(0,1)}\cone_{0,\tt} =\homMfd$.

We check that $u$ satisfies \eqref{eq:ubasicNEW} and \eqref{eq:uconclNEW}.
Clearly, \eqref{eq:useq}, \eqref{eq:usdata}, \eqref{eq:usinf} 
imply \eqref{eq:ueq}, \eqref{eq:udata}, \eqref{eq:uinf} respectively.
\eqref{eq:uETangDer}: Fix $\zz\le0$. 
Then\footnote{
Here $\|{\cdot}\|_{\Ht^{\NN}(\homMfd_{\zz}^{\smallconstant})}$
is defined like $\|{\cdot}\|_{\Ht^{\NN}(\homMfd_{\zz})}$
in Definition \ref{def:sHXDefinition},
but integrating over $\homMfd_{\zz}^{\smallconstant}$.
} 
\begin{align}\label{eq:uMuMeps}
\|u\|_{\Ht^{\NN}(\homMfd_{\zz})}
=
\lim_{\smallconstant\downarrow0}\|u\|_{\Ht^{\NN}(\homMfd_{\zz}^{\smallconstant})}
\end{align}
where $\homMfd_{\zz}^{\smallconstant}\subset\homMfd_{\zz}$ is the subset
where $\ttcoord\le 1-\smallconstant$.
For each small $\smallconstant>0$, 
choose $\tt_{\smallconstant}\in(0,1)$ sufficiently close to $1$ such that
$\homMfd_{\zz}^{\smallconstant}\subset\cone_{0,\tt_{\smallconstant}}$.
Then by \eqref{eq:urestr},
\begin{align*}
\|u\|_{\Ht^{\NN}(\homMfd_{\zz}^{\smallconstant})}
&=
\|u_{\zcut(\tt_{\smallconstant})}\|_{\Ht^{\NN}(\homMfd_{\zz}^{\smallconstant})}
\le
\|u_{\zcut(\tt_{\smallconstant})}\|_{\Ht^{\NN}(\homMfd_{\zz})}
\le
\Clarge \Fconst_{\NN}(\zz)
\end{align*}
where the last step uses \eqref{eq:usEE}, and in this step $\smallconstant$ drops out (uniformity). 
Together with \eqref{eq:uMuMeps} this shows \eqref{eq:uETangDer}.
\eqref{eq:uEAllDer}: Analogously, using \eqref{eq:usE}.
Clearly \eqref{eq:fconstbd}
also holds without the support assumption on $\Fvec$.

\proofheader{Proof of Part 2.}
Fix $k\in\Z_{\ge\NN}$ and $\CHigherIn>0$.
We check that the assumptions in Part 2
of Theorem \ref{thm:nonlinEE} \new{in the general case}
imply the assumptions in Part 2 
of the theorem in the compact support case
for every $\zcut\le0$ (see Table \ref{tab:reduction}).
\ref{item:Higherassp12}: use \eqref{eq:datasnosFconst} with $k_0=k$ and the choice
$C_{k,\theta}\CHigherIn$ for $\CHigherIn$ in the
compact support case;
\ref{item:InhSmallHigher}: use \eqref{eq:datasnosFvec} with $k_0=k$ and the choice $C_{k,\theta}\CHigherIn$;
\ref{item:Higherassp3}: use $\CHigherIn\le C_{k,\theta}\CHigherIn$.
Thus by Part 2 in the compact support case,
for all $\zcut\le0$ the solution \eqref{eq:solution} satisfies
\begin{align*}
\|u_\zcut\|_{\Ht^{\kk}(\homMfd_{\zz})}
\;&\lesssim_{\CM,\nn,\kk,\Cpos,\CLA,\CHigherIn,\theta}\;
\Fconst_{\kk}(\zz)\\
\|u_\zcut\|_{\sHX^{\kk}(\homMfd_{\zz})}
\;&\lesssim_{\CM,\nn,\kk,\Cpos,\CLA,\CHigherIn,\theta}\;
\Fconst_{\kk}(\zz) + \|\Fvec\|_{\sHX^{\kk-1}(\homMfd_{\zz})}
\end{align*}
where we also use \eqref{eq:FvecFconstCut} with $k_0=k$,
and the fact that $C_{k,\theta}$ depends only on $\CM,\nn,k,\theta$.
Then, by repeating the argument \eqref{eq:uMuMeps},
one obtains \eqref{eq:HigherAbstract}. 
Clearly \eqref{eq:fconstbdk}
also holds without the support assumption on $\Fvec$.
\qed
\end{proof}

\section{Construction near spacelike infinity}
\label{sec:SpaceinfConstruction}

We consider the Einstein equations \eqref{eq:MC} 
near spacelike infinity, in the form 
\begin{equation}\label{eq:eqv+c}
\dg(v+c) + \tfrac12[v+c,v+c] = 0
\qquad
c|_{y^0=0}=0
\end{equation}
where $v$ is given and $c$ is the unknown. 
We assume in particular that $v$ 
is smooth including along future null infinity (not at spacelike infinity $\spaceinf$),
and that it asymptotes to a solution towards $\spaceinf$,
in the sense that 
\begin{equation}\label{eq:vdecay}
\text{$\dg v + \tfrac12[v,v]$ decays like a power of $\s=2y^0+|\vec{y}|$ towards $\spaceinf$}
\end{equation}
The main result of this section is Proposition \ref{prop:ApplySHS},
where we show existence of $c$,
that it is smooth away from null and spacelike infinity,
and that its regularity along null infinity increases 
linearly with the rate of decay in \eqref{eq:vdecay}.

Proposition \ref{prop:ApplySHS} will be proven as an application of 
Theorem \ref{thm:nonlinEE}. 
It will be used in the proof of Theorem \ref{thm:main}, 
to construct $u$ near spacelike infinity,
where $v$ is chosen to be $\kerr$ plus an extension 
of $\udata-\kerrdata$ to $y^0\ge0$, see \eqref{eq:u=v+c} and \eqref{eq:MCforc}.

Section \ref{sec:SpaceinfConstruction} is organized as follows.
In Section \ref{sec:GeometrySpacelikeInf} we 
introduce geometric quantities
that allow to pass to the abstract setting of Section \ref{sec:Abstract};
in Section \ref{sec:basis_i0}, \ref{sec:norms_i0} we fix bases and norms;
in Section \ref{sec:gauge_definition_i0} we define a gauge,
used in Section \ref{sec:reformSHS_i0} to show that \eqref{eq:eqv+c} is 
quasilinear symmetric hyperbolic including along null infinity, 
up to constraints that propagate;
Proposition \ref{prop:ApplySHS} is in Section \ref{sec:existence_i0}.
\begin{remark}\label{rem:local4}
Some definitions in
this section are labeled 
'local to Section  \ref{sec:SpaceinfConstruction}',
by which we mean that they are only valid in Section \ref{sec:SpaceinfConstruction}.
For example, the basis \smash{$(\eg^k_i)$} in \eqref{eq:gbasis_i0}
is 'local to Section  \ref{sec:SpaceinfConstruction}'.
In particular it is not to be confused with the basis
\smash{$(\eg^k_i)$} in \eqref{eq:gbasis_bulk} in Section \ref{sec:bulk},
which is 'local to Section \ref{sec:bulk}',
and which uses the same symbol but is a basis of a different space.
\end{remark}
\subsection{Geometry}
\label{sec:GeometrySpacelikeInf}

We explain the geometry near spacelike infinity,
and make precise the identification with the abstract geometric setup 
in Section \ref{sec:AbstractGeom} (Convention \ref{conv:ztXmu_NEW}).

Recall the neighborhood $\diamondy\subset\cyl$ of spacelike infinity $\spaceinf$
in \eqref{eq:diamondy}. On $\diamondy$, the functions $y$ in \eqref{eq:yycoords}
are smooth coordinates, with $\spaceinf$ at the origin $y=0$. 
Define
\begin{align*}
\Dspcl = 
\left(\diamondy \cap \overline{\diamond}_+\right) \setminus \spaceinf
\end{align*}
This set includes the portion of future null infinity
near $\spaceinf$, but excludes $\spaceinf$.
Define the smooth function $\s:\Dspcl\to\R$ given by 
\begin{equation}\label{eq:sdef}
\s = 2y^0 + |\vec{y}|
\end{equation}
where $\vec{y} = (y^1,y^2,y^3)$ and $|\vec{y}|=((y^1)^2+(y^2)^2+(y^3)^2)^{\frac12}$.
For $s>0$ define 
\begin{align}\label{eq:deltas}
\begin{aligned}
\begin{aligned}
\Dspcl_{\le s} &= \{p\in\Dspcl\mid\s(p)\le s \}\\
\Dspcl_{s} &= \{p\in\Dspcl\mid\s(p)= s \}
\end{aligned}
\qquad\qquad
\begin{aligned}
\Dspop_{\le s} &= \{p\in\Dspcl\cap\diamond_+\mid\s(p)\le s \}\\
\Dspop_{s} &= \{p\in\Dspcl\cap\diamond_+\mid\s(p)= s \}
\end{aligned}
\end{aligned}
\end{align}
\new{and analogously for $\le$ replaced by $<$.}
The sets $\Dspcl_{\le s},\Dspcl_{s}$ do intersect future null infinity,
whereas $\Dspop_{\le s},\Dspop_{s}$ do not intersect future null infinity.
\begin{remark}\label{rem:Dstau}
If $s\in(0,1]$
then on $\Dspcl_{\le s}$: $\tau\in[0,\arctan(\frac23s)]$ and
$\xi^4\in [\frac{1-s^2}{1+s^2},1)$.
\end{remark}
The analysis near spacelike infinity will make use of
the $\R_+$-action in Section \ref{sec:R+action}.
The diffeomorphism $\Scal_{\lambda}:\cyl\to\cyl$ with $\lambda>0$
in \eqref{eq:scall} restricts to\footnote{
By abuse of notation, we also denote this map by 
$\Scal_{\lambda}$, analogously for $\Scalg_{\lambda}$ in \eqref{eq:Scalg_i0}.}
\begin{equation}\label{eq:ScalD}
\Scal_{\lambda}:\;\Dspcl\to\Dspcl
\qquad\text{where}\qquad
\Scal_{\lambda}^{\new{*}}y = \tfrac{1}{\lambda}y
\end{equation}
Thus the $\R_+$-action in Definition \ref{def:scalg} restricts to
\begin{equation}\label{eq:Scalg_i0}
\Scalg_{\lambda}:\;\gx(\Dspcl)\to\gx(\Dspcl)
\end{equation}
where $\gx(\Dspcl)$ is the space of smooth sections of $\gx$ on $\Dspcl$, 
c.f.~Remark \ref{rem:gdiamond}.
Recall that the operations \eqref{eq:gop}
commute with this action, see Lemma \ref{lem:homogeneity}.
Furthermore, \eqref{eq:Scalg_i0} restricts
to a map on $\lx(\Dspcl)$ and to a map on $\I(\Dspcl)$.
\begin{definition}[Homogeneous elements]\label{def:homog}
We say that $f\in C^\infty(\Dspcl)$ respectively $u\in \gx(\Dspcl)$
is homogeneous of degree $j\in\Z$ if and only if 
\begin{equation}\label{eq:homogeneouselements}
\Scal_{\lambda}^*f = \tfrac{1}{\lambda^j}f 
	\qquad\qquad \Scalg_{\lambda}u = \tfrac{1}{\lambda^j} u
\end{equation}
Analogously for vector fields, one-forms,
elements in $\lx(\Dspcl)$, elements in $\I(\Dspcl)$.
\end{definition}
Note that the coordinate functions $y^\mu$ are homogeneous of degree one.
\begin{definition}\label{def:ztXmudef}
On $\Dspcl$ we define:
\begin{subequations}
\begin{align}
&\text{The smooth functions $\zzeta = \log(\s)$ and $\ttcoord = \tfrac{y^0}{|\vec{y}|}$.}
	\label{eq:zeta,t}\\
&\text{The smooth vector fields $X_\mu=\s\p_{y^\mu}$ for $\mu=0\dots3$.}
	\label{eq:Xmi}\\
&\text{The smooth 4-density 
	$\muW = \tfrac{1}{\s^4}|dy^0\wedge\cdots\wedge dy^3|$.}
	\label{eq:muW}
\end{align}
Moreover, for every $s>0$ we define the smooth 3-density 
\begin{align}
\muWs 
= \s^{-4}|y^\mu\intermult_{\p_{y^\mu}}(dy^0\wedge\cdots\wedge dy^3)|
\in \dens{3}(\Dspcl_{s})
	\label{eq:muWs}
\end{align}
\end{subequations}
\end{definition}
Note that $X_\mu$, $\muW$  are homogeneous of degree zero.
Further $(\zzeta,\ttcoord,\frac{\vec{y}}{|\vec{y}|})$
are coordinates on $\Dspcl$,
and in these coordinates, \eqref{eq:ScalD} acts by translating $\zzeta$.

In order to apply the results in Section \ref{sec:Abstract}
we use the following convention.
\begin{convention}\label{conv:ztXmu_NEW}
We make the following choices
for the geometric quantities in Section \ref{sec:AbstractGeom}.
We choose $\Mcpt = S^2$.
Then 
$\homMfd = (-\infty,0] \times [0,1) \times S^2$.
We identify
\begin{align}\label{eq:ztcoords}
\Dspop_{\le1}\simeq \homMfd
\qquad 
\text{via the coordinates $(\zzeta ,\ttcoord,\tfrac{\vec{y}}{|\vec{y}|})$ on $\Dspop_{\le1}$}
\end{align}
using \eqref{eq:zeta,t},
see Figure \ref{fig:HomogeneousCoordinates}.
Via this identification one has, for all $s>0$,
\begin{equation}\label{eq:levelsets}
\Dspop_{\le s} \simeq \homMfd_{\le \log(s)}
\qquad
\Dspop_{s} \simeq \homMfd_{\log(s)}
\end{equation}
Furthermore:
\begin{itemize}
\item 
For the frame of vector fields \eqref{eq:abstractX} we choose \eqref{eq:Xmi}.
\item 
For the density \eqref{eq:AbstractmuM} we choose \eqref{eq:muW}.
The definition \eqref{eq:muWs} is compatible with \eqref{eq:AbstractmuM'},
in the sense that one has 
$\muWs = \s^{-4}|\intermult_{\p_{\zzeta}}(dy^0\wedge\cdots\wedge dy^3)|$.
\end{itemize}
We check that these are admissible choices:
$X_0,\dots,X_3$ satisfy \eqref{eq:translinv} 
because they are homogeneous of degree zero;
they satisfy \eqref{eq:X1..Tang} because 
$X_{i}(\ttcoord)=-\ttcoord y^i\s/|\vec{y}|^2$ for $i=1,2,3$;
and \eqref{eq:X1...Comm} because $[X_i,X_j]= \p_{y^i}(\s)X_j-\p_{y^j}(\s)X_i$
for $i,j=1,2,3$.
Further $\muW$ is translation invariant in $\zzeta$ because
it is homogeneous of degree zero.
\end{convention}

\subsection{Homogeneous bases}
\label{sec:basis_i0}
We introduce $C^\infty$-bases of the modules 
$\Omega(\Dspcl)\otimesRR\Kil$ and $\I(\Dspcl)$ and $\gx(\Dspcl)$,
given by elements that are homogeneous of degree zero
in the sense of Definition \ref{def:homog}.
The choice of basis is in particular motivated
by gauge fixing, c.f.~Lemma \ref{lem:gaugebases_i0}.

We decompose $\Kil = \boosts\oplus\transl$ where 
\begin{align}\label{eq:tbdef}
\begin{aligned}
\boosts &= \SPAN_{\R}\{B^{01},B^{02},B^{03},B^{12},B^{23},B^{31}\}
\\
\transl &= \SPAN_{\R}\{T^0,T^1,T^2,T^3\}
\end{aligned}
\end{align}
using the boosts and translations \eqref{eq:Kilbas}.
Note that the boosts are homogeneous of degree zero,
and the translations are homogeneous of degree one.
\begin{definition}\label{def:elements_i0}
This definition is local to Section \ref{sec:SpaceinfConstruction}, 
see Remark \ref{rem:local4}.
Abbreviate $\theta^\mu = \frac{dy^\mu}{\s}$.
Define the numbers
\begin{equation}\label{eq:nm_i0}
\renewcommand{\arraystretch}{1.1}
\begin{array}{c|ccccc}
k & 0 & 1 & 2 & 3 & 4 \\
\hline
\ngG_k^{\Omega} & 1&3&3&1&0\\
\ngG_k^{\I} &\new{0}&\new{0}&10&6&0\\
\ngG_k & 10 & 40 & 36 & 10 & 0\\
\end{array} 
\qquad
\begin{array}{c|ccccc}
k & 0 & 1 & 2 & 3 & 4 \\
\hline
\ng_k^{\Omega} & 1&4&6&4&1\\
\ng_k^{\I} &0&0&10&16&6\\
\ng_k & 10 & 50 & 76 & 46 & 10\\
\end{array} 
\end{equation}
e.g.~$\ng_2=76$.
The numbers $\ng_k^{\Omega}$, $\ng_k^{\I}$, $\ng_k$
are, respectively, the $C^\infty$-ranks of 
$\Omega^k(\Dspcl)$, $\I^k(\Dspcl)$, $\gx^k(\Dspcl)$.
Observe that $\ng_k = \ngG_k + \ngG_{k-1}$,
analogously for $\ng_k^{\Omega},\ng_k^{\I}$.
\begin{itemize}
\item 
For $k=0\dots4$ define 
$(\eGO_i^k)_{i=1\dots\ngG_k^{\Omega}},
(\eO_i^k)_{i=1\dots\ng_k^{\Omega}}
\in\Omega^k(\Dspcl)$ by:
\begin{align}
\eGO_1^0 &= 1\nonumber\\
\eGO_1^1 &= \theta^1,\;\eGO_2^1 =\theta^2,\;\eGO_3^1 =\theta^3\nonumber\\
\eGO_1^2 &=\theta^1\wedge\theta^2,\;
    \eGO_2^2=\theta^2\wedge\theta^3,\;
    \eGO_3^2=\theta^3\wedge\theta^1\nonumber\\
\eGO_1^3 &=  \theta^1\wedge\theta^2\wedge\theta^3\nonumber\\
(\eO^k_i)_{i=1\dots\ng^\Omega_k}&:\; 
	\eGO^k_1,\eGO^k_2,\dots,
	\theta^0\wedge\eGO^{k-1}_1,\theta^0\wedge\eGO^{k-1}_2,\dots
%
%
	\label{eq:Omegabasis_i0}
\end{align}
\item 
For $k=0\dots4$ define the following elements in $\Omega^k(\Dspcl)\otimesRR\boosts$:
\begin{align}
&(\eGB_i^k)_{i=1,\dots,6\ngG_k^{\Omega}}:\;\;\eGO_1^k\otimes B^{\mu\nu},\eGO_2^k\otimes B^{\mu\nu},\dots  \nonumber\\
&(\eB_i^k)_{i=1,\dots,6\ng_k^{\Omega}}:\;\; \eGB_1^k,\eGB_2^k,\dots,\theta^0\eGB_1^{k-1},\theta^0\eGB_2^{k-1},\dots
	\label{eq:OmegaBbasis_i0}
\end{align}
where $\mu\nu$ runs over $01,02,03,12,23,31$,
and using the multiplication \eqref{eq:lmod}.
\item 
For $k=0\dots4$ define the following elements in $\Omega^k(\Dspcl)\otimesRR\transl$:
\begin{align}
&(\eGT_i^k)_{i=1,\dots,4\ngG_k^{\Omega}}:\;\;\tfrac{1}{\s}\eGO_1^k\otimes T^{\mu},\tfrac{1}{\s}\eGO_2^k\otimes T^{\mu},\dots \nonumber\\
&(\eT_i^k)_{i=1,\dots,4\ng_k^{\Omega}}:\;\; \eGT_1^k,\eGT_2^k,\dots,\theta^0\eGT_1^{k-1},\theta^0\eGT_2^{k-1},\dots
	\label{eq:OmegaTbasis_i0}
\end{align}
where $\mu$ runs over $0\dots3$,
and using the multiplication \eqref{eq:lmod}.
\item 
Let $\cyclind=\{(123),(231),(312)\}$ be the cyclic index set.
For $(abc)\in\cyclind$ let\footnote{
The sign is as indicated because $dy^0,\dots,dy^3$
is negatively oriented, see Remark \ref{rem:orientation}.} 
$\theta_\pm^a = \frac12(\theta^0\wedge\theta^a\mp i\theta^b\wedge\theta^c)\in \Omega^2_{\pm}(\Dspcl)$. 
Set
\begin{align}\label{eq:traceless_symmetric_matrices}
\begin{aligned}
h_1 &= 
\left(\begin{smallmatrix}
0&1&0\\
1&0&0\\
0&0&0
\end{smallmatrix}\right),&
h_2 &= 
\left(\begin{smallmatrix}
0&0&0\\
0&0&1\\
0&1&0
\end{smallmatrix}\right),&
h_3 &= 
\left(\begin{smallmatrix}
0&0&1\\
0&0&0\\
1&0&0
\end{smallmatrix}\right),\\
h_4 &= 
\left(\begin{smallmatrix}
1&0&0\\
0&-1&0\\
0&0&0
\end{smallmatrix}\right),&
h_5 &= 
\tfrac{1}{\sqrt{3}}\left(\begin{smallmatrix}
1&0&0\\
0&1&0\\
0&0&-2
\end{smallmatrix}\right)
\end{aligned}
\end{align}
Define the following elements of $\I^2(\Dspcl)$ respectively $\I^3(\Dspcl)$:
\begin{align*}
	(\eGI_j^2)_{j=1\dots10}:\ & 
	\mu_{\etay}^{-1} \otimes
	(\tsum_{p,q=1}^3(h_\ell)_{pq}\theta_+^p\otimes \theta_+^q)
	\oplus
	cc, \\
	&\mu_{\etay}^{-1} \otimes
		(\tsum_{p,q=1}^3(ih_\ell)_{pq}\theta_+^p\otimes \theta_+^q)
		\oplus
		cc\ 
	\\
	(\eGI_j^3)_{j=1\dots6}:\ & 
	\tfrac{1}{2\sqrt{3}}\mu_{\etay}\otimes\left( 
	2\theta^1\theta^2\theta^3\otimes\theta_+^a
	\new{+} i \theta^0\theta^a(\theta^b\otimes\theta^b_+ + \theta^c\otimes\theta^c_+) \right)
	\oplus cc,\\
	&
	i\tfrac{1}{2\sqrt{3}}\mu_{\etay}\otimes\left( 
		2\theta^1\theta^2\theta^3\otimes\theta_+^a
		\new{+} i \theta^0\theta^a(\theta^b\otimes\theta^b_+ + \theta^c\otimes\theta^c_+) \right)
		\oplus cc
\end{align*}
where the index $\ell$ used for $(\eGI_j^2)$ runs over $1\dots5$,
the index $(abc)$ used for $(\eGI_j^3)$ runs over $\cyclind$.
Further $u\oplus cc$ stands for $u\oplus\bar u$;
we suppress the wedge sign;
$\mu_{\etay}^{-1}$ is the density associated to
$\etay=\eta_{\mu\nu}dy^\mu\otimes dy^\nu$ 
in \eqref{eq:diamondy} (see Remark \ref{rem:densityconv}).
For $k=2,3,4$ define the following elements in $\I^k(\Dspcl)$:
\begin{align}
(\eI^k_j)_{j=1\dots\ng^{\I}_k}:
	\;\;\eGI_1^k,\eGI_2^k,\dots,\theta^0\eGI_1^{k-1},\theta^0\eGI_2^{k-1},\dots
	\label{eq:Ibasis_i0}
\end{align}
where we use the multiplication in Definition \ref{def:Imod}.
\item 
For $k=0\dots4$
define the following elements of $\gx^k(\Dspcl)$:
\begin{align}
(\eGg^k_i)_{i=1\dots\ngG_k}:\;\;&
\eGB_1^k\oplus0\Ieps, \eGB_2^k\oplus0\Ieps,\dots, \nonumber\\
&\;\;\eGT_1^k\oplus0\Ieps, \eGT_2^k\oplus0\Ieps,\dots,\nonumber \\
&\;\;0\oplus\eGI_1^{k+1}\Ieps, 0\oplus\eGI_2^{k+1}\Ieps,\dots\label{eq:gGbasis_i0}\\
(\eg^k_i)_{i=1\dots\ng_k}:\;\;& 
\eGg^k_1,\eGg^k_2,\dots, \theta^0 \eGg^{k-1}_1,\theta^0 \eGg^{k-1}_2,\dots
	\label{eq:gbasis_i0}
\end{align}
where we use the multiplication \eqref{eq:gmod}
\end{itemize}
\end{definition}
\begin{lemma}[Homogeneous bases]\label{lem:bases_i0}
The elements introduced in Definition \ref{def:elements_i0}
are homogeneous of degree zero (Definition \ref{def:homog}),
and for each $k=0\dots4$:
\[ 
\def\arraystretch{1.}
\setlength{\tabcolsep}{3pt}
\begin{tabular}{c|ccccccccc}
Module & 
$\Omega^k(\Dspcl)$ &
$\Omega^k(\Dspcl)\otimesRR\boosts$ &
$\Omega^k(\Dspcl)\otimesRR\transl$ &
$\I^k(\Dspcl)$ &
$\gx^k(\Dspcl)$ \\
Rank &
$\ng_k^{\Omega}$ &
$6\ng_k^{\Omega}$ &
$4\ng_k^{\Omega}$ &
$\ng_k^{\I}$ &
$\ng_k$ \\
Basis &
$(\eO^k_i)_{i=1\dots \ng_k^{\Omega}}$ & 
$(\eB^k_i)_{i=1\dots 6\ng_k^{\Omega}}$ & 
$(\eT^k_i)_{i=1\dots 4\ng_k^{\Omega}}$ & 
$(\eI^k_i)_{i=1\dots \ng_k^{\I}}$ & 
$(\eg^k_i)_{i=1\dots \ng_k}$ & 
\end{tabular} 
\]
where, for example, the last column means that
$\ng_k$ is the $C^\infty$-rank of $\gx^k(\Dspcl)$,
and $(\eg^k_i)_{i=1\dots \ng_k}$ is a $C^\infty$-basis.
\end{lemma}
\begin{proof}
By direct inspection.\qed
\end{proof}
\subsection{Homogeneous norms}
\label{sec:norms_i0}
We define the norms that are used near $\spaceinf$
(some of them are actually seminorms, but we refer to them
as norms for simplicity).
Recall Definition \ref{def:ztXmudef}.
\begin{definition}[Norms near spacelike infinity]\label{def:norms_i0}
For every $k\in\Z_{\ge0}$ and $s>0$ and 
$f\in C^\infty(\Dspop_{\le s})$ define:
\begin{align}\label{eq:homognorms}
\begin{aligned}
\|f\|_{\nosHb^{k}(\Dspop_{\le s})}^2
&=\textstyle
\sum_{j=0}^k
\sum_{i_1,\dots,i_{j}=0}^{3} 
\int_{\Dspop_{\le s}}\big|X_{i_1} \cdots X_{i_{j}} f\big|^2\,\muW\\
\|f\|_{\sHb^{k}(\Dspop_{s})}^2
&=\textstyle
\sum_{j=0}^k
\sum_{i_1,\dots,i_{j}=0}^{3}
\int_{\Dspop_{s}}\big|X_{i_1} \cdots X_{i_{j}} f\big|^2
\,\muWs\\
\|f\|_{\nosCb^{k}(\Dspop_{\le s})}
&=\textstyle
\sum_{j=0}^k
\sum_{i_1,\dots,i_{j}=0}^{3} 
\sup_{p \in \Dspop_{\le s}}\big|X_{i_1} \cdots X_{i_{j}}f(p)\big|\\
\|f\|_{\sCb^{k}(\Dspop_{s})}
&=\textstyle
\sum_{j=0}^k
\sum_{i_1,\dots,i_{j}=0}^{3} 
\sup_{p \in \Dspop_{s}}\big|X_{i_1} \cdots X_{i_{j}}f(p)\big|
\end{aligned}
\end{align}
For the norms on $\Dspop_{s}$ we make the same
definition when $f$ is only defined near $\Dspop_{s}$.
We make analogous definitions for vector- and matrix-valued functions, 
where we apply the norms componentwise 
and then take the $\ell^2$-sum of the components;
and for matrix differential
operators of the form $a^i X_i$,
where we apply the norms to the matrices $a^i$ and then sum over $i$;
and for elements in $\gx(\Dspop_{\le s})$,
where we use the homogeneous basis \eqref{eq:gbasis_i0} to identify
them with vector-valued functions on $\Dspop_{\le s}$.
For $k\in\Z_{<0}$ we declare \eqref{eq:homognorms} to be zero.
\end{definition}
The norms in \eqref{eq:homognorms}
measure differentiability with respect
to all four vector fields $X_0,\dots,X_{3}$.
In particular, the slashed 
norms over the level sets $\Dspop_{s}$
are not determined by the restriction of $f$ to $\Dspop_{s}$.

The norms in \eqref{eq:homognorms} are homogeneous
(c.f.~Melrose's b-calculus \cite{MelroseGreenBook}).
That is, for all $\lambda>0$ and $u\in\gx(\Dspop_{\le s})$,
\begin{align}\label{eq:norms_homog}
\|\Scalg_\lambda u\|_{\nosHb^{k}(\Dspop_{\le \lambda s})}
=
\|u\|_{\nosHb^{k}(\Dspop_{\le s})}
\end{align}
and analogously for 
$\|{\cdot}\|_{\sHb^{k}(\Dspop_{s})}$,
$\|{\cdot}\|_{\nosCb^{k}(\Dspop_{\le s})}$,
$\|{\cdot}\|_{\sCb^{k}(\Dspop_{s})}$.
This uses the fact that $X_\mu$, $\muW$ and the
basis elements \eqref{eq:gbasis_i0} are homogeneous of degree zero.

\begin{remark}\label{rem:normsequality}
Using Convention \ref{conv:ztXmu_NEW}, 
in particular the identification \eqref{eq:ztcoords},
the norms \eqref{eq:Xnorms} in Definition \ref{def:sHXDefinition}
and the norms in Definition \ref{def:norms_i0} are equal,
in the sense that for all $s\in(0,1]$
and all $f\in C^\infty(\Dspop_{\le s})$:
\begin{align*}
\|f\|_{\nosHb^k(\Dspop_{\le s})}
&=
\|f\|_{\HX^k(\homMfd_{\le\log(s)})} &
\|f\|_{\sHb^k(\Dspop_{s})}
&=
\|f\|_{\sHX^k(\homMfd_{\log(s)})}\\
\|f\|_{\nosCb^k(\Dspop_{\le s})}
&=
\|f\|_{\CX^k(\homMfd_{\le\log(s)})}&
\|f\|_{\sCb^k(\Dspop_{s})}
&=
\|f\|_{\sCX^k(\homMfd_{\log(s)})}
\end{align*}
\end{remark}

\subsection{Gauge}
\label{sec:gauge_i0}

We define gauge fixing conditions under which the equation \eqref{eq:eqv+c}
contains a square system that is quasilinear symmetric hyperbolic including along null infinity.
The remaining equations are the constraints, 
which themselves solve a linear symmetric hyperbolic
system, i.e.~the constraints propagate.

The gauge that we define here is 
a special case of the gauges constructed in \cite[Section 3.5.3]{Thesis},
the concrete choice here is compatible with homogeneity.
The concept of gauges that we use here was introduced in \cite{RTgLa1,RTgLa2},
see \cite[Section 7]{RTgLa2} and \cite[Section 3.5]{Thesis}
for an overview and more conceptual discussion.

For concreteness, the definitions and statements in  
Section \ref{sec:gauge_i0} will be made on $\Dspcl$.
They hold analogously on the subsets $\Dspcl_{\le s}$, $\Dspop_{\le s}$
in \eqref{eq:deltas}, because all constructions are effectively fiberwise.

\subsubsection{Definition of gauge}
\label{sec:gauge_definition_i0}

We define a gauge (Definition \ref{def:gauges_i0}) and show basic properties (Lemma \ref{lem:gauge_mainprop_i0}).
See the start of Section \ref{sec:reformSHS_i0} for a brief
outline about how the gauge is used.

Fix the following homogeneous of degree zero 
vector field and metric on $\Dspcl$:
\begin{align*}
\Thom = \s \p_{y^0}\qquad\qquad
\ghom = \s^{-2} \etay
\end{align*}
where $\etay = \eta_{\mu\nu} dy^{\mu}\otimes dy^{\nu}$, see \eqref{eq:diamondy}.
The vector field $\Thom$ is future directed and timelike with respect to $[\gcyl]$,
and the metric $\ghom$ is a representative of $[\gcyl]$.

The following preliminary definitions are local to Section \ref{sec:SpaceinfConstruction}, 
see Remark \ref{rem:local4}.%
\begin{itemize}
\item 
Let
\begin{equation}\label{eq:Oiprdef}
\Oipr{k}{\cdot}{\cdot}:\Omega^k(\Dspcl)\times \Omega^k(\Dspcl)\to C^\infty(\Dspcl)
\end{equation}
be the nondegenerate symmetric $C^\infty$-bilinear 
form induced by $\ghom$, i.e.,
\begin{align}
\begin{aligned}
&\Oipr{k}{\omega_1\wedge\cdots\wedge\omega_k}{\omega_1'\wedge\cdots\wedge\omega_k'}\\
&\qquad=
\tsum_{\pi\in S_k}\sgn(\pi) 
\ghom^{-1}(\omega_1,\omega_{\pi(1)}')\cdots \ghom^{-1}(\omega_k,\omega_{\pi(k)}')
\label{eq:OmegaIPdef}
\end{aligned}
\end{align}
for all $\omega_1,\dots,\omega_k,\omega_1',\dots,\omega_k'\in\Omega^1(\Dspcl)$,
where $S_k$ is the symmetric group.
\item 
Let 
\begin{equation}\label{eq:Iiprdef}
\Iipr{k}{\cdot}{\cdot}:\I^k(\Dspcl)\times \I^k(\Dspcl)\to C^\infty(\Dspcl)
\end{equation}
be the nondegenerate symmetric $C^\infty$-bilinear form defined by 
\begin{subequations}\label{eq:Ibildefformulas}
\begin{align}\label{eq:Ibilpm}
\Iipr{k}{\up\oplus\um}{\up'\oplus\um'}
=
\Iiprp{k}{\up}{\up'}
+
\Iiprm{k}{\um}{\um'}
\end{align}
where \smash{$\Iiprpm{k}{\cdot}{\cdot}$} are $C^\infty$-bilinear and defined by
\begin{align}\label{eq:Ipmiprformula}
\Iiprpm{k}{\mu_{\ghom}^{-1}\otimes\omega\otimes\nu
	}{\mu_{\ghom}^{-1}\otimes\omega'\otimes\nu'}
	=
	(k-1)\Oipr{k}{\omega}{\omega'}\Oipr{2}{\nu}{\nu'}
\end{align}
\end{subequations}
for all $\mu_{\ghom}^{-1}\otimes\omega\otimes\nu,
\mu_{\ghom}^{-1}\otimes\omega'\otimes\nu'\in\I^k_{\pm}(\Dspcl)$,
using Sweedler's notation in Remark \ref{rem:Sweedler},
and where,
on the right hand side, the $\C$-linear extension of \eqref{eq:Oiprdef} is used.
Using the basis \eqref{eq:Ibasis_i0},
\begin{align}
\label{eq:IiprExplicit_i0}
\begin{aligned}
\Iipr{2}{\eI^2_i}{\eI^2_j}
&=
\tfrac{1}{\s^2}\left(\begin{smallmatrix}
\one_{5} & 0\\
0& -\one_{5}
\end{smallmatrix}\right)_{ij}
\\
\Iipr{3}{\eI^3_i}{\eI^3_j}
&=
\tfrac{1}{\s^2}\left(\begin{smallmatrix}
-\one_{3} & 0 &0&0\\
0& \one_{3} &0&0\\
0&0&-\one_5&0\\
0&0&0&\one_5
\end{smallmatrix}\right)_{ij}
\\
\Iipr{4}{\eI^4_i}{\eI^4_j}
&=
\tfrac{1}{\s^2}\left(\begin{smallmatrix}
\one_{3} & 0\\
0& -\one_{3}
\end{smallmatrix}\right)_{ij}
\end{aligned}
\end{align}
where $\one_{n}$ is the identity matrix of size $n$.
\item 
Let $\Iinter_{\Thom}:\I^{k+1}(\Dspcl)\to\I^{k}(\Dspcl)$ be the 
adjoint (relative to \eqref{eq:Iiprdef}) of the map 
$\I^{k}(\Dspcl)\to\I^{k+1}(\Dspcl)$, $u\mapsto \Thom^{\flat} u$ 
where $\Thom^{\flat}=\ghom(\Thom,{\cdot}) = -\smash{\frac{dy^0}{\s}}$,
and where we use the module multiplication in Definition \ref{def:Imod}.
That is,
\begin{equation}\label{eq:intmult}
\Iipr{k}{\Iinter_{\Thom}u}{u'} = \Iipr{k+1}{u}{\Thom^{\flat}u'}
\end{equation}
for all $u\in \I^{k+1}(\Dspcl)$ and $u'\in\I^k(\Dspcl)$.
Using the basis \eqref{eq:Ibasis_i0},
\begin{align}
\label{eq:IinterExplicit_i0}
\begin{aligned}
\Iinter_{\Thom} \eI^2_i
&=0 
\\
\Iinter_{\Thom} \eI^3_i
&=
\left(\begin{smallmatrix}
0_{10\times6} & \one_{10}
\end{smallmatrix}\right)_{ji} \eI^2_j 
\\
\Iinter_{\Thom} \eI^4_i
&=
\left(\begin{smallmatrix}
\one_{6} \\ 
0_{10\times6}
\end{smallmatrix}\right)_{ji} \eI^3_j
\end{aligned}
\end{align}
where $0_{n\times n'}$ is the zero matrix of size $n\times n'$,
and where we sum over $j$.

\item 
Define $P_{\Thom}:\I^k(\Dspcl)\to\I^k(\Dspcl)$ as follows.
First let $P_{\Thom}^{\Omega}:\OmegaC^1(\Dspcl)\to\OmegaC^1(\Dspcl)$ be the
fiberwise reflection in the $\ghom$-orthogonal complement of $\Thom^\flat$
(this maps $dy^0\mapsto dy^0$ and $dy^i\mapsto -dy^i$ for $i=1,2,3$).
This induces a map $P_{\Thom}^{\pm}:\I_{\pm}(\Dspcl)\to\I_{\mp}(\Dspcl)$ which acts
trivially on the density. Set
\begin{equation}\label{eq:Pdefhom}
P_{\Thom}(\up\oplus \um )=(-1)^k 
( \overline{P^+_{\Thom}\up}\oplus \overline{P^-_{\Thom}\um} )
\end{equation}
Using the basis \eqref{eq:Ibasis_i0},
\begin{align}
\label{eq:PExplicit_i0}
\begin{aligned}
P_{\Thom}(\eI^2_i) 
&= 
\left(\begin{smallmatrix}
\one_{5} & 0\\
0& -\one_{5}
\end{smallmatrix}\right)_{ji}\eI^2_j
\\
P_{\Thom}(\eI^3_i)
&=
\left(\begin{smallmatrix}
-\one_{3} & 0 &0&0\\
0& \one_{3} &0&0\\
0&0&-\one_5&0\\
0&0&0&\one_5
\end{smallmatrix}\right)_{ji}\eI^3_j
\\
P_{\Thom}(\eI^4_i) 
&= 
\left(\begin{smallmatrix}
\one_{3} & 0\\
0& -\one_{3}
\end{smallmatrix}\right)_{ji}\eI^4_j
\end{aligned}
\end{align}
where we sum over $j$.
\end{itemize}
Recall the decomposition $\Kil=\boosts\oplus\transl$ in \eqref{eq:tbdef}.
\begin{definition}\label{def:gauges_i0}
This  definition is
local to Section \ref{sec:SpaceinfConstruction}, see Remark \ref{rem:local4}.
Define 
\begin{align}\label{eq:gaugespaces_i0}
\begin{aligned}
\OmegaG^k(\Dspcl) 
	&= 
	\{\omega\in\Omega^k(\Dspcl) \mid \intermult_{\Thom}\omega=0\}\\
\IG^k(\Dspcl) 
	&=
	\{u \in \I^{k}(\Dspcl)\mid\Iinter_{\Thom}u=0 \}\\
\gxG^k(\Dspcl) 
	&= \big(\OmegaG^k(\Dspcl)\otimesRR\Kil\big) \oplus \IG^{k+1}(\Dspcl)\Ieps
\end{aligned}
\end{align}
for $k=0\dots4$.
In the first line, $\intermult_{\Thom}$ is the interior multiplication by $\Thom$,
and in the second line, $\Iinter_{\Thom}$ is the map \eqref{eq:intmult}.
Define the $C^\infty$-bilinear forms\footnote{
In \cite[Section 3.5.3]{Thesis} the bilinear forms are denoted by $\textnormal{B}$,
not $\beta$. We use $\beta$ to avoid confusion with 
the $C^\infty$-bilinear map that appears in the symmetric hyperbolic system
\eqref{eq:ueq}.
}:
\begin{itemize}
\item 
$\OmegaBil^k:\OmegaG^k(\Dspcl)\times\Omega^{k+1}(\Dspcl)\to C^\infty(\Dspcl)$ by 
$$\OmegaBil^k(\omega,\omega') = \Oipr{k}{\omega}{\intermult_{\Thom}\omega'}$$
\item 
$\OmegaBBil^k:(\OmegaG^k(\Dspcl)\otimesRR\boosts)\times(\Omega^{k+1}(\Dspcl)\otimesRR\boosts)\to C^\infty(\Dspcl)$ by
\begin{align*}
\OmegaBBil^k(\omega\otimes B^{\mu\nu}, \omega'\otimes B^{\alpha\beta}) 
=
\OmegaBil^k(\omega,\omega') 
(\delta_{\mu\alpha}\delta_{\nu\beta}\new{-}\delta_{\mu\beta}\delta_{\nu\alpha})
\end{align*}
\item 
$\OmegaTBil^k:(\OmegaG^k(\Dspcl)\otimesRR\transl)\times(\Omega^{k+1}(\Dspcl)\otimesRR\transl)\to C^\infty(\Dspcl)$ by
\begin{align*}
\OmegaTBil^k(\omega\otimes T_\mu,\omega'\otimes T_\nu) 
=
\s^2\OmegaBil^k(\omega,\omega') \delta_{\mu\nu}
\end{align*}
\item 
$\IBil^k: \IG^{k}(\Dspcl)\times\I^{k+1}(\Dspcl)\to C^\infty(\Dspcl)$ by,
using $P_{\Thom}$ in \eqref{eq:Pdefhom},
\begin{align*}
\IBil^k(u,u')&= \s^{\new{2}} \Iipr{k}{P_{\Thom}u}{\Iinter_{\Thom}u'}
\end{align*}
\item 
$\gBil^k:\gxG^{k}(\Dspcl)\times\gx^{k+1}(\Dspcl)\to C^\infty(\Dspcl)$ by 
\begin{align}\label{eq:gbil}
\begin{aligned}
&\gBil^k\Big( (u_{\boosts} + u_{\transl}) \oplus \uI ,
	(u_{\boosts}' + u_{\transl}') \oplus \uI'  \Big)\\
&\qquad\qquad=
\OmegaBBil^k(u_{\boosts},u'_{\boosts})
+\OmegaTBil^k(u_{\transl},u'_{\transl})
+\IBil^{k+1}(\uI,\uI')
\end{aligned}
\end{align}
where 
$u_{\boosts}\in\OmegaG^{k}(\Dspcl)\otimesRR\boosts$,
$u_{\transl}\in\OmegaG^{k}(\Dspcl)\otimesRR\transl$,
$u_{\I}\in\IG^{k+1}(\Dspcl)$
and where
$u_{\boosts}'\in\Omega^{k+1}(\Dspcl)\otimesRR\boosts$,
$u_{\transl}'\in\Omega^{k+1}(\Dspcl)\otimesRR\transl$,
$u_{\I}'\in\I^{k+2}(\Dspcl)$.
\end{itemize}
\end{definition}

The module $\gxG^k(\Dspcl)$ is the module of smooth sections,
over $\Dspcl$, of a trivial vector bundle $\gxG^k$ defined on $\Dspcl$.
Further, $\gxG^k$ is a subbundle of $\gx^k$ on $\Dspcl$.
\begin{lemma}[Homogeneous bases for gauge spaces]\label{lem:gaugebases_i0}
Using the
elements from Definition \ref{def:elements_i0},
which are homogeneous of degree zero,
for each $k=0\dots4$:
\[ 
\def\arraystretch{1.1}
\setlength{\tabcolsep}{4pt}
\begin{tabular}{c|ccccccccc}
Module & 
$\OmegaG^k(\Dspcl)$ &
$\OmegaG^k(\Dspcl)\otimesRR\boosts$ &
$\OmegaG^k(\Dspcl)\otimesRR\transl$ &
$\IG^k(\Dspcl)$ &
$\gxG^k(\Dspcl)$ \\
Rank &
$\ngG_k^{\Omega}$ &
$6\ngG_k^{\Omega}$ &
$4\ngG_k^{\Omega}$ &
$\ngG_k^{\I}$ &
$\ngG_k$ \\
Basis &
$(\eGO^k_i)_{i=1\dots \ngG_k^{\Omega}}$ & 
$(\eGB^k_i)_{i=1\dots 6\ngG_k^{\Omega}}$ & 
$(\eGT^k_i)_{i=1\dots 4\ngG_k^{\Omega}}$ & 
$(\eGI^k_i)_{i=1\dots \ngG_k^{\I}}$ & 
$(\eGg^k_i)_{i=1\dots \ngG_k}$ 
\end{tabular} 
\]
where, for example, the last column means that
$\ngG_k$ is the $C^\infty$-rank of $\gxG^k(\Dspcl)$,
and $(\eGg^k_i)_{i=1\dots \ngG_k}$ is a $C^\infty$-basis.
\end{lemma}
\begin{proof}
The first three columns are immediate.
The fourth column follows from \eqref{eq:IinterExplicit_i0}.
The fifth column follows from the second, third and fourth.\qed
\end{proof}
Note that the basis of $\gxG^k(\Dspcl)$
coincides with the first $\ngG_k$ elements of the basis 
$(\eg^k_i)_{i=1\dots \ng_k}$ of
$\gx^k(\Dspcl)$ in Lemma \ref{lem:bases_i0}.
Analogously for the other modules.
\begin{lemma}\label{lem:Binbases}
Relative to the bases in Lemma \ref{lem:bases_i0} and \ref{lem:gaugebases_i0}, 
the bilinear forms in Definition \ref{def:gauges_i0} are given as follows.
Set $\theta^\mu = \frac{dy^\mu}{\s}$.
For $k=0\dots4$, $\ell=1,2,3$:%
\begin{subequations}
\begin{align}
\OmegaBBil^k(\eGB^k_i,\eGB^{k+1}_j) &= 0 
	& \OmegaBBil^k(\eGB^k_i,\theta^0\eGB^{k}_j) &= \delta_{ij}
	& \OmegaBBil^k(\eGB^k_i,\theta^\ell\eGB^{k}_j) &= 0 
	\label{eq:OBbil}\\
\OmegaTBil^k(\eGT^k_i,\eGT^{k+1}_j) &= 0 
	& \OmegaTBil^k(\eGT^k_i,\theta^0\eGT^{k}_j) &= \delta_{ij}
	& \OmegaTBil^k(\eGT^k_i,\theta^\ell\eGT^{k}_j) &= 0 
	\label{eq:OTbil}
\intertext{
Further, for $k=2,3$ and $\ell=1,2,3$ one has 
(note that $\IBil^k=0$ for $k=0,1,4$):
}
\IBil^k(\eGI^k_i,\eGI^{k+1}_j) &= 0 
	& \IBil^k(\eGI^k_i,\theta^0\eGI^{k}_j) &= \delta_{ij}
	& \IBil^k(\eGI^k_i,\theta^\ell\eGI^{k}_j)
	&=
	\left(\begin{smallmatrix}
	0 & A_{k,\ell}\\
	A_{k,\ell}^T & 0
	\end{smallmatrix}\right)_{ij}
	\label{eq:Ibil}
\end{align}
\end{subequations}
where
\begin{align*}
A_{2,1}
	&=
	\left(
	\begin{smallmatrix}
	0 & 0 & -\frac{1}{2} & 0 & 0 \\
	0 & 0 & 0 & -\frac{1}{2} & \frac{\sqrt{3}}{2} \\
	\frac{1}{2} & 0 & 0 & 0 & 0 \\
	0 & \frac{1}{2} & 0 & 0 & 0 \\
	0 & -\frac{\sqrt{3}}{2} & 0 & 0 & 0  
	\end{smallmatrix}\right)
&
A_{2,2}
	&=
	\left(
	\begin{smallmatrix}
	0 & \frac{1}{2} & 0 & 0 & 0 \\
	-\frac{1}{2} & 0 & 0 & 0 & 0 \\
	0 & 0 & 0 & -\frac{1}{2} & -\frac{\sqrt{3}}{2} \\
	0 & 0 & \frac{1}{2} & 0 & 0 \\
	0 & 0 & \frac{\sqrt{3}}{2} & 0 & 0
	\end{smallmatrix}
	\right)
&
A_{2,3}
	&=
	\left(
	\begin{smallmatrix}
	0 & 0 & 0 & 1 & 0 \\
	0 & 0 & \frac{1}{2} & 0 & 0 \\
	0 & -\frac{1}{2} & 0 & 0 & 0 \\
	-1 & 0 & 0 & 0 & 0 \\
	0 & 0 & 0 & 0 & 0 \\
	\end{smallmatrix}
	\right)\\
A_{3,1}
	&=
	\left(
	\begin{smallmatrix}
	0 & 0 & 0 \\
	 0 & 0 & -\frac{1}{2} \\
	 0 & \frac{1}{2} & 0
	\end{smallmatrix}\right)
&
A_{3,2}
	&=
	\left(
	\begin{smallmatrix}
	 0 & 0 & \frac{1}{2} \\
	 0 & 0 & 0 \\
	 -\frac{1}{2} & 0 & 0
	\end{smallmatrix}
	\right)
&
A_{3,3}
	&=
	\left(
	\begin{smallmatrix}
	 0 & -\frac{1}{2} & 0 \\
	 \frac{1}{2} & 0 & 0 \\
	 0 & 0 & 0 
	\end{smallmatrix}
	\right)
\end{align*}
\end{lemma}
\begin{proof}
The identities \eqref{eq:OBbil} and \eqref{eq:OTbil} follow from 
\begin{align*}
\OmegaBil^k(\eGO^k_i,\eGO^{k+1}_j)
&=
\Oipr{k}{\eGO^k_i}{\intermult_{\Thom}\eGO^{k+1}_j}
=0\\
\OmegaBil^k(\eGO^k_i,\theta^0\wedge\eGO^{k}_j) 
&=
\Oipr{k}{\eGO^k_i}{\intermult_{\Thom}(\theta^0\wedge\eGO^k_j)}
=
\Oipr{k}{\eGO^k_i}{\eGO^k_j}
=\delta_{ij}\\
\OmegaBil^k(\eGO^k_i,\theta^\ell\wedge\eGO^{k}_j) 
&=
\Oipr{k}{\eGO^k_i}{\intermult_{\Thom}(\theta^\ell\wedge\eGO^k_j)}
=
0
\end{align*}
where we use $\intermult_{\Thom}\eGO^{k}_j=0$, 
$\intermult_{\Thom}(\theta^0\wedge\eGO^{k}_j)=\eGO^{k}_j$,
$\intermult_{\Thom}(\theta^\ell\wedge\eGO^{k}_j)=0$,
respectively.

Consider \eqref{eq:Ibil}.
By \eqref{eq:IinterExplicit_i0} we have
$\Iinter_{\Thom}\eGI^{k+1}_j=0$ and 
$\Iinter_{\Thom}(\theta^0\eGI^{k}_j)=\eGI^{k}_j$.
Thus 
\begin{align*}
\IBil^k(\eGI^k_i,\eGI^{k+1}_j) 
&= 
\s^2\Iipr{k}{P_{\Thom}\eGI^k_i}{\Iinter_{\Thom}\eGI^{k+1}_j}
=0\\
\IBil^k(\eGI^k_i,\theta^0\eGI^{k}_j) 
&= 
\s^2\Iipr{k}{P_{\Thom}\eGI^k_i}{\Iinter_{\Thom}(\theta^0\eGI^{k}_j)}
=
\s^2\Iipr{k}{P_{\Thom}\eGI^k_i}{\eGI^{k}_j}
=\delta_{ij}
\end{align*}
using \eqref{eq:IiprExplicit_i0}, \eqref{eq:PExplicit_i0} in the last step.
For the third identity in \eqref{eq:Ibil} note that
\[ 
\IBil^k(\eGI^k_i,\theta^\ell \eGI^{k}_j) 
= 
\s^2\Iipr{k}{P_{\Thom}\eGI^k_i}{\Iinter_{\Thom}(\theta^\ell\eGI^{k}_j)}
=
\s^2\Iipr{k+1}{-\theta^0 P_{\Thom}\eGI^k_i}{\theta^\ell\eGI^{k}_j}
\]
where we use \eqref{eq:intmult}.
This can be evaluated explicitly, 
for example,
\begin{align*}
\IBil^2(\eGI^2_1,\theta^1 \eGI^{2}_8) 
&=
-\s^2\Iipr{3}{\theta^0 P_{\Thom}\eGI^2_1}{\theta^1\eGI^{2}_8}
=
-\s^2\Iipr{3}{\theta^0 \eGI^2_1}{\theta^1\eGI^{2}_8}
\intertext{
where we use \eqref{eq:PExplicit_i0}.
One has $\theta^1\eGI^{2}_8= \frac{\sqrt{3}}{2}\eI^{3}_3-\frac12\theta^0\eGI^{2}_1$
and thus by \eqref{eq:IiprExplicit_i0}:
}
\textstyle
\IBil^2(\eGI^2_1,\theta^1 \eGI^{2}_8) 
&=
\tfrac12\s^2\Iipr{3}{\theta^0 \eGI^2_1}{\theta^0\eGI^{2}_1}
=-\tfrac12
\end{align*}
which agrees with the $1,3$-entry of $A_{2,1}$.
\qed
\end{proof}
\begin{remark}\label{rem:EigenvaluesIbil}
Let $c_0,c_1,c_2,c_3\in\R$. 
For $k=2,3$
consider the symmetric 
$\ngG^{\I}_k\times\ngG^{\I}_k$-matrix whose $ij$-entry is given by 
\[ 
\IBil^k(\eGI^k_i,(c_0\theta^0 + \textstyle\sum_{\ell=1}^3 c_\ell\theta^\ell)\eGI^{k}_j)
\]
Its eigenvalues are
$c_0, c_0\pm |\vec{c}|, c_0\pm \frac{|\vec{c}|}{2}$ when $k=2$
respectively $c_0,c_0\pm \frac{|\vec{c}|}{2}$ when $k=3$,
where $\vec{c}=(c_1,c_2,c_3)$.
Hence for $k=2$ it is positive definite iff 
$c_0 - |\vec{c}| > 0$.
\end{remark}
The next lemma lists the main properties of $(\gxG(\Dspcl),\gBil)$.
It is adapted from \cite[Definition 8]{RTgLa2}, where these
properties are used to define abstractly the notion of a gauge.
The lemma will be used to show that \eqref{eq:eqv+c} 
contains a necessary symmetric hyperbolic system
and that the constraints propagate,
for example, \ref{item:gaugesymNEW} will be used to prove symmetry,
and \ref{item:gaugeposNEW} to prove positivity.

Define
\begin{equation}\label{eq:Omegafut}
\Omegafut(\Dspcl) = \{\omega\in\Omega^1(\Dspcl) \mid \gcyl^{-1}(\omega,\omega)<0,\,\omega(\p_{\tau})>0\}
\end{equation}
Recall that $X_\mu = \s\p_{y^\mu}$, see Definition \ref{def:ztXmudef}.
\begin{lemma}\label{lem:gauge_mainprop_i0}
The tuple $(\gxG(\Dspcl),\gBil)$ is a gauge for $\gx(\Dspcl)$, in the following sense.
For all $\omega\in\Omegafut(\Dspcl)$,
left-multiplication $\gxG(\Dspcl)\to\gx(\Dspcl)$, 
$u\mapsto \omega u$ is fiberwise injective,
and 
$\gx(\Dspcl) = \gxG(\Dspcl)\oplus\omega\gxG(\Dspcl)$,
where we use the module multiplication \eqref{eq:gmod}.
Moreover, for all $k=0\dots4$:
\begin{enumerate}[label={\textnormal{(G\arabic*)}}]
\item \label{item:gaugesymNEW}
$\gBil^k(\,\cdot\,,\omega\,\cdot\,)|_{\gxG^k(\Dspcl)\times\gxG^k(\Dspcl)}$ 
is symmetric for all $\omega\in \Omega^1(\Dspcl)$.
\item \label{item:gaugeposNEW}
For every $u\in\gxG^k(\Dspcl)$ and every $\omega\in\Omega^1(\Dspcl)$ one has
\begin{align}
&\gBil^{k}(u,\omega u) \ge 
\big( 
\omega(X_0) - \big(\textstyle\sum_{i=1}^3|\omega(X_i)|^2\big)^{\frac12}
\big)
\gBil^{k}(u, \tfrac{dy^0}{\s} u)
\label{eq:posineq}
\end{align}
and if $u\neq0$ then $\gBil^{k}(u, \tfrac{dy^0}{\s} u)>0$.
Furthermore, 
\begin{equation}\label{eq:posiffofut}
\gBil^k(\,\cdot\,,\omega\,\cdot\,)|_{\gxG^k(\Dspcl)\times\gxG^k(\Dspcl)}>0
\quad\Leftrightarrow\quad
\omega\in \Omegafut(\Dspcl)
\end{equation}
\item \label{item:gaugekernelNEW}
$\gxG^{k+1}(\Dspcl) = \big\{u\in \gx^{k+1}(\Dspcl)\mid\gBil^k(\gxG^{k}(\Dspcl),u)=0 \big\}$
\end{enumerate}
\end{lemma}
\begin{proof}
This is proven, in a more general setting, in \cite[Proposition 13]{Thesis}
(for \eqref{eq:posineq} see the second displayed equation on page 93).
To make this more self-contained we include a proof here.

\ref{item:gaugesymNEW}: By \eqref{eq:gbil} and Lemma \ref{lem:Binbases}.

\ref{item:gaugeposNEW}: 
We check \eqref{eq:posineq}.
By \eqref{eq:gbil}, it suffices to check
\begin{align*}
\OmegaBBil^{k}(u,\omega u) 
	&\ge 
	(\omega(X_0) - \big(\textstyle\sum_{i=1}^3|\omega(X_i)|^2\big)^{\frac12})
	\OmegaBBil^{k}(u, \tfrac{dy^0}{\s} u)\\
\OmegaTBil^{k}(u,\omega u) 
	&\ge 
	(\omega(X_0) - \big(\textstyle\sum_{i=1}^3|\omega(X_i)|^2\big)^{\frac12})
	\OmegaTBil^{k}(u, \tfrac{dy^0}{\s} u)\\
\IBil^{\new{k+1}}(u,\omega u) 
	&\ge 
	(\omega(X_0) - \big(\textstyle\sum_{i=1}^3|\omega(X_i)|^2\big)^{\frac12})
	\IBil^{\new{k+1}}(u, \tfrac{dy^0}{\s} u)
\end{align*}
for $u$ in, respectively, 
$\OmegaG^k(\Dspcl)\otimesRR\boosts$, 
$\OmegaG^k(\Dspcl)\otimesRR\transl$,
$\IG^{k+1}(\Dspcl)$.
The first two inequalities hold by Lemma \ref{lem:Binbases},
the last holds by Lemma \ref{lem:Binbases} and Remark \ref{rem:EigenvaluesIbil}.

The statement
$\gBil^{k}(u, \tfrac{dy^0}{\s} u)>0$ for $u\neq0$ holds by
\eqref{eq:gbil} and Lemma \ref{lem:Binbases}.

We check \eqref{eq:posiffofut}.
$\Leftarrow$ follows from \eqref{eq:posineq}.
$\Rightarrow$ follows from Remark \ref{rem:EigenvaluesIbil}. 

We check fiberwise injectivity of left-multiplication:
Assume that 
$\omega u = 0$ at some point in $\Dspcl$.
Then \ref{item:gaugeposNEW} implies $u=0$ at that point,
because the left hand side of \eqref{eq:posineq} vanishes
and $\omega(X_0) - (\textstyle\sum_{i=1}^3|\omega(X_i)|^2)^{\frac12}>0$.

We check 
$\gx^k(\Dspcl) = \gxG^k(\Dspcl)\oplus\omega\gxG^{k-1}(\Dspcl)$:
Let $u\in \gxG^k(\Dspcl)\cap \omega\gxG^{k-1}(\Dspcl)$.
Then $u\in \gxG^{k}(\Dspcl)$ and $u=\omega u'$ with $u'\in \gxG^{k-1}(\Dspcl)$.
Thus $\omega u = \omega(\omega u') = 0$ using associativity,
thus $u=0$ by injectivity of left-multiplication.
Clearly 
$\gx^k(\Dspcl) \supset \gxG^k(\Dspcl)\oplus\omega\gxG^{k-1}(\Dspcl)$,
thus, to show equality, it remains to show that the ranks are equal.
The rank of 
$\gxG^k(\Dspcl)\oplus\omega\gxG^{k-1}(\Dspcl)$ 
is $\ngG_k+\ngG_{k-1}=\ng_k$, 
where $\ng_k$ is the rank of $\gx^{k}(\Dspcl)$,
using Lemma \ref{lem:gaugebases_i0}, 
injectivity of left-multiplication, 
and \eqref{eq:nm_i0}.

\ref{item:gaugekernelNEW}:
We check $\subset$:
This follows directly from Definition \ref{def:gauges_i0}.
We check $\supset$:
Let $u\in\gx^{k+1}(\Dspcl)$ such that $\gBil^k(\gxG^{k}(\Dspcl),u)=0$.
Fix an $\omega\in\Omegafut(\Dspcl)$ and 
decompose $u=v+\omega v'$ where 
$v\in\gxG^{k+1}(\Dspcl)$, $v'\in\gxG^{k}(\Dspcl)$.
Then for all $z\in \gxG^{k}(\Dspcl)$ we have
\[ 
0 = \gBil^k(z,u) = \gBil^k(z,v)+\gBil^k(z,\omega v') = \gBil^k(z,\omega v')
\]
Using this equality with $z=v'$ one obtains $v'=0$ by \ref{item:gaugeposNEW}.
\qed
\end{proof}

\subsubsection{MC-equation as a symmetric hyperbolic system}
\label{sec:reformSHS_i0}

This section serves as preparation for Section \ref{sec:existence_i0}.
In Section \ref{sec:existence_i0} we solve \eqref{eq:eqv+c} 
using the gauge $(\gxG(\Dspcl),\gBil)$ from Definition \ref{def:gauges_i0},
as follows:
\begin{itemize}
\item 
Impose the condition that $\cc$ is a section of $\gxG^1$,
which we interpret as a gauge fixing condition. 
One then solves the necessary subsystem 
\begin{equation}\label{eq:necSHS}
\gBil^{1}\big( \,\cdot\, , 
\dg (v+\cc) + \tfrac12[v+\cc,v+\cc]\big)
\;=\;0
\end{equation}
This is square: there are $\ngG_1$ equations and $\ngG_1$ unknowns,
by Lemma \ref{lem:gauge_mainprop_i0}.

\item 
Set $U=\dg (v+\cc) + \tfrac12[v+\cc,v+\cc]$.
The equation \eqref{eq:necSHS} and \ref{item:gaugekernelNEW} imply
that $U$ is a section of $\gxG^2$.
Further one has $\dg U + [v+\cc,U]=0$ by \eqref{eq:dgLaax}, in particular
$U$ solves the linear homogeneous system
\begin{equation}\label{eq:constraintsprop}
\gBil^{2}\big( \,\cdot\, , \dg U + [v+\cc,U]\big)
\;=\;0
\end{equation}
This is square: there are $\ngG_2$ equations and $\ngG_2$ unknowns.
Using the fact that $v$ solves the constraint equations,
one shows that $U$ vanishes along $y^0=0$.
Then one uses \eqref{eq:constraintsprop} to 
show that $U$ vanishes everywhere.
\end{itemize}
In this Section \ref{sec:reformSHS_i0} we use the homogeneous 
basis \eqref{eq:gbasis_i0} to rewrite the systems \eqref{eq:necSHS}
and \eqref{eq:constraintsprop} in matrix-vector form 
(Lemma \ref{lem:translationofeq}),
show that they are quasilinear respectively linear symmetric hyperbolic, 
including along null infinity (Lemma \ref{lem:aALproperties}, \ref{lem:qminkdelta}),
and show properties of the linear parts associated
to the Minkowski differential $\dg$ (Lemma \ref{lem:ellMinkest}).
\step
Recall Definition \ref{def:ztXmudef} and Convention \ref{conv:ztXmu_NEW}.
We will use the identifications
\begin{align}\label{eq:identify_g_vec_i0}
\begin{aligned}
\gxG^k(\Dspcl) &\simeq C^\infty(\Dspcl,\R^{\ngG_k}) 
& &\text{using the basis $(\eGg_{i}^k)_{i=1\dots\ngG_k}$ in \eqref{eq:gGbasis_i0}}\\
\gx^k(\Dspcl) &\simeq C^\infty(\Dspcl,\R^{\ng_k}) 
& &\text{using the basis $(\eg_{i}^k)_{i=1\dots\ng_k}$ in \eqref{eq:gbasis_i0}}
\end{aligned}
\end{align}
\begin{definition}
This definition is
local to Section \ref{sec:SpaceinfConstruction}, see Remark \ref{rem:local4}.
Using the basis \eqref{eq:gbasis_i0},
for all $k=1,2$, $\mu=0\dots3$, $\ell=1\dots\ng_k$ define
\begin{equation}\label{eq:anchorformshom}
(\anchor_{k})_{\ell}^{\mu} = \tfrac{1}{\s} \anchorg(\eg^k_{\ell})(y^\mu)
\;\in\; \Omega^k(\Dspcl)
\end{equation}
where the anchor $\anchorg$ is defined in \eqref{eq:anchorgdef}.
\end{definition}
These $k$-forms are indeed smooth on $\Dspcl$, 
because $y^\mu$, $\frac{1}{\s}$, $\eg^k_{\ell}$, $\anchorg$ are smooth there.
With this definition, for all $f\in C^\infty(\Dspcl)$ one has\footnote{
To see this, expand 
$\eg^k_{\ell}=(\sum_{i=1}^{10} \omega_i \otimes \KilBasis_i) \oplus\uI$
where $\KilBasis_1,\dots,\KilBasis_{10}$ is a basis of $\Kil$,
where $\omega_i$ are $k$-forms, and $\uI\in\I^{k+1}(\Dspcl)$. 
Then by definition of $\anchorg$ in \eqref{eq:anchorgdef} and \eqref{eq:anchorLdef},
\begin{align*}
\anchorg(\eg^k_{\ell})(f) 
	&= 
	\tsum_{i=1}^{10} \omega_i \KilBasis_i(f)
	=
	\tfrac{1}{\s} \tsum_{i=1}^{10} \omega_i \KilBasis_i(y^\mu) X_\mu f\\
	&= 
	\tfrac{1}{\s} \anchorg(\eg^k_{\ell})(y^\mu) X_\mu f
	=
	(\anchor_k)^{\mu}_{\ell} X_\mu f
\end{align*}
}
\begin{align}\label{eq:anchorghom}
\anchorg(\eg^k_{\ell})(f) 
	= 
	(\anchor_{k})_{\ell}^{\mu} X_\mu f
\end{align}
\begin{definition}\label{def:SHSarrays_i0}
This definition is local to Section \ref{sec:SpaceinfConstruction}, 
see Remark \ref{rem:local4}.
For $\mu=0\dots3$ and $\ideg,k'=1,2$ define
\begin{align}\label{eq:minkarrays}
\begin{aligned}
\amink_{\ideg}^\mu &\in C^\infty(\Dspcl,\End(\R^{\ngG_{\ideg}}))\\
\Amink_{\ideg}^\mu &\in C^\infty(\Dspcl,\Hom(\R^{\ng_1},\End(\R^{\ngG_{\ideg}}))\\
\Lmink_{\ideg} &\in C^\infty(\Dspcl,\End(\R^{\ngG_{\ideg}}))\\
%
\Bmink_{\ideg} &\in  C^\infty(\Dspcl,\Hom(\R^{\ng_1}\otimes\R^{\ng_{\ideg}},\R^{\ngG_{\ideg}}))\\
\Amink_{k' k}^\mu &\in C^\infty(\Dspcl,\Hom(\R^{\ng_{k'}},\Hom(\R^{\ng_{k}},\R^{\ngG_{k+k'-1}})))\\
\SFmink &\in C^\infty(\Dspcl,\Hom(\R^{\ng_2},\R^{\ngG_1}))\\
\ginj_k &\in C^\infty(\Dspcl, \Hom(\R^{\ngG_k},\R^{\ng_k}))
\end{aligned}
\end{align}
as follows, using $\gBil^k$ in Definition \ref{def:gauges_i0},
the bases \eqref{eq:gGbasis_i0}, \eqref{eq:gbasis_i0}, 
and \eqref{eq:anchorformshom}:
\begin{align*}
(\amink_{\ideg}^\mu u)_i 
	&= (\amink_{\ideg}^\mu)_{ij} u_{j}
&& \text{where}& 
(\amink_{\ideg}^\mu)_{\outind j}&=\gBil^{\ideg}\big(\eGg_\outind^{\ideg}, \tfrac{dy^\mu}{\s} \eGg^{\ideg}_j\big)\\
(\Amink_{\ideg}^\mu (v) u)_i
	&= (\Amink_{\ideg}^\mu)_{\ell,ij} v_\ell u_j 
&& \text{where}& 
(\Amink_{\ideg}^\mu)_{\ell,\outind j} &= \gBil^{k}\big(\eGg^{\ideg}_\outind,
(\anchor_1)_\ell^\mu
\eGg^{\ideg}_j\big)\\
(\Lmink_{\ideg} u)_{i} &= (\Lmink_{\ideg})_{ij}u_j
&& \text{where}& 
(\Lmink_{\ideg})_{\outind j} &= -\gBil^{\ideg}\big(\eGg_\outind^{\ideg}, \dg \eGg^{\ideg}_j\big)\\
(\Bmink_{\ideg}(v,w))_i &= (\Bmink_{\ideg})_{ij\ell} v_j w_\ell 
&& \text{where}& 
(\Bmink_{\ideg})_{\outind j\ell} &= -\gBil^{\ideg}\big(\eGg_\outind^{\ideg}, [\eg_j^1,\eg_\ell^{\ideg}]\big)\\
(\Amink_{k'k}^\mu (w')w)_i  &= (\Amink_{k'k}^\mu)_{\ell,i j} w'_\ell w_j
&& \text{where}& 
(\Amink_{k'k}^\mu)_{\ell,i j} &= \gBil^{k+k'-1}\big(\eGg^{k+k'-1}_i, 
(\anchor_{k'})_\ell^\mu
\eg^k_j\big)\\
(\SFmink v')_i &= \SFmink_{ij} v'_j
&& \text{where}& 
\SFmink_{\outind j} &= -\gBil^{1}\big(\eGg_\outind^1, \eg^2_j\big)\\
(\ginj_k u)_i
&=
\delta_{ij} u_j
\end{align*}
with
$u\in C^\infty(\Dspcl,\R^{\ngG_k})$, 
$v\in C^\infty(\Dspcl,\R^{\ng_1})$,
$v'\in C^\infty(\Dspcl,\R^{\ng_2})$,
$w\in C^\infty(\Dspcl,\R^{\ng_k})$, 
$w'\in C^\infty(\Dspcl,\R^{\ng_{k'}})$,
and where the sum over the repeated indices $j$, $\ell$ is implicit.
\end{definition}
The components of \eqref{eq:minkarrays} are indeed smooth on $\Dspcl$
(in particular along null infinity),
because
	$\frac{dy^\mu}{\s}$, 
	\eqref{eq:gGbasis_i0}, \eqref{eq:gbasis_i0},
	\eqref{eq:anchorformshom}, 
	$\gBil^k$ are smooth (for $\gBil^k$ use Lemma \ref{lem:Binbases}), 
	and $\dg$, $[\cdot,\cdot]$ 
	are differential operators with smooth coefficients on $\cyl$.

Note that $\ginj_k$ is the inclusion $\gxG^k(\Dspcl)\hookrightarrow \gx^k(\Dspcl)$,
via the identification \eqref{eq:identify_g_vec_i0}.
\begin{lemma}[Homogeneity]\label{lem:comphomog}
The differential forms \eqref{eq:anchorformshom},
and the components of \eqref{eq:minkarrays}, are 
homogeneous of degree zero in the sense of Definition \ref{def:homog}.
\end{lemma}
\begin{proof}
We check that \eqref{eq:anchorformshom} are homogeneous of degree zero:
\begin{align*}
\Scal^*_{\lambda}(\anchor_k)^\mu_\ell
&=
\Scal^*_\lambda
\big(\tfrac{1}{\s} \anchorg(\eg^k_{\ell})(y^\mu)\big)
=
\Scal^*_\lambda(\tfrac{1}{\s})
\anchorg(\Scalg_\lambda \eg^k_{\ell})(\Scal^*_\lambda y^\mu)
=
\tfrac{1}{\s}\anchorg(\eg^k_{\ell})(y^\mu)
=
(\anchor_k)^\mu_\ell
\end{align*}
by Lemma \ref{lem:homogeneity}, 
the fact that $\eg^k_\ell$ are homogeneous of degree zero
by Lemma \ref{lem:bases_i0}, and 
by $\Scal^*_{\lambda}(\tfrac{1}{\s})=\frac{\lambda}{\s}$
and $\Scal^*_{\lambda}y^\mu=\frac{y^\mu}{\lambda}$.

To check that the components of \eqref{eq:minkarrays} are 
homogeneous of degree zero
one proceeds similarly, using the fact that the components of $\gBil^k$
are constant by Lemma \ref{lem:Binbases},
using Lemma \ref{lem:homogeneity},
and using homogeneity of $\eg_\ell^{k}$.
For example,
\begin{align*}
\Scal_{\lambda}^*(\amink_{\ideg}^\mu)_{ij}
	&=
	\Scal_{\lambda}^* 
	(\gBil^{\ideg}\big(\eGg_\outind^{\ideg}, 
	\tfrac{dy^\mu}{\s} \eGg^{\ideg}_j\big))
	=
	\gBil^{\ideg}\big(\Scalg_{\lambda}(\eGg_\outind^{\ideg}), 
	\Scalg_{\lambda}(\tfrac{dy^\mu}{\s} \eGg^{\ideg}_j)\big)\\
	&\overset{(1)}{=}
	\gBil^{\ideg}\big(\Scalg_{\lambda}(\eGg_\outind^{\ideg}), 
	\Scal_{\lambda}^*(\tfrac{dy^\mu}{\s}) \Scalg_{\lambda}(\eGg^{\ideg}_j)\big)
	=
	\gBil^{\ideg}\big(\eGg_\outind^{\ideg}, 
		\tfrac{dy^\mu}{\s} \eGg^{\ideg}_j\big)
	=(\amink_{\ideg}^\mu)_{ij}
\end{align*}
where Lemma \ref{lem:homogeneity} is used in (1).
\qed
\end{proof}
We now rewrite \eqref{eq:necSHS} 
and \eqref{eq:constraintsprop} in standard matrix-vector form.
\begin{lemma}\label{lem:translationofeq}
For all $\cc\in\gxG^1(\Dspcl)$, 
$v\in\gx^1(\Dspcl)$ and
$U\in\gxG^2(\Dspcl)$:
\begin{align*}
\gBil^{1}\left( \eGg_{i}^1, 
\dg (v+\cc) + \tfrac12[v+\cc,v+\cc]\right) e^{\ngG_1}_i 
&=
\big(
\amink_1^{\mu}
+ \Amink_1^{\mu}(v)
+ \Amink_1^{\mu}(\ginj_1 \cc)
\big) X_{\mu}  \cc
\nonumber\\
&\;\;-
\big(\Lmink_1  \cc - \Amink_{11}^\mu(\ginj_1 \cc) X_{\mu} v + \Bmink_1(v,\ginj_1 \cc)\big) 
\nonumber\\
&\;\;- \tfrac12\Bmink_1(\ginj_1\cc,\ginj_1\cc) - \SFmink \big(\dg v+\tfrac12[v,v]\big)\nonumber \\[2mm]
\gBil^{2}\left( \eGg_{i}^2, \dg U + [v,U]\right)e^{\ngG_2}_i
&=
\left(
\amink_2^{\mu}
+ \Amink_2^{\mu}(v) \right) X_{\mu}  U\\
&\;\;-
\big(\Lmink_2  U \new{+} \Amink_{21}^\mu(\ginj_2 U) X_{\mu} v + \Bmink_2(v,\ginj_2 U)\big)
\nonumber
\end{align*}
where, on the left hand sides, 
$(e_{i}^{\ngG_k})_{i=1\dots\ngG_k}$ is the standard basis of $\R^{\ngG_k}$
and we sum over $i$,
and, on the right hand sides, the identification 
\eqref{eq:identify_g_vec_i0} is used.
\end{lemma}
\begin{proof}
Here we prove the first identity of the lemma,
the proof of the second is analogous and thus omitted.
We will implicitly sum repeated indices
$i,j$ over $1\dots\ngG_1$,
and repeated indices $\ell$ over $1\dots\ng_1$,
and repeated indices $\mu$ over $0\dots3$.
Expand $\cc = \cc_j \eGg^1_j$ and $v =  v_\ell \eg^1_\ell$.

Using linearity of $\dg$, 
bilinearity of $[\cdot,\cdot]$, 
and $[v,\cc]=[\cc,v]$ by \eqref{eq:bracketas},
\begin{equation}\label{eq:mcuv}
\dg (v+\cc) + \tfrac12[v+\cc,v+\cc]
=
\dg \cc + [v,\cc] + \tfrac12[\cc,\cc] + \left(\dg v + \tfrac12[v,v]\right)
\end{equation}
Thus by $C^\infty$-bilinearity of $\gBil^1$, it suffices to show:
\begin{subequations}
\begin{align}
\gBil^1( \eGg^1_i,\dg \cc) e^{\ngG_1}_i
	&=
	\amink_{1}^\mu X_\mu \cc
	-
	\Lmink_{1} \cc\label{eq:dal}\\
\gBil^1( \eGg^1_i, [v,\cc]) e^{\ngG_1}_i
	&=
	\Amink_{1}^\mu(v) X_\mu \cc 
	+
	\Amink_{11}(\ginj_1 \cc) X_\mu v
	-
	\Bmink_1(v,\ginj_1 \cc)\label{eq:[vc]}\\
\gBil^1( \eGg^1_i, \tfrac12 [\cc,\cc])e^{\ngG_1}_i
	&=
	\Amink^\mu_1(\ginj_1\cc) X_\mu \cc
	-
	\tfrac12\Bmink_1(\ginj_1\cc,\ginj_1\cc)\label{eq:[cc]}\\
\gBil^{1}\left( \eGg_{i}^1, \dg v + \tfrac12[v,v]\right) e^{\ngG_1}_i
	&=
	-\SFmink(\dg v + \tfrac12[v,v])\label{eq:FF}
\end{align}
\end{subequations}

\eqref{eq:dal}: We have
\begin{align*}
\dg \cc
&=
\dg(\cc_j\eGg^1_j)
\overset{\smash{\eqref{eq:dgmodule}}}{=}
\ddR(\cc_j)\eGg^1_j +\cc_j\dg(\eGg^1_j)
=
(X_\mu \cc_j) \theta^\mu \eGg^1_j+\cc_j\dg(\eGg^1_j)
\intertext{
where we abbreviate $\theta^\mu = \frac{dy^\mu}{\s}$.
Using $C^\infty$-bilinearity of $\gBil^1$
and Definition \ref{def:SHSarrays_i0},}
\gBil^1( \eGg^1_i, \dg \cc)
&=
\gBil^1( \eGg^1_i, \theta^\mu \eGg^1_j) (X_\mu \cc_j)
+
\gBil^1( \eGg^1_i, \dg \eGg^1_j) \cc_j
=
(\amink_{1}^\mu)_{i j} X_\mu \cc_j
-
(\Lmink_{1})_{i j} \cc_j
\end{align*}
From this \eqref{eq:dal} follows,
using the identification \eqref{eq:identify_g_vec_i0}.
\eqref{eq:[vc]}: We have
\begin{align*}
[v,\cc]
&=
[v,\cc_j \eGg^1_j]
\overset{\eqref{eq:anchorbracket}}{=}
\anchorg(v)(\cc_j)\eGg^1_j + \cc_j[v,\eGg^1_j]\\
&\overset{\eqref{eq:anchorlin},\eqref{eq:anchorbracket},\eqref{eq:bracketas}}{=}
v_\ell\anchorg(\eg^1_\ell)(\cc_j)\eGg^1_j 
+
\cc_j \anchorg(\eGg^1_j)(v_\ell) \eg^1_\ell 
+
\cc_j v_\ell [\eg^1_\ell,\eGg^1_j]\\
&=
v_\ell (\rho_1)_\ell^\mu \eGg^1_j X_\mu \cc_j
+
\cc_j (\rho_1)^\mu_j \eg^1_\ell X_\mu v_\ell
+
\cc_j v_\ell [\eg^1_\ell,\eGg^1_j]
\end{align*}
where in the last line we use \eqref{eq:anchorghom} and 
$\eGg^1_j=\eg^1_j$ for $j=1\dots\ngG_1$.
Thus
\begin{align*}
\gBil^1( \eGg^1_i, [v,\cc])
&=
\gBil^1( \eGg^1_i, (\rho_1)_\ell^\mu \eGg^1_j) v_\ell (X_\mu \cc_j) \\
&\quad+
\gBil^1( \eGg^1_i, (\rho_1)^\mu_j \eg^1_\ell ) \cc_j (X_\mu v_\ell)\\
&\quad+
\gBil^1( \eGg^1_i, [\eg^1_\ell,\eGg^1_j]) \cc_j v_\ell\\
&=
(\Amink_{1}^\mu)_{\ell,i j}  v_\ell (X_\mu \cc_j)
+
(\Amink_{11}^\mu)_{j,i \ell}\cc_j (X_\mu v_\ell)
-
(\Bmink_{1})_{i \ell j} \cc_j v_\ell
\end{align*}
using Definition \ref{def:SHSarrays_i0}
and $\eGg^1_j=\eg^1_j$ for $j=1\dots\ngG_1$.
From this \eqref{eq:[vc]} follows, using \eqref{eq:identify_g_vec_i0}.
\eqref{eq:[cc]}: Use \eqref{eq:[vc]} with $v=\ginj_1\cc$
and
$\Amink_{11}^\mu(\ginj_1\cc)X_\mu\ginj_1\cc
=
\Amink_{1}^\mu(\ginj_1\cc)X_\mu \cc$.
\eqref{eq:FF}:
Expanding $\dg v + \frac12[v,v] = V_r \eg^2_r$
with implicit sum over $r=1\dots\ng_2$,
\begin{align}
\gBil^{1}\left( \eGg_{i}^1, \dg v + \tfrac12[v,v]\right) 
&=
\gBil^{1}\big( \eGg_{i}^1, \eg^2_r\big) V_r 
=
-\SFmink_{i r}V_r
\end{align}
using Definition \ref{def:SHSarrays_i0}.
From this \eqref{eq:FF} follows, using \eqref{eq:identify_g_vec_i0}.\qed
\end{proof}
We prove basic properties of 
$\amink^\mu_k$, 
$\Amink^\mu_k$,
$\Lmink_k$. 
We use the decomposition
\begin{equation}\label{eq:RGdecomp}
\R^{\ngG_{k}} 
=
\R^{6\ngG^{\Omega}_k}
\oplus
\R^{4\ngG^{\Omega}_k}
\oplus 
\R^{\ngG^{\I}_{k+1}}
\end{equation}
Note that the basis \eqref{eq:gGbasis_i0} is compatible
with this decomposition, see Lemma \ref{lem:gaugebases_i0}.
\begin{lemma}\label{lem:aALproperties}
For every $\mu=0\dots3$ and $k=1,2$:
\begin{enumerate}[({f}1)]
\item \label{item:aprop}
$\amink^\mu_{k} \in C^\infty(\Dspcl,\End(\R^{\ngG_k}))$
is a symmetric matrix at every point on $\Dspcl$.
Its block decomposition relative to \eqref{eq:RGdecomp} is block diagonal:
\begin{align}\label{eq:ablockdiag}
\amink_{k}^{\mu}
=
\begin{pmatrix}
\one \delta_{\mu0} & 0 & 0\\
0 & \one \delta_{\mu0} & 0\\
0 & 0 & \amink_{k,\I}^{\mu}
\end{pmatrix}
\end{align}
where the entries of the matrix $\amink_{k,\I}^\mu$ are constant, i.e.~in $\R$.
Furthermore $\amink_{k,\I}^0=\one$,
and for every $\omega\in\Omega^1(\Dspcl)$:
\begin{equation}\label{eq:aineq}
\omega(\amink_{k,\I}^{\mu}X_\mu) \ge 
\big( \omega(X_0)-(
\textstyle\sum_{i=1}^3|\omega(X_i)|^2)^{\frac12} \big) \one
\end{equation}
\item \label{item:Aprop}
Let $(e_\ell^{\ng_1})_{\ell=1\dots\ng_1}$ be the standard basis of $\R^{\ng_1}$.
For every $\ell=1\dots\ng_1$, 
$\Amink^\mu_{k}(e_\ell^{\ng_1}) \in C^\infty(\Dspcl,\End(\R^{\ngG_k}))$
is a symmetric matrix at every point on $\Dspcl$.
Its block decomposition relative to \eqref{eq:RGdecomp} is block diagonal.
Furthermore, for every $\ell,k$ there exists
$\new{f_{\ell k}}\in C^\infty(\Dspcl,\End(\R^{\ngG_k}))$
whose components are homogeneous of degree zero and such that 
\begin{equation}\label{eq:A(1-t)}
d\ttcoord(\Amink_k^\mu(e_{\ell}^{\ng_{1}}) X_\mu) =
(1-\ttcoord) f_{\ell k}
\end{equation}
\item \label{item:Lprop}
$\Lmink_1\in C^\infty(\Dspcl,\End(\R^{\ngG_1}))$ is upper block triangular relative to \eqref{eq:RGdecomp}:
\begin{align*}
\Lmink_{1}
=
\begin{pmatrix}
\one d\zzeta(X_0) & 0 & \Lmink_{\boosts\I}\\
0 & 2\one d\zzeta(X_0) & \Lmink_{\transl\I}\\
0 & 0 & 4 d\zzeta(\amink_{1,\I}^{\mu} X_{\mu})
\end{pmatrix}
\end{align*}
In particular, the diagonal blocks are proportional 
to those of $d\zzeta(\amink_{k}^{\mu} X_\mu)$,
see \eqref{eq:ablockdiag}.
The entries of $\Lmink_{\boosts\I}$, $\Lmink_{\transl\I}$ 
are homogeneous of degree zero.
\end{enumerate}
\end{lemma}
\begin{proof}
Abbreviate $\theta^\mu = \frac{dy^\mu}{\s}$.
\ref{item:aprop}:
Except for \eqref{eq:aineq}, the claim 
follows from \eqref{eq:gbil} and Lemma \ref{lem:Binbases}
(see also \ref{item:gaugesymNEW} of Lemma \ref{lem:gauge_mainprop_i0}).
For \eqref{eq:aineq} note that 
\begin{align*}
\omega((\amink_{k,\I}^{\mu})_{ij}X_\mu)
&=
\IBil^{k+1}(\eGI_{i}^{k+1},\theta^{\mu} \eGI_{j}^{k+1})
\omega(X_\mu)
=
\IBil^{k+1}(\eGI_{i}^{k+1},\omega \eGI_{j}^{k+1})
\end{align*}
Thus \eqref{eq:aineq} follows from 
\eqref{eq:posineq} restricted to $u\in\I^{k+1}(\Dspcl)$,
and $\amink_{k,\I}^{0}=\one$.

\ref{item:Aprop}:
By Definition \ref{def:SHSarrays_i0},
the components of $\Amink^\mu_{k}(e_\ell^{\ng_1})$ are
\begin{align*}
(\Amink^\mu_{k}(e_\ell^{\ng_1}))_{ij}
=
\gBil^{\new{k}}\left(\eGg^{k}_i, 
(\anchor_1)_\ell^\mu \eGg^{k}_j\right)
\end{align*}
where $(\anchor_1)_\ell^\mu\in\Omega^1(\Dspcl)$.
This is symmetric in $ij$ by \ref{item:gaugesymNEW} of Lemma \ref{lem:gauge_mainprop_i0}.
It is block diagonal by \eqref{eq:gbil}
and because
multiplication by a one-form maps each of 
the modules
$\Omega(\Dspcl)\otimesRR\boosts$, 
$\Omega(\Dspcl)\otimesRR\transl$, 
$\I(\Dspcl)$ back to itself.
We check \eqref{eq:A(1-t)}. 
The functions $f_{\ell k}$ are homogeneous of degree zero because
the left hand side of \eqref{eq:A(1-t)} is homogeneous of degree zero (by Lemma \ref{lem:comphomog} and homogeneity of $d\ttcoord$, $X_\mu$)
and because $1-\ttcoord$ is homogeneous of degree zero.
For smoothness, write
\begin{align}\label{eq:Aanchorcalc_i0}
\begin{aligned}
d\ttcoord\big((\Amink^\mu_{k}(e_\ell^{\ng_1}))_{ij}X_\mu\big)
	&=
	(\Amink^\mu_{k}(e_\ell^{\ng_1}))_{ij}d\ttcoord(X_\mu)\\
	&=
	\gBil^{\new{k}}\big(\eGg^{k}_i, 
	(\anchor_1)_\ell^\mu X_\mu(\ttcoord)\eGg^{k}_j\big)\\
	&=
	\gBil^{\new{k}}\big(\eGg^{k}_i, 
	\anchorg(\eg^1_\ell)(\ttcoord)\eGg^{k}_j\big)
\end{aligned}
\end{align}
where the last step uses \eqref{eq:anchorghom}.
There are three cases:
\begin{itemize}
\item 
If $\eg^1_\ell = (\theta^\mu\otimes B^{\alpha\beta})\oplus0\Ieps$
for some $\mu,\alpha,\beta=0\dots3$, then, using \eqref{eq:anchorLdef},
\[ 
\anchorg(\eg^1_\ell)(\ttcoord) 
=
\anchorL(\theta^\mu\otimes B^{\alpha\beta})(\ttcoord)
= 
\theta^\mu B^{\alpha\beta}(\ttcoord)
=
(1-\ttcoord)(1+\ttcoord)
\left(\begin{smallmatrix}
0 & \frac{\vec{y}^T}{|\vec{y}|}\\
-\frac{\vec{y}}{|\vec{y}|} & 0
\end{smallmatrix}\right)_{\alpha\beta}\theta^\mu
\]
\item 
If $\eg^1_\ell = (\frac{1}{\s}\theta^\mu\otimes T_\nu)\oplus0\Ieps$
for some $\mu,\nu=0\dots3$, then 
\[ 
\anchorg(\eg^1_\ell)(\ttcoord) 
=
\anchorL(\tfrac{1}{\s}\theta^\mu\otimes T_\nu)(\ttcoord)
= 
\theta^\mu \tfrac{1}{\s}T_\nu(\ttcoord)
=
(1-\ttcoord) 
\tfrac{1+\ttcoord}{1+2\ttcoord} 
\left(\begin{smallmatrix}
1\\ -\ttcoord \frac{\vec{y}}{|\vec{y}|}
\end{smallmatrix}\right)_{\nu}\theta^\mu 
\]
\item 
If $\eg^1_\ell = 0\oplus\eI^2_{j}\Ieps$ for some $j=1\dots10$,
then $\anchorg(\eg^1_\ell)=\anchorL(0)=0$.
\end{itemize}
From this and 
$\gBil^{\new{k}}\big(\eGg^{k}_i, 
	\theta^\mu\eGg^{k}_j\big) \in C^\infty(\Dspcl)$,
the claim follows.

\ref{item:Lprop}: 
By Definition \ref{def:SHSarrays_i0},
by the definition of $\dg$ in \eqref{eq:dgdef},
and by \eqref{eq:gbil}:
\begin{align*}
\Lmink_1
=
\begin{pmatrix}
\Lmink_{\boosts} & 0 & \Lmink_{\boosts\I}\\
0 & \Lmink_{\transl} & \Lmink_{\transl\I}\\
0 & 0 & \Lmink_{\I}
\end{pmatrix}
\qquad
\text{where}
\qquad
\begin{aligned}
(\Lmink_{\boosts})_{ij} 
	&= -\OmegaBBil^1\left(\eGB_i^1,\dl\eGB_j^1\right)\\
(\Lmink_{\transl})_{ij} 
	&= -\OmegaTBil^1\left(\eGT_i^1,\dl\eGT_j^1\right)\\
(\Lmink_{\I})_{ij} 
	&= -\IBil^2\left(\eGI_i^2,\dI\eGI_j^2\right)\\
(\Lmink_{\boosts\I})_{ij} 
	&= +\OmegaBBil^1\left(\eGB_i^1,\pi_{\boosts}\inj(\eGI_j^2)\right)\\
(\Lmink_{\transl\I})_{ij} 
	&= +\OmegaTBil^1\left(\eGT_i^1,\pi_{\transl}\inj(\eGI_j^2)\right)
\end{aligned}
\end{align*}
where $\pi_{\boosts}$, $\pi_{\transl}$ denote
the projections onto the two direct summand of 
$\Omega^2(\Dspcl)\otimesRR\Kil
	=(\Omega^2(\Dspcl)\otimesRR\boosts)\oplus (\Omega^2(\Dspcl)\otimesRR\transl)$.
Consider the diagonal blocks. 
We claim that 
\begin{align}\label{eq:dzero}
\dl (\s\eGB_j^1) &= 0&
\dl (\s^2\eGT_j^1) &= 0 &
\dI (\s^4 \eGI_j^2) &= 0
\end{align}
Proof of \eqref{eq:dzero}:
For the first note that $\s\eGB_j^1=dy^\mu \otimes B^{\alpha\beta}$
for some $\mu,\alpha,\beta=0\dots3$, hence
$\dl(\s\eGB^1_j)=\ddR (dy^\mu) \otimes B^{\alpha\beta}=0$.
For the second note that $\s^2\eGT^1_j = dy^\mu \otimes T_{\nu}$
for some $\mu,\nu=0\dots3$, 
hence $\dl(\s^2\eGT^1_j)=0$.
Consider the third.
Each element $\s^4\eGI_j^2$ is of the form $u\oplus\bar{u}$
where $u\in \I_+^2(\Dspcl)$ is a $\C$-linear combination of elements
\begin{align*}
\mu_{\etay}^{-1} \otimes (dy^{i_1}\wedge dy^{i_2}) \otimes (dy^{i_3}\wedge dy^{i_4}) 
\qquad
i_1,\dots,i_4=0\dots3
\end{align*}
We have $\dI(u)=0$:
Use the formula \eqref{eq:dIpm} for $\dI$,
where, by Remark \ref{rem:dIprop}, 
one can replace the representative $\gcyl$ of $[\gcyl]$
by the representative $\etay$,
then $\nabla^{\etay}_{\V_{\alpha}} dy^{\mu}=0$ for $\alpha,\mu=0\dots3$
yields $\dI(u)=0$. 
Thus $\dI(\s^4\eGI_j^2)=0$, which proves \eqref{eq:dzero}.

We conclude \ref{item:Lprop}.
Observe that
\begin{align*}
\dl(\eGB_j^1)
&=
\dl(\tfrac{1}{\s}\s\eGB_j^1)
=
(\ddR\tfrac{1}{\s})\s\eGB_j^1
+
\tfrac{1}{\s}\dl(\s\eGB_j^1)
\overset{\smash{(1)}}{=}
\s(\ddR\tfrac{1}{\s})\eGB_j^1
=
- d\zzeta \eGB_j^1\\
\dl(\eGT_j^1)
&=
-2d\zzeta \eGT_j^1\\
\dI(\eGI_j^2)
&=
\dI(\s^{-4}\s^4\eGI_j^2)
\overset{\smash{(2)}}{=}
(\ddR\s^{-4}) \s^4\eGI_j^2
+
\s^{-4} \dI(\s^4\eGI_j^2)
\overset{\smash{(1)}}{=}
\s^4 \ddR(\s^{-4}) \eGI_j^2
=
-4 d\zzeta  \eGI_j^2
\end{align*}
where in (1) we use \eqref{eq:dzero}, and in (2)
we use the Leibniz rule in Remark \ref{rem:dIprop}. 
Thus
\begin{align*}
(\Lmink_{\boosts})_{ij} 
&=
-\OmegaBBil^1\left(\eGB_i^1,\dl\eGB_j^1\right)
=
\OmegaBBil^1\left(\eGB_i^1, d\zzeta \eGB_j^1\right)
=
\OmegaBBil^1\left(\eGB_i^1, \theta^\mu \eGB_j^1\right) d\zzeta(X_\mu)
\overset{\smash{(1)}}{=}
\delta_{ij} d\zzeta(X_0)\\
(\Lmink_{\transl})_{ij}
&= 2\delta_{ij}d\zzeta(X_0)\\
(\Lmink_{\I})_{ij} 
&= 
4\IBil^2\left(\eGI_i^2,d\zzeta  \eGI_j^2\right)
=
4\IBil^2\left(\eGI_i^2,\theta^\mu  \eGI_j^2\right) d\zzeta(X_\mu)
\overset{\smash{(1)}}{=}
4(\amink_{1,\I}^{\mu})_{ij} d\zzeta(X_{\mu})
\end{align*}
as claimed,
where in (1) we use \ref{item:aprop}.
Homogeneity follows from Lemma \ref{lem:comphomog}.\qed
\end{proof}
\begin{lemma}\label{lem:qminkdelta}
There exist $\qmink\ge1$ and $\deltamink\in(0,1]$ such that for all $k=1,2$ and 
for all $u\in \R^{\ng_1}$ with $\sqrt{u^Tu}\le 2\deltamink$, 
at every point on $\Dspcl$ one has
\begin{subequations}\label{eq:minkcausal}
\begin{align}
\qmink^{-1} \one 
	&\le d\zzeta\left( 
	\big(
	\amink_k^{\mu}
	+ \Amink_k^{\mu}(u) 
	\big) X_{\mu} 
	\right)
	\le
	\qmink\one
	\label{eq:dzetapos}\\
\qmink^{-1}\one &\le 
	\;\;\;\;\;\;\; 
	\amink_k^{0} 
	+ \Amink_k^{0}(u) 
	\qquad\ \le \qmink\one
	\label{eq:d0pos}\\
(1-\ttcoord)\qmink^{-1} &\le 
	d\ttcoord\left( 
	\big(
	\amink_k^{\mu} 
	+ \Amink_k^{\mu}(u) 
	\big) X_{\mu} 
	\right)
	\le
	\qmink\one
	\label{eq:dtpos}
\end{align}
\end{subequations}
where the matrices in the middle are symmetric by Lemma \ref{lem:aALproperties}.
\end{lemma}
\begin{proof}
We first consider only 
\begin{equation}\label{eq:ag}
d\zzeta(\amink_k^{\mu}X_\mu)
\qquad
\amink_k^{0}
\qquad
d\ttcoord(\amink_k^{\mu}X_\mu)
\end{equation}
By \ref{item:aprop} these matrices are constant,
and  $\amink_k^{0}=\one$.
Further \eqref{eq:ablockdiag} and \eqref{eq:aineq} yield
\begin{align*}
d\zzeta(\amink_{k}^{\mu}X_\mu)
&\ge
\big( d\zzeta(X_0)-(
\textstyle\sum_{i=1}^3|d\zzeta(X_i)|^2)^{\frac12} \big) \one
=
\one\\
d\ttcoord(\amink_{k}^{\mu}X_\mu)
&\ge
\big( d\ttcoord(X_0)-(
\textstyle\sum_{i=1}^3|d\ttcoord(X_i)|^2)^{\frac12} \big) \one
=
(1-\ttcoord)(1+2\ttcoord)\one
\ge
(1-\ttcoord)\one
\end{align*}
This gives a lower bound for the matrices \eqref{eq:ag}.
We also have an upper bound, 
using the fact that \eqref{eq:ag} 
are smooth on $\Dspcl$ and homogeneous of degree zero,
and compactness of $\Dspcl_{s}$.
Hence there exists $\qminkaux\ge1$ such that
at every point on $\Dspcl$:
\begin{align}\label{eq:qminkdefqaux}
\begin{aligned}
\one \;&\le\; d\zzeta(\amink_{k}^{\mu}X_\mu) \;\le\; \qminkaux \one\\
\one \;&\le\; \amink_{k}^{0} \;\le\;  \one\\
(1-\ttcoord) \one \;&\le\; d\ttcoord(\amink_{k}^{\mu}X_\mu) \;\le\; \qminkaux \one
\end{aligned}
\end{align}
Set $\qmink=2\qminkaux$.
Choose $\deltamink>0$ sufficiently small so that for $k=1,2$:
\begin{subequations}\label{eq:deltacondNEW}
\begin{align}
(2\deltamink)\tsum_{\ell=1}^{\ng_1}\|d\zzeta(\Amink_k^{\mu}(e_\ell^{\ng_1}) X_\mu)\|
	&\le \tfrac{1}{4\qminkaux} \label{eq:deltacondzetaNEW}\\
(2\deltamink)\tsum_{\ell=1}^{\ng_1}\|\Amink_k^{0}(e_\ell^{\ng_1})\|
	&\le \tfrac{1}{4\qminkaux} \label{eq:deltacond0NEW}\\	
(2\deltamink)\tsum_{\ell=1}^{\ng_1}\|\frac{1}{1-\ttcoord}d\ttcoord(\Amink_k^{\mu}(e_\ell^{\ng_1}) X_\mu)\|
	&\le \tfrac{1}{4\qminkaux} \label{eq:deltacondtNEW}
\end{align}
\end{subequations}
at every point on $\Dspcl$, 
where $\|{\cdot}\|$ is the $\ell^2$-matrix norm,
and $(e_\ell^{\ng_1})_{\ell=1\dots\ng_1}$
the standard basis of $\R^{\ng_1}$.
Such a $\deltamink>0$ exists
because $d\zzeta(\Amink_k^{\mu}(e_\ell^{\ng_1}) X_\mu)$,
$\Amink_k^{0}(e_\ell^{\ng_1})$, 
$\frac{1}{1-\ttcoord}d\ttcoord(\Amink_k^{\mu}(e_\ell^{\ng_1}) X_\mu)$
are smooth on $\Dspcl$ and homogeneous of degree zero (Lemma \ref{lem:comphomog}
and \ref{item:Aprop})
and by compactness of $\Dspcl_{s}$.

We conclude \eqref{eq:minkcausal}.
Expand $u = \sum_{\ell=1}^{\ng_1} u_\ell e_{\ell}^{\ng_1}$.
At every point on $\Dspcl$:
\begin{align*}
\|d\zzeta(\Amink_k^{\mu}(u)X_\mu)\|
&\le\textstyle
\sum_{\ell=1}^{\ng_1} |u_\ell|\|d\zzeta(\Amink_k^{\mu}(e_{\ell}^{\ng_1})X_\mu)\|\\
&\le\textstyle
2\deltamink
\sum_{\ell=1}^{\ng_1}\|d\zzeta(\Amink_k^{\mu}(e_{\ell}^{\ng_1})X_\mu)\|
\le
\tfrac{1}{4\qminkaux}
\end{align*}
by \eqref{eq:deltacondzetaNEW}.
Similarly
$\|\Amink_k^{0}(u)\| \le \tfrac{1}{4\qminkaux}$
and 
$\|d\ttcoord(\Amink_k^{\mu}(u)X_\mu)\|
\le \tfrac{1-\ttcoord}{4\qminkaux}$.
This implies
\begin{align*}
-\tfrac{1}{4\qminkaux}\one 
	&\le d\zzeta(\Amink_k^{\mu}(u)X_\mu) 
	\le \tfrac{1}{4\qminkaux}\one\\
-\tfrac{1}{4\qminkaux}\one 
	&\le \Amink_k^{0}(u)
	\le \tfrac{1}{4\qminkaux}\one\\
-\tfrac{1-\ttcoord}{4\qminkaux}\one 
	&\le d\ttcoord(\Amink_k^{\mu}(u)X_\mu) 
	\le \tfrac{1-\ttcoord}{4\qminkaux}\one
	\le \tfrac{1}{4\qminkaux}\one
\end{align*}
and from this \eqref{eq:minkcausal} follows, using $\qmink=2\qminkaux$.
\qed
\end{proof}
\begin{lemma}\label{lem:ellMinkest}
Using Convention \ref{conv:ztXmu_NEW}, 
the functions $\LdivestNEW{}$ and $\CcomTang{},\CcomTransv{}$ in Definition \ref{def:ellCcomdef} respectively Definition \ref{def:defCcom} have the following properties:
\begin{itemize}
\item 
Let $\chi>0$ and let $\conj\in\End(\R^{\ngG_1})$ be given by the matrix
\begin{equation*}
\conj 
	= 
	\left(\begin{smallmatrix}
	\one &0&0\\
	0&\one&0\\
	0&0&\chi\one
	\end{smallmatrix}\right)
\end{equation*}
using \eqref{eq:RGdecomp}.
Let $\qmink\ge1$ be as in Lemma \ref{lem:qminkdelta}.
Then for all $\zz\le0$:
\begin{align}\label{eq:ellmink}
\LdivestNEW{\conj^{-1} \Lmink_1 \conj,\amink_1^\mu X_\mu}(\zz)
\le
\tfrac52 + \tfrac12\chi\, \qmink \big(\|\Lmink_{\boosts\I}\|_{L^\infty(\Dspcl)}+\|\Lmink_{\transl\I}\|_{L^\infty(\Dspcl)}\big)
\end{align}
where the pointwise norm is the $\ell^2$-matrix norm,
with $\Lmink_{\boosts\I}$, $\Lmink_{\transl\I}$ from \ref{item:Lprop}.
\item 
Let $\Ccomp = (\Ccomp_1,\Ccomp_2,\Ccomp_3)$ with 
$
\Ccomp_i = \frac{y^i}{|\vec{y}|}
$.
Then for all $\zz\le0$:
\begin{align}\label{eq:CcomT0}
\CcomTang{\amink_1^\mu X_\mu,\Ccomp}(\zz)& \new{=} 1 
&
\CcomTransv{\amink_1^\mu X_\mu,\Ccomp}(\zz)& = 0
\end{align}
\end{itemize}
\end{lemma}
\begin{proof}
We abbreviate $\amink_1=\amink_1^\mu X_\mu$.
Note that, as required by \eqref{eq:ALassp}, 
the matrix $d\zzeta(\amink_1)$ is positive 
at every point on $\Dspcl_{\le1}$ by \eqref{eq:dzetapos} with $u=0$.

\eqref{eq:ellmink}:
Abbreviate $\Lmink_{\chi}=\conj^{-1}\Lmink_1\conj$.
By definition,
\begin{align*}
\LdivestNEW{\Lmink_{\chi},\amink_1}(\zz)
=
\sup_{p\in\Dspop_{e^{\zz}}}
\sig\Big(\Lmink_{\chi}|_p+\tfrac12\div_{\muW}(\amink_1)|_p\,,\,d\zzeta(\amink_1)|_p\Big)
\end{align*}
where we use \eqref{eq:levelsets}.
We claim that 
\begin{align}\label{eq:Adiv}
\div_{\muW}(\amink_1) &= -3 d\zzeta(\amink_1)
\end{align}
Proof of \eqref{eq:Adiv}:
The matrices $\amink_1^\mu$ are constant by \ref{item:aprop} of Lemma \ref{lem:aALproperties}, thus 
$\div_{\muW}(\amink_1)
=\amink_1^\mu \div_{\muW}(X_\mu)$.
Using $\muW=\s^{-4}|dy^0\wedge\cdots\wedge dy^3|$
and the Leibniz rules for the divergence,
one obtains $\div_{\muW}(X_\mu) = -3 d\zzeta(X_\mu)$, 
and thus \eqref{eq:Adiv}.
%
%
%

By \ref{item:Lprop},
$$
\Lmink_{\chi}
=
\begin{pmatrix}
\one & 0 & 0\\
0 & 2\one & 0\\
0 & 0 & 4\one
\end{pmatrix}
d\zzeta(\amink_1)
+
\chi\begin{pmatrix}
0 & 0 & \Lmink_{\boosts\I}\\
0 & 0 & \Lmink_{\transl\I}\\
0 & 0 & 0
\end{pmatrix}
$$
Together with \eqref{eq:Adiv} and \eqref{eq:ablockdiag} this yields
$$
\Lmink_{\chi} + \tfrac12\div_{\muW}(\ALinM)
=
\begin{pmatrix}
-\tfrac12\one&0&0\\
0&\tfrac12\one&0\\
0&0&\tfrac52\one
\end{pmatrix} d\zzeta(\amink_1)
+
\chi\begin{pmatrix}
0 & 0 & \Lmink_{\boosts\I}\\
0 & 0 & \Lmink_{\transl\I}\\
0 & 0 & 0
\end{pmatrix}
$$
Using the definition of $\sig$ in \eqref{eq:sigdef},
for each $p\in\Dspop_{e^{\zz}}$ we have
\begin{align*}
\sig\Big(\Lmink_{\chi}|_p+\tfrac12\div_{\muW}(\amink_1)|_p\,,\,d\zzeta(\amink_1)|_p\Big)
&\le
\frac52 +  \chi\sup_{w\in\R^{\ngG_1}}
 \frac{| w_{\boosts}^T\Lmink_{\boosts\I}|_p w_{\I}| + |w_{\transl}^T\Lmink_{\transl\I}|_p w_{\I}| }{w^Td\zzeta(\amink_1)|_p w }
\end{align*}
where we decompose $w=w_{\boosts}\oplus w_{\transl}\oplus w_{\I}$ using \eqref{eq:RGdecomp}.
We have
$|w_{\boosts}^T\Lmink_{\boosts\I}|_p w_{\I}| 
\le \|\Lmink_{\boosts\I}|_p\|\|w_{\boosts}\|\|w_{\I}\|
\le \frac12\|\Lmink_{\boosts\I}|_p\|\|w\|^2$ and 
$|w_{\transl}^T\Lmink_{\transl\I}|_p w_{\I}|\le\frac12 \|\Lmink_{\transl\I}|_p\|\|w\|^2$.
Thus
\begin{align*}
\sig\Big(\Lmink_{\chi}|_p+\tfrac12\div_{\muW}(\amink_1)|_p\,,\,d\zzeta(\amink_1)|_p\Big)
&\le
\tfrac52 + \tfrac12 \chi \qmink (\|\Lmink_{\boosts\I}|_p\|+\|\Lmink_{\transl\I}|_p\|)
\end{align*}
where we also use \eqref{eq:dzetapos} with $u=0$.
This shows \eqref{eq:ellmink}.

\eqref{eq:CcomT0}:
We claim that for each $\mu=0\dots3$ one has
\begin{subequations}\label{eq:comXa}
\begin{align}
&[X_\mu,\amink_1] 
=
d\zzeta(X_\mu)\amink_1-d\zzeta(\amink_1) X_\mu
	\label{eq:Xa1}\\
&d\zzeta(X_0) = 2\;,\;\;
d\zzeta(X_1) = \Ccomp_1\;,\;\;
d\zzeta(X_2) = \Ccomp_2\;,\;\;
d\zzeta(X_3) = \Ccomp_3
	\label{eq:dzetaX}
\end{align}
\end{subequations}
Proof of \eqref{eq:Xa1}:
Since the matrices $\amink_1^\mu$ are constant by \ref{item:aprop},
one has $[X_\mu,\amink_1]=\amink_1^\nu[X_\mu,X_\nu]$.
By direct calculation
$[X_{\mu},X_\nu]
=
d\zzeta(X_\mu)X_\nu-d\zzeta(X_\nu) X_\mu$,
from which \eqref{eq:Xa1} follows.
Proof of \eqref{eq:dzetaX}:
By direct calculation using $d\zzeta(X_\mu) = \p_{y^\mu}\s$.
%
%

The identities \eqref{eq:comXa} imply that 
for each $i=1,2,3$,
\begin{equation*}
[X_i,\amink_1]-\Ccomp_i\amink_1
=
d\zzeta(X_i)\amink_1-d\zzeta(\amink_1) X_i -\Ccomp_i\amink_1
=
-d\zzeta(\amink_1) X_i
\end{equation*}
Thus for all $w\in C^\infty(\Dspcl,{\R^{\ngG_1}})$,
the left hand sides of \eqref{eq:CcomTangdef} 
respectively \eqref{eq:CcomTransvdef} are 
\begin{align*}
-\sum_{i,j=1}^{3} (X_iw)^T \Xd^j([X_i,\amink_1]-\Ccomp_i\amink_1) X_j w
&=\sum_{i,j=1}^{3} (X_iw)^T \Xd^j(d\zzeta(\amink_1) X_i) X_j w\\
&=\sum_{i=1}^{3} (X_iw)^T d\zzeta(\amink_1) X_i w
=\sum_{i=1}^{3} \|X_iw\|^2_{\amink_1}\\
-\sum_{i=1}^{3} (X_iw)^T \Xd^0([X_i,\amink_1]-\Ccomp_i\amink_1) X_0 w
&=
\sum_{i=1}^{3} (X_iw)^T \Xd^0(d\zzeta(\amink_1) X_i) X_0 w
=0
\end{align*}
(Here $\Xd^j=\frac{dy^j}{\s}$.)
From this \eqref{eq:CcomT0} follows.
\qed
\end{proof}

\subsection{Main existence result}
	\label{sec:existence_i0}

We state and prove Proposition \ref{prop:ApplySHS}, 
the main result of Section \ref{sec:SpaceinfConstruction}.

Let $(\gxG(\Dspcl),\gBil)$ be the gauge in Definition \ref{def:gauges_i0}.
Denote by $\gxG(\Dspop_{\le\sfix})$ the space of sections
of $\gxG$ over $\Dspop_{\le\sfix}$, c.f.~Remark \ref{rem:gdiamond}.
We use the norms in Definition \ref{def:norms_i0}.
\begin{prop}\label{prop:ApplySHS}
For all 
\begin{equation}\label{eq:spconst}
\NN\in\Z_{\ge6}
\qquad
\epspower\in(0,1] 
\qquad
\CinGR>0
\end{equation}
there exist $\ClargeGR>0$ and $\CsmallGR\in (0,1]$ such
for all 
\[ 
s_*\in (0,1] \qquad\qquad v\in \gx^{1}(\Dspcl_{\le s_*})
\]
the following holds.
For every $\kk\in\N$ define $\Vconst_{\kk}(v)\in[0,\infty]$ by
\begin{align}\label{eq:defVkk}
\Vconst_{\kk}(v) = 
\int_{0}^{s_*} \Big(\frac{s_*}{s}\Big)^{\frac{5}{2}+\epspower+\kk} 
\left(1+|\log(\tfrac{s_*}{s})|\right)^{\kk}
\|\dg v + \tfrac12[v,v]\|_{\sHb^{\kk}(\Dspop_{s})} \frac{ds}{s}
\end{align}
If 
\begin{enumerate}[({h}1),leftmargin=10mm]
\item \label{item:vconstraints}
$\Pconstraints(v|_{y^0=0}) = 0$,
see Definition \ref{def:Pconstraints}
\item \label{item:CN+1v}
$\|v\|_{\nosCb^{\NN+1}(\Dspop_{\le s_*})} \le \CinGR$ 
\item \label{item:vint}
$\int_0^{s_*} \|v\|_{\sCb^{1}(\Dspop_{s})} \frac{ds}{s} \le \CinGR$
\item \label{item:vsmall}
$\|v\|_{\nosCb^0(\Dspop_{\le\sfix})} \le \CsmallGR$
\item \label{item:MC(v)small}
$\Vconst_{\NN}(v) \le \CsmallGR$ 
\end{enumerate}
Then there exists $c\in \gxG^{1}(\Dspop_{\le s_*})$ such that
\begin{subequations}\label{eq:existence_v+c}
\begin{align}
\dg(v+c) + \tfrac12[v+c,v+c] &= 0 \label{eq:v+cMC}\\
c|_{y^0=0} &=0 \label{eq:cdata}
\end{align}
\end{subequations}
Furthermore:
\begin{itemize}
\item 
\textbf{Part 0.} $c$ is unique.
\item \textbf{Part 1.}
For all $s\in(0,s_*]$:
\begin{align}
\|c\|_{\sHb^{\NN}(\Dspop_{s})} 
	&\le \smash{
	\ClargeGR 
	\big(\tfrac{s}{s_*}\big)^{\frac{5}{2}+\epspower+\NN}
	\Vconst_{\NN}(v) 
	}
	\label{eq:cestimates}
\end{align}
\item \textbf{Part 2.}
For every $\kk\in\Z_{\ge\NN}$ and every $\bb\new{>}0$, if
\begin{enumerate}[({h}1),leftmargin=10mm,resume]
\item \label{item:GRHigherassp12} 
$\Vconst_{\kk}(v) < \infty$ and $\Vconst_{\kk-1}(v) \le \bb$
\item 
$\|v\|_{\nosCb^{\kk+1}(\Dspop_{\le s_*})} \le \bb$ \label{item:GRHigherassp3}
\end{enumerate}
then for all $s\in(0,s_*]$:
\begin{align}\label{eq:GRHigherconcl}
\|c\|_{\sHb^{\kk}(\Dspop_{s})} 
\;
\smash{\lesssim_{\kk,\epspower,\CinGR,\bb}
\big(\tfrac{s}{s_*}\big)^{\frac{5}{2}+\epspower+\kk}} \Vconst_{\kk}(v) 
\end{align}
\end{itemize}
\end{prop}
The proof will use
the strategy discussed at the beginning of Section \ref{sec:reformSHS_i0}.
\begin{proof}
It suffices to prove this for the special case $\sfix=1$.
The reduction of the general case to the $\sfix=1$ case
goes as follows.
Given $\sfix\in (0,1]$ and $v\in \gx^1(\Dspcl_{\le\sfix})$,
define
$v' = \Scalg_{\lambda} v \in \gx^1(\Dspcl_{\le1})$
where
$\lambda := 1/\sfix$,
using \eqref{eq:Scalg_i0}.
The assumptions of the proposition for $v$ in the general case
imply the assumptions of the proposition for $v'$ in the $\sfix=1$ case,
by Lemma \ref{lem:homogeneity} and homogeneity of the norms \eqref{eq:norms_homog}. 
Let $c'\in\gxG^1(\Dspop_{\le1})$ be the solution produced 
by the proposition in the $\sfix=1$ case.
Then $c = \Scalg_{1/\lambda} c'\in \gxG^1(\Dspop_{\le\sfix})$
satisfies the conclusions in the general case, 
by Lemma \ref{lem:homogeneity} and \eqref{eq:norms_homog}.

We now prove the proposition for $\sfix=1$
where we abbreviate
\[ 
\Dspop_{\le1} = \Dspop
\]
Instead of specifying $\ClargeGR$ and $\CsmallGR$ upfront,
we will make finitely many admissible largeness respectively 
smallness assumptions as we go,
where admissible means that they depend only on \eqref{eq:spconst}.
The dependencies of the constants on the maps in Definition \ref{def:SHSarrays_i0}
will not be made explicit, since they are fixed once and for all.

We first consider the following necessary subsystem of \eqref{eq:v+cMC}:
\begin{align}\label{eq:MCnecessary}
\gBil^{1}(\,\cdot\,, \dg(v+c)+\tfrac12[v+c,v+c]) = 0
\end{align}
Abbreviate $V = \dg v + \tfrac12[v,v]$.
By Lemma \ref{lem:translationofeq} the system \eqref{eq:MCnecessary} is equivalent to 
\begin{align}\label{eq:cSHS}
\begin{aligned}
\big(
\amink_1^{\mu}
+ \Amink_1^{\mu}(v)
+ \Amink_1^{\mu}(\ginj_1 c)
\big) X_{\mu}  c
&=
\Lmink_1  c - \Amink_{11}^\mu(\ginj_1 c) X_{\mu} v + \Bmink_1(v,\ginj_1 c) \\
&\quad+\tfrac12\Bmink_1(\ginj_1c,\ginj_1c) + \SFmink V
\end{aligned}
\end{align}
where the identification \eqref{eq:identify_g_vec_i0} is used.

Fix $\qmink\ge1$ and $\deltamink\in(0,1]$ as in Lemma \ref{lem:qminkdelta}. 
Define
\begin{equation}\label{eq:chi}
\chi \;=\; \frac{\epspower}{1+\qmink
\big(\|\Lmink_{\boosts\I}\|_{L^\infty(\Dspcl)}
+\|\Lmink_{\transl\I}\|_{L^\infty(\Dspcl)}\big)}
\;\in\; (0,1]
\end{equation}
where the pointwise norm in $\|{\cdot}\|_{L^\infty(\Dspcl)}$
is the $\ell^2$-matrix norm. 
The norms $\|\Lmink_{\boosts\I}\|_{L^\infty(\Dspcl)}, \|\Lmink_{\transl\I}\|_{L^\infty(\Dspcl)}$ are finite by \ref{item:Lprop}.
Define $\conj\in\End(\R^{\ngG_1})$ by
\begin{equation*}
\conj 
= 
\left(\begin{smallmatrix}
\one & 0 & 0\\
0 & \one & 0\\
0 & 0 & \chi\one
\end{smallmatrix}\right)
\end{equation*}
using the decomposition \eqref{eq:RGdecomp}.
We conjugate \eqref{eq:cSHS} with $\conj$, i.e.~we replace 
\begin{equation}\label{eq:defctilde}
c=\conj\tilde{c}
\end{equation}
and apply $\conj^{-1}$ from the left. 
One has
$\Amink_1^{\mu}(\ginj_1 \conj\tilde{c})=\Amink_1^{\mu}(\ginj_1\tilde{c})$
and 
$\Amink_{11}^{\mu}(\ginj_1 \conj\tilde{c})=\Amink_{11}^{\mu}(\ginj_1\tilde{c})$
since 
$\anchorg(0\oplus\uI\Ieps)=0$ for all $\uI\in \I(\Dspcl)$, see \eqref{eq:anchorgdef}.
Further the matrices
$\amink_1^{\mu}$,
$\Amink_1^{\mu}(v)$,
$\Amink_1^{\mu}(\ginj_1\tilde{c})$
commute with $\conj$, 
since they are block diagonal relative to \eqref{eq:RGdecomp},
by Lemma \ref{lem:aALproperties}. 
Thus we obtain the following equation for $\tilde{c}$:
\begin{equation}\label{eq:SHSforctilde}
(\AmatLin^\mu+\AmatBil^\mu(\tilde{c})) X_\mu\tilde{c} \;=\; \LMat \tilde{c} + \BSHS(\tilde{c},\tilde{c}) + \Fvec
\end{equation}
where we define  
\begin{align}\label{eq:termsofconjSHS}
\begin{aligned}
\AmatLin^\mu
&= 
\amink_1^{\mu}
+ \Amink_1^{\mu}(v)\\
\AmatBil^\mu(\cdot)
&=
\Amink_1^{\mu}(\ginj_1 \,\cdot\,)\\
\LMat
&=
\conj^{-1}\Lmink_1  \conj 
- \conj^{-1}\Amink_{11}^\mu(\ginj_1  \,\cdot\,) X_{\mu} v
+ \conj^{-1}\Bmink_1(v,\ginj_1\conj \,\cdot\,)\\
\BSHS
&=
\tfrac12\conj^{-1}\Bmink_1(\ginj_1\conj \,\cdot\,,\ginj_1\conj \,\cdot\,)\\
\Fvec
&=
\conj^{-1}\SFmink V
\end{aligned}
\end{align}
We will abbreviate $\AmatLin=\AmatLin^\mu X_\mu$ and $\amink_1=\amink_1^\mu X_\mu$.

We apply Theorem \ref{thm:nonlinEE} using Convention \ref{conv:ztXmu_NEW}, 
and with the parameters in Table \ref{tab:Application1ofAbstractTheorem}.
The equality of norms in Remark \ref{rem:normsequality} will be used without further notice.
Let $\Capply{\Clarge}$, $\epsapply\Csmall$ be the constants produced by 
Theorem \ref{thm:nonlinEE}  (called $\Clarge,\Csmall$ there).
They depend only on \eqref{eq:spconst},
in particular $\ClargeGR$ and $\CsmallGR$ are allowed to depend on $\Capply{\Clarge}$ and  $\epsapply\Csmall$.

\begin{table}
\centering
\begin{tabular}{rc|c}
	&
	\begin{tabular}{@{}c@{}}
	Parameters \\
	in Theorem \ref{thm:nonlinEE} 
	\end{tabular}
	& 
	\begin{tabular}{@{}c@{}}
	Parameters\\ 
	used to invoke Theorem \ref{thm:nonlinEE} 
	\end{tabular}
	 \\
\hline
Input
	&
	$\Mcpt, X_0,\dots,X_{\MdimNEW}, \muM$
	&
	see Convention \ref{conv:ztXmu_NEW}\\
	&$\nn$, $\NN$ 
	& $\ngG_1$, $\NN$ \\
	&$\Cpos$
	& $\qmink$ 
	\\
	&$\CLA$ 
	&
	$
	\max\{2, C_{\NN,\epspower}(1+\CinGR),\new{C_0}\CinGR \}$
	\\
	&$\delta$ 
	& $\deltamink$ 
	\\
	&$\AmatLin^\mu$, $\AmatBil^\mu$, $\LMat$, $\BSHS$, $\Fvec$
	& $\AmatLin^\mu$, $\AmatBil^\mu$, $\LMat$, $\BSHS$, $\Fvec$ in \eqref{eq:termsofconjSHS}\\
	&$\Ccomp$
	& $\Ccomp=(\frac{y^1}{|\vec{y}|},\frac{y^2}{|\vec{y}|},\frac{y^3}{|\vec{y}|})$ \\
	&$k$, $\CHigherIn$ (Part 2 only) 
	& $k$, 
	$\max\{
	C'_{\kk,\epspower}(1+\bb),
	C'_{\kk,\epspower,\CinGR}\bb
	\}
$
\\
\hline
%
Output
	& $\Clarge$, $\Csmall$ 
	& $\Capply{\Clarge}$, 
	$\epsapply{\Csmall}$
\end{tabular}
\captionsetup{width=119.6mm}
\caption{%
The first column lists the input and output parameters of Theorem \ref{thm:nonlinEE}. 
The second column specifies the choice of the input parameters used to invoke
Theorem \ref{thm:nonlinEE},
in terms of the input parameters of Proposition \ref{prop:ApplySHS}
and the parameters introduced in this proof.
The output parameters produced by this invocation of
Theorem \ref{thm:nonlinEE}
are denoted $\Capply{\Clarge}$, $\epsapply{\Csmall}$,
where $\Capply{\Clarge}$ depends only on the parameters in the first four rows,
and $\epsapply{\Csmall}$ only on those in the first five rows.}
\label{tab:Application1ofAbstractTheorem}
\end{table}

We check that the assumptions of Theorem \ref{thm:nonlinEE} are satisfied.
As required $\NN\ge 6$;
\eqref{eq:termsofconjSHS} are smooth on $\Dspcl_{\le1}$
since the maps in Definition \ref{def:SHSarrays_i0} and $v$ are smooth on $\Dspcl_{\le1}$;
$\Ccomp$ is smooth on $\Dspcl_{\le1}$;
the matrices $\AmatLin^\mu$, and $\AmatBil^\mu(w)$ for every $w$, 
are symmetric by Lemma \ref{lem:aALproperties}.
We check 
\ref{item:Apos}, 
\ref{item:CNnormsSHS}, 
\ref{item:KappaTransv},
\ref{item:Fsmall},
\ref{item:InhSmall}
in this order.

\ref{item:Apos}: 
We make the admissible smallness assumption on $\CsmallGR$ that:
\begin{align}\label{eq:epssmallcausal}
\CsmallGR\le \deltamink
\end{align}
Then \ref{item:Apos} follows from \ref{item:vsmall}
and $\sqrt{v^Tv}\le \|v\|_{\nosCb^0(\Dspop)}$,
and Lemma \ref{lem:qminkdelta} with $k=1$.

\ref{item:CNnormsSHS}: 
The components of the maps in Definition \ref{def:SHSarrays_i0} 
are smooth on $\Dspcl$ and homogeneous of degree zero by Lemma \ref{lem:comphomog}.
Thus there exists a constant $C_{\NN,\epspower}>0$ 
that depends only on $\NN,\epspower$, such that 
\begin{align}
&\max\big\{
\|\AmatLin^\mu\|_{\nosCb^{\NN}(\Dspop)},
\|\AmatBil^\mu\|_{\nosCb^{\NN}(\Dspop)},
\|\LMat\|_{\nosCb^{\NN}(\Dspop)},
\|\BSHS\|_{\nosCb^{\NN}(\Dspop)}
 \big\}\nonumber \\
&\quad\le
C_{\NN,\epspower} (1 + \|v\|_{\nosCb^{\NN+1}(\Dspop)})\nonumber\\
&\quad\le
C_{\NN,\epspower}(1+\CinGR) \label{eq:matcalc}
\end{align}
where the last step holds by  \ref{item:CN+1v}.
Further $\|\Ccomp\|_{\nosCb^0(\Dspop)} \le 2$,
see Table \ref{tab:Application1ofAbstractTheorem}.
Thus \ref{item:CNnormsSHS} holds, 
using the fourth row in Table \ref{tab:Application1ofAbstractTheorem},
and the equality of norms in Remark \ref{rem:normsequality}.

\ref{item:KappaTransv}:
By \eqref{eq:CcomT0}, for all $\zz\le0$:
\begin{align*}
|\CcomTransv{\AmatLin,\Ccomp}(\zz)|
=
|\CcomTransv{\AmatLin,\Ccomp}(\zz)
-\CcomTransv{\amink_1,\Ccomp}(\zz)|
\end{align*}
We use Lemma \ref{lem:ellkappaLip}
with $\a,\tilde\a,\Ccomp,\Cpos$ there 
given by $\amink_1,\AmatLin,\Ccomp,\qmink$ here,
where \eqref{eq:posasspaatilde} is satisfied by 
\eqref{eq:Azetapos} with $w=0$ and \eqref{eq:dzetapos} with $u=0$.
With $\AmatLin-\amink_1=\Amink_1^{\mu}(v)X_\mu$,
and using $\|\Ccomp\|_{\nosCb^0(\Dspop)} \le 2$,
$\|\amink_1^\mu\|_{\nosCb^1(\Dspop)}\lesssim1$, $\|\Amink_1^{\mu}\|_{\nosCb^1(\Dspop)}\lesssim1$,
we obtain that there exists a constant $\new{C_0}>0$
(not depending on any parameters), such that 
\begin{align}\label{eq:ctdiffmink}
|\CcomTransv{\AmatLin,\Ccomp}(\zz)
-\CcomTransv{\amink_1,\Ccomp}(\zz)|
\le
\new{C_0}\|v\|_{\sCb^1(\Dspop_{e^{\zz}})}
\end{align}
Thus
\begin{align}
\tint_{-\infty}^{0}|\CcomTransv{\AmatLin,\Ccomp}(\zz)| d\zz 
\le
\new{C_0}
\tint_{-\infty}^{0}\|v\|_{\sCb^1(\Dspop_{e^{\zz}})}d\zz
=
\new{C_0}\tint_{0}^{1}\|v\|_{\sCb^1(\Dspop_{s})}\tfrac{ds}{s}
\le 
\new{C_0}\CLA
\label{eq:CLA2}
\end{align}
where we substitute $s=e^{\zz}$ and use \ref{item:vint} in the last step.
Thus \ref{item:KappaTransv} holds, using the 
fourth row in Table \ref{tab:Application1ofAbstractTheorem}.

To check \ref{item:Fsmall} and \ref{item:InhSmall}
we need some preliminaries, which will also be useful for Part 2.
We claim that for all $\kk\in\Z_{\ge0}$ and all $\zz_0\le\zz_1\le0$ and $\zz\le0$:
\begin{subequations}\label{eq:FestimatesMink}
\begin{align}
\kprop_{\kk}(\zz_1,\zz_0) &\;\lesssim_{\kk,\epspower,\CinGR}\;
	e^{(\zz_1-\zz_0)(\frac52+\epspower+\kk)}
	\label{eq:propest} \\ 
\Fconst_{\kk}(\zz) &\;\lesssim_{\kk,\epspower,\CinGR}\; 
	e^{\zz(\frac52+\epspower+\kk)}\Vconst_{\kk}(v)
	\label{eq:Fconstest} \\
\tint_{-\infty}^0 \|\Fvec\|_{\sHb^{k}(\Dspop_{e^{\zz'}})} d\zz'
	&\;\lesssim_{\kk,\epspower,\CinGR}\;
	\Vconst_{\kk}(v)
	\label{eq:intFvec}\\
\|\Fvec\|_{\sHb^{k}(\Dspop_{e^{\zz}})}
	&\;\lesssim_{\kk,\epspower,\CinGR}\;
	e^{\zz(\frac52+\epspower+(\kk+1))}\Vconst_{\kk+1}(v)
	\label{eq:FvecToInt}
\end{align}
\end{subequations}
where $\kprop_{\kk}(\zz_1,\zz_0)$, $\Fconst_{\kk}(\zz)$
are defined in \eqref{eq:kpropdef}, \eqref{eq:Fconstdef}
in Theorem \ref{thm:nonlinEE}, 
using Table \ref{tab:Application1ofAbstractTheorem}.

Proof of \eqref{eq:propest}:
Adding and subtracting yields 
\begin{align*}
\kprop_{\kk}(\zz_1,\zz_0) 
&=
\textstyle
\exp\Big(\int_{\zz_0}^{\zz_1}
(\LdivestNEW{\conj^{-1}\Lmink_1 \conj,\amink_1}(\zz') + \kk\max\{0,\CcomTang{\amink_1,\Ccomp}(\zz')\})d\zz' \Big)\\
&\quad
\textstyle
\times \exp\Big(\int_{\zz_0}^{\zz_1} 
	(
	\LdivestNEW{\LMat,\AmatLin}(\zz')
	-
	\LdivestNEW{\conj^{-1}\Lmink_1 \conj,\amink_1}(\zz')
	)
	d\zz'\Big)\\
&\quad\textstyle
\times \exp\Big(\int_{\zz_0}^{\zz_1} 
	\kk
	(\max\{0,\CcomTang{\AmatLin,\Ccomp}(\zz')\}
	-
	\max\{0,\CcomTang{\amink_1,\Ccomp}(\zz')\})
	\,d\zz'\Big)
\end{align*}
By Lemma \ref{lem:ellMinkest} 
and the choice of $\chi$ in \eqref{eq:chi},
for all $\zz\le0$:
\begin{align*}
\LdivestNEW{\conj^{-1}\Lmink_1 \conj,\amink_1}(\zz) 
&\le \tfrac52+\epspower & 
\max\{0,\CcomTang{\amink_1,\Ccomp}(\zz)\} &= 1
\end{align*}
By Lemma \ref{lem:ellkappaLip} (using the positivity 
\eqref{eq:Azetapos} and \eqref{eq:dzetapos}), for all $\zz\le0$:
\begin{align*}
|\LdivestNEW{\LMat,\AmatLin}(\zz)
-
\LdivestNEW{\conj^{-1}\Lmink_1 \conj,\amink_1}(\zz)|
&\lesssim_{\epspower}
\|v\|_{\sCb^1(\Dspop_{e^{\zz}})}\\
|\max\{0,\CcomTang{\AmatLin,\Ccomp}(\zz)\}-\max\{0,\CcomTang{\amink_1,\Ccomp}(\zz)\}|
&\lesssim
\|v\|_{\sCb^1(\Dspop_{e^{\zz}})}
\end{align*}
similarly to \eqref{eq:ctdiffmink}.
Therefore, using $\zz_0\le\zz_1$, 
\begin{align*}
\kprop_{\kk}(\zz_1,\zz_0)
\le&\textstyle
	\exp\Big(
	\int_{\zz_0}^{\zz_1} (\tfrac{5}{2}+\epspower + \kk)d\zz'\Big)
	\exp\Big(
	C_{\kk,\epspower}\int_{\zz_0}^{\zz_1}
	\|v\|_{\sCb^1(\Dspop_{e^{\zz'}})}
	d\zz'\Big)
\end{align*}
for a constant $C_{\kk,\epspower}>0$ that depends only on $\kk,\epspower$.
Analogously to \eqref{eq:CLA2} one obtains
that the integral over $v$ is bounded by $\CinGR$
(use \ref{item:vint}), thus \eqref{eq:propest} follows.

Proof of \eqref{eq:Fconstest}:
Using \eqref{eq:propest},
\begin{align*}
\Fconst_{\kk}(\zz)
&=
\tint_{-\infty}^{\zz}\kprop_{\kk}(\zz,\zz') 
(1+|\zz-\zz'|)^{\kk}
\|\Fvec\|_{\sHb^{\kk}(\Dspop_{e^{\zz'}})}
\,d\zz'
\\
&\lesssim_{\kk,\epspower,\CinGR}
e^{\zz(\frac52+\epspower+\kk)}
\tint_{-\infty}^{\zz}
e^{-\zz'(\frac52+\epspower+\kk)}
(1+|\zz-\zz'|)^{\kk}
\|\Fvec\|_{\sHb^{\kk}(\Dspop_{e^{\zz'}})} \,d\zz'\\
&\le
e^{\zz(\frac52+\epspower+\kk)}
\tint_{-\infty}^{0}
e^{-\zz'(\frac52+\epspower+\kk)}
(1+|\zz'|)^{\kk}
\|\Fvec\|_{\sHb^{\kk}(\Dspop_{e^{\zz'}})} \,d\zz'\\
&=
e^{\zz(\frac52+\epspower+\kk)}
\tint_{0}^{1}
s^{-(\frac52+\epspower+\kk)} (1+|\log s|)^{\kk}
\|\Fvec\|_{\sHb^{\kk}(\Dspop_{s})} \,\tfrac{ds}{s}
\end{align*}
where in the last step we substitute $s=e^{\zz'}$.
Using $\Fvec=\conj^{-1}\SFmink V$ (see \eqref{eq:termsofconjSHS}),
and the fact that the components of $\SFmink$
are constant (see Lemma \ref{lem:Binbases}), 
\begin{align*}
\|\Fvec\|_{\sHb^{\kk}(\Dspop_{s})}
&
\lesssim_{\epspower}
\|V\|_{\sHb^{\kk}(\Dspop_{s})}
=
\|\dg v + \tfrac12[v,v]\|_{\sHb^{\kk}(\Dspop_{s})}
\end{align*}
Thus \eqref{eq:Fconstest} follows. 

Proof of \eqref{eq:intFvec}: We have
\begin{align*}
\tint_{-\infty}^0 \|\Fvec\|_{\sHb^{k}(\Dspop_{e^{\zz'}})} d\zz'
=
\tint_{0}^1 \|\Fvec\|_{\sHb^{k}(\Dspop_{s})} \frac{ds}{s}
\lesssim_{\epspower}
\tint_{0}^1 \|\dg v + \tfrac12[v,v]\|_{\sHb^{k}(\Dspop_{s})} \frac{ds}{s}
\le
\Vconst_{\kk}(v) 
\end{align*}
where the last inequality uses 
$1\le
(\tfrac{1}{s})^{\frac{5}{2}+\epspower+\kk} 
(1+|\log(\tfrac{1}{s})|)^{\kk}$.

Proof of \eqref{eq:FvecToInt}:
Set $p=\frac52+\epspower+(\kk+1)$.
Note that $p>0$.
By Lemma \ref{lem:LinfL2toL1L2},
\begin{align*}
\|\Fvec\|_{\sHb^{k}(\Dspop_{e^{\zz}})}
&\lesssim_{k}
\tint_{\zz-1}^{\zz} 
\|\Fvec\|_{\sHb^{k+1}(\Dspop_{e^{\zz'}})}
d\zz'
\intertext{
For all $\zz'\in[\zz-1,\zz]$ we have 
$1\le e^{(\zz-\zz')p}(1+|\zz'|)^{\kk+1}$, hence}
\|\Fvec\|_{\sHb^{k}(\Dspop_{e^{\zz}})}
&\lesssim_{k}
\tint_{\zz-1}^{\zz} 
e^{(\zz-\zz') p}
(1+|\zz'|)^{\kk+1}
\|\Fvec\|_{\sHb^{k+1}(\Dspop_{e^{\zz'}})}
d\zz'\\
&\le
e^{\zz p}
\tint_{-\infty}^{0} 
e^{-\zz' p}
(1+|\zz'|)^{\kk+1}
\|\Fvec\|_{\sHb^{k+1}(\Dspop_{e^{\zz'}})}
d\zz'\\
&\lesssim_{\epspower}
e^{\zz p}\Vconst_{\kk+1}(v)
\end{align*}
see the proof of \eqref{eq:Fconstest} for the last step.
This proves \eqref{eq:FvecToInt}.

\ref{item:Fsmall}:
By \eqref{eq:Fconstest}, for all $\zz\le0$,
\begin{align*}
\Fconst_{\NN}(\zz)
	&\lesssim_{\NN,\epspower,\CinGR}
	e^{\zz(\frac52+\epspower+\NN)}\Vconst_{\NN}(v)
	\le
	\CsmallGR 
\intertext{
where we use $e^{\zz(\frac52+\epspower+\NN)}\le1$ and \ref{item:MC(v)small}.
Also by \eqref{eq:Fconstest},}	
\tint_{-\infty}^{0}\Fconst_{\NN}(\zz) d\zz 
	&\lesssim_{\NN,\epspower,\CinGR} 
	\Vconst_{\NN}(v) \tint_{-\infty}^{0}e^{\zz(\frac52+\epspower+\NN)} d\zz
	\le
	\CsmallGR
\end{align*}
using $\NN,\epspower\ge0$ and \ref{item:MC(v)small}.
Thus an admissible smallness assumption on $\CsmallGR$ yields
$\sup_{\zz\in(-\infty,0]}\Fconst_{\NN}(\zz)  \le \epsapply{\Csmall}$
and 
$\tint_{-\infty}^{0}\Fconst_{\NN}(\zz) d\zz \le \epsapply{\Csmall}$,
which proves \ref{item:Fsmall}.

\ref{item:InhSmall}:
For all $\zz\le0$ we have,
using \eqref{eq:FvecToInt},
$e^{\zz(\frac52+\epspower+\NN)}\le1$ and \ref{item:MC(v)small},
\begin{align*}
\|\Fvec\|_{\sHb^{\NN-1}(\Dspop_{e^{\zz}})}
	&\lesssim_{\NN,\epspower,\CinGR}
	e^{\zz(\frac52+\epspower+\NN)}\Vconst_{\NN}(v)
	\le
	\CsmallGR
\intertext{	
Thus the first inequality in \ref{item:InhSmall}
holds under an admissible smallness assumption on $\CsmallGR$.
For the second, note that by
\eqref{eq:intFvec} and then \ref{item:MC(v)small},}
\tint_{-\infty}^0 \|\Fvec\|_{\sHb^{\NN-1}(\Dspop_{e^{\zz'}})} d\zz'
	&\;\lesssim_{\NN,\epspower,\CinGR}\;
	\Vconst_{\NN-1}(v)
	\le
	\Vconst_{\NN}(v)
	\le
	\CsmallGR
	\le
	1
\end{align*}
Thus the second inequality in \ref{item:InhSmall} holds 
using the fourth row in Table \ref{tab:Application1ofAbstractTheorem}.

We have checked the assumptions
\ref{item:Fsmall},
\ref{item:InhSmall}, 
\ref{item:KappaTransv},
\ref{item:Apos}, 
\ref{item:CNnormsSHS}
of Theorem \ref{thm:nonlinEE}.

\proofheader{Proof of Part 0.}
Suppose that $c_1,c_2\in\gxG^1(\Dspop)$ satisfy \eqref{eq:existence_v+c}.
Then they satisfy \eqref{eq:MCnecessary},
and then $\tilde{c}_1=\conj^{-1}c_1$ and $\tilde{c}_2=\conj^{-1}c_2$ satisfy 
\eqref{eq:SHSforctilde}. Then Part 0 of Theorem \ref{thm:nonlinEE}
implies $\tilde{c}_1=\tilde{c}_2$, hence $c_1=c_2$.

\proofheader{Proof of existence and Part 1, 
for \eqref{eq:MCnecessary} instead of \eqref{eq:v+cMC}.}
By Theorem \ref{thm:nonlinEE} (existence, Part 0, Part 1), there exists a unique
\[ 
\tilde{c} \;\in\; C^\infty(\Dspop,\R^{\ngG_1})
\qquad 
(\text{called $u$ in Theorem \ref{thm:nonlinEE}})
\]
that satisfies \eqref{eq:SHSforctilde}, $\tilde{c}|_{y^0=0}=0$,
$\sqrt{\tilde{c}^T\tilde{c}} \le \deltamink$ on $\Dspop$, and 
such that for all $\zz\le0$:
\begin{align}\label{eq:apP1t5}
\|\tilde{c}\|_{\sHb^{\NN}(\Dspop_{e^{\zz}})}
	&\le
	\Capply{\Clarge} 
	(\Fconst_{\NN}(\zz) + \|\Fvec\|_{\sHb^{\NN-1}(\Dspop_{e^{\zz}})})
	\lesssim_{\NN,\epspower,\CinGR}
	e^{\zz(\frac52+\epspower+\NN)}\Vconst_{\NN}(v)
\end{align}
More precisely, the first inequality in \eqref{eq:apP1t5} holds by 
\eqref{eq:uEAllDer} (we do not use \eqref{eq:uETangDer} here),
and for the second inequality in \eqref{eq:apP1t5} we use 
\eqref{eq:Fconstest}, \eqref{eq:FvecToInt},
and the fact that $\Capply{\Clarge}$ depends only
on $\NN,\epspower,\CinGR$.
Set 
$$c=\conj\tilde{c}$$ 
as in \eqref{eq:defctilde}. 
Then $c$ solves \eqref{eq:cSHS}, $c|_{y^0=0}= 0$ and,
using $\chi\in (0,1]$, 
\begin{align}
\sqrt{c^Tc} 
	&\le
	\sqrt{\tilde c^T\tilde c} 
	\le \deltamink 
	&&
	\text{on $\Dspop$}
	\label{eq:csmallGR}\\
\|c\|_{\sHb^{\NN}(\Dspop_{e^{\zz}})}
	&\le
	\|\tilde{c}\|_{\sHb^{\NN}(\Dspop_{e^{\zz}})}
	\lesssim_{\NN,\epspower,\CinGR}
	e^{\zz(\frac52+\epspower+\NN)}\Vconst_{\NN}(v)	
	&&\text{for $\zz\le0$}
	\nonumber
\end{align}
By choosing $\ClargeGR$ sufficiently large,
depending only on $\NN,\epspower,\CinGR$, 
and replacing $z=\log(s)$, we obtain that for all $s\in(0,1]$:
\[ 
\|c\|_{\sHb^{\NN}(\Dspop_{s})}
\le 
\ClargeGR s^{(\frac52+\epspower+\NN)} \Vconst_{\NN}(v)
\]
Viewing $c$ as an element in $\gxG^1(\Dspop)$ via \eqref{eq:identify_g_vec_i0},
it satisfies \eqref{eq:MCnecessary}, \eqref{eq:cdata}, \eqref{eq:cestimates}.

\proofheader{Proof of existence and Part 1.}
It remains to show that $c$ solves \eqref{eq:v+cMC}, 
i.e.~that the constraints propagate.
Define
\begin{align}\label{eq:Udef}
U = \dg(v+c) + \tfrac12[v+c,v+c]\;\in\;\gx^2(\Dspop)
\end{align}
Our goal is to show $U=0$.
We claim that 
\begin{subequations}\label{eq:Uprop}
\begin{align}
U \in \gxG^2(\Dspop) \label{eq:Ugauged}\\
\dg U + [v+c,U]=0 \label{eq:Ueq}\\
U|_{\ttcoord=0}=0 \label{eq:Udata}
\end{align}
\end{subequations}

Proof of \eqref{eq:Ugauged}: By \eqref{eq:MCnecessary} and \ref{item:gaugekernelNEW}.

Proof of \eqref{eq:Ueq}: 
Abbreviate $u=v+c$.
Then, expanding the definition of $U$ and using linearity of 
the differential and bilinearity of the bracket,
\begin{align*}
\dg U + [u,U]
&=
\dg\left( \dg u + \tfrac12[u,u] \right)
+
[u,\dg u + \tfrac12[u,u]]\\
&=
\dg^2 u + \tfrac12\dg[u,u] 
+
[u,\dg u]
+
\tfrac12[u,[u,u]]\\
&=
\tfrac12\dg[u,u] 
+
[u,\dg u]
\intertext{
where in the last step we use \eqref{eq:dgdifferential}
and the graded Jacobi identity \eqref{eq:bracketjacobi}.
Now the Leibniz rule \eqref{eq:dgbracket} 
and graded antisymmetry of the bracket \eqref{eq:bracketas} yield
}
\dg U + [u,U]
&=
\tfrac12[\dg u,u]
-
\tfrac12[u,\dg u]
+
[u,\dg u]
=0
\end{align*}

Proof of \eqref{eq:Udata}: 
Let $p\in \Dspop\cap\diamonddata$.
By \ref{item:vconstraints}, \eqref{eq:cdata}
and Lemma \ref{lem:pext},
\begin{equation}\label{eq:Uconstr}
0
= 
\big((d\ttcoord + \anchorg(v)(\ttcoord)) U\big)|_{p}
\end{equation}
We claim that 
\begin{equation}\label{eq:dtrho}
(d\ttcoord + \anchorg(v)(\ttcoord))|_p\in\Omegafut|_p
\end{equation}
with $\Omegafut|_p$ the fiber of \eqref{eq:Omegafut} at $p$.
Proof of \eqref{eq:dtrho}: 
By Definition \ref{def:SHSarrays_i0} and \eqref{eq:anchorghom},
\begin{equation}\label{eq:abbb}
d\ttcoord(\AmatLin^\mu_{ij} X_\mu)|_p
=
\gBil^1(\eGg^1_{i}, (d\ttcoord + \anchorg(v)(\ttcoord))\eGg^1_{j})|_{p}
\end{equation}
with $\AmatLin^\mu_{ij}$ the components of the matrix
$\AmatLin^\mu$ in \eqref{eq:termsofconjSHS}.
By \eqref{eq:Atpos} with $w=0$, the matrix \eqref{eq:abbb} is positive definite,
which implies \eqref{eq:dtrho} by \ref{item:gaugeposNEW}.
We can now conclude $U|_{p}=0$: 
This follows from 
\eqref{eq:Uconstr}, \eqref{eq:dtrho}, \eqref{eq:Ugauged}
and fiberwise injectivity of 
left-multiplication in Lemma \ref{lem:gauge_mainprop_i0}.
Thus \eqref{eq:Udata} holds.

We conclude $U=0$.
By \eqref{eq:Ueq} we have 
$\gBil^{2}(\,\cdot\,, \dg U + [v+c,U])=0$.
By \eqref{eq:Ugauged} and Lemma \ref{lem:translationofeq},
this is equivalent to
\begin{equation}\label{eq:Ueqmat}
\slashed{\AmatLin}^{\mu} X_\mu U
=
\slashed{\LMat}U
\end{equation}
where we use the identification \eqref{eq:identify_g_vec_i0}, 
and where we define
\begin{align}\label{eq:ConstraintMat}
\begin{aligned}
\slashed{\AmatLin}^{\mu} &= \amink_2^{\mu}+ \Amink_2^{\mu}(v+\ginj_1c)\\
\slashed{\LMat} &= \Lmink_2  \new{+} \Amink_{21}^\mu(\ginj_2\,\cdot\,) X_{\mu} (v+\ginj_1c) + \Bmink_2(v+\ginj_1c,\ginj_2\,\cdot\,) 
\end{aligned}
\end{align}
To show $U=0$ it suffices to show that 
\smash{$U|_{\cone^{\qmink}_{0,\tt_0}}=0$} for every 
$\tt_0\in (0,1)$, with 
$\cone^{\qmink}_{0,\tt_0}$ defined in \eqref{eq:eqcone}.
For this we apply Theorem \ref{thm:AbstractUniqueness} 
using Convention \ref{conv:ztXmu_NEW}, and 
with the parameters in Table \ref{tab:Application2ofAbstractTheorem}.
\begin{table} 
\centering
\begin{tabular}{c|c}
	\begin{tabular}{@{}c@{}}
	Parameters \\
	in Theorem \ref{thm:AbstractUniqueness}
	\end{tabular}
	& 
	\begin{tabular}{@{}c@{}}
	Parameters \\
	used to invoke Theorem \ref{thm:AbstractUniqueness}
	\end{tabular} \\
\hline
	$\Mcpt, X_0,\dots,X_{\MdimNEW}, \muM$
	&
	see Convention \ref{conv:ztXmu_NEW}\\
$\nn$, $\Cpos$ 
	& $\ngG_2$, $\qmink$ \\
$(\zz_0,\tt_0)$ & $(0,\tt_0)$\\
$\AmatLin^\mu$, $\AmatBil^\mu$, $\LMat$, $\BSHS$, $\Fvec$
	& $\slashed{\AmatLin}^{\mu}$, $0$, $\slashed{\LMat}$, $0$, $0$
	in \eqref{eq:ConstraintMat}\\
$u_1$, $u_2$ 
	& $0$, $U$ in \eqref{eq:Udef} 
\end{tabular}
\captionsetup{width=115mm}
\caption{
The first column lists the input parameters of Theorem \ref{thm:AbstractUniqueness}. 
The second column specifies the choice of the input parameters used to invoke
Theorem \ref{thm:AbstractUniqueness}.}
\label{tab:Application2ofAbstractTheorem}
\end{table}

We check that the assumptions of Theorem \ref{thm:AbstractUniqueness} are satisfied:
Clearly $\slashed{\AmatLin}^\mu$, $\slashed{\LMat}$ are smooth on $\cone^{\qmink}_{0,\tt_0}\subset\Dspop$,
and $\slashed\AmatLin^\mu$ are symmetric by Lemma \ref{lem:aALproperties}.
\eqref{eq:ueqUniq}: 
The $\ell=1$ case is clear;
the $\ell=2$ case holds by \eqref{eq:Ueqmat}.
\eqref{eq:udataUniq}:
The $\ell=1$ case is clear, the $\ell=2$ case holds by \eqref{eq:Udata}.
\eqref{eq:AzetaposUniq}, \eqref{eq:AtposUniq}:
By Lemma \ref{lem:qminkdelta} with $k=2$
and $\|v+\ginj_1c\| \le \|v\|+\|c\| \le 2\deltamink$,
using \eqref{eq:epssmallcausal} and \ref{item:vsmall} for $v$,
and \eqref{eq:csmallGR} for $c$.

We have checked that the assumptions of Theorem \ref{thm:AbstractUniqueness} hold, thus \smash{$U|_{\cone^{\qmink}_{0,\tt_0}}=0$.}
Therefore $U=0$, equivalently \eqref{eq:v+cMC} holds.

\proofheader{Proof of Part 2.}
Let $\kk\in\Z_{\ge\NN}$ and $\bb>0$
and assume that \ref{item:GRHigherassp12}, \ref{item:GRHigherassp3} hold.
We check that assumptions 
\ref{item:Higherassp12},
\ref{item:InhSmallHigher},
\ref{item:Higherassp3} 
of Part 2 in Theorem \ref{thm:nonlinEE} hold with the parameters 
in Table \ref{tab:Application1ofAbstractTheorem}.
Using \eqref{eq:FestimatesMink} and \ref{item:GRHigherassp12},
one obtains that there exists a constant 
$C'_{\kk,\epspower,\CinGR}>0$ that depends only on $\kk,\epspower,\CinGR$,
such that for all $\zz\le0$:
\begin{align*}
\Fconst_{\kk}(0) 
	&\le C'_{\kk,\epspower,\CinGR} \Vconst_{\kk}(v) < \infty\\
\Fconst_{\kk-1}(\zz) 
	&\le C'_{\kk,\epspower,\CinGR} \Vconst_{\kk-1}(v)
	\le C'_{\kk,\epspower,\CinGR} \bb\\
\tint_{-\infty}^0\Fconst_{\kk-1}(\zz')d\zz'
	&\le
	C'_{\kk,\epspower,\CinGR} \Vconst_{\kk-1}(v) 
	\le 
	C'_{\kk,\epspower,\CinGR} \bb\\
\|\Fvec\|_{\sHb^{k-2}(\Dspop_{e^{\zz}})}
	&\le
	C'_{\kk,\epspower,\CinGR}\Vconst_{\kk-1}(v) 
	\le 
	C'_{\kk,\epspower,\CinGR} \bb\\
\tint_{-\infty}^0 \|\Fvec\|_{\sHb^{k-2}(\Dspop_{e^{\zz'}})} d\zz'
	&\le
	C'_{\kk,\epspower,\CinGR} \Vconst_{\kk-2}(v) 
	\le C'_{\kk,\epspower,\CinGR} \Vconst_{\kk-1}(v) 
	\le C'_{\kk,\epspower,\CinGR} \bb
\end{align*}
This implies \ref{item:Higherassp12}, \ref{item:InhSmallHigher},
using the second last row in Table \ref{tab:Application1ofAbstractTheorem}.
Analogously to \eqref{eq:matcalc}, and using \ref{item:GRHigherassp3},
one checks that there exists a constant 
$C'_{\kk,\epspower}>0$ that depends only on $\kk,\epspower$,
such that 
\begin{align*}
&\max\big\{
\|\AmatLin^\mu\|_{\nosCb^{\kk}(\Dspop)},
\|\AmatBil^\mu\|_{\nosCb^{\kk}(\Dspop)},
\|\LMat\|_{\nosCb^{\kk}(\Dspop)},
\|\BSHS\|_{\nosCb^{\kk}(\Dspop)} \big\}
	\le
	C'_{\kk,\epspower}(1+\bb)
\end{align*}
This implies \ref{item:Higherassp3},
using the second last row in Table \ref{tab:Application1ofAbstractTheorem}.

We have checked that the assumptions of Part 2 in Theorem \ref{thm:nonlinEE} hold.
Hence \eqref{eq:EkkestX} holds (we do not use \eqref{eq:Ekkest} here), 
that is, for all $\zz\le0$:
\begin{align*}
\|\tilde c\|_{\sHb^{\kk}(\Dspop_{e^{\zz}})}
\lesssim_{\kk,\epspower,\CinGR,\bb}
\Fconst_{\kk}(\zz) + \|\Fvec\|_{\sHb^{\kk-1}(\Dspop_{e^{\zz}})}
\lesssim_{\kk,\epspower,\CinGR}\; 
e^{\zz(\frac52+\epspower+\kk)}\Vconst_{\kk}(v)
\end{align*}
where for the second inequality we use 
\eqref{eq:Fconstest}, \eqref{eq:FvecToInt}. 
Using 
$\|c\|_{\sHb^{\kk}(\Dspop_{e^{\zz}})}
\le \|\tilde c\|_{\sHb^{\kk}(\Dspop_{e^{\zz}})}$,
and replacing $\zz=\log(s)$ with $s\in(0,1]$, one obtains \eqref{eq:GRHigherconcl}.
\qed
\end{proof}

\section{Construction away from spacelike infinity}
\label{sec:bulk}

Assume that $v$ is an element in $\gx^1(\diamond_+)$ that solves
the Einstein equations \eqref{eq:MC} near spacelike infinity.
The main result of this section is Proposition \ref{prop:MainCpt},
where we prove existence of an element $c$, such that 
$v+c$ is a solution of \eqref{eq:MC} globally on $\diamond_+$.
Upon gauge fixing, the equation for $c$ is 
quasilinear symmetric hyperbolic including along null and timelike infinity,
and by finite speed of propagation the solution $c$ vanishes near
spacelike infinity.
Thus this is a problem on a compact domain, 
simpler and more routine
than the problem considered in Section \ref{sec:SpaceinfConstruction}.

Proposition \ref{prop:MainCpt} will be used in the proof of Theorem \ref{thm:main},
to construct $u$ away from $\spaceinf$
(after having constructed the solution near $\spaceinf$
using Proposition \ref{prop:ApplySHS}).

Section \ref{sec:bulk} is organized as follows.
In Section \ref{sec:Foliation} we define auxiliary subsets of $\diamond_+$,
useful for energy estimates;
in Section \ref{sec:basis_bulk}, \ref{sec:norms_bulk} 
we fix bases and norms;
in Section \ref{sec:Gauge_cyl} we fix a gauge and show
that the relevant equations are symmetric hyperbolic in this gauge;
in Section \ref{sec:existence_bulk} we state and prove
Proposition \ref{prop:MainCpt}.
\begin{remark}\label{rem:local5}
Some definitions in 
this section are labeled 
'local to Section  \ref{sec:bulk}'
by which we mean that they are only valid in Section \ref{sec:bulk}.
See also Remark \ref{rem:local4}.
\end{remark}

\subsection{Spacelike exhaustion of $\diamond_+$}
\label{sec:Foliation}

As preparation for the energy estimates, we define several subsets of $\diamond_+$. 
In particular we define an exhaustion of $\diamond_+$, given 
by subsets whose boundary is spacelike for the
conformal background metric $[\gcyl]$, see Lemma \ref{lem:Dexhaustion}.

\begin{definition}\label{def:diamonds,stau}
For $s>0$ and $\tau\in[0,\pi)$ define 
\begin{align*}
\diamond_{s} &= \diamond_+\setminus \Dspop_{\new{<}\frac{s}{6}}\\
\diamond_{\tau,s} &= (\{\tau\}\times S^3)\cap \diamond_{s}
\end{align*}
where $\Dspop_{\new{<}\frac{s}{6}}$ is defined in \eqref{eq:deltas}.
The factor $\frac16$ is for later convenience.
\end{definition}
Recall from Section \ref{sec:geomconfcpt} that
$\cyl$ is given by all points 
$(\tau,\xi)$ with $\tau\in\R$ and $\xi\in\R^{4}$ with $|\xi|=1$.
Further $\timeinf=(\pi,(0,0,0,-1))$, $\spaceinf=(0,(0,0,0,1))$.
Define
\begin{equation}\label{eq:phidef}
\fctfol \;=\; \left(\frac{1-\cos\tau}{1-\xi^4}\right)^{\frac12}
	\;\in\; 
	C^\infty\big(\overline\diamond_+ \setminus\spaceinf\big) 
\end{equation}
This is indeed smooth including along $\tau=0$ since
$1-\cos\tau=2\sin(\frac{\tau}{2})^2$.
Define
\begin{equation}\label{eq:Phidef}
\Phi = d\fctfol(\V_0) - \sqrt{\tsum_{i=1}^3|d\fctfol(\V_i)|^2}
\end{equation}
using $\V_{0},\dots,\V_3$ in \eqref{eq:defV}.
This is smooth away from $\xi^4=-1$, and continuous along $\xi^4=-1$
(one has $d\fctfol(\V_i)=\xi^i\frac{\fctfol}{2(1-\xi^4)}$ 
which vanishes along $\xi^4=-1$).
\begin{lemma}\label{lem:dphiproperties}
On $\diamond_+$ one has
\begin{equation}\label{eq:Phi/h}
\frac{\Phi}{\nullgen} 
	= \frac{1}{\sqrt{2}(1-\xi^4)}\frac{1}{\sqrt{1-\xi^4}\cos(\frac\tau2)+\sqrt{1+\xi^4}\sin(\frac\tau2)}
	\;>\;0
\end{equation}
where $\nullgen=\cos(\tau)-\xi^4$, see \eqref{eq:nullgendef}.
In particular, $d\fctfol$ is future directed ($d\fctfol(\p_{\tau})>0$)
and timelike relative to $[\gcyl]$ on $\diamond_+$.
Furthermore, for each $\mu,\nu=0\dots3$:
\begin{equation}\label{eq:KilBounds}
|d\fctfol(B^{\mu\nu})| \le 3\Phi
\qquad
|d\fctfol(T_\mu)| \le 3\Phi
\qquad
\text{on $\diamond_+$}
\end{equation}
using the boosts and translations \eqref{eq:Kilbas}.
\end{lemma}
\begin{proof}
One obtains \eqref{eq:Phi/h} by direct calculation.
We check \eqref{eq:KilBounds}:
Abbreviate $\xivec = (\sum_{i=1}^3 (\xi^i)^2)^{\frac12}$.
By direct calculation, for $i,j=1,2,3$:
\begin{align*}
\tfrac{d\fctfol(B^{0i})}{\Phi}
	&=\textstyle
	\frac12\big(
	\xi^i(1+\cos\tau)
	+
	\frac{\xi^i}{\xivec}(1+\xi^4)\sin\tau
	\big)\\
\tfrac{d\fctfol(B^{ij})}{\Phi}
	&=0\\
\tfrac{d\fctfol(T^{0})}{\Phi}
	&=\textstyle
	\frac12\big(  
	(1+\cos\tau)(1-\xi^4)+\xivec\sin\tau
	\big)\\
\tfrac{d\fctfol(T^{i})}{\Phi}
	&=\textstyle
	-\frac12
	\big(
	\frac{\xi^i}{\xivec}(1-\cos\tau)(1+\xi^4)
	+\xi^i \sin\tau
	\big)
\end{align*}
Each term is bounded
above and below by $3$, thus the claim follows.\qed
\end{proof}
For $\tau_0\in(0,\pi)$ define the following auxiliary subsets of $\diamond_+$:
\begin{align*}
	\diamondaux^{\tau_0} 
	&= 
	\{ p\in \diamond_+ 
	\mid 
	\fctfol(p) \mathbin{\new{\le}} (\tfrac{1-\cos\tau_0}{2})^{\frac12}  \}\\
	\nullinfaux^{\tau_0} 
	&= \{ p\in \diamond_+ 
	\mid \fctfol(p) = (\tfrac{1-\cos\tau_0}{2})^{\frac12}  \}
\end{align*}
\begin{definition}\label{def:diamondtau*}
For $\tau_*\in(0,\pi)$ define 
\begin{align*}
	\diamond^{\tau_*} 
	&=
	\diamondaux^{\frac12(\tau_*+\pi)}  \cap ([0,\tau_*]\times S^3)\\
	\nullinf^{\tau_*} 
	&=
	\nullinfaux^{\frac12(\tau_*+\pi)}\cap ([0,\tau_*]\times S^3) 
\end{align*} 
\end{definition}
See Figure \ref{fig:Foliation}.
Note that 
$\overline{\diamond^{\tau_*}} = \diamond^{\tau_*} \cup\spaceinf$.
By Lemma \ref{lem:dphiproperties}, 
the lateral boundary component $\nullinf^{\tau_*}$ of $\diamond^{\tau_*}$
is spacelike relative to $[\gcyl]$.
Furthermore the sets $\diamond^{\tau_*}$ are an exhaustion of
$\diamond_+$, in the following sense.
\begin{figure}
\centering
\begin{tikzpicture}[inner sep=0pt,scale=0.93]
\def\mycoordinatesround{(-4.,0.) (-3.9,0.095) (-3.8,0.19) (-3.7,0.285) (-3.6,0.38) (-3.5,0.475) (-3.4,0.57) (-3.3,0.665) (-3.2,0.76) (-3.1,0.854) (-3.,0.949) (-2.9,1.043) (-2.8,1.137) (-2.7,1.231) (-2.6,1.324) (-2.5,1.418) (-2.4,1.511) (-2.3,1.603) (-2.2,1.695) (-2.1,1.787) (-2.,1.878) (-1.9,1.969) (-1.8,2.059) (-1.7,2.148) (-1.6,2.236) (-1.5,2.323) (-1.4,2.408) (-1.3,2.492) (-1.2,2.575) (-1.1,2.655) (-1.,2.733) (-0.9,2.807) (-0.8,2.878) (-0.7,2.945) (-0.6,3.006) (-0.5,3.061) (-0.4,3.109) (-0.3,3.147) (-0.2,3.176) (-0.1,3.194) (0.,3.2) (0.1,3.194) (0.2,3.176) (0.3,3.147) (0.4,3.109) (0.5,3.061) (0.6,3.006) (0.7,2.945) (0.8,2.878) (0.9,2.807) (1.,2.733) (1.1,2.655) (1.2,2.575) (1.3,2.492) (1.4,2.408) (1.5,2.323) (1.6,2.236) (1.7,2.148) (1.8,2.059) (1.9,1.969) (2.,1.878) (2.1,1.787) (2.2,1.695) (2.3,1.603) (2.4,1.511) (2.5,1.418) (2.6,1.324) (2.7,1.231) (2.8,1.137) (2.9,1.043) (3.,0.949) (3.1,0.854) (3.2,0.76) (3.3,0.665) (3.4,0.57) (3.5,0.475) (3.6,0.38) (3.7,0.285) (3.8,0.19) (3.9,0.095) (4.,0.)};
\def\mycoordinatesedgeA{(-4.,0.) (-3.9,0.095) (-3.8,0.19) (-3.7,0.285) (-3.6,0.38) (-3.5,0.475) (-3.4,0.57) (-3.3,0.665) (-3.2,0.76) (-3.1,0.854) (-3.,0.949) (-2.9,1.043) (-2.8,1.137) (-2.7,1.231) (-2.6,1.324) (-2.5,1.418) (-2.4,1.511) (-2.3,1.603) (-2.2,1.695) (-2.1,1.787) (-2.,1.878) (-1.9,1.969) (-1.8,2.059) (-1.7,2.148) (-1.6,2.236) (-1.5,2.323) (-1.40967,2.4)} ;
\def\mycoordinatesedgeB {
(-1.40967,2.4) (1.40967,2.4)
};

\def\mycoordinatesedgeC{(1.40967,2.4) (1.41,2.4) (1.51,2.314) (1.61,2.227) (1.71,2.139) (1.81,2.05) (1.91,1.96) (2.01,1.869) (2.11,1.778) (2.21,1.686) (2.31,1.594) (2.41,1.502) (2.51,1.409) (2.61,1.315) (2.71,1.222) (2.81,1.128) (2.91,1.034) (3.01,0.939) (3.11,0.845) (3.21,0.75) (3.31,0.656) (3.41,0.561) (3.51,0.466) (3.61,0.371) (3.71,0.276) (3.81,0.181) (3.91,0.086) (4,0)} ;
\node (tip) at (0,4) {}; %
\node (l0) at (-4,0) {}; %
\node (r0) at (4,0) {}; %

\draw[line width=1.5pt,ufogreen] (l0.center)--(tip.center);
\draw[line width=1.5pt,ufogreen] (r0.center)--(tip.center);
\fill[gray!60] plot[smooth] coordinates {\mycoordinatesround};
\fill[pattern={Lines[angle=45,distance=5pt]},pattern color=black]
	plot[smooth] coordinates {
	\mycoordinatesedgeA \mycoordinatesedgeB \mycoordinatesedgeC};
\draw [very thick,color=gray] plot [smooth] coordinates {\mycoordinatesround}; 
\draw [very thick,black] plot [smooth] coordinates {\mycoordinatesedgeA}; 
\draw [very thick,black] plot [smooth] coordinates {\mycoordinatesedgeB}; 
\draw [very thick,black] plot [smooth] coordinates {\mycoordinatesedgeC}; 
\draw[->] (-4.5,-0.2) -- (-4.5,4.5) node[anchor=south east,yshift=1mm] {$\tau$};
  \foreach \y in {0,2.4,3.2,4}
    \draw (-4.6,\y) -- (-4.4,\y);
	\node[left,thick] at (-4.65,0) {$0$};
	\node[left,thick] at (-4.65,2.4) {$\frac{3\pi}{5}$};
	\node[left,thick] at (-4.65,3.2) {$\frac{4\pi}{5}$};
	\node[left,thick] at (-4.65,4) {$\pi$};
	
\node[anchor=south west,xshift=-3mm,yshift=1mm] at (l0) {$\spaceinf$};
\node[anchor=south east,xshift=3mm,yshift=1mm] at (r0) {$\spaceinf$};
\draw[->] (-4.65,0) -- (5,0) node[anchor=north,xshift=2mm,yshift=-1.1mm] {$\arccos\xi^4$};
  \foreach \x in {-4,0,4}
    \draw (\x,-0.1) -- (\x,0.1);
	\node[below,thick] at (-4,-0.2) {$0$};
	\node[below,thick] at (0,-0.2) {$\pi$};
	\node[below,thick] at (4,-0.2) {$0$};
\node[anchor=south west,yshift=1mm] at (2,2) {$\fnullinf$};
\draw[color=black,thick, fill=white] (l0) circle (.05);
\draw[color=black,thick, fill=white] (tip) circle (.05);
\draw[color=black,thick, fill=white] (r0) circle (.05);
\node[anchor=south,yshift=1mm] at (tip) {$\timeinf$};

 \draw[line width=1.5pt, black] (-3.95,0) -- (3.95,0);
\end{tikzpicture}
\captionsetup{width=115mm}
\caption{
Depicted is a cross-section of $\diamond_+$,
using $\tau$ and $\arccos\xi^4$ as coordinates.
The gray shaded region depicts the subset $\smash{\diamondaux}^{\frac{4\pi}{5}}$, 
its upper boundary component is $\smash{\nullinfaux}^{\frac{4\pi}{5}}$.
The hatched region depicts the subset $\smash{\diamond}^{\frac{3\pi}{5}}$,
its lateral boundary component is $\smash{\nullinf}^{\frac{3\pi}{5}}$.}
\label{fig:Foliation}
\end{figure}
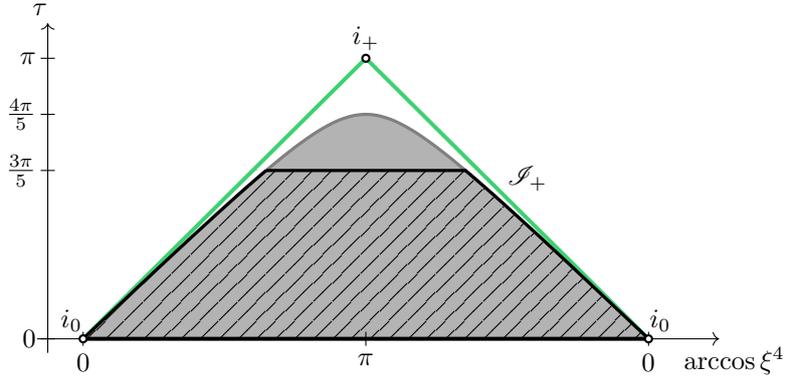
\begin{lemma}\label{lem:Dexhaustion}
One has 
$$
\textstyle
\bigcup_{\tau_*\in[\frac\pi2,\pi)} \diamond^{\tau_*}
=
\diamond_+$$
and for all $0< \tau_0<\tau_1< \pi$ one has 
$\diamond^{\tau_0}\subsetneq\diamond^{\tau_1}\subsetneq \diamond_+$.
\end{lemma}
\begin{proof}
By construction.\qed
\end{proof}
Using Definition \ref{def:diamonds,stau}, 
for $s>0$ and $\tau_*\in(0,\pi)$
define
\begin{subequations}\label{eq:dsttNEW}
\begin{align}
\diamond^{\tau_*}_{s} 
	&=
	\diamond_{s}\cap\diamond^{\tau_*}
	\label{eq:dtaus}\\
\diamond^{\tau_*}_{\tau,s} 
	&= \diamond_{\tau,s}\cap\diamond^{\tau_*}
	&&\text{for $\tau\in[0,\tau_*]$}
	\label{eq:dtautaus}\\
\diamond^{\tau_*}_{\le\tau,s} 
	&= 
	\diamond^{\tau_*}_{s}\cap ([0,\tau] \times S^3)
	\;=\;
	\textstyle
	\bigcup_{\tau'\in[0,\tau]} \diamond^{\tau_*}_{\tau',s} 
	&&
	\text{for $\tau\in(0,\tau_*]$}
	\label{eq:dtauletaus}
\end{align}
\end{subequations}
The sets $\diamond^{\tau_*}_{\le\tau,s}$ are closed, 
contained in $\diamond_+$,
and one has $\diamond^{\tau_*}_{\le\tau_*,s}=\diamond^{\tau_*}_{s}$.
The sets $\diamond^{\tau_*}_{\tau,s}$ are non-empty
and diffeomorphic to a \new{closed} three-dimensional Euclidean ball,
see also Lemma \ref{lem:lowerboundballsNEW}.
They are an exhaustion of $\diamond_{\tau,s}$, in the following sense.
\begin{lemma}\label{lem:Dtauexhaustion}
For every $s>0$ and $\tau\in[0,\pi)$ one has
\begin{align*}
\textstyle
\bigcup_{\tau_*\in\new{(}\tau,\pi)}\diamond^{\tau_*}_{\tau,s}
&\;=\;
\diamond_{\tau,s} 
\end{align*}
where for all $\tau<\tau_0\le \tau_1 < \pi$ one has
$\diamond^{\tau_0}_{\tau,s}
\subset
\diamond^{\tau_1}_{\tau,s}
\subset \diamond_{\tau,s}$.
\end{lemma}
\begin{proof}
By construction.\qed
\end{proof}
\begin{lemma}\label{lem:lowerboundballsNEW}
For all 
$s\in(0,1]$ and 
$\taufix\in[\frac\pi2,\pi)$ and
$\tau_0\in[0,\frac\pi2]$, 
there exists $r\in [-\frac12,1)$ such that 
\begin{equation}\label{eq:diamondBall}
\diamond_{\tau_0,s}^{\taufix}
\;=\;
\{\tau_0\}
\times 
\big\{  
(\xi^1,\xi^2,\xi^3,\xi^4)\in S^3 
\mid
-1\le\xi^4 \le r
\big\}
\end{equation}
\end{lemma}
\begin{proof}
Clearly there exists 
$r\in [-1,1)$ such that \eqref{eq:diamondBall} holds.
We show $-\frac12\le r$.
Using $\Dspop_{<\frac{s}{6}}
\subset \{ (\tau,\xi)\in\cyl \mid 0 < \xi^4 \le 1 \}$
by $s\le1$ and by Remark \ref{rem:Dstau},
\begin{align*}
\diamond^{\taufix}_{\tau_0,s} 
&\supseteq \textstyle
\{\tau_0\}\times 
\{ \xi\in S^3 \mid 
\tfrac{1-\cos\tau_0}{1-\xi^4} \le \tfrac{1-\cos(\frac12(\taufix+\pi))}{2},\;
-1 \le \xi^4 \le 0 \}\\
&= \textstyle
\{\tau_0\}\times 
\{ \xi\in S^3 \mid 
-1 \le \xi^4 \le \min\{0
,
1-\frac{2(1-\cos(\tau_0))}{1-\cos(\frac12(\taufix+\pi))}
\} \}
\end{align*}
One has
\[ 
\textstyle
1-2\tfrac{1-\cos(\tau_0)}{1-\cos(\frac12(\taufix+\pi))}
\ge
1-2\tfrac{1-\cos(\frac\pi2)}{1-\cos(\frac\pi4)}
\ge
-\frac12
\]
using $\taufix\in[\frac\pi2,\pi)$ and $\tau_0\in[0,\frac\pi2]$,
and from this the claim follows.
\qed
\end{proof}

\subsection{Bases}
\label{sec:basis_bulk}

We fix a global $C^\infty$-basis of $\gx(\cyl)$.
Recall the positively oriented frame of vector fields 
$\V_0,\V_1,\V_2,\V_3$ in \eqref{eq:defV},
and the dual frame of one-forms $\Vd^0,\Vd^1,\Vd^2,\Vd^3$,
where $\V_0=\p_{\tau}$ and $\Vd^0=d\tau$.
The following definition parallels Definition \ref{def:elements_i0}.
\begin{definition}\label{def:elements_bulk}
This definition is local to Section \ref{sec:bulk}, see Remark \ref{rem:local5}.
Define the numbers 
$\ngG_k^{\Omega}$,
$\ngG_k^{\I}$
$\ngG_k$,
$\ng_k^{\Omega}$,
$\ng_k^{\I}$, 
$\ng_k$ 
as in \eqref{eq:nm_i0}.

\begin{itemize}
\item 
For $k=0\dots4$ define 
$(\eGO_i^k)_{i=1\dots\ngG_k^{\Omega}},
(\eO_i^k)_{i=1\dots\ng_k^{\Omega}}
\in\Omega^k(\cyl)$ by:
\begin{align*}
\eGO_1^0 &= 1\nonumber\\
\eGO_1^1 &= \Vd^1,\;\eGO_2^1 =\Vd^2,\;\eGO_3^1 =\Vd^3\nonumber\\
\eGO_1^2 &=\Vd^1\wedge\Vd^2,\;
    \eGO_2^2=\Vd^2\wedge\Vd^3,\;
    \eGO_3^2=\Vd^3\wedge\Vd^1\nonumber\\
\eGO_1^3 &=  \Vd^1\wedge\Vd^2\wedge\Vd^3
\\
(\eO^k_i)_{i=1\dots\ng^\Omega_k}&:\; 
	\eGO^k_1,\eGO^k_2,\dots,
	d\tau\wedge\eGO^{k-1}_1,d\tau\wedge\eGO^{k-1}_2,\dots
\end{align*}
\item 
Let $\cyclind=\{(123),(231),(312)\}$ be the cyclic index set.
For $(abc)\in\cyclind$ let
$\Vdpm^a = \frac12(\Vd^0\wedge\Vd^a\pm i\Vd^b\wedge\Vd^c)\in \Omega^2_{\pm}(\cyl)$. 
Define $h_1,\dots,h_5$ exactly as in \eqref{eq:traceless_symmetric_matrices}.
Define the following elements of $\I^2(\cyl)$ respectively $\I^3(\cyl)$:\footnote{
The formula for $\eGI_j^3$ here differs from the analogous 
formula in Definition \ref{def:elements_i0} by a sign,
this is because the basis $dy^0,\dots,dy^3$ is negatively oriented
and $\Vd^0,\dots,\Vd^3$ is positively oriented.}
\begin{align*}
	(\eGI_j^2)_{j=1\dots10}:\ & 
	\mu_{\gcyl}^{-1} \otimes
	(\tsum_{p,q=1}^3(h_\ell)_{pq} \Vdp^p\otimes \Vdp^q)
	\oplus
	cc, \\
	&\mu_{\gcyl}^{-1} \otimes
		(\tsum_{p,q=1}^3(ih_\ell)_{pq}\Vdp^p\otimes \Vdp^q)
		\oplus
		cc\ 
	\\
	(\eGI_j^3)_{j=1\dots6}:\ & 
	\tfrac{1}{2\sqrt{3}}\mu_{\gcyl}^{-1}\otimes\left( 
	2\Vd^1\Vd^2\Vd^3\otimes\Vdp^a
	\new{-} i \Vd^0\Vd^a(\Vd^b\otimes\Vdp^b + \Vd^c\otimes\Vdp^c) \right)
	\oplus cc,\\
	&
	i\tfrac{1}{2\sqrt{3}} \mu_{\gcyl}^{-1}\otimes\left( 
		2\Vd^1\Vd^2\Vd^3\otimes\Vdp^a
		\new{-} i \Vd^0\Vd^a(\Vd^b\otimes\Vdp^b + \Vd^c\otimes\Vdp^c) \right)
		\oplus cc
\end{align*}
where the index $\ell$ used for $(\eGI_j^2)$ runs over $1\dots5$,
the index $(abc)$ used for $(\eGI_j^3)$ runs over $\cyclind$,
and where we use notation analogous to Definition \ref{def:elements_i0}.
For $k=2,3,4$ define the following elements in $\I^k(\cyl)$:
\begin{align}\label{eq:Ibasis_cyl}
(\eI^k_j)_{j=1\dots\ng^{\I}_k}:
	\;\;\eGI_1^k,\eGI_2^k,\dots,d\tau\eGI_1^{k-1},d\tau\eGI_2^{k-1},\dots
\end{align}
where we use the module multiplication in Definition \ref{def:Imod}.
\item 
For $k=0\dots4$
define the following elements of $\gx^k(\cyl)$:
\begin{align}
(\eGg^k_i)_{i=1\dots\ngG_k}:\;\;&
(\eGO^k_1\otimes \KilBasis_\ell)\oplus0\Ieps,\, 
(\eGO^k_2\otimes \KilBasis_\ell)\oplus0\Ieps,\dots,\nonumber\\
&\ 0\oplus\eGI_1^{k+1}\Ieps,\, 0\oplus\eGI_2^{k+1}\Ieps,\dots\label{eq:gGbasis_bulk}\\
(\eg^k_i)_{i=1\dots\ng_k}:\;\;& 
\eGg^k_1,\, \eGg^k_2,\dots, d\tau \eGg^{k-1}_1,\,d\tau\eGg^{k-1}_2,\dots
	\label{eq:gbasis_bulk}
\end{align}
where $\ell$ runs over $1\dots10$, 
where 
$\KilBasis_1,\dots,\KilBasis_{10}$ is the basis of $\Kil$ in \eqref{eq:Kilbas},
and where we use the module multiplication \eqref{eq:gmod}.
\end{itemize}
\end{definition}
The following lemma parallels Lemma \ref{lem:bases_i0}.
\begin{lemma}\label{lem:bases_bulk}
Using the elements in Definition \ref{def:elements_bulk}, 
for $k=0\dots4$ one has:
\[ 
\def\arraystretch{1.}
\begin{tabular}{c|ccccccccc}
Module & 
$\Omega^k(\cyl)$ &
$\I^k(\cyl)$ &
$\gx^k(\cyl)$ \\
Rank &
$\ng_k^{\Omega}$ &
$\ng_k^{\I}$ &
$\ng_k$ \\
Basis &
$(\eO^k_i)_{i=1\dots \ng_k^{\Omega}}$ & 
$(\eI^k_i)_{i=1\dots \ng_k^{\I}}$ & 
$(\eg^k_i)_{i=1\dots \ng_k}$ 
\end{tabular} 
\]
\end{lemma}
\begin{proof}
By direct inspection.\qed
\end{proof}

\subsection{Norms}
\label{sec:norms_bulk}

We define the norms that we use away from $\spaceinf$
(some of them are actually seminorms, but we refer to them
as norms for simplicity).

Define the following densities:
\begin{align}\label{eq:bulkdensities}
\begin{aligned}
\mucyl &= |\Vd^0\wedge\Vd^1\wedge\Vd^2\wedge\Vd^3|
	\;\in\; \dens{4}(\cyl)\\
\mucylS &= |\Vd^1\wedge\Vd^2\wedge\Vd^3|
	\qquad\,\;\;\in\; \dens{3}(S^3)
\end{aligned}
\end{align}
We will also use the fact that $\mucylS$ defines a $3$-density
on every level set of $\tau$.
\begin{definition}[Norms away from spacelike infinity]\label{def:bulknorms}
For every $k\in\Z_{\ge0}$ and 
$s>0$ and $\tau\in [0,\pi)$ and
$f\in C^\infty(\diamond_s)$ define:
\begin{align}\label{eq:CHtaulevelset}
\begin{aligned}
\|f\|_{\H^k(\diamond_{\tau,s})}^2 
&=\textstyle
	\sum_{j=0}^k
	\sum_{i_1,\dots,i_{j}=1}^{3}
	\int_{\diamond_{\tau,s}}\big|\V_{i_1}\cdots\V_{i_j} f\big|^2
	\,\mucylS\\
\|f\|_{\sH^k(\diamond_{\tau,s})}^2 
&=\textstyle
	\sum_{j=0}^k
	\sum_{i_1,\dots,i_{j}=0}^{3}
	\int_{\diamond_{\tau,s} }\big|\V_{i_1}\cdots\V_{i_j} f\big|^2
	\,\mucylS\\
\|f\|_{C^k(\diamond_{\tau,s})}
&=\textstyle
	\sum_{j=0}^k
	\sum_{i_1,\dots,i_{j}=1}^{3}
	\sup_{p\in\diamond_{\tau,s}}
	\big|\V_{i_1}\cdots\V_{i_j} f(p)\big|\\
\|f\|_{\sC^k(\diamond_{\tau,s})}
&=\textstyle
	\sum_{j=0}^k
	\sum_{i_1,\dots,i_{j}=0}^{3}
	\sup_{p\in\diamond_{\tau,s}}
	\big|\V_{i_1}\cdots\V_{i_j} f(p)\big|
\end{aligned}
\intertext{%
We make the same definitions when $f$ is 
only defined near $\diamond_{\tau,s}$. Further define}
\label{eq:CHdiamond}
\begin{aligned}
\|f\|_{\nosH^k(\diamond_{s})}^2 
&=\textstyle
	\sum_{j=0}^k
	\sum_{i_1,\dots,i_{j}=0}^{3}
	\int_{\diamond_{s} }\big|\V_{i_1}\cdots\V_{i_j} f\big|^2
	\,\mucyl\\
\|f\|_{\nosC^k(\diamond_{s})}
&=\textstyle
	\sum_{j=0}^k
	\sum_{i_1,\dots,i_{j}=0}^{3}
	\sup_{p\in\diamond_{s}}
	\big|\V_{i_1}\cdots\V_{i_j} f(p)\big|
\end{aligned}
\end{align}
We make analogous definitions for vector-
and matrix-valued functions, 
where we apply the norms componentwise and then take the $\ell^2$-sum of the 
components;
and for elements in $\gx(\diamond_s)$,
where we use the basis \eqref{eq:gbasis_bulk} to identify
them with vector-valued functions on $\diamond_s$.
\end{definition}
The slashed norms in \eqref{eq:CHtaulevelset}
measure differentiability with respect to all
vector fields
$\V_0,\dots,\V_3$. Note that $\V_0$ is not tangential to $\diamond_{\tau,s}$, 
so the slashed norms are not determined by 
the restriction of $f$ to $\diamond_{\tau,s}$.
Further, in \eqref{eq:CHtaulevelset}
the $\sH^0$- and the $\H^0$-norms are equal,
and the $\sC^0$- and the $C^0$-norms are equal.

Recall the sets in \eqref{eq:dsttNEW}.
\begin{definition}\label{def:taufixnorms}
For all $k\in\Z_{\ge0}$ and 
$s>0$,
$\taufix\in (0,\pi)$,
$\tau\in[0,\taufix]$
define
\[ 
\|{\cdot}\|_{\H^k(\diamond_{\tau,s}^{\taufix})} 
\qquad
\|{\cdot}\|_{\sH^k(\diamond_{\tau,s}^{\taufix})}
\qquad
\|{\cdot}\|_{C^k(\diamond_{\tau,s}^{\taufix})}
\qquad
\|{\cdot}\|_{\sC^k(\diamond_{\tau,s}^{\taufix})}
\]
analogously to \eqref{eq:CHtaulevelset},
with $\diamond_{\tau,s}$ replaced by $\diamond_{\tau,s}^{\taufix}$.
Further, for $\tau\in(0,\taufix]$, define
\[ 
\|{\cdot}\|_{\nosH^k(\diamond_{\le\tau,s}^{\taufix})}
\qquad
\|{\cdot}\|_{\nosC^k(\diamond_{\le\tau,s}^{\taufix})}
\]
analogously to \eqref{eq:CHdiamond},
with $\diamond_{s}$ replaced by $\diamond_{\le\tau,s}^{\taufix}$.
\end{definition}

We now prove auxiliary inequalities for these norms.
It will be important that the constants in the inequalities
are independent of $\taufix$ and $\tau$.
\begin{lemma}\label{lem:SobolevCNEW}
For all $k\in\Z_{\ge0}$, $s\in(0,\new{1}]$, $\tau_*\in[\frac\pi2,\pi)$
and $f\in C^\infty(\diamond_s^{\tau_*})$:
\begin{itemize}
\item 
For all $\tau\in[0,\tfrac\pi2]$:
\begin{subequations}
\begin{align}
\|f\|_{\sC^k(\diamond_{\tau,s}^{\tau_*})}
&\lesssim_k 
\|f\|_{\sH^{k+2}(\diamond_{\tau,s}^{\tau_*})}
\label{eq:CH1}\\
\|f\|_{\sH^k(\diamond_{\tau,s}^{\tau_*})} 
&\lesssim_{k,s}
\tint_{0}^{\frac\pi2} \|f\|_{\sH^{k+1}(\diamond_{\tau',s}^{\tau_*})} d\tau'
\label{eq:LinfL2L1L2}\\
\|f\|_{\sC^k(\diamond_{\tau,s}^{\tau_*})} 
&\lesssim_{k}
\tint_{0}^{\frac\pi2} \|f\|_{\sH^{k+3}(\diamond_{\tau',s}^{\tau_*})} d\tau'
\label{eq:Ckint}
\end{align}
\end{subequations}
\item 
For all $\tau\in[0,\tau_*]$:
\begin{align}
\|f\|_{\sC^k(\diamond_{\tau,s}^{\tau_*})}
&\lesssim_k
\sup_{\tau'\in[0,\tau]}\|f\|_{\sH^{k+3}(\diamond_{\tau',s}^{\tau_*})}
\label{eq:SobolevC} 
\end{align}
\end{itemize}
\end{lemma}
\begin{proof}
\eqref{eq:CH1}:
This follows from a standard three-dimensional Sobolev inequality,
the constant is independent of $s,\taufix,\tau$ by Lemma \ref{lem:lowerboundballsNEW}.

\eqref{eq:LinfL2L1L2}: 
This is checked similarly to Lemma \ref{lem:LinfL2toL1L2}, hence we omit the details.

\eqref{eq:Ckint}: This follows from \eqref{eq:CH1} and \eqref{eq:LinfL2L1L2}. 

\eqref{eq:SobolevC}:
Given \eqref{eq:CH1}, it remains to check this for $\tau\in[\frac\pi2,\taufix]$.
Then 
\begin{align*}
\|f\|_{\sC^k(\diamond_{\tau,s}^{\tau_*})}
\le
\|f\|_{\nosC^k(\diamond_{\le\tau,s}^{\tau_*})}
\lesssim_{k}
\|f\|_{\nosH^{k+3}(\diamond_{\le\tau,s}^{\tau_*})}
\end{align*}
where we use a standard four-dimensional Sobolev inequality.
The constant in the Sobolev inequality is independent of $\tau$
because $\tau\ge\frac\pi2$.
By Fubini,
\[ 
\|f\|_{\nosH^{k+3}(\diamond_{\le\tau,s}^{\tau_*})}^2
=
\tint_{0}^{\tau}\|f\|_{\sH^{k+3}(\diamond_{\tau',s}^{\tau_*})}^2\,d\tau'
\le
\tau \sup_{\tau'\in[0,\tau]} \|f\|_{\sH^{k+3}(\diamond_{\tau',s}^{\tau_*})}^2
\]
Using $\tau\le\pi$ and taking the square root, the claim follows.
\qed
\end{proof}

\subsection{Gauge}
\label{sec:Gauge_cyl}

Similarly to Section \ref{sec:gauge_i0} we define a gauge,
and show that relative to this gauge, the Einstein equations \eqref{eq:MC} 
are quasilinear symmetric hyperbolic up to constraints that propagate,
including along $\p\diamond_+$.
The gauge that we define here equals that in 
\cite[Section 3.5.3]{Thesis} with $\mathrm{T}$, $g$ 
there chosen as in \eqref{eq:Tg} below.

The constructions in Section \ref{sec:Gauge_cyl} 
parallel those in Section \ref{sec:gauge_i0}.
For clarity we will nevertheless write them down explicitly,
hence there will be repetitions.

Many of the definitions and statements
will be made globally on $\cyl$. 
They can be made analogously on $\diamond_+$ and on $\diamond^{\taufix}$ in Definition \ref{def:diamondtau*}, 
since all constructions are effectively fiberwise.

\subsubsection{Definition of gauge}
\label{sec:gauge_definition_bulk}

We define a gauge (Definition \ref{def:gauge_bulk}) and show basic properties 
(Lemma \ref{lem:gauge_mainprop_bulk}).

The construction uses the following smooth vector field and metric on $\cyl$:
\begin{align}\label{eq:Tg}
\V_0=\p_{\tau} \qquad\qquad \gcyl
\end{align}

The next preliminary definitions are local to 
to Section \ref{sec:bulk}, see Remark \ref{rem:local5}.%
\begin{itemize}
\item 
Let $\Oipr{k}{\cdot}{\cdot}:
\Omega^k(\cyl)\times\Omega^k(\cyl)\to C^\infty(\cyl)$ 
be the nondegenerate symmetric $C^\infty$-bilinear form induced by $\gcyl$,
defined using
the formula 
\eqref{eq:OmegaIPdef} 
with $\ghom$ replaced by $\gcyl$.
\item
Let  
$\smash{\Iipr{k}{\cdot}{\cdot}}:\I^k(\cyl)\times \I^k(\cyl)\to C^\infty(\cyl)$ 
be the nondegenerate symmetric $C^\infty$-bilinear form defined using 
the formula 
\eqref{eq:Ibildefformulas}
with $\mu_{\ghom}^{-1}$ replaced by $\mu_{\gcyl}^{-1}$.
Explicitly, using the basis \eqref{eq:Ibasis_cyl},
\begin{align}
\label{eq:IiprExplicit_bulk}
\begin{aligned}
\Iipr{2}{\eI^2_i}{\eI^2_j}
&=
\left(\begin{smallmatrix}
\one_{5} & 0\\
0& -\one_{5}
\end{smallmatrix}\right)_{ij}
\\
\Iipr{3}{\eI^3_i}{\eI^3_j}
&=
\left(\begin{smallmatrix}
-\one_{3} & 0 &0&0\\
0& \one_{3} &0&0\\
0&0&-\one_5&0\\
0&0&0&\one_5
\end{smallmatrix}\right)_{ij}
\\
\Iipr{4}{\eI^4_i}{\eI^4_j}
&=
\left(\begin{smallmatrix}
\one_{3} & 0\\
0& -\one_{3}
\end{smallmatrix}\right)_{ij}
\end{aligned}
\end{align}

\item 
Let $\Iinter_{\V_0}:\I^{k+1}(\cyl)\to\I^{k}(\cyl)$ be
the adjoint (relative to \smash{$\Iipr{k}{\cdot}{\cdot}$}) of the map 
$\I^{k}(\cyl)\to \I^{k+1}(\cyl)$, $u\mapsto \V_0^\flat u$
where $\V_0^\flat=\gcyl(\V_0,{\cdot})=-d\tau$,
and where we use the module multiplication in Definition \ref{def:Imod}. That is,
\begin{equation}
\label{eq:intmult_bulk}
\Iipr{k}{\Iinter_{\V_0}u}{u'} = \Iipr{k+1}{u}{\V_0^\flat u'}
\end{equation}
for all $u\in \I^{k+1}(\cyl)$ and $u'\in\I^k(\cyl)$.
Explicitly,
\begin{align}
\label{eq:IinterExplicit_bulk}
\begin{aligned}
\Iinter_{\V_0} \eI^2_i
&=0 
\\
\Iinter_{\V_0} \eI^3_i
&=
\left(\begin{smallmatrix}
0_{10\times6} & \one_{10}
\end{smallmatrix}\right)_{ji} \eI^2_j 
\\
\Iinter_{\V_0} \eI^4_i
&=
\left(\begin{smallmatrix}
\one_{6} \\ 
0_{10\times6}
\end{smallmatrix}\right)_{ji} \eI^3_j
\end{aligned}
\end{align}
where we sum over $j$.

\item 
Define $P_{\V_0}:\I^k(\cyl)\to\I^k(\cyl)$ 
using the formula \eqref{eq:Pdefhom} and the preceding paragraph, with $\Thom$ and $\ghom$ replaced by 
$\V_0$ and $\gcyl$, respectively.
Explicitly,
\begin{align}
\label{eq:PExplicit_bulk}
\begin{aligned}
P_{\V_0}(\eI^2_i) 
&= 
\left(\begin{smallmatrix}
\one_{5} & 0\\
0& -\one_{5}
\end{smallmatrix}\right)_{ji}\eI^2_j
\\
P_{\V_0}(\eI^3_i)
&=
\left(\begin{smallmatrix}
-\one_{3} & 0 &0&0\\
0& \one_{3} &0&0\\
0&0&-\one_5&0\\
0&0&0&\one_5
\end{smallmatrix}\right)_{ji}\eI^3_j
\\
P_{\V_0}(\eI^4_i) 
&= 
\left(\begin{smallmatrix}
\one_{3} & 0\\
0& -\one_{3}
\end{smallmatrix}\right)_{ji}\eI^4_j
\end{aligned}
\end{align}
where we sum over $j$.
\end{itemize}
\begin{definition}\label{def:gauge_bulk}
This definition is local to Section \ref{sec:bulk}, 
see Remark \ref{rem:local5}.
Define 
\begin{align}\label{eq:gaugespaces_bulk}
\begin{aligned}
\OmegaG^k(\cyl) 
	&=
	\{\omega\in\Omega^k(\cyl) \mid \intermult_{\V_0}\omega=0 \} \\ 
\IG^k(\cyl) 
	&=
	\{u \in \I^{k}(\cyl) \mid \Iinter_{\V_0} u =0\}\\
\gxG^k(\cyl) 
	&= \big(\OmegaG^k(\cyl)\otimesRR\Kil\big) \oplus \IG^{k+1}(\cyl)\Ieps
\end{aligned}
\end{align}
for $k=0\dots4$. In the first line, $\intermult_{\V_0}$
is the interior multiplication with $\V_0$, 
in the second line, $\Iinter_{\V_0}$ is the map in \eqref{eq:intmult_bulk}.
Define the $C^\infty$-bilinear forms:
\begin{itemize}
\item 
$\OmegaBil^k:\OmegaG^k(\cyl)\times\Omega^{k+1}(\cyl)\to C^\infty(\cyl)$ by 
$\OmegaBil^k(\omega,\omega') = \Oipr{k}{\omega}{\intermult_{\V_0}\omega'}$.
\item 
$\IBil^k: \IG^{k}(\cyl)\times\I^{k+1}(\cyl)\to C^\infty(\cyl)$ by
$\IBil^k(u,u')= \Iipr{k}{P_{\V_0}u}{\Iinter_{\V_0}u'}$.
\item 
$\gBil^k:\gxG^{k}(\cyl)\times\gx^{k+1}(\cyl)\to C^\infty(\cyl)$ by 
\begin{align}
\label{eq:cylgbil}
&\gBil^k( (\omega\otimes \KilBasis_{\ell})\oplus \uI\Ieps,(\omega'\otimes \KilBasis_{\ell'})\oplus \uI'\Ieps  )
=
\OmegaBil^k(\omega,\omega')\delta_{\ell\ell'} + \IBil^{k+1}(\uI,\uI')
\end{align}
where $\omega\in\OmegaG^k(\cyl)$, $\omega'\in\Omega^{k+1}(\cyl)$
and 
$\uI\in\IG^{k+1}(\cyl)$, $\uI'\in\I^{k+2}(\cyl)$,
and where $\KilBasis_{1},\dots,\KilBasis_{10}$ 
is the basis of $\Kil$ in \eqref{eq:Kilbas}.
\end{itemize}
\end{definition}

The module $\gxG^k(\cyl)$ is the module of smooth sections,
over $\cyl$, of a trivial vector bundle $\gxG^k$ defined on $\cyl$.
Further, $\gxG^k$ is a subbundle of $\gx^k$.
\begin{lemma}\label{lem:gaugebases_bulk}
Using the elements from Definition \ref{def:elements_bulk},
for each $k=0\dots4$:
\[ 
\def\arraystretch{1.1}
\begin{tabular}{c|ccccccccc}
Module & 
$\OmegaG^k(\cyl)$ &
$\IG^k(\cyl)$ &
$\gxG^k(\cyl)$ \\
Rank &
$\ngG_k^{\Omega}$ &
$\ngG_k^{\I}$ &
$\ngG_k$ \\
Basis &
$(\eGO^k_i)_{i=1\dots \ngG_k^{\Omega}}$ & 
$(\eGI^k_i)_{i=1\dots \ngG_k^{\I}}$ & 
$(\eGg^k_i)_{i=1\dots \ngG_k}$ 
\end{tabular} 
\]
\end{lemma}
\begin{proof}
The first column is immediate;
the second column follows from \eqref{eq:IinterExplicit_bulk};
the third column follows from the first two.\qed
\end{proof}
The following lemma parallels Lemma \ref{lem:Binbases}.
\begin{lemma}\label{lem:BilBasisCpt}
Relative to the bases in Lemma \ref{lem:bases_bulk} and \ref{lem:gaugebases_bulk}, 
the bilinear forms in Definition \ref{def:gauge_bulk} are given as follows.
For $k=0\dots4$ and $\ell=1,2,3$ one has:
\begin{align*}
\OmegaBil^k(\eGO^k_i,\eGO^{k+1}_j) &= 0 
	& \OmegaBil^k(\eGO^k_i,d\tau\wedge\eGO^{k}_j) &= \delta_{ij}
	& \OmegaBil^k(\eGO^k_i,\Vd^\ell\wedge\eGO^{k}_j) &= 0\\
\intertext{
Further, for $k=2,3$ and $\ell=1,2,3$ one has
(note that $\IBil^k=0$ for $k=0,1,4$):
}
\IBil^k(\eGI^k_i,\eGI^{k+1}_j) &= 0 
	& \IBil^k(\eGI^k_i,d\tau\eGI^{k}_j) &= \delta_{ij}
	& \IBil^k(\eGI^k_i,\Vd^\ell\eGI^{k}_j)
	&=
	\left(\begin{smallmatrix}
	0 & -A_{k,\ell}\\
	(-A_{k,\ell})^T & 0
	\end{smallmatrix}\right)_{ij}
\end{align*}
where the matrices $A_{k,\ell}$ are defined exactly
like in Lemma \ref{lem:Binbases}.
\end{lemma}
\begin{proof}
Analogous to the proof of Lemma \ref{lem:Binbases},
using \eqref{eq:IiprExplicit_bulk}, 
\eqref{eq:IinterExplicit_bulk}, 
\eqref{eq:PExplicit_bulk}.
Note that $ \IBil^k(\eGI^k_i,\Vd^\ell\eGI^{k}_j)$ 
differs from the corresponding expression in Lemma \ref{lem:Binbases}
by a sign, this is due to the fact that
the frame of one-forms $dy^0,\dots,dy^3$ used there is negatively oriented, 
while $\Vd^0,\dots,\Vd^3$ used here is positively oriented. \qed
\end{proof}
The following remark parallels Remark \ref{rem:EigenvaluesIbil}.
\begin{remark}\label{rem:EigenvaluesIbil_bulk}
Let $c_0,c_1,c_2,c_3\in\R$. 
For $k=2,3$
consider the symmetric 
$\ngG^{\I}_k\times\ngG^{\I}_k$-matrix whose $ij$-entry is given by 
(recall that $\Vd^0=d\tau$)
\[ 
\IBil^k(\eGI^k_i,(c_0d\tau + \textstyle\sum_{\ell=1}^3 c_\ell\Vd^\ell)\eGI^{k}_j)
\]
Its eigenvalues are
$c_0, c_0\pm |\vec{c}|, c_0\pm \frac{|\vec{c}|}{2}$ when $k=2$
respectively $c_0,c_0\pm \frac{|\vec{c}|}{2}$ when $k=3$,
where $\vec{c}=(c_1,c_2,c_3)$.
Hence for $k=2$ it is positive definite iff 
$c_0 - |\vec{c}| > 0$.
\end{remark}
The following lemma parallels Lemma \ref{lem:gauge_mainprop_i0}.
Define
\begin{equation}\label{eq:Omegafutcyl}
\Omegafut(\cyl) = \{\omega\in\Omega^1(\cyl) \mid \gcyl^{-1}(\omega,\omega)<0,\,\omega(\p_{\tau})>0\}
\end{equation}
\begin{lemma}\label{lem:gauge_mainprop_bulk}
The tuple $(\gxG(\cyl),\gBil)$ is a gauge for $\gx(\cyl)$, in the following sense.
For all $\omega\in\Omegafut(\cyl)$,
left-multiplication $\gxG(\cyl)\to\gx(\cyl)$, $u\mapsto \omega u$ is 
fiberwise injective,
and 
$\gx(\cyl) = \gxG(\cyl)\oplus\omega\gxG(\cyl)$,
where we use the module multiplication \eqref{eq:gmod}.
Moreover, for all $k=0\dots4$:
\begin{enumerate}[label={\textnormal{(G'\arabic*)}}]
\item \label{item:gaugesymNEW_bulk}
$\gBil^k(\,\cdot\,,\omega\,\cdot\,)|_{\gxG^k(\cyl)\times\gxG^k(\cyl)}$ 
is symmetric for all $\omega\in \Omega^1(\cyl)$.
\item \label{item:gaugeposNEW_bulk}
For every $u\in\gxG^k(\cyl)$ and every $\omega\in\Omega^1(\cyl)$ one has
\begin{align}
&\gBil^{k}(u,\omega u) \ge 
\big( 
\omega(\V_0) - \big(\textstyle\sum_{i=1}^3|\omega(\V_i)|^2\big)^{\frac12}
\big)
\gBil^{k}(u, d\tau u)
\label{eq:posineq_bulk}
\end{align}
and if $u\neq0$ then $\gBil^{k}(u, d\tau u)>0$. 
Furthermore, 
\begin{equation}\label{eq:posiffofut_bulk}
\gBil^k(\,\cdot\,,\omega\,\cdot\,)|_{\gxG^k(\cyl)\times\gxG^k(\cyl)}>0
\quad\Leftrightarrow\quad
\omega\in \Omegafut(\cyl)
\end{equation}
\item \label{item:gaugekernelNEW_bulk}
$\gxG^{k+1}(\cyl) = \big\{u\in \gx^{k+1}(\cyl)\mid\gBil^k(\gxG^{k}(\cyl),u)=0 \big\}$
\end{enumerate}
\end{lemma}
\begin{proof}
This is analogous to the proof of Lemma \ref{lem:gauge_mainprop_i0},
using Lemma \ref{lem:BilBasisCpt} and Remark \ref{rem:EigenvaluesIbil_bulk}.\qed
\end{proof}

\subsubsection{MC-equation as a symmetric hyperbolic system}
This section serves as preparation for Section \ref{sec:existence_bulk}.
We show, using the gauge $(\gxG(\cyl),\gBil)$ 
from Definition \ref{def:gauge_bulk}, that the Einstein equations \eqref{eq:MC} 
are quasilinear symmetric hyperbolic including along 
null and timelike infinity, up to constraints that propagate 
(Lemma \ref{lem:translationofeqcyl}, \ref{lem:aAminkcpt}, \ref{lem:deltacausalcpt}). 

We will use the identifications
\begin{align}\label{eq:identify_g_vec_bulk}
\begin{aligned}
\gxG^k(\cyl) &\simeq C^\infty(\cyl,\R^{\ngG_k}) 
& &\text{using the basis $(\eGg_{i}^k)_{i=1\dots\ngG_k}$ in \eqref{eq:gGbasis_bulk}}\\
\gx^k(\cyl) &\simeq C^\infty(\cyl,\R^{\ng_k}) 
& &\text{using the basis $(\eg_{i}^k)_{i=1\dots\ng_k}$ in \eqref{eq:gbasis_bulk}}
\end{aligned}
\end{align}
\begin{definition}
This definition is local to Section \ref{sec:bulk}, 
see Remark \ref{rem:local5}.
For $\mu=0\dots3$ and $\ideg=1,2$ and $\ell=1\dots\ng_k$ let
\begin{align}\label{eq:anchorforms}
(\anchor_{k})_{\ell}^\mu \in \Omega^k(\cyl)
\end{align}
be the unique $k$-form such that for all $f\in C^\infty(\cyl)$:
$\anchorg(\eg_{\ell}^k)(f) 
=
(\anchor_{k})_{\ell}^\mu \V_{\mu}f$.
\end{definition}
The $k$-forms \eqref{eq:anchorforms}
are explicitly given as follows. 
Write $\eg_{\ell}^k=(\sum_{i=1}^{10}\omega_i \otimes \KilBasis_i)\oplus\uI\Ieps$,
where $\KilBasis_1,\dots,\KilBasis_{10}$ is a basis of $\Kil$, and where $\omega_i$ are $k$-forms.
Then 
$$(\anchor_{k})_{\ell}^\mu = \tsum_{i=1}^{10} \omega_i \Vd^\mu(\KilBasis_i)$$
which follows \eqref{eq:anchorgdef} and \eqref{eq:anchorLdef}.
In particular, \eqref{eq:anchorforms} are indeed smooth on $\cyl$.

The following definition parallels Definition \ref{def:SHSarrays_i0}.
\begin{definition}\label{def:SHSarrays_bulk}
This definition is local to Section \ref{sec:bulk}, 
see Remark \ref{rem:local5}.
For $\mu=0\dots3$ and $\ideg,k'=1,2$ define
\begin{align}\label{eq:minkarrays_bulk}
\begin{aligned}
\amink_{\ideg}^\mu &\in C^\infty(\cyl,\End(\R^{\ngG_{\ideg}}))\\
\Amink_{\ideg}^\mu &\in C^\infty(\cyl,\Hom(\R^{\ng_1},\End(\R^{\ngG_{\ideg}}))\\
\Lmink_{\ideg} &\in C^\infty(\cyl,\End(\R^{\ngG_{\ideg}}))\\
%
\Bmink_{\ideg} &\in  C^\infty(\cyl,\Hom(\R^{\ng_1}\otimes\R^{\ng_{\ideg}},\R^{\ngG_{\ideg}}))\\
\Amink_{k' k}^\mu &\in C^\infty(\cyl,\Hom(\R^{\ng_{k'}},\Hom(\R^{\ng_{k}},\R^{\ngG_{k+k'-1}})))\\
\SFmink &\in C^\infty(\cyl,\Hom(\R^{\ng_2},\R^{\ngG_1}))\\
\ginj_k &\in C^\infty(\cyl, \Hom(\R^{\ngG_k},\R^{\ng_k}))
\end{aligned}
\end{align}
as follows, using $\gBil^k$ in Definition \ref{def:gauge_bulk},
the bases 
\eqref{eq:gGbasis_bulk}, \eqref{eq:gbasis_bulk}, and \eqref{eq:anchorforms}:
\begin{align*}
(\amink_{\ideg}^\mu u)_i 
	&= (\amink_{\ideg}^\mu)_{ij} u_{j}
&& \text{where}& 
(\amink_{\ideg}^\mu)_{\outind j}&=\gBil^{\ideg}\big(\eGg_\outind^{\ideg}, \Vd^\mu \eGg^{\ideg}_j\big)\\
(\Amink_{\ideg}^\mu (v) u)_i
	&= (\Amink_{\ideg}^\mu)_{\ell,ij} v_\ell u_j 
&& \text{where}& 
(\Amink_{\ideg}^\mu)_{\ell,\outind j} &= \gBil^{k}\big(\eGg^{\ideg}_\outind,
(\anchor_1)_\ell^\mu
\eGg^{\ideg}_j\big)\\
(\Lmink_{\ideg} u)_{i} &= (\Lmink_{\ideg})_{ij}u_j
&& \text{where}& 
(\Lmink_{\ideg})_{\outind j} &= -\gBil^{\ideg}\big(\eGg_\outind^{\ideg}, \dg \eGg^{\ideg}_j\big)\\
(\Bmink_{\ideg}(v,w))_i &= (\Bmink_{\ideg})_{ij\ell} v_j w_\ell 
&& \text{where}& 
(\Bmink_{\ideg})_{\outind j\ell} &= -\gBil^{\ideg}\big(\eGg_\outind^{\ideg}, [\eg_j^1,\eg_\ell^{\ideg}]\big)\\
(\Amink_{k'k}^\mu (w')w)_i  &= (\Amink_{k'k}^\mu)_{\ell,i j} w'_\ell w_j
&& \text{where}& 
(\Amink_{k'k}^\mu)_{\ell,i j} &= \gBil^{k+k'-1}\big(\eGg^{k+k'-1}_i, 
(\anchor_{k'})_\ell^\mu
\eg^k_j\big)\\
(\SFmink v')_i &= \SFmink_{ij} v'_j
&& \text{where}& 
\SFmink_{\outind j} &= -\gBil^{1}\big(\eGg_\outind^1, \eg^2_j\big)\\
(\ginj_k u)_i
&=
\delta_{ij} u_j
\end{align*}
with
$u\in C^\infty(\cyl,\R^{\ngG_k})$, 
$v\in C^\infty(\cyl,\R^{\ng_1})$,
$v'\in C^\infty(\cyl,\R^{\ng_2})$,
$w\in C^\infty(\cyl,\R^{\ng_k})$, 
$w'\in C^\infty(\cyl,\R^{\ng_{k'}})$,
and where the sum over the repeated indices $j$, $\ell$ is implicit.
\end{definition}
The components of \eqref{eq:minkarrays_bulk} are indeed smooth
on $\cyl$ (in particular they are smooth on $\overline\diamond$)
because 
$\Vd^\mu$,
\eqref{eq:gGbasis_bulk}, \eqref{eq:gbasis_bulk}, \eqref{eq:anchorforms},
$\gBil^k$ are smooth 
(for $\gBil^k$ use Lemma \ref{lem:BilBasisCpt}), 
and $\dg$, $[\cdot,\cdot]$ 
are differential operators with smooth coefficients on $\cyl$.

Note that $\ginj_k$ is the inclusion $\gxG^k(\cyl)\hookrightarrow \gx^k(\cyl)$,
via the identification \eqref{eq:identify_g_vec_bulk}.

The following lemma parallels Lemma \ref{lem:translationofeq}.
\begin{lemma}\label{lem:translationofeqcyl}
For all $\cc\in\gxG^1(\cyl)$, 
$v\in\gx^1(\cyl)$ and
$U\in\gxG^2(\cyl)$:
\begin{align*}
\gBil^{1}\left( \eGg_{i}^1, 
\dg (v+\cc) + \tfrac12[v+\cc,v+\cc]\right) e^{\ngG_1}_i 
&=
\big(
\amink_1^{\mu}
+ \Amink_1^{\mu}(v)
+ \Amink_1^{\mu}(\ginj_1 \cc)
\big) \V_{\mu}  \cc
\nonumber\\
&\;\;-
\big(\Lmink_1  \cc - \Amink_{11}^\mu(\ginj_1 \cc) \V_{\mu} v + \Bmink_1(v,\ginj_1 \cc)\big) 
\nonumber\\
&\;\;- \tfrac12\Bmink_1(\ginj_1\cc,\ginj_1\cc) - \SFmink \big(\dg v+\tfrac12[v,v]\big)\nonumber \\[2mm]
\gBil^{2}\left( \eGg_{i}^2, \dg U + [v,U]\right)e^{\ngG_2}_i
&=
\left(
\amink_2^{\mu}
+ \Amink_2^{\mu}(v) \right) \V_{\mu}  U\\
&\;\;-
\big(\Lmink_2  U \new{+} \Amink_{21}^\mu(\ginj_2 U) \V_{\mu} v + \Bmink_2(v,\ginj_2 U)\big)
\nonumber
\end{align*}
where, on the left hand sides, 
$(e_{i}^{\ngG_k})_{i=1\dots\ngG_k}$ is the standard basis of $\R^{\ngG_k}$
and we sum over $i$,
and, 
on the right hand sides, the identification \eqref{eq:identify_g_vec_bulk} is used.
\end{lemma}
\begin{proof}
Analogous to the proof of Lemma \ref{lem:translationofeq}. \qed
\end{proof}
In the remainder we show symmetry and positivity properties. 
\begin{lemma}\label{lem:aAminkcpt}
For every $\mu=0\dots3$ and $k=1,2$:
\begin{enumerate}[({i}1)]
\item \label{item:aminkcpt}
$\amink^\mu_{k} \in C^\infty(\cyl,\End(\R^{\ngG_k}))$
is a symmetric matrix at every point on $\cyl$,
and its entries are constant, i.e.~in $\R$.
Further $\amink^0_{k}=\one$, and for every $\omega\in\Omega^1(\cyl)$:
\begin{equation}\label{eq:aineqcyl}
\textstyle
\omega(\amink^\mu_{k}\V_\mu) \ge 
\big( \omega(\V_0)-\big(\sum_{i=1}^3|\omega(\V_i)|^2\big)^{\frac12} \big) \one
\end{equation}

\item \label{item:Aminkcpt}
Let $(e_\ell^{\ng_1})_{\ell=1\dots\ng_1}$ be the standard basis of $\R^{\ng_1}$.
For every $\ell=1\dots\ng_1$, 
$\Amink^\mu_{k}(e_\ell^{\ng_1}) \in C^\infty(\cyl,\End(\R^{\ngG_k}))$
is a symmetric matrix at every point on $\cyl$.
\end{enumerate}
\end{lemma}
\begin{proof}
\ref{item:aminkcpt}: 
Except for \eqref{eq:aineqcyl} this follows from 
\eqref{eq:cylgbil} and Lemma \ref{lem:BilBasisCpt}.
The inequality \eqref{eq:aineqcyl} follows from 
$
\omega((\amink^\mu_{k})_{ij}\V_\mu)
=
\gBil^{\ideg}\big(\eGg_\outind^{\ideg}, \omega \eGg^{\ideg}_j\big)
$ and \eqref{eq:posineq_bulk} and $\amink^0_{k}=\one$. 
\ref{item:Aminkcpt}: This follows from 
$(\Amink^\mu_{k}(e_\ell^{\ng_1}))_{ij}
=
\gBil^{\new{k}}(\eGg^{k}_i, 
(\anchor_1)_\ell^\mu \eGg^{k}_j)$ and 
\ref{item:gaugesymNEW_bulk}.
\qed
\end{proof}
\begin{lemma}\label{lem:deltacausalcpt}
There exists $\deltabulk\in(0,1]$
such that for all $k=1,2$ and 
for all $u\in\R^{\ng_1}$ with $\sqrt{u^Tu}\le 2\deltabulk$ one has
\begin{subequations}
\begin{align}
\tfrac12 \one 
	&\le 
	\amink_k^{0}
	+ \Amink_k^{0}(u) 
	\le
	\new{2}\one
	&& \text{at every point on $\diamond_+$}
	\label{eq:dtauposcpt}\\
\tfrac{1}{20} \one  &\le 
	d\s\left( 
	\big(
	\amink_k^{\mu} 
	+ \Amink_k^{\mu}(u) 
	\big) \V_{\mu} 
	\right)
	&& \text{at every point on  $\Dspop_{\le1}$}
	\label{eq:dsposcpt}\\
0 &< 
	d\fctfol\left( 
	\big(
	\amink_k^{\mu} 
	+ \Amink_k^{\mu}(u) 
	\big) \V_{\mu} 
	\right)
	&& \text{at every point on  $\diamond_+$}
	\label{eq:dphiposcpt}
\end{align}
\end{subequations}
where $\fctfol$ is defined in \eqref{eq:phidef},
and $\s=2y^0+|\vec{y}|$ in \eqref{eq:sdef}.
\end{lemma}
\begin{proof}
We show separately that
\eqref{eq:dtauposcpt}, \eqref{eq:dsposcpt}, \eqref{eq:dphiposcpt}
hold for all sufficiently small $\deltabulk$.
\eqref{eq:dtauposcpt}:
By Lemma \ref{lem:aAminkcpt} we have $\amink^0_{k}=\one$.
Choose $\deltabulk$ sufficiently small so that 
$2\deltabulk\tsum_{\ell=1}^{\ng_1}
	\|\Amink_k^{0}(e_\ell^{\ng_1})\|
	\le \tfrac{1}{4}$ at every point on $\diamond_+$,
using the $\ell^2$-matrix norm.
Such a $\deltabulk$ exists because 
$\Amink_k^{0}(e_\ell^{\ng_1})$ is smooth on $\cyl$
and $\overline{\diamond}_+\subset\cyl$ is compact.
Then \eqref{eq:dtauposcpt} holds by a calculation analogous
to that in the proof of Lemma \ref{lem:qminkdelta}.

\eqref{eq:dsposcpt}:
\claimheader{Claim:} At every point on $\Dspop_{\le1}$,
\begin{subequations}
\begin{align}
d\s(\amink^\mu_{k}\V_\mu)
	&\ge \tfrac{1}{10}\one\label{eq:dsaC}\\
\|d\s(\Amink_k^{\mu}(u) \V_{\mu})\|
	&\lesssim 
	\sqrt{u^Tu}
	\label{eq:dsAcpt}
\end{align}
\end{subequations}
Proof of \eqref{eq:dsaC}: By \eqref{eq:aineqcyl}, 
\begin{align*}
d\s(\amink^\mu_{k}\V_\mu)
&\ge
\textstyle
\big(d\s(\V_0)-
\big(\sum_{i=1}^3|d\s(\V_i)|^2\big)^{\frac12}\big) \one \\
&\overset{(1)}{=}
\tfrac{1}{\ynullgen^2}
\big(
1 + \cos(\tau)\xi^4 - \sin(\tau) \sqrt{1-(\xi^4)^2}
\big) \one \\
&\overset{(2)}{\ge}
\tfrac{1}{\ynullgen^2} 
(1-\sin(\arctan(\tfrac23))) \one 
\overset{(3)}{\ge}
\tfrac{1}{10}\one
\end{align*}
where (1) holds by direct calculation using \eqref{eq:yycoords};
for (2) we use 
Remark \ref{rem:Dstau};
%
%
for (3) we use $\ynullgen =\cos(\tau)+\xi^4 \le2$.

Proof of \eqref{eq:dsAcpt}:
Let $(\Amink_k^{\mu}(u))_{ij}$ be the components
of the matrix $\Amink_k^{\mu}(u)$.
Expand $u= u_{\ell} e^{\ng_1}_\ell$ with implicit 
sum over $\ell$.
Analogously to \eqref{eq:Aanchorcalc_i0} one has
\begin{align*}
d\s\big((\Amink_k^{\mu}(u))_{ij} \V_{\mu}\big)
=
u_\ell \gBil^{\new{k}}\big(\eGg^{k}_i, \anchorg(\eg^1_{\ell})(\s)\eGg^{k}_j \big)
\end{align*}
For each $\ell$ one has 
$\eg^1_{\ell}=(\omega_\ell\otimes \KilBasis_\ell)\oplus u_{\I,\ell}\Ieps$,
where $\omega_\ell$ is either zero or one of $\Vd^\mu$, 
and $\KilBasis_\ell$ is a basis element \eqref{eq:Kilbas},
see \eqref{eq:gbasis_bulk}.
Then, using \eqref{eq:anchorgdef},
\begin{align}\label{eq:dsAA}
d\s\big((\Amink_k^{\mu}(u))_{ij} \V_{\mu}\big)
&=
u_\ell \gBil^{\new{k}}\left(\eGg^{k}_i, \omega_\ell \eGg^{k}_j \right) \KilBasis_\ell(\s)
\end{align}
Using \eqref{eq:TBy} one obtains
$|\KilBasis_{\ell}(\s)| \lesssim 1$ on $\Dspop_{\le1}$.
Thus for each $i,j$,
\begin{align*}
|d\s\big((\Amink_k^{\mu}(u))_{ij} \V_{\mu}\big)|
\lesssim 
|u_\ell|
|\gBil^{\new{k}}\left(\eGg^{k}_i, \omega_\ell \eGg^{k}_j\right)|
\lesssim
\sqrt{u^Tu}
\end{align*}
where the last inequality holds by 
\eqref{eq:cylgbil} and Lemma \ref{lem:BilBasisCpt}.
This implies \eqref{eq:dsAcpt}.

Choose $\deltabulk\in(0,1]$ sufficiently small so that 
for all $u\in\R^{\ng_1}$ with $\sqrt{u^Tu}\le 2\deltabulk$,
\[ 
\|d\s\big((\Amink_k^{\mu}(u)) \V_{\mu}\big)\| \le \tfrac{1}{20}
\]
Such a $\deltabulk$ exists by \eqref{eq:dsAcpt}.
Together with \eqref{eq:dsaC} this yields \eqref{eq:dsposcpt}.

\eqref{eq:dphiposcpt}:
By \eqref{eq:aineqcyl} and \eqref{eq:Phidef},
\begin{equation}\label{eq:phiaV}
	\Phi \one 
	\le
	d\fctfol
	\big(
	\amink_k^{\mu} 
	\V_{\mu}\big)
\end{equation}
Analogously to \eqref{eq:dsAA} one has
\[ 
d\fctfol\big((\Amink_k^{\mu}(u))_{ij} \V_{\mu}\big)
=
u_\ell \gBil^{\new{k}}\left(\eGg^{k}_i, \omega_\ell \eGg^{k}_j \right) \KilBasis_\ell(\fctfol)
\]
where we sum over $\ell$,
and where for each $\ell$, 
the one-form $\omega_\ell$ is either zero or one of $\Vd^{\mu}$,
and $\KilBasis_\ell$ is a basis element \eqref{eq:Kilbas}.
Together with \eqref{eq:KilBounds} this yields
\begin{align}\label{eq:phiAV}
|d\fctfol\big((\Amink_k^{\mu}(u))_{ij} \V_{\mu}\big)|
\lesssim
\sqrt{u^Tu}\, \Phi
\end{align}
By \eqref{eq:phiaV} and \eqref{eq:phiAV}, 
and by choosing $\deltabulk\in(0,1]$ sufficiently small,
one obtains
\[ 
\tfrac12 \Phi \one 
\le
d\fctfol\left( 
	\big(
	\amink_k^{\mu} 
	+ \Amink_k^{\mu}(u) 
	\big) \V_{\mu} 
	\right)
\]
By \eqref{eq:Phi/h} we have $\Phi>0$ on $\diamond_+$,
thus \eqref{eq:dphiposcpt} follows.\qed
\end{proof}

\subsection{Main existence result}
\label{sec:existence_bulk}

We state and prove Proposition \ref{prop:MainCpt}, 
the main result of Section \ref{sec:bulk}.

Let $(\gxG(\cyl),\gBil)$ be the gauge in Definition \ref{def:gauge_bulk}.
Denote by $\gxG(\diamond_+)$ the space of sections of 
$\gxG$ over $\diamond_+$, c.f.~Remark \ref{rem:gdiamond}.
We use the norms in Definition \ref{def:bulknorms}.
\begin{prop}\label{prop:MainCpt}
For all 
\[ 
\NN\in\Z_{\ge\new{7}}
\qquad
s_*\in(0,1]
\]
there exist 
$\ClargeCpt>0$
and
$\CsmallCpt\in(0,1]$ such that for all 
$v \in \gx^1(\diamond_+)$, if
\begin{enumerate}[({j}1)]
\item \label{item:vconstraintsD}
$\Pconstraints(v|_{\tau=0})=0$, see Definition \ref{def:Pconstraints}
\item \label{item:MCvsp}
$\dg v + \frac12[v,v]=0$ on $\Dspop_{\le s_*}$
\item \label{item:vsupp}
$v|_{\tau\ge\frac\pi2}=0$
\item \label{item:vNewBound}
$\int_{0}^{\new{\frac\pi2}}\|v\|_{\sH^{\NN+1}(\diamond_{\tau',s_*})}d\tau'
	\le \CsmallCpt$
\end{enumerate}
then there exists $c\in \gxG^1(\diamond_+)$ such that 
\begin{subequations}\label{eq:Cptcprop}
\begin{align}
\dg(v+c) + \tfrac12[v+c,v+c] 
	&= 0 \label{eq:ceqD}\\
c|_{\tau=0} 
	&=0 \label{eq:cdataD}\\
c|_{\Dspop_{\le s_*}} 
	&= 0 \label{eq:czeroD}
\end{align}
\end{subequations}
Furthermore:
\begin{itemize}
\item 
\new{\textbf{Part 0.} $c$ is unique.}
\item 
\textbf{Part 1.}
For all $\tau\in[0,\pi)$:
\begin{subequations}\label{eq:cbulkP1}
\begin{align}
\|c\|_{\H^{\NN}(\diamond_{\tau,s_*})} 
	&\le \ClargeCpt 
	\tint_0^{\tau}\|v\|_{\sH^{\NN+1}(\diamond_{\tau',s_*})}d\tau'
	\label{eq:Hcbulk}
	\\
\|c\|_{\sH^{\NN}(\diamond_{\tau,s_*})} 
	&\le \ClargeCpt 
	(\tint_0^{\tau}\|v\|_{\sH^{\NN+1}(\diamond_{\tau',s_*})}d\tau'
	+
	\|v\|_{\sH^{\NN}(\diamond_{\tau,s_*})})
	\label{eq:SHcEEDNEW}
\end{align}
\end{subequations}
Moreover, 
$\|v\|_{\sH^{\NN}(\diamond_{\tau,s_*})}
\lesssim_{\NN,\sfix}
\tint_0^{\frac\pi2}\|v\|_{\sH^{\NN+1}(\diamond_{\tau',s_*})}d\tau'$.
\item 
\textbf{Part 2.}
For every $\kk\in\Z_{\ge\NN}$ and every $\CHigherCpt\new{>}0$, if 
\begin{enumerate}[({j}1),resume]
\item \label{item:vHigherNewBound}
$\tint_0^{\new{\frac\pi2}}\|v\|_{\sH^{\kk+1}(\diamond_{\tau',s_*})}d\tau'
	\le \CHigherCpt$
\end{enumerate}
then for all $\tau\in[0,\pi)$:
\begin{subequations}\label{eq:bulkh12}
\begin{align}
\|c\|_{\H^{\kk}(\diamond_{\tau,s_*})} 
	&\lesssim_{\kk,s_*,\CHigherCpt} 
	\tint_0^{\tau}
	\|v\|_{\sH^{\kk+1}(\diamond_{\tau',s_*})}d\tau'
	\label{eq:chigherEECptNEW}\\
\|c\|_{\sH^{\kk}(\diamond_{\tau,s_*})} 
	&\lesssim_{\kk,s_*,\CHigherCpt} 
	\tint_0^{\tau}
	\|v\|_{\sH^{\kk+1}(\diamond_{\tau',s_*})}d\tau'
	+
	\|v\|_{\sH^{\kk}(\diamond_{\tau,s_*})} 
	\label{eq:chigherEECptNEW1}
\end{align}
\end{subequations}
Moreover, 
$\|v\|_{\sH^{k}(\diamond_{\tau,s_*})}
\lesssim_{k,\sfix}
\tint_0^{\frac\pi2}\|v\|_{\sH^{k+1}(\diamond_{\tau',s_*})}d\tau'$.
\end{itemize}
\end{prop}
The proof of Proposition \ref{prop:MainCpt} is at the end of this section.
Instead of constructing the solution $c$ directly on $\diamond_+$,
we use the exhaustion 
$\diamond_+=\cup_{\taufix\in[\frac\pi2,\pi)}\diamond^{\taufix}$
in Lemma \ref{lem:Dexhaustion},
and construct $c$ separately on each $\diamond^{\taufix}$, with estimates
that are uniform in (i.e.~independent of) $\taufix$.
The advantage is that $v$ and $c$ are smooth on 
$\diamond^{\taufix} = \overline{\diamond^{\taufix}}\setminus\spaceinf$.
The construction of $c$ on these smaller sets is in the next lemma.
Let $\gxG(\diamond^{\tau_*})$ be the sections of $\gxG$ over $\diamond^{\tau_*}$.
We use the norms in Definition \ref{def:taufixnorms}.
\begin{lemma}\label{lem:bulkpropaux}
For all 
$\NN\in\Z_{\ge\new{7}}$, $\sfix\in(0,1]$
there exist 
$\ClargeCpt>0$,
$\CsmallCpt\in(0,1]$
such that for all $v \in \gx^1(\diamond_+)$, 
if 
\ref{item:vconstraintsD}, 
\ref{item:MCvsp}, 
\ref{item:vsupp}, 
\ref{item:vNewBound} hold
then for all $\taufix\in[\frac\pi2,\pi)$
the following holds:
There exists $c\in \gxG^1(\diamond^{\taufix})$ 
that satisfies \eqref{eq:Cptcprop}.
Furthermore:
\begin{itemize}
\item 
\new{\textbf{Part 0.} $c$ is unique.}
\item 
\textbf{Part 1.}
For all $\tau\in[0,\taufix]$:
\begin{subequations}\label{eq:cNhsNH}
\begin{align}
&\|c\|_{\H^{\NN}(\diamond_{\tau,\sfix}^{\taufix})} 
	\le \ClargeCpt 
	\tint_0^{\tau}\|v\|_{\sH^{\NN+1}(\diamond_{\tau',\sfix}^{\taufix})}d\tau'
	\label{eqq:HNtau}
	\\
&\|c\|_{\sH^{\NN}(\diamond_{\tau,\sfix}^{\taufix})} 
	\le \ClargeCpt 
	(\tint_0^{\tau}\|v\|_{\sH^{\NN+1}(\diamond_{\tau',\sfix}^{\taufix})}d\tau'
	+
	\|v\|_{\sH^{\NN}(\diamond_{\tau,\sfix}^{\taufix})})
	\label{eqq:sHNtau}
\end{align}
\end{subequations}
\item 
\new{\textbf{Part 2.}}
For every $\kk\in\Z_{\ge\NN}$ and $\CHigherCpt\new{>}0$, if 
\ref{item:vHigherNewBound} holds then for all $\tau\in[0,\taufix]$:%
\begin{subequations}\label{eq:bulkauxhall}
\begin{align}
\|c\|_{\H^{\kk}(\diamond_{\tau,\sfix}^{\taufix})} 
	&\lesssim_{\kk,\sfix,\CHigherCpt} 
	\tint_0^{\tau}
	\|v\|_{\sH^{\kk+1}(\diamond_{\tau',\sfix}^{\taufix})}d\tau'
	\label{eq:bulktaufixhigherH}\\
\|c\|_{\sH^{\kk}(\diamond_{\tau,\sfix}^{\taufix})} 
	&\lesssim_{\kk,\sfix,\CHigherCpt} 
	\tint_0^{\tau}
	\|v\|_{\sH^{\kk+1}(\diamond_{\tau',\sfix}^{\taufix})}d\tau'
	+
	\|v\|_{\sH^{\kk}(\diamond_{\tau,\sfix}^{\taufix})} 
	\label{eq:bulktaufixhighersH}
\end{align}
\end{subequations}
\end{itemize}
\end{lemma}
Before we prove Lemma \ref{lem:bulkpropaux}, we will derive 
the relevant energy estimates in Lemma \ref{lem:EEbulk} below.
For this consider the following necessary subsystem of \eqref{eq:ceqD}:
\begin{equation}\label{eq:cSHSD}
\gBil^{1}\big(\,\cdot\,, \dg(v+c)+\tfrac12[v+c,v+c]\big) = 0
\end{equation}
By Lemma \ref{lem:translationofeqcyl}, 
the system \eqref{eq:cSHSD} is equivalent to
\begin{equation}\label{eq:cSHSDMatrices}
(\AmatLin^\mu+\AmatBil^\mu(c)) \V_\mu c \;=\; \LMat c 
+ \BSHS(c,c) + \Fvec
\end{equation}
where we use the identification \eqref{eq:identify_g_vec_bulk}
and define, using Definition \ref{def:SHSarrays_bulk}, 
\begin{align}\label{eq:aALBFbulk}
\begin{aligned}
\AmatLin^\mu
&= 
\amink_1^{\mu}
+ \Amink_1^{\mu}(v)\\
\AmatBil^\mu(\cdot)
&=
\Amink_1^{\mu}(\ginj_1 \,\cdot\,)\\
\LMat
&=
\Lmink_1   
- \Amink_{11}^\mu(\ginj_1 \,\cdot\,) \V_{\mu} v
+ \Bmink_1(v,\ginj_1 \,\cdot\,)\\
\BSHS
&=
\tfrac12\Bmink_1(\ginj_1\,\cdot\,,\ginj_1\,\cdot\,)\\
\Fvec
&=
\SFmink (\dg v+\tfrac12[v,v])
\end{aligned}
\end{align}
Here and below, 
the restriction of the maps in Definition \ref{def:SHSarrays_bulk}
to suitable subsets of $\cyl$ is left implicit.
Beware that \eqref{eq:aALBFbulk} depend on $v$.
%


For the remainder of this section we fix 
$\deltabulk\in(0,1]$
as in Lemma \ref{lem:deltacausalcpt}.

\begin{lemma}\label{lem:EEbulk}
For all 
\begin{equation}\label{eq:EEconstbulk}
\NN\in\Z_{\ge\new{1}}
\qquad
\sfix\in(0,1]
\qquad
\bbulk > 0
\end{equation}
there exists $\Cbulk>0$ such that for all
\[ 
\taufix\in [\tfrac\pi2,\pi)
\qquad
\taum\in (0,\taufix]
\]
and all
\begin{equation}\label{eq:cveetau}
c\;\in\; C^\infty(\diamond_{\le\taum,\sfix}^{\taufix},\R^{\ngG_1})
\qquad
v\in C^\infty(\diamond_{\le\taum,\sfix}^{\taufix},\R^{\ng_1})
\end{equation}
the following holds. Associated to $v$ define the maps \eqref{eq:aALBFbulk}.
If 
\begin{subequations}\label{eq:EEbulkassp}
\begin{align}
(\AmatLin^\mu+\AmatBil^\mu(c)) \V_\mu c 
	&= \LMat c + \BSHS(c,c) + \Fvec 
	\label{eq:ceqbulk}\\
c|_{\tau=0} 
	&=0
	\label{eq:cdatabulk} \\
c|_{\Dspop_{\le s_*}} 
	&= 0 
	\label{eq:czerobulk}\\
\sqrt{c^Tc},\ \sqrt{v^Tv}
	&\le 
	\deltabulk
	\;\ \quad
	\text{on $\diamond_{\le\taum,\sfix}^{\taufix}$}
	\label{eq:cvdeltabulk}\\
	\|c\|_{C^{\lfloor \frac{\NN+1}{2} \rfloor}(\diamond_{\tau,\sfix}^{\taufix})}
	&\le 
	\bbulk
	\qquad\text{for all $\tau\in[0,\taum]$}
	\label{eq:cCbulk}\\
\|v\|_{\sC^{\lfloor \frac{\NN+1}{2} \rfloor}(\diamond_{\tau,\sfix}^{\taufix})}
	&\le 
	\bbulk
	\qquad\text{for all $\tau\in[0,\taum]$}
	\label{eq:vCbulk}
\end{align}
\end{subequations}
Then:
\begin{itemize}
\item \textbf{Part 1.}
For all $k\in\Z_{\ge0}$ with $k\le\NN$ and all $\tau\in[0,\taum]$:
\begin{subequations}\label{eq:sCHtoCHbulk}
\begin{align}
\|c\|_{\sC^k(\diamond_{\tau,\sfix}^{\taufix})} 
	&\le \Cbulk(\|c\|_{C^k(\diamond_{\tau,\sfix}^{\taufix})} 
		+ \|v\|_{\sC^k(\diamond_{\tau,\sfix}^{\taufix})}) 
	\label{eq:sCtoCbulk}\\
\|c\|_{\sH^k(\diamond_{\tau,\sfix}^{\taufix})} 
	&\le \Cbulk(\|c\|_{\H^k(\diamond_{\tau,\sfix}^{\taufix})} 
		+ \|v\|_{\sH^k(\diamond_{\tau,\sfix}^{\taufix})})
	\label{eq:sHtoHbulk}
\end{align}
\end{subequations}
\item 
\textbf{Part 2.} For all $\tau\in[0,\taum]$:
\begin{align}\label{eq:ekktauint}
\|c\|_{\H^{\NN}(\diamond_{\tau,\sfix}^{\taufix})}
&\le
\Cbulk 
\tint_{0}^{\tau} \|v\|_{\sH^{\NN+1}(\diamond_{\tau',s_*}^{\taufix})}  d\tau'
\end{align}
\end{itemize}
\end{lemma}
Recall that the sets used in \eqref{eq:cveetau} 
are closed, see \eqref{eq:dtauletaus}.
\begin{proof}[of Lemma \ref{lem:EEbulk}]
Instead of specifying $\Cbulk$ upfront, we will make 
finitely many admissible largeness assumptions on $\Cbulk$,
where admissible means that they depend only on \eqref{eq:EEconstbulk}
(the dependencies on the fixed maps in
\eqref{eq:minkarrays_bulk} will be suppressed).
We will repeatedly use the following fact:
\begin{align}\label{eq:mapunifbounds}
\begin{aligned}
&\text{The components of the maps \eqref{eq:minkarrays_bulk},
and their derivatives}\\[-1mm]
&\text{relative to $\V_0,\dots,\V_3$, are bounded in absolute value on \smash{$\diamond_{\le\taum,\sfix}^{\taufix}$},} \\[-1mm]
&\text{and the bounds are independent of $\sfix,\taufix,\taum$.}
\end{aligned}
\end{align}
This holds because the maps \eqref{eq:minkarrays_bulk} are smooth on $\cyl$.

By Lemma \ref{lem:deltacausalcpt}
(with 
$k=1$, $u=v+\ginj_1c$)
and \eqref{eq:cvdeltabulk},
on \smash{$\diamond_{\le\taum,\sfix}^{\taufix}$} one has:%
\begin{subequations}
\begin{align}
\tfrac12\one
	&\le \AmatLin^0+\AmatBil^0(c) \le 2\one
	\label{eq:dtaupos1Ap}\\
0 &< d\fctfol\big((\AmatLin^\mu+\AmatBil^\mu(c)) \V_\mu\big)
	\label{eq:cnullnonneg1Ap}	
\end{align}
\end{subequations}

We make definitions analogous to \eqref{eq:indsetdef}:
For $k_0,k\in\Z_{\ge0}$, let 
$$
\indset_{k_0,k}\subset\{0,1,2,3\}^{k_0+k}
$$
be given by all $\indI=(i_1,\dots,i_{k_0+k})$ such that 
precisely $k_0$ of the $i_1,\dots,i_{k_0+k}$ are equal to $0$.
Further set
$\indset_{k_0,\le k} = \cup_{k'\le k} \indset_{k_0,k'}$
and
$\indset_{\le k_0,\le k} = \cup_{k_0'\le k_0} \indset_{k_0',\le k}$.

Analogously to \eqref{eq:indexnotation}, for
$\indI=(i_1,\dots,i_{k})$ and functions $f$ we denote
\begin{align}\label{eq:Vfi}
\V^{\indI} &= \V_{i_1}\cdots \V_{i_{k}} & 
|\indI| &= k & 
f_{\indI} = \V^{\indI} f = \V_{i_1}\cdots \V_{i_{k}}f
\end{align}

\proofheader{Proof of Part 1.}
This is similar to the proof of Part 1 of 
Lemma \ref{lem:Nonlindiffenergyineq}, so we are brief.
Using \eqref{eq:ceqbulk}, \eqref{eq:dtaupos1Ap}, \eqref{eq:cCbulk}, \eqref{eq:vCbulk}, \eqref{eq:mapunifbounds}
one derives the following pointwise bound,
for all \new{$k_0\in\Z_{\ge1}$}, $j\in\Z_{\ge0}$ with $k_0+j \le\NN$
and $\indI\in\indset_{k_0,j}$:
\begin{align*}
\|c_{\indI}\|
	&\lesssim_{\NN,\bbulk}
	\textstyle
	\sum_{\indJ\in\indset_{\le k_0-1,\le j+1}} \|c_{\indJ}\|
	+
	\sum_{\indJ\in\indset_{k_0,\le j-1}} \|c_{\indJ}\|
	+
	\sum_{\substack{\indJ\in\indset_{\le k_0,\le j+1} \\
		|\indJ|\le k_0+j}} \|v_{\indJ}\|
	\\
	&
	\qquad
	+\textstyle
	\sum_{
	\substack{
	\indJ,\indK \in \indset_{\le k_0,\le j+1}\\
	|\indJ|+|\indK|\le k_0+j\\
	|\indJ|,|\indK|\le k_0+j-1\\
	n_0(\indJ)+n_0(\indK) \le k_0
	}}
	\big(
	\|c_{\indJ}\|
	+
	\|v_{\indJ}\|
	+
	\sum_{\mu=0}^3\|\V^\mu v_{\indJ}\|
	\big)\|c_{\indK}\|
\end{align*}
where $n_0(\indJ)$ is equal to the number of entries in $\indJ$
that are equal to $0$, 
and where $\|{\cdot}\|$ denotes the $\ell^2$-vector norm.
To derive this,
one must in particular use the fact that for $\indJ\in\indset_{\le k_0-1,\le j}$:
\begin{align}\label{eq:Fv}
\|\Fvec_{\indJ}\| 
=
\|\V^{\indJ} \SFmink (\dg v+\tfrac12[v,v])\|
\lesssim_{\NN,\bbulk}
\tsum_{
\substack{\indJ\in\indset_{\le k_0,\le j+1} \\
		|\indJ|\le k_0+j}}
\|v_{\indJ}\|
\end{align}
which uses \eqref{eq:vCbulk}, \eqref{eq:mapunifbounds} and the fact that $\dg$ and $[\cdot,\cdot]$
are smooth first order linear respectively
bilinear differential operators on $\cyl$.
Now \eqref{eq:sCtoCbulk} and \eqref{eq:sHtoHbulk} follow 
by an inductive argument
similar to \eqref{eq:Pk0kt} respectively \eqref{eq:inductionHk},
using \eqref{eq:cCbulk}, \eqref{eq:vCbulk},
and an admissible largeness assumption on $\Cbulk$.

\proofheader{Proof of Part 2.}
By \eqref{eq:sCtoCbulk} and $\lfloor\frac{\NN+1}{2}\rfloor\le\NN$,
for all $\tau\in[0,\taum]$:
\begin{align}\label{eq:scCbulk}
\begin{aligned}
\|c\|_{\sC^{\lfloor \frac{\NN+1}{2} \rfloor}(\diamond_{\tau,\sfix}^{\taufix})}
	&\lesssim_{\NN,\sfix,\bbulk}
	\|c\|_{C^{\lfloor \frac{\NN+1}{2} \rfloor}(\diamond_{\tau,\sfix}^{\taufix})}
	+
	\|v\|_{\sC^{\lfloor \frac{\NN+1}{2} \rfloor}(\diamond_{\tau,\sfix}^{\taufix})}
	\lesssim_{\bbulk}
	1
\end{aligned}
\end{align}
where the last step holds by \eqref{eq:cCbulk} and \eqref{eq:vCbulk}.

For $k_0\in\{0,1\}$, $k\le\NN$, $\indI\in\indset_{k_0,k}$ 
and $\tau\in[0,\taum]$ define\footnote{%
Beware that $\diamond^{\tau_*}_{\sfix}$ has a corner 
along the intersection of 
$\Dspop_{\sfix/6}$ and $\nullinf^{\tau_*}$.
Still $E_{\indI}(\tau)$ is differentiable in $\tau$ by \eqref{eq:czerobulk}.
}\textsuperscript{,}\footnote{%
Beware that the index $\indI$ is used in two different ways,
in $c_{\indI}$ it stands for the derivative of $c$ (see \eqref{eq:Vfi}),
while in $E_{\indI}$, and in $\current_{\indI}$ below, it is part of the name.} 
\begin{align*}
E_{\indI}(\tau) 
	&= 
	\tint_{\diamond_{\tau,\sfix}^{\tau_*}} 
	c_{\indI}^T (\AmatLin^0+\AmatBil^0(c)) c_{\indI}
	\, \mucylS\\
	E_{k_0,\le k}(\tau) 
	&=
	\tsum_{\indI\in\indset_{k_0,\le k}} E_{\indI}(\tau)\\
	e_{\le\NN}(\tau)
	&=
	\sqrt{E\smash{_{0,\le\NN}}(\tau)}
\end{align*}
By \eqref{eq:dtaupos1Ap} and \eqref{eq:sHtoHbulk}:
\begin{align}
\sqrt{E\smash{_{1,\le\NN-1}}(\tau)}
&\lesssim_{\NN}
\|c\|_{\sH^{\NN}(\diamond^{\taufix}_{\tau,\sfix})} \nonumber\\
&\lesssim_{\NN,\sfix,\bbulk}
\|c\|_{\H^{\NN}(\diamond_{\tau,\sfix}^{\taufix})} 
		+ \|v\|_{\sH^{\NN}(\diamond_{\tau,\sfix}^{\taufix})} \nonumber\\
&\lesssim_{\NN}
e_{\le\NN}(\tau)
		+ \|v\|_{\sH^{\NN}(\diamond_{\tau,\sfix}^{\taufix})}
		\label{eq:E1E0bulk}
\end{align}

Let $\indI\in\indset_{0,\le\NN}$.
Define the current 
\[ 
\current_{\indI} 
	= c_{\indI}^T (\AmatLin^\mu+\AmatBil^\mu(c))c_{\indI} \V_\mu
\]
For each $\tau\in{\new(}0,\taum]$,
\begin{align}
\label{eq:intmucylform}
\tint_{\diamond_{\le\tau,\sfix}^{\taufix}} \div_{\mucyl}(\current_{\indI})\mucyl
=
\tint_{\diamond_{\le\tau,\sfix}^{\taufix}} \div_{\mucyl}(\current_{\indI})\mucylform
\end{align}
where, on the right hand side, we integrate
relative to the positive volume form $\mucylform = \Vd^0\wedge\dots\wedge\Vd^3$,
using the fixed orientation on $\cyl$, see Remark \ref{rem:orientation}.
We have
\begin{align*}
\p\diamond_{\le\tau,\sfix}^{\taufix}
=
\diamond_{0,\sfix}^{\taufix}
\;\cup\;
(\p\diamond_{\le\tau,\sfix}^{\taufix}\cap\Dspop_{\frac\sfix6})
\;\cup\;
(\p\diamond_{\le\tau,\sfix}^{\taufix}\cap\nullinf^{\taufix})
\;\cup\;
\diamond_{\tau,\sfix}^{\taufix}
\end{align*}
where for small $\tau$, the third boundary component is empty.
\new{The union is disjoint up to lower-dimensional sets.}
The function $c_{\indI}$ vanishes on the first boundary component,
by \eqref{eq:cdatabulk} and the fact that $\V^{1},\V^2,\V^3$
are tangential to $\tau=0$.
Further $c_{\indI}$ vanishes on the second boundary component by \eqref{eq:czerobulk}.
Thus Stokes' theorem, applied to the right hand side of \eqref{eq:intmucylform}, yields
\begin{equation}
\label{eq:EEcptqForms}
\tint_{\diamond_{\le\tau,\sfix}^{\taufix}} \div_{\mucyl}(\current_{\indI})\mucyl
=
\tint_{\diamond_{\tau,\sfix}^{\taufix}} \intermult_{\current_{\indI}}\mucylform
+
\tint_{\p\diamond_{\le\tau,\sfix}^{\taufix}\cap\nullinf^{\taufix}} \intermult_{\current_{\indI}}\mucylform
\end{equation}
where, on the right hand side, we use the induced orientation.
Note that the equality also holds for $\tau=0$, then both sides are zero.
One has:
\begin{itemize}
\item 
$\tint_{\diamond_{\tau,\sfix}^{\taufix}} \intermult_{\current_{\indI}}\mucylform
=\tint_{\diamond_{\tau,\sfix}^{\taufix}} 
c_{\indI}^T (\AmatLin^0+\AmatBil^0(c)) c_{\indI} 
\,\intermult_{\p_{\tau}}\mucylform = E_{\indI}(\tau)$,
using $|\intermult_{\p_{\tau}}\mucylform|=\mucylS$
and the fact that $\intermult_{\p_{\tau}}\mucylform$
is positive with respect to the induced orientation.
\item 
The second term on the left hand side of \eqref{eq:EEcptqForms}
is increasing in $\tau$ by \eqref{eq:cnullnonneg1Ap} and 
by $d\fctfol(\p_\tau)>0$, see Lemma \ref{lem:dphiproperties}.
%
%
%
%
\end{itemize}
Thus differentiating \eqref{eq:EEcptqForms} in $\tau$, 
and using Fubini and $\mucyl=|d\tau|\mucylS$, 
we obtain that for all $\tau\in[0,\taum]$:
$$
\tfrac{d}{d\tau}E_{\indI}(\tau)
\le
\tfrac{d}{d\tau}
\tint_{\diamond_{\le\tau,\sfix}^{\taufix}} \div_{\mucyl}(\current_{\indI})\mucyl
=
\tint_{\diamond^{\taufix}_{\tau,\sfix}} \div_{\mucyl}(\current_{\indI})\mucylS
\le
\tint_{\diamond^{\taufix}_{\tau,\sfix}} |\div_{\mucyl}(\current_{\indI})|\mucylS
$$
Abbreviate $\Amat^\mu(c) = \AmatLin^\mu + \AmatBil^\mu(c)$
and $\Amat(c) = \Amat^\mu(c)\V_\mu$.
We have 
\begin{align*}
\div_{\mucyl}(\current_{\indI})
&=
2 c_{\indI}^T \Amat(c)  c_{\indI}
+
c_{\indI}^T\div_{\mucyl}(\Amat(c))c_{\indI}
\end{align*}
using the fact that the matrices $\Amat^\mu(c)$ are
symmetric, by Lemma \ref{lem:aAminkcpt}.
Thus
\begin{align*}
\begin{aligned}
\tfrac{d}{d\tau} E_{\indI}(\tau)
&\lesssim 
\|\Amat(c) c_{\indI}\|_{L^2(\diamond^{\taufix}_{\tau,\sfix})}
\|c_{\indI}\|_{L^2(\diamond^{\taufix}_{\tau,\sfix})}\\
&\quad+
\|\div_{\mucyl}(\Amat(c))\|_{C^0(\diamond^{\taufix}_{\tau,s_*})}
\|c_{\indI}\|_{L^2(\diamond^{\taufix}_{\tau,\sfix})}^2
\end{aligned}
\end{align*}
where the $L^2$-norm is defined with respect to $\mucylS$.
By \eqref{eq:dtaupos1Ap} and $\indI\in\indset_{0,\le\NN}$ we have
$\|c_{\indI}\|_{L^2(\diamond^{\taufix}_{\tau,\sfix})}\lesssim_{\NN} e_{\le \NN}(\tau)$. Using \eqref{eq:mapunifbounds} and then 
\eqref{eq:scCbulk}, \eqref{eq:vCbulk} and $\NN\ge1$:
\begin{equation*}
\|\div_{\mucyl}(\Amat(c))\|_{C^0(\diamond_{\tau,s_*}^{\taufix})} 
\lesssim
1+\|v\|_{\sC^1(\diamond^{\taufix}_{\tau,s_*})}
 +\|c\|_{\sC^1(\diamond^{\taufix}_{\tau,s_*})}
\lesssim_{\sfix,\bbulk}
1
\end{equation*}
Thus for all $\tau\in[0,\taum]$:
\begin{align}
\tfrac{d}{d\tau} E_{\indI}(\tau)
&\lesssim_{\NN,\sfix,\bbulk}
e_{\le \NN}(\tau)\big(\|\Amat(c) c_{\indI}\|_{L^2(\diamond^{\taufix}_{\tau,\sfix})}
+
e_{\le \NN}(\tau)\big)
\label{eq:d/dtEIc}
\end{align}
Differentiating \eqref{eq:ceqDapriori} with respect to $\V^{\indI}$ yields
\begin{align*}
\Amat(c)c_{\indI} 
	&= 
	-[\V^{\indI}, \Amat(c)] c 
	+\V^{\indI}\LMat c 
	+\V^{\indI}\BSHS(c,c) 
	+\V^{\indI}\Fvec
\end{align*}
We claim that for all $\tau\in[0,\taum]$:
\begin{align}\label{eq:EqEst}
\|\Amat(c)c_{\indI}\|_{L^2(\diamond^{\taufix}_{\tau,\sfix})}
\lesssim_{\NN,\sfix,\bbulk}
e_{\le\NN}(\tau)+\|v\|_{\sH^{\NN+1}(\diamond^{\taufix}_{\tau,\sfix})}
\end{align}
Proof of \eqref{eq:EqEst}:
\begin{itemize}
\item 
By definition of $\Amat^\mu(c)$,
at every point on $\diamond^{\taufix}_{\tau,\sfix}$ we have
\begin{equation*}
\|[\V^{\indI}, \Amat^\mu(c) \V_\mu] c \|
	\le
	\|[\V^{\indI}, \amink_1^\mu \V_\mu] c \|
	+
	\|[\V^{\indI}, \Amink_1^{\mu}(v+\ginj_1 c) \V_\mu] c \|
\end{equation*}
Using the Leibniz rule and \eqref{eq:mapunifbounds},
\begin{align*}
\|[\V^{\indI}, \amink_1^\mu \V_\mu] c \|
	&\lesssim_{\NN}
	\tsum_{\indJ \in \indset_{0,\le \NN}\cup\indset_{1,\le \NN-1}} \|c_{\indJ}\|\\
\|[\V^{\indI}, \Amink_1^\mu(v+\ginj_1 c) \V_\mu] c \|
	&\lesssim_{\NN}
	\tsum_{\substack{
	\indJ\in\indset_{0,\le\NN}\\ 
	\indK\in\indset_{0,\le\NN}\cup\indset_{1,\le\NN-1}\\
	|\indJ|+|\indK|\le\NN+1}}
	(\|v_{\indJ}\|+\|c_{\indJ}\|)\|c_{\indK}\|
\end{align*}
Taking the $L^2$-norm and using \eqref{eq:dtaupos1Ap},
\eqref{eq:scCbulk}, \eqref{eq:vCbulk} and then \eqref{eq:E1E0bulk},
\begin{align*}
\|[\V^{\indI}, \Amat^\mu(c) \V_\mu] c \|_{L^2(\diamond^{\taufix}_{\tau,\sfix})}
	&\lesssim_{\NN}
	\sqrt{E\smash{_{0,\le\NN}}(\tau)}
	+ \sqrt{E\smash{_{1,\le\NN-1}}(\tau)}
	+ \|v\|_{\sH^{\NN}(\diamond_{\tau,\sfix}^{\taufix})}\\
	&\lesssim_{\NN,\sfix,\bbulk}
	e_{\le\NN}(\tau) + \|v\|_{\sH^{\NN}(\diamond_{\tau,\sfix}^{\taufix})}
\end{align*}
\item 
Recall from \eqref{eq:aALBFbulk} that $\LMat$ depends on $v$ and its first derivatives.
With \eqref{eq:mapunifbounds} we obtain that 
at every point on $\diamond^{\taufix}_{\tau,\sfix}$:
\begin{align*}
\|\V^{\indI} \LMat c\| 
\lesssim_{\NN}
\tsum_{
\substack{
\indJ\in \indset_{0,\le\NN+1}\cup \indset_{1,\le\NN}\\
\indK\in \indset_{0,\le\NN}\\
|\indJ|+|\indK| \le\NN+1
}}
(1+\|v_{\indJ}\|) \|c_{\indK}\|
\end{align*}
Thus with \eqref{eq:dtaupos1Ap},
\eqref{eq:cCbulk}, \eqref{eq:vCbulk} we obtain
\begin{align*}
\|\V^{\indI} \LMat c\|_{L^2(\diamond^{\taufix}_{\tau,\sfix})}
&\lesssim_{\NN,\bbulk}
e_{\le\NN}(\tau) + \|v\|_{\sH^{\NN+1}(\diamond^{\taufix}_{\tau,\sfix})}
\end{align*}

\item 
Using \eqref{eq:mapunifbounds}, 
\eqref{eq:dtaupos1Ap} and
\eqref{eq:cCbulk},
$\|\V^{\indI}\BSHS(c,c)\|_{L^2(\diamond^{\taufix}_{\tau,\sfix})}
\lesssim_{\NN,\bbulk}
e_{\le\NN}(\tau)$.
\item 
Similarly to \eqref{eq:Fv}, 
$
\|\V^{\indI}\Fvec\|_{L^2(\diamond^{\taufix}_{\tau,\sfix})}
\lesssim_{\NN,\bbulk}
\|v\|_{\sH^{\NN+1}(\diamond^{\taufix}_{\tau,\sfix})}$.
\end{itemize}
Collecting terms yields \eqref{eq:EqEst}.
The estimates \eqref{eq:d/dtEIc} and \eqref{eq:EqEst} yield
\begin{align*}
\tfrac{d}{d\tau} E_{\indI}(\tau)
&\lesssim_{\NN,\sfix,\bbulk}
e_{\le\NN}(\tau)\big(
e_{\le\NN}(\tau)
+
\|v\|_{\sH^{\NN+1}(\diamond^{\taufix}_{\tau,\sfix})}
\big)
\end{align*}
We now sum over $\indI\in\indset_{0,\le\NN}$,
which yields the same estimate but where $ E_{\indI}(\tau)$
on the left is replaced by $E_{0,\le\NN}(\tau)=e_{\le\NN}(\tau)^2$.
Thus we obtain that for all $\tau\in[0,\taum]$
(see also the footnote preceding 
\eqref{eq:EnergyEstimateEsquareroot}):
\begin{align*}
\tfrac{d}{d\tau} e_{\le \NN}(\tau)
&\lesssim_{\NN,\sfix,\bbulk}
e_{\le\NN}(\tau)
+
\|v\|_{\sH^{\NN+1}(\diamond^{\taufix}_{\tau,\sfix})}
\end{align*}
Integrating this inequality in $\tau$
yields that for all $\tau\in[0,\taum]$:
\begin{align*}
e_{\le \NN}(\tau)
\lesssim_{\NN,\sfix,\bbulk}
\tint_{0}^{\tau} \|v\|_{\sH^{\NN+1}(\diamond^{\taufix}_{\tau',\sfix})}\,d\tau'
\end{align*}
where we use compactness in $\tau$, 
and $e_{\le\NN}(0)=0$
by \eqref{eq:cdatabulk} and the fact that $\V^1,\V^2,\V^3$
are tangential to $\tau=0$.
This implies \eqref{eq:ekktauint},
by \eqref{eq:dtaupos1Ap} and an 
admissible largeness assumption on $\Cbulk$.\qed
\end{proof}
\begin{proof}[of Lemma \ref{lem:bulkpropaux}]
We will specify $\ClargeCpt$ during the proof.
Instead of specifying $\CsmallCpt$ explicitly,
we will make finitely many admissible smallness assumptions on $\CsmallCpt$,
where admissible means that they depend only on $\NN$ and $\sfix$
(the dependencies on the fixed maps in
Definition \ref{def:SHSarrays_bulk} will be suppressed).

As a preliminary, note that 
for all $k\in\Z_{\ge\new{4}}$,
\begin{equation}\label{eq:vpointwise}
\|v\|_{\nosC^{\lfloor\frac{k+1}{2}\rfloor}(\new{\diamond^{\taufix}_{\sfix}})}
	\;\lesssim_{k}\;  
	\tint_{0}^{\frac\pi2}\|v\|_{\sH^{\new{k+1}}(\diamond_{\tau',\sfix})} d\tau'
\end{equation}
Proof of \eqref{eq:vpointwise}:
By \eqref{eq:Ckint}, for all $\tau\in[0,\frac\pi2]$ we have 
\begin{align*}
\|v\|_{\sC^{\lfloor\frac{k+1}{2}\rfloor}(\diamond^{\taufix}_{\tau,\sfix})}
	\lesssim_{k}
	\tint_{0}^{\frac\pi2}
	\|v\|_{\sH^{\lfloor\frac{k+1}{2}\rfloor+3}(\diamond^{\taufix}_{\tau',\sfix})} d\tau'
	\le
	\tint_{0}^{\frac\pi2}\|v\|_{\sH^{\new{k+1}}(\diamond_{\tau',\sfix})} d\tau'
\end{align*}	
where in the second step we use 
\new{$\lfloor\frac{k+1}{2}\rfloor+3\le k+1$ (use $k\ge4$)} 
and $\diamond^{\taufix}_{\tau',\sfix}\subset \diamond_{\tau',\sfix}$.
This implies \eqref{eq:vpointwise} by \ref{item:vsupp}
and by $\diamond^{\taufix}_{\sfix}=\cup_{\tau\in[0,\taufix]}\diamond^{\taufix}_{\tau,\sfix}$.

\proofheader{Proof of existence and Part 1, 
	for \eqref{eq:cSHSD} instead of \eqref{eq:ceqD}.}
More precisely, here we prove the following:
\begin{align}\label{eq:ep1}
\begin{aligned}
&\text{There exists \smash{$c\in\gxG^1(\diamond^{\taufix})
\simeq C^\infty(\diamond^{\taufix},\R^{\ngG_1})
$} that satisfies}\\[-1mm]
&\text{\eqref{eq:cSHSD}, \eqref{eq:cdataD}, \eqref{eq:czeroD}, \eqref{eq:cNhsNH},
and $\sqrt{c^Tc}\le \deltabulk$ on $\diamond^{\taufix}$.}
\end{aligned}
\end{align}
where we use the identification \eqref{eq:identify_g_vec_bulk},
and where we recall that $\deltabulk\in (0,1]$ has been 
fixed as in Lemma \ref{lem:deltacausalcpt}.
Associated to $v$
define $\AmatLin^\mu$, $\AmatBil^\mu$, $\LMat$, $\BSHS$, $\Fvec$ as in 
\eqref{eq:aALBFbulk}. 
Recall that \eqref{eq:cSHSD} is equivalent to \eqref{eq:cSHSDMatrices}.
By \eqref{eq:vpointwise} with $k=\NN$ and \ref{item:vNewBound},
\begin{align}
\|v\|_{\nosC^{\lfloor\frac{\NN+1}{2}\rfloor}(\new{\diamond^{\taufix}_{\sfix}})} 
&\lesssim_{\NN} 
\CsmallCpt
\nonumber
\intertext{
Thus under an admissible smallness assumption on $\CsmallCpt$,}
\|v\|_{\nosC^{\lfloor\frac{\NN+1}{2}\rfloor}(\new{\diamond^{\taufix}_{\sfix}})} 
&\le \deltabulk
\label{eq:vdelta}
\end{align}
Then for all $u \in \R^{\ngG_1}$ with $\sqrt{u^Tu}\le\deltabulk$:%
\begin{subequations}\label{eq:posc}
\begin{align}
\tfrac12\one
	&\le \AmatLin^0+\AmatBil^0(u) \le 2\one
	&&\text{at every point on $\diamond^{\taufix}_{\sfix}$}
	\label{eq:dtaupos1ApMain}\\
0
	&< d\s((\AmatLin^\mu+\AmatBil^\mu(u)) \V_\mu) 
	&&\text{at every point on $\Dspop_{\le \new{1}}\cap\diamond^{\taufix}_{\sfix}$}
	\label{eq:dspos1ApMAin}\\
0
	&< d\fctfol((\AmatLin^\mu+\AmatBil^\mu(u)) \V_\mu)
	&&\text{at every point on $\diamond^{\taufix}_{\sfix}$}
	\label{eq:cnullnonneg1ApMain}	
\end{align}
\end{subequations}
by Lemma \ref{lem:deltacausalcpt} with $k=1$ 
and with $u$ there given by $v+\ginj_1 u$ here.

We will use Lemma \ref{lem:EEbulk} with the parameters in Table \ref{tab:EEbulkAppN}.
Let $\Capply{\Cbulk}$ be the constant produced by Lemma \ref{lem:EEbulk}
(called $\Cbulk$ there).
It depends only on $\NN$, $\sfix$
(in particular it is independent of $\tau_*$),
thus $\ClargeCpt,\CsmallCpt$ are allowed to depend on $\Capply{\Cbulk}$.
Set $\ClargeCpt=\Capply{\Cbulk}(\Capply{\Cbulk}+1)$.

\claimheader{Claim:}
For all $\taum\in(0,\taufix]$ and all 
\begin{equation}\label{eq:ccpt}
c\in C^\infty(\diamond_{\le\taum,\sfix}^{\taufix},\R^{\ngG_1})
\end{equation}
if 
\begin{subequations}\label{eq:cDapriori}
\begin{align}
(\AmatLin^\mu+\AmatBil^\mu(c)) \V_\mu c 
	&= \LMat c + \BSHS(c,c) + \Fvec \label{eq:ceqDapriori}\\
c|_{\tau=0} 
	&=0 \label{eq:cdataDapriori}\\
c|_{\Dspop_{\le s_*}} 
	&= 0 \label{eq:czeroDapriori}\\
\|c\|_{C^{\lfloor\frac{\NN+1}{2}\rfloor}(\diamond_{\tau,s_*}^{\taufix})}
	&\le \deltabulk 
	\qquad\qquad\;\;\,
	\text{for all $\tau\in[0,\taum]$}
	 \label{eq:csmallDapriori}\\
\|c\|_{\H^{\NN}(\diamond_{\tau,s_*}^{\taufix})} 
	&\le 2\Capply{\Cbulk} \CsmallCpt
	\qquad\qquad
	\text{for all $\tau\in[0,\taum]$} 
	\label{eq:cEEDapriori}
\end{align}
\end{subequations}
then,
under an admissible smallness assumption on $\CsmallCpt$,
for all $\tau\in[0,\taum]$:
\begin{subequations}\label{eq:cDapriorigap}
\begin{align}
\|c\|_{C^{\lfloor\frac{\NN+1}{2}\rfloor}(\diamond_{\tau,s_*}^{\taufix})}
	&\le \tfrac12\deltabulk 
	 \label{eq:csmallDapriorigap}\\
\|c\|_{\H^{\NN}(\diamond^{\taufix}_{\tau,s_*})} 
	&\le \Capply{\Cbulk} 
	\tint_0^{\tau}\|v\|_{\sH^{\NN+1}(\diamond^{\taufix}_{\tau',s_*})}d\tau'
	\le \Capply{\Cbulk} \CsmallCpt 
	\label{eq:cEEDapriorigapepsNEW}\\
\|c\|_{\sH^{\NN}(\diamond^{\taufix}_{\tau,s_*})} 
	&\le
	\Capply{\Cbulk}(\Capply{\Cbulk}+1)
	\big(\tint_0^{\tau}\|v\|_{\sH^{\NN+1}(\diamond^{\taufix}_{\tau',s_*})}d\tau' 
	+
	\|v\|_{\sH^{\NN}(\diamond^{\taufix}_{\tau,s_*})}\big)
	\label{eq:Hallder}
\end{align}
\end{subequations}

\claimheader{Proof of claim:}
We check that the assumptions of Lemma \ref{lem:EEbulk} hold 
with the parameters in Table \ref{tab:EEbulkAppN}:
$v$ is smooth on $\diamond^{\taufix}_{\le\taum,s_*}$ 
because it is smooth on $\diamond_+$;
$c$ is smooth there by \eqref{eq:ccpt};
\eqref{eq:ceqbulk}, 
\eqref{eq:cdatabulk},
\eqref{eq:czerobulk}
hold by 
\eqref{eq:ceqDapriori}
\eqref{eq:cdataDapriori}
\eqref{eq:czeroDapriori};
\eqref{eq:cvdeltabulk} holds by \eqref{eq:vdelta} and \eqref{eq:csmallDapriori};
\eqref{eq:cCbulk} holds by \eqref{eq:csmallDapriori} and $\deltabulk\le1$;
\eqref{eq:vCbulk} holds by \eqref{eq:vdelta} and $\deltabulk\le1$.
Thus the assumptions hold.
We now show \eqref{eq:cDapriorigap}.

\begin{table} 
\centering
\begin{tabular}{cc|c|c}
	&
	Parameters 
	&
	\multicolumn{2}{c}{Parameters used to invoke Lemma \ref{lem:EEbulk}}
	\\
	&
	in Lemma \ref{lem:EEbulk}
	&
	\textit{Existence and Part 1}
	&
	\textit{Part 2}
	\\
	\hline
Input
	& $\NN$, $\sfix$, $\bbulk$
	& $\NN$, $\sfix$, $1$
	& $k$, $\sfix$, $C_{k,\sfix,\CHigherCpt}$ in \eqref{eq:vck}\\
	& $\taufix$, $\taum$
	& $\taufix$, $\taum$
	& $\taufix$, $\taufix$ \\
	& $c$, $v$
	& $c$ in \eqref{eq:ccpt}, $v$
	& $c$ in \eqref{eq:ep1}, $v$ \\
	\hline
Output
	&$\Cbulk$
	&$\Capply{\Cbulk}$
	&$\Capply{\Cbulk}_k$
\end{tabular}
\captionsetup{width=115mm}
\caption{
The first column lists the input and output parameters of Lemma \ref{lem:EEbulk}. 
The second column specifies the choice of input parameters used to invoke
Lemma \ref{lem:EEbulk}, in the proof of 
existence and Part 1 of Lemma \ref{lem:bulkpropaux},
in terms of the input parameters of Lemma \ref{lem:bulkpropaux}
and the parameters introduced in this proof.
The output parameter produced by this invocation of
Lemma \ref{lem:EEbulk} is denoted $\Capply{\Cbulk}$,
it depends only on the parameters in the first row.
Analogously for the third column, used to invoke Lemma \ref{lem:EEbulk}
in the proof of Part 2.}
\label{tab:EEbulkAppN}
\end{table}

\eqref{eq:csmallDapriorigap}:
For all $\tau\le\taum$, by \eqref{eq:SobolevC}
and $\lfloor\frac{\NN+1}{2}\rfloor+3\le\NN$ (use $\NN\ge6$),
\begin{align*}
\|c\|_{C^{\lfloor\frac{\NN+1}{2}\rfloor}(\diamond_{\tau,s_*}^{\taufix})}
	&\lesssim_{\NN}
	\sup_{\tau'\in [0,\taum]} \|c\|_{\sH^{\NN}(\diamond_{\tau',s_*}^{\taufix})}
\intertext{
Using \eqref{eq:sHtoHbulk} with $k=\NN$, 
and the fact that $\Capply{\Cbulk}$ depends only on $\NN$, $\sfix$,
we obtain}	
\|c\|_{C^{\lfloor\frac{\NN+1}{2}\rfloor}(\diamond_{\tau,s_*}^{\taufix})}
	&\lesssim_{\NN,\sfix}
	\sup_{\tau'\in [0,\taum]}
	(\|c\|_{\H^{\NN}(\diamond_{\tau',\sfix}^{\taufix})} 
	+ \|v\|_{\sH^{\NN}(\diamond_{\tau',\sfix}^{\taufix})})
\end{align*}
The $c$-term is bounded by $2\Capply{\Cbulk} \CsmallCpt$, 
by \eqref{eq:cEEDapriori}.
By \eqref{eq:LinfL2L1L2} and \ref{item:vsupp},
\begin{align*}
\sup_{\tau'\in [0,\taum]}
\|v\|_{\sH^{\NN}(\diamond_{\tau',\sfix}^{\taufix})}
&\lesssim_{\NN,\sfix}
\tint_{0}^{\frac\pi2}
\|v\|_{\sH^{\NN+1}(\diamond_{\tau',\sfix}^{\taufix})}
d\tau'
\le \CsmallCpt
\end{align*}
where the last step holds by \ref{item:vNewBound}
and $\diamond_{\tau',\sfix}^{\taufix}\subset\diamond_{\tau',\sfix}$.
Thus 
$$
\|c\|_{C^{\lfloor\frac{\NN+1}{2}\rfloor}(\diamond_{\tau,s_*}^{\taufix})}
\lesssim_{\NN,\sfix}
\CsmallCpt
$$
which implies \eqref{eq:csmallDapriorigap}
under an admissible smallness assumption on $\CsmallCpt$.

\eqref{eq:cEEDapriorigapepsNEW}:
The first inequality holds by \eqref{eq:ekktauint},
the second by \ref{item:vsupp}, \ref{item:vNewBound}.

\eqref{eq:Hallder}:
This follows from 
\eqref{eq:sHtoHbulk} with $k=\NN$ and from \eqref{eq:cEEDapriorigapepsNEW}.
This concludes the proof of the claim.

Define
\[ 
I = \Big\{ \taum \in (0,\taufix] \mid 
\text{There exists $c\in C^\infty(\diamond_{\le\taum,\sfix}^{\taufix},\R^{\ngG_1})$
that staifies \eqref{eq:cDapriori}} \Big\}
\]
Note that $\taum\in I$ implies $(0,\taum]\subset I$.

\claimheader{Claim:} $I=(0,\taufix]$.

\claimheader{Proof of claim:}
This is similar to, but easier than, the open-closed argument in the
proof of Theorem \ref{thm:nonlinEE}.
We have:
\begin{itemize}
\item 
$I$ is nonempty:
By local well-posedness of symmetric hyperbolic systems \cite[Section 16.1-16.2]{Taylor3},
using the fact that $\AmatLin^\mu$, $\AmatBil^\mu$ are symmetric
(Lemma \ref{lem:aAminkcpt})
and the positivity \eqref{eq:dtaupos1ApMain}, 
one obtains that there exists a closed trapezoidal domain 
$T\subset \diamond_{\sfix}^{\taufix}$ as indicated
in Figure \ref{fig:bulknonempty}
and $c\in C^{\infty}(T,\R^{\ngG_1})$
that satisfies \eqref{eq:ceqDapriori} and \eqref{eq:cdataDapriori}.
On the intersection $T\cap\Dspop_{\le\sfix}$ also the 
zero solution satisfies \eqref{eq:ceqDapriori} and \eqref{eq:cdataDapriori}
since $\Fvec=0$ by \ref{item:MCvsp}.
Finite speed of propagation applied to $T\cap\Dspop_{\le\sfix}$
(the inner lateral boundary component of this intersection 
is positive for the zero solution by \eqref{eq:dspos1ApMAin} with $u=0$, 
we may assume that the outer lateral boundary component is also 
positive by \eqref{eq:dtaupos1ApMain} with $u=0$ and by choosing $T$
sufficiently flat)
implies that $c$ coincides with the zero solution on this intersection.
We can therefore extend $c$ by zero to get a smooth solution 
on $\diamond_{\le\taum,\sfix}^{\taufix}$ for a small $\taum>0$.
By construction \eqref{eq:czeroDapriori} holds.
Since the left hand sides of \eqref{eq:csmallDapriori} and \eqref{eq:cEEDapriori}
are zero for $\tau=0$ (by \eqref{eq:cdataDapriori}
and the fact that $\V_1,\V_2,\V_3$ are tangential to $\tau=0$),
and continuous in $\tau$, \eqref{eq:csmallDapriori} and \eqref{eq:cEEDapriori} hold by making $\taum$ smaller if necessary.
Then $\taum\in I$.

\item 
$I$ is open in $(0,\taufix]$:
Let $\taum\in I$ with $\taum<\taufix$
(if $\taum=\taufix$ then we are done),
and let $c$ be the solution on $\diamond_{\le\taum,\sfix}^{\taufix}$
that satisfies \eqref{eq:cDapriori}.
Then $c$ also satisfies \eqref{eq:cDapriorigap}.
We show that there exists $\taum'\in(\taum,\taufix]$
with $\taum'\in I$.
For this let $\tau_0$ be the value of $\tau$
at which $\Dspop_{\sfix}$ and $\nullinf^{\taufix}$ intersect,
see Figure \ref{fig:bulknonempty}
(one has $\tau_0\in(0,\taufix)$ using $\taufix\ge\frac\pi2$ 
and Remark \ref{rem:Dstau}).
There are two cases:
\begin{itemize}
\item 
$\taum<\tau_0$:
By an argument analogous to the first item
(using local well-posedness and finite speed of propagation),
$c$ extends as a solution of \eqref{eq:ceqDapriori}
to $\diamond_{\le\taum',\sfix}^{\taufix}$ for some $\taum'>\taum$,
that satisfies \eqref{eq:czeroDapriori}.
Then \eqref{eq:csmallDapriori}, \eqref{eq:cEEDapriori} hold
because $c$ satisfies \eqref{eq:csmallDapriorigap}, \eqref{eq:cEEDapriorigapepsNEW}
on $\diamond_{\le\taum,\sfix}^{\taufix}$,
by continuity,
and making $\taum'>\taum$ smaller if necessary.
Then $\taum'\in I$.
\item 
$\taum\ge\tau_0$:
By local well-posedness, $c$ extends as a solution of \eqref{eq:ceqDapriori}
to $\diamond_{\le\taum',\sfix}^{\taufix}$ for some $\taum'>\taum$,
where one may extend the symmetric hyperbolic system and the initial data 
smoothly across the boundary of $\diamond_{\sfix}^{\taufix}$,
and use \eqref{eq:cnullnonneg1ApMain} to show that the solution is independent of the extension.
The solution satisfies \eqref{eq:cDapriori},
where for \eqref{eq:csmallDapriori}, \eqref{eq:cEEDapriori}
one must use \eqref{eq:csmallDapriorigap}, \eqref{eq:cEEDapriorigapepsNEW}
and continuity. Then $\taum'\in I$.
\end{itemize}
\item 
$I$ is closed in $(0,\taufix]$:
Let $\taum\in\bar{I}$.
Then there exists a smooth solution $c$ on 
$\diamond_{\sfix}^{\taufix}\cap([0,\taum)\times S^3)$ 
that satisfies \eqref{eq:cDapriori}
(this uses a standard uniqueness argument, c.f.~the proof of \eqref{eq:stronguniqueness} below).
A persistence of regularity argument
(essentially the energy estimates \eqref{eq:bulktaufixhighersH} 
restricted to $\tau\in[0,\taum)$)
shows that $c$ extends smoothly to $\tau=\taum$.
Then \eqref{eq:ceqDapriori}, \eqref{eq:czeroDapriori},
\eqref{eq:csmallDapriori}, \eqref{eq:cEEDapriori}
hold up to $\tau=\taum$ by continuity.
Thus $\taum\in I$.
\end{itemize}
Thus $I=(0,\taufix]$, which proves the claim.

\claimheader{We conclude \eqref{eq:ep1}:}
We have $\taufix\in I$, 
hence there exists a smooth solution 
$c$ on $\diamond_{\le\taufix,\sfix}^{\taufix}=\diamond_{\sfix}^{\taufix}$
that satisfies \eqref{eq:cDapriori}.
By \eqref{eq:czeroDapriori} we can extend $c$ by zero
to obtain a smooth solution on $\diamond^{\taufix}$,
which satisfies \eqref{eq:ceqDapriori} on $\diamond^{\taufix}$ by \ref{item:MCvsp}.
Clearly this satisfies the properties stated in \eqref{eq:ep1},
where we use the choice $\ClargeCpt=\Capply{\Cbulk}(\Capply{\Cbulk}+1)$.

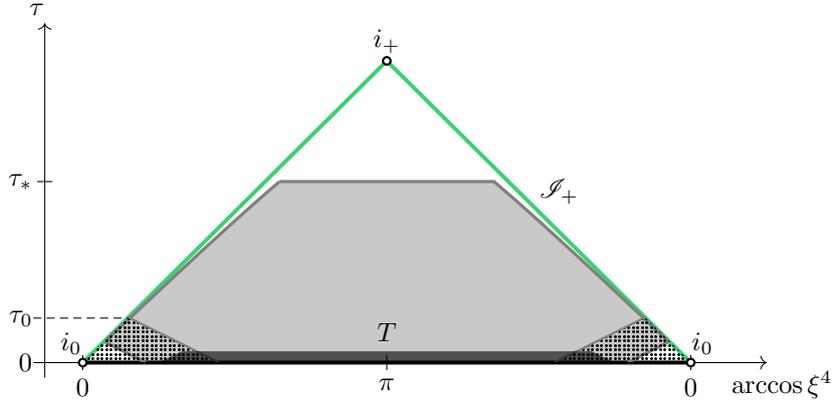
\begin{figure}
\centering
\begin{tikzpicture}[inner sep=0pt,scale=1]
\def\mycoordinatesedgeA{
(-3.7,0.285) 
(-3.6,0.38) (-3.5,0.475) (-3.4,0.57)
(-3.3,0.665) (-3.2,0.76) (-3.1,0.854) (-3.,0.949) (-2.9,1.043) (-2.8,1.137) (-2.7,1.231) (-2.6,1.324) (-2.5,1.418) (-2.4,1.511) (-2.3,1.603) (-2.2,1.695) (-2.1,1.787) (-2.,1.878) (-1.9,1.969) (-1.8,2.059) (-1.7,2.148) (-1.6,2.236) (-1.5,2.323) (-1.40967,2.4)} ;
\def\mycoordinatesedgeB {
(-1.40967,2.4) (1.40967,2.4)
};

\def\mycoordinatesedgeC{(1.40967,2.4) (1.41,2.4) (1.51,2.314) (1.61,2.227) (1.71,2.139) (1.81,2.05) (1.91,1.96) (2.01,1.869) (2.11,1.778) (2.21,1.686) (2.31,1.594) (2.41,1.502) (2.51,1.409) (2.61,1.315) (2.71,1.222) (2.81,1.128) (2.91,1.034) (3.01,0.939) (3.11,0.845) (3.21,0.75) (3.31,0.656) 
(3.41,0.561) (3.51,0.466) (3.61,0.371) (3.6,0.38) (3.71,0.276)
}; 
\def\mycoordinatesedgeD{ (3.172,0) (-3.172,0) };
\node (tip) at (0,4) {}; %
\node (l0) at (-4,0) {}; %
\node (r0) at (4,0) {}; %
\fill[gray!70, opacity=0.6]
	plot[smooth] coordinates {
	(-3.172,0)
	\mycoordinatesedgeA 
	\mycoordinatesedgeB 
	\mycoordinatesedgeC 
	(3.172,0)
	};

\draw[line width=1.5pt,ufogreen] (l0.center)--(tip.center);
\draw[line width=1.5pt,ufogreen] (r0.center)--(tip.center);
\draw[line width=1.5pt, black] (-3.95,0) -- (3.95,0);	

\draw [very thick,gray] plot [smooth] coordinates {\mycoordinatesedgeA}; 
\draw [very thick,gray] plot [smooth] coordinates {\mycoordinatesedgeB}; 
\draw [very thick,gray] plot [smooth] coordinates {\mycoordinatesedgeC}; 
\draw [very thick,gray] plot [] coordinates {(3.172,0) (3.71,0.276)};
\draw [very thick,gray] plot [] coordinates {(-3.172,0) (-3.71,0.276)};
\fill[pattern={Dots[distance=1.7pt]}, pattern color=black, opacity=1] 
	(-4,0) -- (-3.4,0.6) -- (-2.2,0) ;
\draw[gray,thick] (-3.4,0.6) -- (-2.2,0) ;
\fill[pattern={Dots[distance=1.7pt]}, pattern color=black, opacity=1] 
	(4,0) -- (3.4,0.6) -- (2.2,0) ;
\draw[gray,thick] (3.4,0.6) -- (2.2,0) ;

\fill[black,opacity=0.6] (-3,0) -- (-2.7,0.15) -- (2.7,0.15) -- (3,0) ;
\draw[] (0,0.4) node {$T$};

\node[anchor=south west,xshift=-3mm,yshift=1mm] at (l0) {$\spaceinf$};
\node[anchor=south east,xshift=3mm,yshift=1mm] at (r0) {$\spaceinf$};

\draw[->] (-4.5,-0.2) -- (-4.5,4.5) node[anchor=south east,yshift=1mm] {$\tau$};
    \draw (-4.6,2.4) -- (-4.4,2.4);
	\node[left,thick] at (-4.65,0) {$0$};
	\node[left,thick] at (-4.65,2.4) {$\tau_*$};
	
\draw[->] (-4.65,0) -- (5,0) node[anchor=north,xshift=2mm,yshift=-1.1mm] {$\arccos\xi^4$};

\foreach \x in {-4,0,4}
	\draw (\x,-0.1) -- (\x,0.1);
	\node[below,thick] at (-4,-0.2) {$0$};
	\node[below,thick] at (0,-0.2) {$\pi$};
	\node[below,thick] at (4,-0.2) {$0$};

\draw[color=black,thick, fill=white] (l0) circle (.05);
\draw[color=black,thick, fill=white] (tip) circle (.05);
\draw[color=black,thick, fill=white] (r0) circle (.05);
\node[anchor=south,yshift=1mm] at (tip) {$\timeinf$};
\node[anchor=south west,yshift=1mm] at (2,2) {$\fnullinf$};

\draw[densely dashed] (-4.65,0.59) node[anchor=east] {$\tau_0$} -- (-3.4,0.59);

\end{tikzpicture}
\captionsetup{width=115mm}
\caption{The gray shaded region depicts $\diamond^{\tau_*}_{\sfix}$.
The dotted region depicts $\Dspop_{\le\sfix}$, 
here one has $\Fvec=0$ by \ref{item:MCvsp}.
The trapezoidal domain $T$ is chosen sufficiently flat
so that its lateral boundary component is 
positive (spacelike) for the zero solution,
in the sense that the contraction of the 
outward pointing normal one-form 
with $\AmatLin^\mu \V_\mu$ is positive.}
\label{fig:bulknonempty}
\end{figure}

\proofheader{Proof of Part 0, for \eqref{eq:cSHSD} instead of \eqref{eq:ceqD}.}
More precisely, we prove:
\begin{align}\label{eq:stronguniqueness}
\begin{aligned}
&\text{Let $c$ be as in \eqref{eq:ep1}.
If \smash{$c'\in\gxG^1(\diamond^{\taufix})
\simeq C^\infty(\diamond^{\taufix},\R^{\ngG_1})$}}\\ 
&\text{satisfies
\eqref{eq:cSHSD}, 
\eqref{eq:cdataD}, 
\eqref{eq:czeroD} then $c'=c$.}
\end{aligned}
\end{align}
This shows in particular that $c$ in \eqref{eq:ep1}
is unique.

Proof of \eqref{eq:stronguniqueness}:
By \eqref{eq:czeroD} we have $c'=c$ on $\Dspop_{\le\sfix}$.
Thus it remains to show $c'=c$ on $\diamond_{\sfix}^{\taufix}$.
By \eqref{eq:cSHSD}, equivalently \eqref{eq:cSHSDMatrices},
the difference $c-c'$ satisfies a linear homogeneous 
symmetric hyperbolic system on $\diamond_{\sfix}^{\taufix}$,
with principal term
$$(\AmatLin^\mu+\AmatBil^\mu(c)) \V_\mu$$
(see the proof of Theorem \ref{thm:AbstractUniqueness} for details).
Since $\sqrt{c^Tc}\le\deltabulk$,
for this principal term we control the causal structure 
by \eqref{eq:posc} with $u=c$.
Then standard energy estimates similar to those in Lemma \ref{lem:EEbulk},
using the fact that $c-c'$ vanishes along $\tau=0$ and on $\Dspop_{\le\sfix}$,
imply $c-c'=0$ on $\diamond_{\sfix}^{\taufix}$.
This proves \eqref{eq:stronguniqueness}.

\proofheader{Proof of Part 0.}
This follows from \eqref{eq:stronguniqueness}, 
since \eqref{eq:ceqD} implies \eqref{eq:cSHSD}.

\proofheader{Proof of existence and Part 1.}
It remains to show that $c$ in \eqref{eq:ep1} solves \eqref{eq:ceqD},
i.e.~that the constraints propagate.
Define
\begin{align*}
U = \dg(v+c)+\tfrac12[v+c,v+c] \;\in\;\gx^2(\diamond^{\taufix})
\end{align*}
Our goal is to show that $U=0$. 
By \ref{item:MCvsp} and \eqref{eq:czeroD},
\begin{align}\label{eq:Uzerosp}
U|_{\Dspop_{\le s_*}} =0
\end{align}
Hence it remains to show $U=0$ on $\diamond^{\taufix}_{\sfix}$.
Analogously to \eqref{eq:Uprop} one checks that
\begin{subequations}\label{eq:UpropCpt}
\begin{align}
U \in \gxG^2(\diamond^{\taufix}) \label{eq:UgaugedCpt}\\
\dg U + [v+c,U]=0 \label{eq:UeqCpt}\\
U|_{\tau=0}=0 \label{eq:UdataCpt}
\end{align}
\end{subequations}
Briefly, \eqref{eq:UgaugedCpt} follows from \eqref{eq:cSHSD}
and \ref{item:gaugekernelNEW_bulk} of Lemma \ref{lem:gauge_mainprop_bulk};
\eqref{eq:UeqCpt} follows from \eqref{eq:dgLaax};
for \eqref{eq:UdataCpt} note that
$(d\tau+\anchorg(v)(\tau))U|_{\tau=0}=0$ by \ref{item:vconstraintsD}
and \eqref{eq:cdataD},
that $d\tau+\anchorg(v)(\tau) \in \Omegafut(\cyl)$ along $\tau=0$
by \eqref{eq:dtaupos1ApMain} and \eqref{eq:posiffofut_bulk},
and then conclude $U|_{\tau=0}=0$ using 
\eqref{eq:UgaugedCpt} and injectivity of left-multiplication in 
Lemma \ref{lem:gauge_mainprop_bulk}.

By \eqref{eq:UeqCpt} we have
$\gBil^2(\,\cdot\,,\dg U + [v+c,U])=0$.
By Lemma \ref{lem:translationofeqcyl},
and using the identification \eqref{eq:identify_g_vec_bulk},
this is equivalent to
\begin{equation}\label{eq:Ueqmatbulk}
\slashed{\AmatLin}^{\mu} \V_\mu U
=
\slashed{\LMat}U
\end{equation}
where we define
\begin{align*}
\slashed{\AmatLin}^{\mu} 
	&= \amink_2^{\mu}+ \Amink_2^{\mu}(v+\ginj_1c)\\
\slashed{\LMat} 
	&= \Lmink_2  + \Amink_{21}^\mu(\ginj_2\,\cdot\,) \V_{\mu} (v+\ginj_1c) + \Bmink_2(v+\ginj_1c,\ginj_2\,\cdot\,)
\end{align*}
By Lemma \ref{lem:aAminkcpt},
$\slashed{\AmatLin}^{\mu}$ is a symmetric matrix
at every point on $\diamond^{\taufix}_{\sfix}$.
By Lemma \ref{lem:deltacausalcpt} 
(with $k=2$ and $u=v+\ginj_1 c$)
and \eqref{eq:vdelta} and 
$\sqrt{c^Tc}\le\delta_0$,
\begin{align}\label{eq:poscon}
\begin{aligned}
\tfrac12\one
	&\le \slashed{\AmatLin}^{0} \le 2\one
	&&\text{at every point on $\diamond_{\sfix}^{\taufix}$}\\
0
	&< d\s(\slashed{\AmatLin}^{\mu}\V_\mu) 
	&&\text{at every point on $\Dspop_{\le 1}\cap\diamond^{\taufix}_{\sfix}$}\\
0 &< d\fctfol(\slashed{\AmatLin}^{\mu} \V_\mu)
	&&\text{at every point on $\diamond_{\sfix}^{\taufix}$}
\end{aligned}
\end{align}
Given \eqref{eq:poscon},
standard energy estimates for the linear
homogeneous symmetric hyperbolic system \eqref{eq:Ueqmatbulk},
using \eqref{eq:Uzerosp} and \eqref{eq:UdataCpt},
yield $U=0$ on $\diamond_{\sfix}^{\taufix}$.
Thus $c$ satisfies \eqref{eq:ceqD}.
This concludes the proof of existence and Part 1.

\proofheader{Proof of Part 2.} We prove this by induction in $k\ge\NN$.
Let $\PP_k$ be the statement
\[ 
\PP_{\kk}:\;\;
\text{For all $\CHigherCpt\new{>}0$, 
if \ref{item:vHigherNewBound}$_{k,\CHigherCpt}$
then \eqref{eq:bulkauxhall}$_{k,\CHigherCpt}$}
\]
where, for example, \ref{item:vHigherNewBound}$_{k,\CHigherCpt}$
means \ref{item:vHigherNewBound} with parameters $k$ and $\CHigherCpt$.
The base case $\PP_{\NN}$ holds by Part 1.
For the induction step we fix $k>\NN$,
and show that $\PP_{k-1}$ implies $\PP_{k}$.
Let $\CHigherCpt>0$ and assume that 
\ref{item:vHigherNewBound}$_{k,\CHigherCpt}$ holds.
Then also \ref{item:vHigherNewBound}$_{k-1,\CHigherCpt}$ holds,
hence by the induction hypothesis 
\eqref{eq:bulkauxhall}$_{k-1,\CHigherCpt}$ holds.

We claim that for all $\tau\in[0,\taufix]$:
\begin{subequations}\label{eq:vck}
\begin{align}
\|v\|_{\sC^{\lfloor\frac{k+1}{2}\rfloor}(\diamond^{\taufix}_{\tau,\sfix})} 
&\le C_{k,\sfix,\CHigherCpt} \label{eq:vk}\\
\|c\|_{\new{C}^{\lfloor\frac{k+1}{2}\rfloor}(\diamond_{\tau,\sfix}^{\taufix})}
&\le C_{k,\sfix,\CHigherCpt}\label{eq:ck}
\end{align}
\end{subequations}
for a constant $C_{k,\sfix,\CHigherCpt}>0$ that depends only on $k,\sfix,\CHigherCpt$.
\eqref{eq:vk}: This follows from 
\eqref{eq:vpointwise} and \ref{item:vHigherNewBound}$_{k,\CHigherCpt}$.
\eqref{eq:ck}:
By \eqref{eq:SobolevC} and $\lfloor\frac{k+1}{2}\rfloor+3\le k-1$
(use $k>\NN\ge7$),
\begin{align*}
\|c\|_{C^{\lfloor\frac{k+1}{2}\rfloor}(\diamond_{\tau,\sfix}^{\taufix})}
	&\lesssim_{k}
	\sup_{\tau'\in[0,\taufix]}
	\|c\|_{\sH^{k-1}(\diamond_{\tau',\sfix}^{\tau_*})}\\
	&\lesssim_{k,\sfix,\CHigherCpt}
	\sup_{\tau'\in[0,\taufix]}
	\big(
	\tint_0^{\tau'}
		\|v\|_{\sH^{\kk}(\diamond_{\tau'',\sfix}^{\taufix})}d\tau''
		+
		\|v\|_{\sH^{\kk-1}(\diamond_{\tau',\sfix}^{\taufix})} 
	\big)\\
	&\le
	\tint_0^{\frac\pi2}\|v\|_{\sH^{\kk}(\diamond_{\tau'',\sfix}^{\taufix})}d\tau''
	+
	\sup_{\tau'\in[0,\frac\pi2]}\|v\|_{\sH^{\kk-1}(\diamond_{\tau',\sfix}^{\taufix})} 
\intertext{
where the second inequality holds by \eqref{eq:bulktaufixhighersH}$_{k-1,\CHigherCpt}$,
and the third by \ref{item:vsupp}.
Thus}
\|c\|_{C^{\lfloor\frac{k+1}{2}\rfloor}(\diamond_{\tau,\sfix}^{\taufix})}
	&\lesssim_{k,\sfix,\CHigherCpt}
	\tint_0^{\frac\pi2}
	\|v\|_{\sH^{\kk}(\diamond_{\tau',\sfix}^{\taufix})}d\tau'
	\le
	\CHigherCpt
\end{align*}
using \eqref{eq:LinfL2L1L2} and \ref{item:vHigherNewBound}$_{k,\CHigherCpt}$.
This proves \eqref{eq:ck}.

We use Lemma \ref{lem:EEbulk} with the parameters in Table \ref{tab:EEbulkAppN}.
We check that the assumptions \eqref{eq:EEbulkassp} are satisfied:
\eqref{eq:ceqbulk}, 
\eqref{eq:cdatabulk},
\eqref{eq:czerobulk}
hold by 
\eqref{eq:ceqD},
\eqref{eq:cdataD},
\eqref{eq:czeroD};
\eqref{eq:cvdeltabulk} holds by \eqref{eq:ep1} and \eqref{eq:vdelta};
\eqref{eq:cCbulk}, \eqref{eq:vCbulk} hold by \eqref{eq:vck}.
Now \eqref{eq:ekktauint} implies \eqref{eq:bulktaufixhigherH}$_{k,\CHigherCpt}$,
which together with \eqref{eq:sHtoHbulk} implies \eqref{eq:bulktaufixhighersH}$_{k,\CHigherCpt}$.
\qed
\end{proof}
\begin{proof}[of Proposition \ref{prop:MainCpt}]
We use Lemma \ref{lem:bulkpropaux} with $\NN$, $\sfix$ as in Proposition \ref{prop:MainCpt}. Let $\ClargeCpt$, $\CsmallCpt$ be the constants produced by
Lemma \ref{lem:bulkpropaux}, which depend only on $\NN$, $\sfix$. 
We show that Proposition \ref{prop:MainCpt} holds with the same constants
$\ClargeCpt$, $\CsmallCpt$.
For each $\taufix\in[\frac\pi2,\pi)$ let 
$c_{\taufix}\in \gxG^{1}(\diamond^{\taufix})$ be the solution produced by 
Lemma \ref{lem:bulkpropaux}.
Define:
\begin{equation}\label{eq:defcprop}
c\in \gxG^1(\diamond_+)
\quad
\text{such that}
\quad
c|_{\diamond^{\taufix}} = c_{\taufix}
\quad
\text{for all $\taufix\in[\tfrac\pi2,\pi)$}
\end{equation}
Such a $c$ exists by Part 0 of Lemma \ref{lem:bulkpropaux} (uniqueness),
and it is unique because the $\diamond^{\taufix}$ exhaust $\diamond_+$ 
by Lemma \ref{lem:Dexhaustion}.
Clearly $c$ satisfies \eqref{eq:Cptcprop}.
We conclude Part 0,1,2 of Proposition \ref{prop:MainCpt}.
Part 0:
Suppose that $c'\in\gxG^1(\diamond_+)$ 
satisfies \eqref{eq:Cptcprop}.
Then for every $\taufix\in[\frac\pi2,\pi)$, 
the restriction $c'|_{\diamond^{\taufix}}$ 
also satisfies \eqref{eq:Cptcprop}, 
hence $c'|_{\diamond^{\taufix}}=c_{\taufix}$ 
by Part 0 of Lemma \ref{lem:bulkpropaux}.
Hence $c'=c$ by \eqref{eq:defcprop}.
Part 1 and Part 2:
This follows from Part 1 respectively Part 2 of Lemma \ref{lem:bulkpropaux}
and from Lemma \ref{lem:Dtauexhaustion}.
The last statements in Part 1 and 2 follow from
\ref{item:vsupp}, \eqref{eq:LinfL2L1L2} and Lemma \ref{lem:Dtauexhaustion}.
\qed
\end{proof}

\section{Construction on $\diamond_+$}
	\label{sec:global}

We combine the results from Section \ref{sec:SpaceinfConstruction}
and \ref{sec:bulk} to prove Theorem \ref{thm:main}.

\subsection{Norms for initial data near spacelike infinity}
	\label{sec:DataNorms_i0}

We define norms for the initial data near spacelike infinity (Definition \ref{def:normsdatasp}), used in Theorem \ref{thm:main}.
Further we define an operator that extends the initial data
near $\spaceinf$ to $y^0\ge0$ (Definition \ref{def:extensionspinf}), 
and show continuity properties (Lemma \ref{lem:PextCont}).

For $s>0$ define 
\begin{align} \label{eq:udeltas}
\Dspdata_{\le s} 
&= \Dspop_{\le s} \cap \diamonddata
\end{align}
where $\Dspop_{\le s}$ was introduced in \eqref{eq:deltas},
and where $\diamonddata$ is the initial hypersurface \eqref{eq:diamonddatadef}.
Using the coordinates $y$ in \eqref{eq:yycoords}, 
this is equivalently given by all points in $\diamondy$ with 
$y^0=0$ and $0<|\vec{y}|\le s$.
In particular, $(y^1,y^2,y^3)$ are smooth coordinates on \eqref{eq:udeltas}.
Analogously 
define $\Dspdata_{<s}=\Dspop_{<s} \cap \diamonddata$.
For $0<s_0<s_1$ define
\[ 
\Dspdata_{s_0,s_1}
=
\Dspdata_{\le s_1} \setminus \Dspdata_{<s_0}
\]
which is equivalently given by all points in $\diamondy$ with
$y^0=0$ and $s_0\le|\vec{y}|\le s_1$.

Recall the bundle $\gxdata$ in Definition \ref{def:gxdata}.
For the spaces of sections 
\begin{equation}\label{eq:gxdatasp}
\gxdata(\Dspdata_{\le s})
\qquad\qquad
\gxdata(\Dspdata_{s_0,s_1})
\end{equation}
we again use the homogeneous basis \eqref{eq:gbasis_i0},
now restricted to $y^0=0$.
\begin{definition}[Norms for data near spacelike infinity]\label{def:normsdatasp}
For every $k\in\Z_{\ge0}$ and $\sfix>0$ 
and $f\in C^\infty(\Dspdata_{\le\sfix})$ define:
\begin{align*}
\|f\|_{\Cb^{k}(\Dspdata_{\le \sfix})}
&=
\tsum_{j=\new{0}}^{k} \sum_{i_1,\dots,i_j=1}^3 
\sup_{p\in\Dspdata_{\le \sfix}} 
\big| 
(|\vec{y}|\p_{y^{i_1}}) \cdots (|\vec{y}|\p_{y^{i_j}}) f(p)
\big| \\
\|f\|_{\Hb^{k}(\Dspdata_{\le \sfix})}^2
&=
\tsum_{j=\new{0}}^{k} \sum_{i_1,\dots,i_j=1}^3 
\tint_{\Dspdata_{\le \sfix}}
|(|\vec{y}|\p_{y^{i_1}}) \cdots (|\vec{y}|\p_{y^{i_j}}) f|^2\,\muDu
\end{align*}
where we define \orient{$\muDu=|\vec{y}|^{-3} |dy^1\wedge dy^2\wedge dy^3|$}.
For $0<s_0<s_1\le\sfix$ define 
\[ 
\|f\|_{\Cb^{k}(\Dspdata_{s_0,s_1})}
\qquad
\|f\|_{\Hb^{k}(\Dspdata_{s_0,s_1})}
\]
analogously, with $\Dspdata_{\le\sfix}$ replaced by $\Dspdata_{s_0,s_1}$.
For $k\ge1$ and every $a\ge0$ define:
\begin{align*}
\|f\|_{\HdataNEW^{a,k}(\Dspdata_{\le\sfix})}
&=
\int_{0}^{\sfix} \big(\frac{\sfix}{s}\big)^{a+(k-1)} 
\big(1+|\log\big(\frac{\sfix}{s}\big)|\big)^{k-1}
\|f\|_{\Hb^{k}(\Dspdata_{\frac{s}{\new{3}},s})} \frac{ds}{s}
\end{align*}
We make analogous definitions for vector-valued functions, where we 
apply the norms componentwise and then take the $\ell^2$-sum
of the components;
and for elements in \eqref{eq:gxdatasp},
where we use the homogeneous basis \eqref{eq:gbasis_i0} 
to identify them with vector-valued functions.
\end{definition}
\begin{lemma}\label{lem:DataNorms}
For all 
$k\in\Z_{\ge0}$, 
$a\ge0$, 
$0<s\le\sfix\le1$,
and $f\in C^\infty(\Dspdata_{\le s_*})$:%
\begin{subequations}
\begin{align}
\|f\|_{\Hb^{k}(\Dspdata_{\frac{s}{3},s})}
	&\lesssim_{k}
	\big(\tfrac{s}{\sfix}\big)^{a+k} \|f\|_{\HdataNEW^{a,k+1}(\Dspdata_{\le s_*})}
	\label{eq:fHfHdata}\\
\|f\|_{\Cb^{k}(\Dspdata_{\frac{s}{3},s})}
	&\lesssim_{k}
	\|f\|_{\Hb^{k+2}(\Dspdata_{\frac{s}{3},s})}
	\label{eq:SobolevData}\\
\|f\|_{\Cb^{k}(\Dspdata_{\le s})}
	&\lesssim_{k}
	\big(\tfrac{s}{\sfix}\big)^{a+k+2} \|f\|_{\HdataNEW^{a,k+3}(\Dspdata_{\le s_*})}
	\label{eq:fCfH1}
\end{align}
\end{subequations}
\end{lemma}
\begin{proof}
After rescaling, it suffices to prove the lemma for $\sfix=1$.
Denote $\homMfddata=(-\infty,0]\times S^2$.
It is convenient to identify
\begin{equation}\label{eq:identdata}
\Dspdata_{\le1}\simeq \homMfddata
\qquad
\vec{y}\mapsto \big(\log(|\vec{y}|),\tfrac{\vec{y}}{|\vec{y}|}\big)
\end{equation}
which is the identification \eqref{eq:ztcoords} 
restricted to $y^0=0$ respectively $\ttcoord=0$.
Accordingly, we denote 
the coordinate on the first factor $(-\infty,0]$ of $\homMfddata$ by 
$\zzetadata=\log(|\vec{y}|)$.
For $\zz\le0$ let $\homMfddata_{\zz}=\{\zz\}\times S^2$
and for $\zz_0<\zz_1\le0$ let
$\homMfddata_{\zz_0,\zz_1}=[\zz_0,\zz_1]\times S^2$.

We prove \eqref{eq:fHfHdata}.
Set $\qq=\log(3)$.
We first show that for all $\zz_0\le0$:
\begin{align}\label{eq:dataminkineq}
\|f\|_{L^2(\homMfddata_{\zz_0-\qq,\zz_0})}
\lesssim
\tint_{\zz_0-3\qq}^{\zz_0}
(\|f\|_{L^2(\homMfddata_{\zz-\qq,\zz})}
+
\|\p_{\zzetadata}f\|_{L^2(\homMfddata_{\zz-\qq,\zz})})
d\zz
\end{align}
where the $L^2$-norm is defined using 
$|d\zzetadata\wedge\mu_{S^2}|$, 
which is equal to $\muDu$ via \eqref{eq:identdata}.

Proof of \eqref{eq:dataminkineq}:
By translating in $\zzetadata$, it suffices to prove this for $\zz_0=0$.
By using a cutoff function that is equal to one on $[-\qq,0]$
and zero for $\zz\le -\frac32\qq$, it suffices to prove
the inequality under the additional assumption that 
\begin{align}
\supp(f)&\subset [-\tfrac{3}{2}\qq,0]\times S^2
	\label{eq:suppfzero}
\end{align}
We now prove \eqref{eq:dataminkineq} with $\zz_0=0$
and under the additional assumption \eqref{eq:suppfzero}.

Analogously to \eqref{eq:hl1} one obtains that for all $(\zz,p)\in\homMfddata$:
\begin{align*}
|f(\zz,p)|^2 
	\lesssim 
	\tint_{\zz-\qq}^{\zz} (|f(\zz',p)|^2 + |f(\zz',p)||(\p_{\zzetadata}f)(\zz',p)|)\,d\zz'
\end{align*}
Integrating over $p\in S^2$ relative to $|\mu_{S^2}|$, and using Fubini, one obtains
\begin{align*}
\|f\|_{L^2(\homMfddata_{\zz})}^2
	&\lesssim 
	\|f\|_{L^2(\homMfddata_{\zz-\qq,\zz})}^2
	+
	\|f\p_{\zzetadata}f\|_{L^1(\homMfddata_{\zz-\qq,\zz})}\\
	&\le
	\|f\|_{L^2(\homMfddata_{\zz-\qq,\zz})}(\|f\|_{L^2(\homMfddata_{\zz-\qq,\zz})}
	+
	\|\p_{\zzetadata}f\|_{L^2(\homMfddata_{\zz-\qq,\zz})})
\end{align*}
where the second step uses Cauchy Schwarz.
Integrating over $\zz\in[-3\qq,0]$,
\begin{align*}
\|f\|_{L^2(\homMfddata_{-3\qq,0})}^2
	&\lesssim
	\tint_{-3\qq}^{0}
	\|f\|_{L^2(\homMfddata_{\zz-\qq,\zz})}
	(\|f\|_{L^2(\homMfddata_{\zz-\qq,\zz})}
	+
	\|\p_{\zzetadata}f\|_{L^2(\homMfddata_{\zz-\qq,\zz})})
	d\zz\\
	&\le
	\|f\|_{L^2(\homMfddata_{-3\qq,0})}
	\tint_{-3\qq}^{0}
	(\|f\|_{L^2(\homMfddata_{\zz-\qq,\zz})}
	+
	\|\p_{\zzetadata}f\|_{L^2(\homMfddata_{\zz-\qq,\zz})})
	d\zz
\end{align*}
using \eqref{eq:suppfzero} in the last step.
Canceling yields 
\[ 
\|f\|_{L^2(\homMfddata_{-3\qq,0})}
	\lesssim
	\tint_{-3\qq}^{0}
	(\|f\|_{L^2(\homMfddata_{\zz-\qq,\zz})}
	+
	\|\p_{\zzetadata}f\|_{L^2(\homMfddata_{\zz-\qq,\zz})})
	d\zz
\]
The left hand side bounds $\|f\|_{L^2(\homMfddata_{-\qq,0})}$,
hence this proves \eqref{eq:dataminkineq}.

Via the identification \eqref{eq:identdata},
the inequality
\eqref{eq:dataminkineq} implies that for all $s\in(0,1]$:
\[ 
\|f\|_{\Hb^{0}(\Dspdata_{\frac{s}{3},s})}
\lesssim
\tint_{\frac{s}{27}}^{s}
\|f\|_{\Hb^{1}(\Dspdata_{\frac{s'}{3},s'})} \frac{ds'}{s'}
\]
Using this inequality also for the derivatives of $f$
with respect to the vector fields $|\vec{y}|\p_{y^1},|\vec{y}|\p_{y^2},|\vec{y}|\p_{y^3}$,
one obtains that for all $s\in(0,1]$ and all $k\in\Z_{\ge0}$:
\begin{align*}
\|f\|_{\Hb^{k}(\Dspdata_{\frac{s}{3},s})}
\lesssim_{k}
\tint_{\frac{s}{27}}^{s}
\|f\|_{\Hb^{k+1}(\Dspdata_{\frac{s'}{3},s'})}\frac{ds'}{s'}
\end{align*}
To obtain \eqref{eq:fHfHdata}
we multiply and divide with the polynomial weight, that is,
\begin{align*}
\|f\|_{\Hb^{k}(\Dspdata_{\frac{s}{3},s})}
&\lesssim_{k}\textstyle
\tint_{\frac{s}{27}}^{s}
{s'}^{a+k}(\tfrac{1}{s'})^{a+k}
\|f\|_{\Hb^{k+1}(\Dspdata_{\frac{s'}{3},s'})}\frac{ds'}{s'}\\
&\le\textstyle
s^{a+k}
\tint_{\frac{s}{27}}^{s}
(\tfrac{1}{s'})^{a+k}
\|f\|_{\Hb^{k+1}(\Dspdata_{\frac{s'}{3},s'})}\frac{ds'}{s'}\\
&\le\textstyle
s^{a+k}
\tint_{\frac{s}{27}}^{s}
(\tfrac{1}{s'})^{a+k} (1+|\log(\frac{1}{s'})|)^{k}
\|f\|_{\Hb^{k+1}(\Dspdata_{\frac{s'}{3},s'})}\frac{ds'}{s'}\\
&\le
s^{a+k} \|f\|_{\HdataNEW^{a,k+1}(\Dspdata_{\le 1})}
\end{align*}
where, for the last inequality, we extend the 
domain of integration to $s'\in(0,1]$.

\eqref{eq:SobolevData}:
Via \eqref{eq:identdata}, this is a standard three-dimensional Sobolev inequality. 

\eqref{eq:fCfH1}:
Using \eqref{eq:SobolevData} and then \eqref{eq:fHfHdata},
for all $s'\in(0,s]$ we have
\begin{align*} 
\|f\|_{\Cb^{k}(\Dspdata_{\frac{s'}{3},s'})}
&\lesssim_{k}
\|f\|_{\Hb^{k+2}(\Dspdata_{\frac{s'}{3},s'})}
\lesssim_{k}
{s'}^{a+k+2} \|f\|_{\HdataNEW^{a,k+3}(\Dspdata_{\le1})}\\
&\le
s^{a+k+2} \|f\|_{\HdataNEW^{a,k+3}(\Dspdata_{\le1})}
\end{align*}
This implies the claim because 
$\Dspdata_{\le s}
=
\cup_{s'\in(0,s]} \Dspdata_{\frac{s'}{3},s'}$.\qed
\end{proof}
\begin{definition}\label{def:extensionspinf}
We define an $\R$-linear extension operator 
\begin{equation}\label{eq:Pextdef}
\Pext:\;\gxdata^k(\Dspdata_{\le 1})\to\gx^k(\Dspcl_{\le 1})
\end{equation}
as follows.
Let $(\eg^k_i)_{i=1\dots\ng_k}$ be the basis \eqref{eq:gbasis_i0},
which is given by the elements
\[ 
(\eB_i^k \oplus 0\Ieps)_{i=1\dots6\ng^{\Omega}_k},
\;\;
(\eT_i^k \oplus 0\Ieps)_{i=1\dots4\ng^{\Omega}_k},
\;\;
(0 \oplus \eI^{k+1}_i\Ieps)_{i=1\dots\ng^{\I}_k}
\]
For $\uf\in C^\infty(\Dspdata_{\le1})$ define 
\begin{align*}
\Pext \big((\eGB_i^k \oplus 0\Ieps) \uf \big)
	&=
	\big(\tfrac{\s}{|\vec{y}|}\big)^{k}(\eGB_i^k \oplus 0\Ieps)f
	\\
\Pext \big((\eGT_i^k \oplus 0\Ieps)\uf\big)
	&=
	\big(\tfrac{\s}{|\vec{y}|}\big)^{k+1}(\eGT_i^k \oplus 0\Ieps)f
	\\
\Pext\big((0\oplus\eI^{k+1}_i\Ieps)\uf\big)
	&=
	\big(\tfrac{\s}{|\vec{y}|}\big)^{k+3}(0\oplus\eI^{k+1}_i\Ieps)f
\end{align*}
where $f\in C^\infty(\Dspcl_{\le1})$ is defined by $f(y^0,\vec{y})=\uf(\vec{y})$,
and on the left hand sides the restriction of the basis elements
to $y^0=0$ is implicit.
For all $s\in(0,1]$ the map \eqref{eq:Pextdef} restricts to 
a map $\gxdata^k(\Dspdata_{\le s})\to\gx^k(\Dspcl_{\le s})$,
that we also denote by $\Pext$.
\end{definition}
Note that elements in the image of $\Pext$ are indeed
smooth on $\Dspcl_{\le1}$, in particular they are smooth along null infinity, 
because $\frac{\s}{|\vec{y}|}$ is smooth there.

The operator $\Pext$ is an extension operator, in the sense that
$
\Pext(\udata)|_{y^0=0}=\udata
$,
since $\s|_{y^0=0}=|\vec{y}|$.
One has for example
\begin{equation}\label{eq:ExPext}
\textstyle
\Pext\big( 
	(\frac{1}{|\vec{y}|} \frac{dy^0}{|\vec{y}|}\otimes T^1)\oplus0\Ieps 
	\big)
	=\big( (
	\frac{\s}{|\vec{y}|})^2\frac{1}{\s} \frac{dy^0}{\s}\otimes T^1\big)\oplus0\Ieps
	= \big(\frac{1}{|\vec{y}|} \frac{dy^0}{|\vec{y}|}\otimes T^1\big)\oplus0\Ieps
\end{equation}

The specific definition of $\Pext$ is motivated by 
Appendix \ref{ap:ConstructionOnD},
where we construct solutions as in Theorem \ref{thm:main} on $\diamond$
(not only on $\diamond_+$).
\begin{lemma}\label{lem:PextCont}
For all $s\in(0,1]$, all $k\in\Z_{\ge0}$ and all $\udata\in\gxdata^k(\Dspdata_{\le s})$, 
\begin{subequations}
\begin{align}
\|\Pext(\udata)\|_{\sCb^{k}(\Dspop_{s})}
&\lesssim_{k}
\|\udata\|_{\Cb^{k}(\Dspdata_{\frac s3,s})}
\label{eq:extensionopCk}
\\
\|\Pext(\udata)\|_{\sHb^{k}(\Dspop_{s})}
&\lesssim_{k}
\|\udata\|_{\Hb^{k}(\Dspdata_{\frac s3,s})}
\label{eq:extensionopHk}\\
\|\Pext(\udata)\|_{\nosCb^{k}(\Dspop_{\le s})}
&\lesssim_{k}
\|\udata\|_{\Cb^{k}(\Dspdata_{\le s})}
\label{eq:extensionopCkfull}
\end{align}
\end{subequations}
where, on the left hand sides, the norms in Definition \ref{def:norms_i0} are used.
\end{lemma}
\begin{proof}
We first show the following inequalities:
Let $\uf \in C^\infty(\Dspdata_{\le s})$ and define 
$f \in C^\infty(\Dspcl_{\le s})$ by $\uf(y^0,\vec{y})=f(\vec{y})$.
Then for all $n\in\Z_{\ge0}$:
\begin{subequations}
\begin{align}
\|\big(\tfrac{\s}{|\vec{y}|}\big)^{n}f\|_{\sCb^{k}(\Dspop_{s})} &\lesssim_{n,k} \|\uf\|_{\Cb^{k}(\Dspdata_{\frac{s}{3},s})}
\label{eq:Ext3f}\\
\|\big(\tfrac{\s}{|\vec{y}|}\big)^{n}f\|_{\sHb^{k}(\Dspop_{s})} &\lesssim_{n,k} \|\uf\|_{\Hb^{k}(\Dspdata_{\frac{s}{3},s})}
\label{eq:Ext4f}
\end{align}
\end{subequations}
Proof of \eqref{eq:Ext3f}:
First note that:
\begin{equation}\label{eq:Dsproj}
\text{For all $s\in(0,1]$ and all points in $\Dspop_{s}$ one has 
$\tfrac{s}{3}\le|\vec{y}|\le s$.}
\end{equation}
Using the Leibniz rule
and the fact that 
$\s/|\vec{y}|$ is homogeneous of degree zero 
in the sense of Definition \ref{def:homog},
one has
\[ 
\|\big(\tfrac{\s}{|\vec{y}|}\big)^{n}f\|_{\sCb^{k}(\Dspop_{s})}
\lesssim_{k}
\|\big(\tfrac{\s}{|\vec{y}|}\big)^{n}\|_{\sCb^{k}(\Dspop_{s})}
\|f\|_{\sCb^{k}(\Dspop_{s})}
\lesssim_{n,k}
\|f\|_{\sCb^{k}(\Dspop_{s})}
\]
By the Leibniz rule
and the fact that $\s$ is homogeneous of degree one,
\begin{align}
\|f\|_{\sCb^{k}(\Dspop_{s})}
	&\lesssim_{k} \textstyle
	\sum_{j=0}^k
	\sum_{i_1,\dots,i_{j}=0}^{3} 
	\sup_{p \in \Dspop_{s}}\big|\s^j\p_{y^{i_1}}\cdots\p_{y^{i_j}}f(p)\big|
	\nonumber
\intertext{
Since $\p_{y^0}f=0$, and since $\s=s\le3|\vec{y}|$ on $\Dspop_{s}$
by \eqref{eq:Dsproj}, we obtain 
}
\|f\|_{\sCb^{k}(\Dspop_{s})}
	&\lesssim_{k} \textstyle
	\sum_{j=0}^k
	\sum_{i_1,\dots,i_{j}=1}^{3} 
	\sup_{p \in \Dspop_{s}}\big||\vec{y}|^j\p_{y^{i_1}}\cdots\p_{y^{i_j}}f(p)\big|
	\nonumber\\
	&\lesssim_{k} \textstyle
	\sum_{j=0}^k
	\sum_{i_1,\dots,i_{j}=1}^{3} 
	\sup_{p \in \Dspdata_{\frac{s}{3},s}}\big||\vec{y}|^j\p_{y^{i_1}}
	\cdots\p_{y^{i_j}}\uf(p)\big|
	\label{eq:syppp}
\end{align}
where the second step follows from \eqref{eq:Dsproj}.
The expression \eqref{eq:syppp} is bounded by 
$\|\uf\|_{\Cb^{k}(\Dspdata_{s/3,s})}$, 
by the Leibniz rule and homogeneity of $|\vec{y}|$.
This proves \eqref{eq:Ext3f}.

Proof of \eqref{eq:Ext4f}:
By direct calculation,
\[ 
\tint_{\Dspdata_{\frac{s}{3},s}}|\uf|^2\,\muDu
\;\lesssim\;
\tint_{\Dspop_{s}}|f|^2\,\muWs
\;\lesssim\;
\tint_{\Dspdata_{\frac{s}{3},s}}|\uf|^2\,\muDu
\]
where $\muWs$ is the density \eqref{eq:muWs},
used in the definition of $\|{\cdot}\|_{\sHb^{k}(\Dspop_{s})}$.
From this, and a calculation similar to the proof of \eqref{eq:Ext3f},
the inequality \eqref{eq:Ext4f} follows.

We conclude the lemma.
The inequalities \eqref{eq:extensionopCk} and \eqref{eq:extensionopHk}
are immediate from \eqref{eq:Ext3f} respectively \eqref{eq:Ext4f}.
For \eqref{eq:extensionopCkfull}, note that \eqref{eq:extensionopCk} yields 
\begin{align*}
\|\Pext(\udata)\|_{\nosCb^{k}(\Dspop_{\le s})}
&=
\sup_{s'\in(0,s]}\|\Pext(\udata)\|_{\sCb^{k}(\Dspop_{s'})}
\lesssim_{k}
\sup_{s'\in(0,s]}\|\udata\|_{\Cb^{k}(\Dspdata_{\frac{s'}{3},s'})}
\le
\|\udata\|_{\Cb^{k}(\Dspdata_{\le s})}
\end{align*}
using $\Dspdata_{\frac{s'}{3},s'}\subset \Dspdata_{\le s}$
for $s'\in(0,s]$.
\qed
\end{proof}

\subsection{Norms for initial data away from spacelike infinity}
\label{sec:norms_data_bulk}

We define norms for the initial data away from spacelike infinity (Definition \ref{def:normdataawayfromi0}), used in Theorem \ref{thm:main}.
We define an operator that extends the data away from $\spaceinf$
to $x^0\ge0$ (Definition \ref{def:normdataawayfromi0}), and show continuity properties
(Lemma \ref{lem:Pextbulkest}).

Recall $\diamond_{0,s}\subset\diamonddata$
in Definition \ref{def:diamonds,stau}, where $s>0$.
Using the coordinates $x$ in \eqref{eq:xxcoords},
it is equivalently given by all points in 
$\diamond_+$ with $x^0=0$ and $|\vec{x}| \le \tfrac{6}{s}$,
see \eqref{eq:Kelvin}.
For the spaces of sections $\gxdata(\diamonddata)$ and $\gxdata(\diamond_{0,s})$
we use the basis \eqref{eq:gbasis_bulk}, now restricted to $\tau=0$.
The following definition is a special case of Definition \ref{def:bulknorms}.
\begin{definition}[Norms for data away from spacelike infinity]
\label{def:normdataawayfromi0}
For every $k\in\Z_{\ge0}$ and $s>0$ and $f\in C^\infty(\diamond_{0,s})$
define:
\begin{align*}
\|f\|_{\H^k(\diamond_{0,s})}^2 
&=\textstyle
	\sum_{j=0}^k
	\sum_{i_1,\dots,i_{j}=1}^{3}
	\int_{\diamond_{0,s}}\big|\V_{i_1}\cdots\V_{i_j} f\big|^2
	\,\mucylS
\end{align*}
using \eqref{eq:defV} and \eqref{eq:bulkdensities}.
We make an analogous definition for vector-valued functions,
where we apply the norms componentwise and then take the $\ell^2$-sum of the 
components;
and for elements in $\new{\gxdata(\diamond_{0,s})}$,
where we use the basis \eqref{eq:gbasis_bulk}
to identify them with vector-valued functions.
\end{definition}
\begin{definition}\label{def:bulkextension}
We define an $\R$-linear extension operator
\[ 
\Pextbulk :\;\gxdata^k(\diamonddata)\to\gx^k(\overline\diamond_+\setminus\spaceinf)
\]
as follows. Using the basis $(\eg^k_i)_{i=1\dots\ng_k}$ in \eqref{eq:gbasis_bulk},
for all $\uf\in C^\infty(\diamonddata)$ set
\[ 
\Pextbulk(\uf \eg^k_i) = f \eg^k_i
\]
where $f$ is defined by extending $\uf$ constantly in $\tau$ 
(i.e.~$f|_{\tau=0}=\uf$ and $\p_{\tau}f=0$),
and on the left hand side the restriction of the
basis elements to $\tau=0$ is implicit.
\end{definition}
The operator $\Pextbulk$ is an extension operator, 
that is, $\Pextbulk(\udata)|_{\tau=0}=\udata$.
\begin{lemma}\label{lem:Pextbulkest}
For all $k\in\Z_{\ge0}$, all $\new{s\in(0,1]}$, all $\tau\in[0,\pi)$
and $\udata\in\gxdata(\diamonddata)$:
\begin{align*}
\|\Pextbulk(\udata)\|_{\sH^{k}(\diamond_{\tau,4s})}
\le
\|\udata\|_{\H^{k}(\diamond_{0,s})}
\end{align*}
where, on the left hand side, we use the norms in Definition \ref{def:bulknorms}.
\end{lemma}
\begin{proof}
Since $\Pextbulk$ extends constantly in $\tau$,
\[ 
\|\Pextbulk(\udata)\|_{\sH^{k}(\diamond_{\tau,4s})} 
=
\|\Pextbulk(\udata)\|_{\H^{k}(\diamond_{\tau,4s})} 
\le
\|\udata\|_{\H^{k}(\diamond_{0,s})} 
\]
where the last step follows from:
\begin{equation}\label{eq:D4s,s}
(\tau,\xi) \in \diamond_{\tau,4s}
\quad\Rightarrow\quad
(0,\xi)\in \diamond_{0,s}
\end{equation}
We prove \eqref{eq:D4s,s}:
First note that (c.f. Remark \ref{rem:Dstau})
\begin{equation}\label{eq:D0xi4}
(0,\xi)\in \diamond_{0,s}
\quad
\Leftrightarrow
\quad
\xi^4 \in [-1,\smash{\tfrac{1-(\frac{s}{6})^2}{1+(\frac{s}{6})^2}}]
\end{equation}
Now let $(\tau,\xi) \in \diamond_{\tau,4s}$.
We distinguish two cases:
\begin{itemize}
\item 
$\xi^4 \in [-1,0]$: Then $(0,\xi) \in \diamond_{0,s}$ by \eqref{eq:D0xi4}.
\item 
$\xi^4 \in (0,1)$:
Then $(\tau,\xi)\in \diamond_+\cap\diamondy$, where $\s$ is defined.
We have
\[ 
\textstyle
\frac{4}{6} s
\le
\s(\tau,\xi) 
=
\frac{2\sin(\tau) + \sqrt{1-(\xi^4)^2}}{\cos(\tau)+\xi^4}
\le
\frac32\frac{\sqrt{1-(\xi^4)^2}}{\xi^4}
\]
where we use the fact that $\s$ is increasing in $\tau$
and $\tau\le\arccos \xi^4$ 
(because $\cos(\tau) - \xi^4>0$ by definition of $\diamond_+$).
This implies $\xi^4 \le \smash{\frac{1}{\sqrt{1 + (\frac{4s}{9})^2}}}$,
which, together with \eqref{eq:D0xi4}, implies $(0,\xi) \in \diamond_{0,s}$.
\qed
\end{itemize}
\end{proof}
\subsection{Estimates for the frame}\label{sec:estframe}
In this section we prove estimates for 
the endomorphism $\frame{\uo}$ in Definition \ref{def:Fu}, 
using the norms in Definition \ref{def:norms_i0} and \ref{def:bulknorms}.
Recall in particular that for $u\in \gx^1(\diamond_+)$,
the norms $\|u\|_{\nosCb^k(\Dspop_{\le s})}$ are defined
using the homogeneous basis \eqref{eq:gbasis_i0},
the norms $\|u\|_{\nosC^k(\diamond_{s})}$ are defined
using the basis \eqref{eq:gbasis_bulk}, which is regular on $\overline\diamond_+$.
\begin{lemma}\label{lem:framehomogbasis}
Let $(\eB_i^1)_{i=1,\dots,6\cdot4}$, $(\eT_i^1)_{i=1,\dots,4\cdot4}$ be the basis of 
$\Omega^1(\Dspcl)\otimesRR\Kil$ defined in \eqref{eq:OmegaBbasis_i0}
and \eqref{eq:OmegaTbasis_i0}, explicitly given by the elements
\[ 
\tfrac{dy^\mu}{\s}\otimes B^{\alpha\beta}
\qquad
\qquad
\qquad
\tfrac{1}{\s}\tfrac{dy^\mu}{\s}\otimes T_\nu
\]
For each $i$, the components of $\frame{\eB_i^1}$
and of $\frame{\eT_i^1}$ 
with respect to the basis $\p_{y^0},\dots,\p_{y^3}$
are smooth on $\Dspcl$ and homogeneous of degree zero
(Definition \ref{def:homog}). Explicitly,
\begin{align}\label{eq:FrameBT}
\frame{\frac{dy^\mu}{\s}\otimes B^{\alpha\beta}}(\p_{y^{\gamma}})
	&=
	\delta^{\mu}_{\gamma} \tfrac{B^{\alpha\beta}(y^{\sigma})}{\s}\p_{y^{\sigma}}
	&
\frame{\frac{1}{\s}\frac{dy^\mu}{\s}\otimes T_\nu}(\p_{y^{\gamma}})
	&= 
	\delta^{\mu}_{\gamma} \tfrac{T_{\nu}(y^\sigma)}{\s^2} \p_{y^\sigma}
\end{align}
\end{lemma}
\begin{proof}
The formulas \eqref{eq:FrameBT} are immediate from Definition \ref{def:Fu}.
By \eqref{eq:TBy}, the functions 
$\tfrac{B^{\alpha\beta}(y^{\sigma})}{\s}$, 
$\tfrac{T_{\nu}(y^\sigma)}{\s^2}$ are smooth on $\Dspcl$ and homogeneous of degree zero. 
\qed
\end{proof}
\begin{lemma}\label{lem:FrameInvertibility}
\framechange{Let $s\in (0,1]$ and let $\uo\in\Omega^1(\diamond_+)\otimesRR\Kil$. Then:}
\begin{itemize}
\item 
At every point on $\Dspop_{\le s}$:
Denoting by $\|\frame{\uo}\|$ the $\ell^2$-matrix norm of $\frame{\uo}$
with respect to the basis $\s\p_{y^0},\dots,\s\p_{y^3}$,
one has \smash{$\|\frame{\uo}\|\lesssim \|\uo\oplus0\|_{\nosCb^0(\Dspop_{\le s})}$}.
\item 
At every point on $\diamond_{s}$:
Denoting by $\|\frame{\uo}\|$ the $\ell^2$-matrix norm of $\frame{\uo}$
with respect to the basis $\V_0,\dots,\V_3$,
one has $\|\frame{\uo}\|\lesssim \|\uo\oplus0\|_{\nosC^0(\diamond_{s})}$.
\end{itemize}
Furthermore, for every $k\in\Z_{\ge0}$:
\begin{subequations}
\begin{align}
\|\frame{\uo}\|_{\nosCb^k(\Dspop_{\le s})}
	&\lesssim_{k}
	\|\uo\oplus0\|_{\nosCb^k(\Dspop_{\le s})}\label{eq:Fuestsp}
	\\
\|\frame{\uo}\|_{\nosC^k(\diamond_{s})}
	&\lesssim_{k}
	\|\uo\oplus0\|_{\nosC^k(\diamond_{s})}\label{eq:Fuestbulk}
\end{align}
\end{subequations}
where, on the left hand sides, the norms are taken 
componentwise with respect to the basis 
$\s\p_{y^0},\dots,\s\p_{y^3}$ in \eqref{eq:Fuestsp},
respectively the basis $\V_0,\dots,\V_3$
in \eqref{eq:Fuestbulk}.
\end{lemma}
\begin{proof}
First item and \eqref{eq:Fuestsp}:
By Lemma \ref{lem:framehomogbasis}
and $C^\infty$-linearity of $\frame{\uo}$ in $\uo$.
Second item:
By the second item in Lemma \ref{lem:Fglobbasic}.
\eqref{eq:Fuestbulk}:
By Lemma \ref{lem:frameglobalbasis}
and $C^\infty$-linearity of $\frame{\uo}$ in $\uo$.
\qed
\end{proof}
Recall that the function $\nullgen=\cos(\tau)-\xi^4$ in \eqref{eq:nullgendef} is positive
on $\diamond$,
vanishes first order along null infinity,
and second order at spacelike and timelike infinity.
\begin{lemma}\label{lem:Fdhlemma}
Let $s\in(0,1]$ and let $\uo\in\Omega^1(\diamond_+)\otimesRR\Kil$. Then:
\begin{itemize}
\item 
At every point on $\Dspop_{\le s}$:
Denoting by $\|\dualframe{\uo}(d\nullgen)/\nullgen\|$
the $\ell^2$-vector norm relative to the 
basis ${dy^0}/{\s},\dots,{dy^3}/{\s}$,
one has $\|\dualframe{\uo}(d\nullgen)/\nullgen\|
\lesssim\|\uo\oplus0\|_{\nosCb^0(\Dspop_{\le s})}$.
\item 
At every point on $\diamond_{s}$:
Denoting by $\|\dualframe{\uo}(d\nullgen)/\nullgen\|$
the $\ell^2$-vector norm relative to the 
basis $\Vd^0,\dots,\Vd^3$,
one has $\|\dualframe{\uo}(d\nullgen)/\nullgen\|
\lesssim\|\uo\oplus0\|_{\nosC^0(\diamond_{s})}$.
\end{itemize}
\end{lemma}
\begin{proof}
\proofheader{First item.}
By definition we have, for all $\mu,\nu,\alpha,\beta=0\dots3$, 
\begin{align*}
\tfrac{1}{\nullgen}\dualframe{\frac{dy^\mu}{\s}\otimes B^{\alpha\beta}}(d\nullgen)
&=
\tfrac{B^{\alpha\beta}(\nullgen)}{\nullgen} \tfrac{dy^\mu}{\s}
&
\tfrac{1}{\nullgen}\dualframe{\frac{dy^\mu}{\s^2}\otimes T_\nu}(d\nullgen)
&=
\tfrac{T_\nu(\nullgen)}{\s\nullgen} \tfrac{dy^\mu}{\s}
\end{align*}
Using Remark \ref{rem:Kil(h)} and \eqref{eq:yycoords}
one checks that 
$|\tfrac{B^{\alpha\beta}(\nullgen)}{\nullgen}|\lesssim 1$ and
$|\tfrac{T_\nu(\nullgen)}{\s\nullgen}|\lesssim 1$ on $\Dspop_{\le1}$.
Then together with $C^\infty$-linearity of $\dualframe{\uo}$ in $\uo$
the claim follows.

\proofheader{Second item.}
By linearity of $\dualframe{\uo}$ in $\uo$ and Lemma \ref{lem:frameglobalbasis}.
\qed
\end{proof}

\subsection{Proof of Theorem \ref{thm:main} 
and of Theorem \ref{thm:mainpointwise}}
\label{sec:ProofTheoremMain}
In this section we prove Theorem \ref{thm:main},
and prove Theorem \ref{thm:mainpointwise} as a corollary.

We start with some preliminary estimates.
\begin{lemma}\label{lem:normsbulktosp}
For all $k\in\Z_{\ge0}$, 
all $\sfix\in(0,1]$,
and all $f\in C^\infty(\Dspop_{\le s_*})$:
\begin{align}
\label{eq:normsbulktosp}
\tint_{0}^{\frac\pi2} \|f\|_{\sH^k(\diamond_{\tau,\frac\sfix2}\cap\Dspop_{\le\sfix})} d\tau
\lesssim_{k,\sfix}
\tint_{\frac{\sfix}{12}}^{\sfix} \|f\|_{\sHb^{k+1}(\Dspop_{s})} ds
\end{align}
in the sense that if the right hand side is finite,
then the left hand side is finite and the inequality holds.
The $\sH^k(\diamond_{\tau,\frac\sfix2}\cap\Dspop_{\le\sfix})$ norm 
is defined analogously to $\sH^k(\diamond_{\tau,s})$ in 
Definition \ref{def:bulknorms},
and it is understood to be zero
if the intersection $\diamond_{\tau,\frac\sfix2}\cap\Dspop_{\le\sfix}$
is empty (i.e.~if $\tau \ge \arctan(2\sfix/3)$, see Remark \ref{rem:Dstau}).
\end{lemma}
\begin{proof}
We first prove the inequality in the special case
when $f$ is smooth on $\Dspcl_{\le\sfix}$.
First using Cauchy Schwarz in $\tau$ and then using Fubini,
\begin{align}
\tint_{0}^{\frac\pi2} 
\|f\|_{\sH^k(\diamond_{\tau,\frac\sfix2}\cap\Dspop_{\le\sfix})} d\tau
	&\le
	(\tfrac\pi2)^{\frac12}
	(\tint_{0}^{\frac\pi2} 
	\|f\|_{\sH^k(\diamond_{\tau,\frac\sfix2}\cap\Dspop_{\le\sfix})}^2 d\tau)^{\frac12}
	\nonumber\\
	&= 
	(\tfrac\pi2)^{\frac12}
	\|f\|_{\nosH^k(\Dspop_{\le\sfix}\setminus\Dspop_{<\frac{\sfix}{12}})}
	\label{eq:Hkspec}\\
	&\lesssim_{k,\sfix}
	\|f\|_{\nosHb^k(\Dspop_{\le\sfix}\setminus\Dspop_{<\frac{\sfix}{12}})}
	\label{eq:Hbkspec}
\intertext{
(the norm in \eqref{eq:Hkspec} is defined analogously
to $\nosH^k(\diamond_{s})$ in Definition \ref{def:bulknorms},
the norm in \eqref{eq:Hbkspec} is defined analogously to 
$\nosHb^k(\Dspop_{\le s})$ in Definition \ref{def:norms_i0})
where we use the fact that the norms \eqref{eq:Hkspec} and
\eqref{eq:Hbkspec} are comparable
with a comparability constant that depends only on $k,\sfix$.
By Fubini, we obtain}
\tint_{0}^{\frac\pi2} 
\|f\|_{\sH^k(\diamond_{\tau,\frac\sfix2}\cap\Dspop_{\le\sfix})} d\tau
	&\lesssim_{k,\sfix}
	(\tint_{\frac{\sfix}{12}}^{\sfix} \|f\|_{\sHb^k(\Dspop_{s})}^2\, \frac{ds}{s})^{\frac12}
	\lesssim
	\sup_{s\in[\frac{\sfix}{12},\sfix]} \|f\|_{\sHb^k(\Dspop_{s})}
	\nonumber
\end{align}
By Lemma \ref{lem:LinfL2toL1L2} (using Convention \ref{conv:ztXmu_NEW},
see also Remark \ref{rem:normsequality}),
for all $s\in[\frac{\sfix}{12},\sfix]$,
\[ 
\|f\|_{\sHb^k(\Dspop_{s})}
\lesssim_{k}
\tint_{\frac{\sfix}{12}}^{\sfix} \|f\|_{\sHb^{k+1}(\Dspop_{s})}\,\frac{ds}{s}
\lesssim_{\sfix}
\tint_{\frac{\sfix}{12}}^{\sfix} \|f\|_{\sHb^{k+1}(\Dspop_{s})}\,ds
\]
This proves the lemma in the special case $f\in C^\infty(\Dspcl_{\le\sfix})$.

For general $f\in C^\infty(\Dspop_{\le\sfix})$ the lemma is proven as follows:
For $0<\eps\le \frac{1}{10}$
denote by \eqref{eq:normsbulktosp}$_{\eps}$ 
the inequality obtained by taking \eqref{eq:normsbulktosp} 
and replacing
$\diamond_{\tau,\frac\sfix2}\cap\Dspop_{\le\sfix}$ and $\Dspop_{s}$ 
by their intersection with the subset of $\Dspcl$ 
given by all points with $\frac{y^0}{|\vec y|} \le 1-\eps$ 
(equivalently $\ttcoord\le 1-\eps$ using Convention \ref{conv:ztXmu_NEW}).
Note that $f$ is smooth on these intersections, 
up to and including the boundary.
The inequality \eqref{eq:normsbulktosp}$_{\eps}$ is then proven
analogously to the proof of the special case of \eqref{eq:normsbulktosp} above.
In the limit $\eps\downarrow0$ the inequality \eqref{eq:normsbulktosp}$_{\eps}$
implies \eqref{eq:normsbulktosp} by monotone convergence.
\qed
\end{proof}
\begin{lemma}\label{lem:SobolevDandW}
Recall the norms in Definition \ref{def:norms_i0}
and Definition \ref{def:bulknorms}.
For all $k\in\Z_{\ge3}$, all $\sfix\in(0,1]$
and all $f\in C^\infty(\diamond_+)$:
\begin{itemize}
\item
If for all $s\in(0,\sfix]$,
\begin{equation}\label{eq:sasspsoblem}
\sup_{s'\in[\frac{s}{2},s]}
\|f\|_{\sHb^{k}(\Dspop_{s'})} < \infty
\end{equation}
then $f$ extends in $C^{k-3}$ to 
$\Dspcl_{\le\sfix}$.
For all $s\in(0,\sfix]$,
\begin{align}\label{eq:Wsbded}
\|f\|_{\sCb^{k-2}(\Dspop_{s})}
\lesssim_{k}
\|f\|_{\sHb^{k}(\Dspop_s)}
\end{align}
\item 
If 
\begin{equation}
\sup_{\tau\in[0,\pi)} \|f\|_{\sH^{k}(\diamond_{\tau,\sfix})} <\infty
\label{eq:tauasspsop}
\end{equation}
then $f$ extends in $C^{k-3}$ to $\overline{\diamond}_{\sfix}$.
Furthermore
\begin{align}\label{eq:supHbded}
\|f\|_{\nosC^{k-3}(\diamond_{\sfix})}
\lesssim_{k,\sfix}
\sup_{\tau\in[0,\pi)} \|f\|_{\sH^{k}(\diamond_{\tau,\sfix})}
\end{align}
\end{itemize}
\end{lemma}
\begin{proof}
\proofheader{First item:}
For all $s\in(0,\sfix]$, using Fubini and \eqref{eq:sasspsoblem} we have
\[ 
\|f\|_{\nosHb^{k}(\Dspop_{\le s} \setminus \Dspop_{<\frac{s}{2}})}
=
{(\tint_{\frac{s}{2}}^{s} \|f\|_{\sHb^{k}(\Dspop_{s'})}^2 \frac{ds'}{s'})^{\frac12}}
\le 
\smash{\sup_{s'\in[\frac{s}{2},s]}\|f\|_{\sHb^{k}(\Dspop_{s'})}}
<\infty
\]
Now using Convention \ref{conv:ztXmu_NEW} 
(see also the equality of norms in Remark \ref{rem:normsequality}),
a standard four-dimensional Sobolev embedding
(e.g.~\cite[Proposition 4.3]{Taylor1}) implies
that $f$ extends in $C^{k-3}$ to 
$\Dspcl_{\le s} \setminus \Dspcl_{<\frac{s}{2}}$.
Hence
$f$ extends in $C^{k-3}$ to $\Dspcl_{\le\sfix}$, using
$$
\Dspcl_{\le\sfix}
=
\textstyle\bigcup_{s\in(0,\sfix]}
\Dspcl_{\le s} \setminus \Dspcl_{<\frac{s}{2}}
$$
The inequality \eqref{eq:Wsbded} follows from 
\eqref{eq:sobolevMz}, using Convention \ref{conv:ztXmu_NEW}.

\proofheader{Second item:}
Using Fubini and \eqref{eq:tauasspsop} we have 
\[ 
\|f\|_{\nosH^k(\diamond_{\sfix})}
=
{(\tint_{0}^{\pi}
\|f\|_{\sH^k(\diamond_{\tau,\sfix})}^2\,d\tau)^{\frac12}}
\le
\smash{\pi^{\frac12} \sup_{\tau\in[0,\pi)}}
\|f\|_{\sH^k(\diamond_{\tau,\sfix})}
<\infty
\]
Hence the claim follows from a standard Sobolev embedding,
using the fact that the boundary of $\diamond_{\sfix}$ is Lipschitz.

We remark that we will not use an estimate analogous to \eqref{eq:Wsbded}
over $\tau$-level sets, since the constant would degenerate
as $\tau\uparrow\pi$ (i.e.~at timelike infinity).
\qed
\end{proof}
\begin{proof}[of Theorem \ref{thm:main}]
Instead of specifying $\Clargemain$, $\Csmallmain$ upfront,
we will make finitely many admissible largeness assumptions on $\Clargemain$,
respectively smallness assumptions on $\Csmallmain$,
where admissible means that they depend only on \eqref{eq:mainconst}.

Recall that we abbreviate 
$\Dspop = \Dspop_{\le\sfix}$ and 
$\Dspdata=\Dspdata_{\le\sfix}$ and 
$\kerrdata = \kerr|_{\tau=0}$.

\proofheader{Construction of $u$ near spacelike infinity.}
Define
\begin{equation}
\label{eq:v0def}
v_0 = \kerr + \Pext(\udata-\kerrdata) \;\in\; \gx^1(\Dspcl_{\le\sfix})
\end{equation}
where we use the extension operator in Definition \ref{def:extensionspinf},
and where the restriction of $\udata$ to $\Dspdata$ is implicit.
Observe that 
\begin{align}
v_0|_{y^0=0}
	&=\udata
	\label{eq:v0atzero}
\end{align}

We will correct $v_0$ to a solution of
\eqref{eq:MCandudatamain} near spacelike infinity 
using Proposition \ref{prop:ApplySHS}.
For this we need some preliminary estimates.

\claimheader{Claim:} 
For all $k\in\Z_{\ge0}$ and $s\in(0,\sfix]$:
\begin{subequations}\label{eq:v0uk}
\begin{align}
\|v_0\|_{\nosCb^{k}(\Dspop)}
	&\lesssim_{k}
	\|\kerr\|_{\nosCb^{k}(\Dspop)}
	+
	\|\udata-\kerrdata\|_{\Cb^{k}(\Dspdata)}
	\label{eq:v0Ck}\\
\|\udata-\kerrdata\|_{\Cb^k(\Dspdata)}
	&\lesssim_{k}
	\|\udata\|_{\Cb^k(\Dspdata)} + \|\kerr\|_{\nosCb^k(\Dspop)}
	\label{eq:uktruk}\\
\|\Pext(\udata-\kerrdata)\|_{\sHb^{k}(\Dspop_{s})}
	&\lesssim_{k}
	\big(\tfrac{s}{s_*}\big)^{\frac{5}{2}+\epspower+k}
	\|\udata-\kerrdata\|_{\HdataNEW^{\frac{5}{2}+\epspower,k+1}(\Dspdata)}
	\label{eq:Pu-kbasic}
\end{align}
\end{subequations}

Proof of \eqref{eq:v0Ck}: 
This holds by the triangle inequality and \eqref{eq:extensionopCkfull}.

Proof of \eqref{eq:uktruk}: 
This holds by the triangle inequality and $\Dspdata\subset\Dspop$.

Proof of \eqref{eq:Pu-kbasic}: By \eqref{eq:extensionopHk},
\begin{align*}
\|\Pext(\udata-\kerrdata)\|_{\sHb^{k}(\Dspop_{s})}
\lesssim_{k}
\|\udata-\kerrdata\|_{\Hb^{k}(\Dspdata_{\frac{s}{3},s})}
\end{align*}
Now the claim follows from 
\eqref{eq:fHfHdata} with $a=\frac{5}{2}+\epspower$.

\claimheader{Claim:} 
For all $k\in\Z_{\ge0}$ one has
\begin{align}
&\Vconst_{k}(v_0)
	\lesssim_k
	\big(1
	+
	\|\kerr\|_{\nosCb^{k+1}(\Dspop)} 
	+
	\|\udata-\kerrdata\|_{\Cb^{\lfloor\frac{k+1}{2}\rfloor}(\Dspdata)}
	\big)
	\|\udata-\kerrdata\|_{\HdataNEW^{\frac52+\epspower,k+1}(\Dspdata)}
	\label{eq:v0Vconstk}
\end{align}
where $\Vconst_{k}(v_0)$ is defined exactly as in \eqref{eq:defVkk}.

Proof of \eqref{eq:v0Vconstk}:
Abbreviate $v_0' = \Pext(\udata-\kerrdata)$. 
By linearity and bilinearity of the differential respectively
the bracket, graded antisymmetry \eqref{eq:bracketas}
and \ref{item:KMCmain},
\begin{align*}
\dg v_0 + \tfrac12[v_0,v_0]
&=
\dg\kerr + \tfrac12[\kerr,\kerr]
+
\dg v_0'
+
[\kerr,v_0']
+
\tfrac12[v_0',v_0']\\
&=
\dg v_0'
+
[\kerr,v_0']
+
\tfrac12[v_0',v_0']
\end{align*}
Thus for all $s\in(0,\sfix]$:
\begin{align*}
\|\dg v_0 + \tfrac12[v_0,v_0]\|_{\sHb^{k}(\Dspop_{s})}
\le
\|\dg v_0'\|_{\sHb^{k}(\Dspop_{s})}
+
\|[\kerr,v_0']\|_{\sHb^{k}(\Dspop_{s})}
+
\|[v_0',v_0']\|_{\sHb^{k}(\Dspop_{s})}
\end{align*}
The maps $\dg$, $[{\cdot},{\cdot}]$ 
are first order linear respectively bilinear
differential operators, and 
by Lemma \ref{lem:homogeneity},
their coefficients with respect to the homogeneous basis \eqref{eq:gbasis_i0}
and the vector fields $\s\p_{y^0},\dots,\s\p_{y^3}$ 
are homogeneous of degree zero.
This yields,
together with standard properties of the norms in Definition \ref{def:norms_i0},
\begin{align*}
\|\dg v_0'\|_{\sHb^{k}(\Dspop_{s})}
	&\lesssim_{k}
	\|v_0'\|_{\sHb^{k+1}(\Dspop_{s})}\\
\|[\kerr,v_0']\|_{\sHb^{k}(\Dspop_{s})}
	&\lesssim_{k}
	\|\kerr\|_{\sCb^{k+1}(\Dspop_{s})}\|v_0'\|_{\sHb^{k+1}(\Dspop_{s})}\\
\|[v_0',v_0']\|_{\sHb^{k}(\Dspop_{s})}
	&\lesssim_{k}
	\|v_0'\|_{\sCb^{\lfloor\frac{k+1}{2}\rfloor}(\Dspop_{s})}\|v_0'\|_{\sHb^{k+1}(\Dspop_{s})}
\intertext{
Using \eqref{eq:extensionopCk} and \eqref{eq:extensionopHk} and 
$\Dspdata_{\frac{s}{3},s}\subset\Dspdata$ 
and $\Dspop_{s}\subset\Dspop$,
}
\|v_0'\|_{\sHb^{k+1}(\Dspop_{s})}
	&\lesssim_{k}
	\|\udata-\kerrdata\|_{\Hb^{k+1}(\Dspdata_{\frac{s}{3},s})}\\
\|v_0'\|_{\sCb^{\lfloor\frac{k+1}{2}\rfloor}(\Dspop_{s})}
	&\lesssim_{k}
	\|\udata-\kerrdata\|_{\Cb^{\lfloor\frac{k+1}{2}\rfloor}(\Dspdata)}
	\\
\|\kerr\|_{\sCb^{k+1}(\Dspop_{s})}
	&\le
	\|\kerr\|_{\nosCb^{k+1}(\Dspop)}
\end{align*}
Thus for all $s\in(0,\sfix]$:
\begin{align*}
&\|\dg v_0 + \tfrac12[v_0,v_0]\|_{\sHb^{k}(\Dspop_{s})}\\
	&\qquad\lesssim_{k}
	\big(1
	+
	\|\kerr\|_{\nosCb^{k+1}(\Dspop)}
	+
	\|\udata-\kerrdata\|_{\Cb^{\lfloor\frac{k+1}{2}\rfloor}(\Dspdata)}
	\big)
	\|\udata-\kerrdata\|_{\Hb^{k+1}(\Dspdata_{\frac{s}{3},s})}
\end{align*}
Plugging this into the definition of $\Vconst_{k}(v_0)$ yields \eqref{eq:v0Vconstk}.

We use Proposition \ref{prop:ApplySHS} with
the parameters in Table \ref{tab:ApplyApplySHS}.
Let $\Capply{\ClargeGR}_0$, $\epsapply{\CsmallGR}_0$ be the constants
produced by Proposition \ref{prop:ApplySHS} (called $\ClargeGR$, $\CsmallGR$ there).
They depend only on $\NN,\epspower,\Cinmain$,
in particular $\Clargemain$, $\Csmallmain$ are allowed to depend on 
$\Capply{\ClargeGR}_0$, $\epsapply{\CsmallGR}_0$.

\begin{table}
\centering
\begin{tabular}{cc|c}
	&
	\begin{tabular}{@{}c@{}}
	Parameters \\
	in Proposition \ref{prop:ApplySHS}
	\end{tabular}
	& 
	\begin{tabular}{@{}c@{}}
	Parameters\\ 
	used to invoke Proposition \ref{prop:ApplySHS}
	\end{tabular}
	 \\
\hline
Input
	&$\NN$, $\epspower$, $\CinGR$
	&$\NN+2$, $\epspower$, 
	$\max\{\Cauxmain_{1,\NN,\Cinmain},\Cauxmain_{2,\Cinmain}\}$ \\
%
	&$\sfix$ 
	& $\sfix$\\
	&$v$ 
	& $v_0 = \kerr + \Pext(\udata-\kerrdata)$ \\
%
	&$k$, $\bb$ (\new{Part 2 only})
	& $k+2$, 
	$\max\{
	\Cauxmain_{3,k,\CinmainHigher},
	\Cauxmain_{4,k,\CinmainHigher}
	\}$ \\
\hline
Output
	& $\ClargeGR$, $\CsmallGR$ 
	& $\Capply{\ClargeGR}_0$, $\epsapply{\CsmallGR}_0$	
\end{tabular}
\captionsetup{width=115mm}
\caption{%
The first column lists the input and output parameters of Proposition \ref{prop:ApplySHS}. 
The second column specifies the choice of the input parameters used to invoke
Proposition \ref{prop:ApplySHS},
in terms of the input parameters of Theorem \ref{thm:main}
and the parameters introduced in this proof.
The output parameters produced by this invocation of
Proposition \ref{prop:ApplySHS}
are denoted $\Capply{\ClargeGR}_0$, $\epsapply{\CsmallGR}_0$	.
They depend only on the parameters in the first row.}
\label{tab:ApplyApplySHS}
\end{table}
We check that the assumptions of Proposition \ref{prop:ApplySHS} hold.
As required $\NN+2\ge\new{9}\ge\new{6}$; 
and  
$v_0$ is smooth on $\Dspcl_{\le\sfix}$
because $\kerr$ is smooth by \eqref{eq:KerrMain}
and $\Pext(\udata-\kerrdata)$ is smooth by Definition \ref{def:extensionspinf}.

\ref{item:vconstraints}:
By \eqref{eq:v0atzero} and \ref{item:uconstrMain}
we have
$\Pconstraints(v_0|_{y^0=0})
=\Pconstraints(\udata)
=0 $.

\ref{item:CN+1v}:
Using \eqref{eq:v0Ck} with $k=\NN+3$,
\eqref{eq:uktruk} with $k=\NN+3$,
and \ref{item:KBoundMain}, \ref{item:uBoundMain},
\begin{align*}
\|v_0\|_{\nosCb^{\NN+3}(\Dspop)}
	&\lesssim_{\NN}
	\|\kerr\|_{\nosCb^{\NN+3}(\Dspop)}
	+
	\|\udata-\kerrdata\|_{\Cb^{\NN+3}(\Dspdata)}\\
	&\lesssim_{\NN}
	\|\kerr\|_{\nosCb^{\NN+3}(\Dspop)}
	+
	\|\udata\|_{\Cb^{\NN+3}(\Dspdata)}
	\le
	2\Cinmain
\end{align*}
Thus there exists a constant $\Cauxmain_{1,\NN,\Cinmain}>0$
that depends only on $\NN$, $\Cinmain$, such that 
\begin{align*}
\|v_0\|_{\nosCb^{\NN+3}(\Dspop)} \le \Cauxmain_{1,\NN,\Cinmain}
\end{align*}
Thus \ref{item:CN+1v} holds, using Table \ref{tab:ApplyApplySHS}.

\ref{item:vint}:
By the triangle inequality,
\begin{align*}
\tint_0^{s_*} \|v_0\|_{\sCb^{1}(\Dspop_{s})} \frac{ds}{s} 
&\le 
\tint_0^{s_*} \|\kerr\|_{\sCb^{1}(\Dspop_{s})} \frac{ds}{s} 
+
\tint_0^{s_*} \|\Pext(\udata-\kerrdata)\|_{\sCb^{1}(\Dspop_{s})} \frac{ds}{s} 
\end{align*}
By \ref{item:KL1Main}, the first term is bounded by $\Cinmain$.
Using \eqref{eq:extensionopCk} and \eqref{eq:SobolevData}, 
\begin{align*}
\tint_0^{s_*} \|\Pext(\udata-\kerrdata)\|_{\sCb^{1}(\Dspdata_{s})} \frac{ds}{s} 
&\lesssim
\tint_0^{s_*} \|\udata-\kerrdata\|_{\Cb^{1}(\Dspdata_{\frac{s}{3},s})} \frac{ds}{s}\\
&\lesssim
\tint_0^{s_*} \|\udata-\kerrdata\|_{\Hb^{3}(\Dspdata_{\frac{s}{3},s})} \frac{ds}{s}\\
&\le
\|\udata-\kerrdata\|_{\HdataNEW^{\frac52+\epspower,3}(\Dspdata)}\\
&\le1
\end{align*}
where the third inequality is clear from the definition of $\HdataNEW$,
and the fourth inequality holds by \ref{item:diffsmallmain} and $\Csmallmain\le1$.
Thus there exists a constant $\Cauxmain_{2,\Cinmain}>0$
that depends only on $\Cinmain$, such that
\begin{equation*}
\tint_0^{s_*} \|v_0\|_{\sCb^{1}(\Dspop_{s})} \frac{ds}{s} 
\le
\Cauxmain_{2,\Cinmain}
\end{equation*}
Thus \ref{item:vint} holds, using Table \ref{tab:ApplyApplySHS}.

\ref{item:vsmall}:
By \eqref{eq:v0Ck} with $k=0$, 
\begin{align*}
\|v_0\|_{\nosCb^{0}(\Dspop)}
	&\lesssim
	\|\kerr\|_{\nosCb^{0}(\Dspop)}
	+
	\|\udata-\kerrdata\|_{\Cb^{0}(\Dspdata)}\\
\intertext{
Using \eqref{eq:fCfH1} with $k=0$, $s=\sfix$, $a=\frac52+\epspower$,
and then \ref{item:KsmallMain} and \ref{item:diffsmallmain},
we obtain}
\|v_0\|_{\nosCb^{0}(\Dspop)}
	&\lesssim
	\|\kerr\|_{\nosCb^{0}(\Dspop)}
	+
	\|\udata-\kerrdata\|_{\HdataNEW^{\frac52+\epspower,3}(\Dspdata)}
	\le
	2\Csmallmain
\end{align*}
This implies \ref{item:vsmall} under an admissible
smallness assumption on $\Csmallmain$.

\ref{item:MC(v)small}: 
By \eqref{eq:v0Vconstk} with $k=\NN+2$,
$\lfloor\frac{\NN+3}{2}\rfloor\le \NN+3$,
and \eqref{eq:uktruk} with $k=\NN+3$,
\begin{align}
\Vconst_{\NN+2}(v_0)
	&\lesssim_{\NN}
	\big(1
	+
	\|\kerr\|_{\nosCb^{\NN+3}(\Dspop)} 
	+
	\|\udata\|_{\Cb^{\NN+3}(\Dspdata)}
	\big)
	\|\udata-\kerrdata\|_{\HdataNEW^{\frac52+\epspower,\NN+3}(\Dspdata)}
	\nonumber
\intertext{
Now \ref{item:KBoundMain}, \ref{item:uBoundMain}, \ref{item:diffsmallmain} yield
}
\Vconst_{\NN+2}(v_0)
	&\lesssim_{\NN}
	(1+\Cinmain)
	\|\udata-\kerrdata\|_{\HdataNEW^{\frac52+\epspower,\NN+3}(\Dspdata)}
	\le
	(1+\Cinmain) \Csmallmain
	\label{eq:VN}
\end{align}
This implies \ref{item:MC(v)small} under 
an admissible smallness assumption on $\Csmallmain$.

We have checked that the assumptions 
\ref{item:vconstraints},
\ref{item:CN+1v},
\ref{item:vint},
\ref{item:vsmall},
\ref{item:MC(v)small} of Proposition \ref{prop:ApplySHS} hold.
Hence by Proposition \ref{prop:ApplySHS} there exists a unique
\begin{equation}\label{eq:c0}
c_0\in\gxG^1(\Dspop) \qquad \text{(called $c$ in Proposition \ref{prop:ApplySHS})}
\end{equation}
(using the gauge in Definition \ref{def:gauges_i0})
that satisfies \eqref{eq:existence_v+c}.
Further $c_0$ satisfies \eqref{eq:cestimates} with $\NN$ replaced by $\NN+2$.
Define
\[ 
u_0 = v_0 + c_0 \;\in\; \gx^1(\Dspop)
\]
By \eqref{eq:v+cMC}, respectively by \eqref{eq:v0atzero} and \eqref{eq:cdata},
\begin{subequations}
\begin{align}
\dg u_0 + \tfrac12[u_0,u_0] &= 0 \label{eq:u0MC}\\
u_0|_{y^0=0} &= \udata \label{eq:u0data}
\end{align}
\end{subequations}

\claimheader{Claim:} For all $s\in(0,\sfix]$:
\begin{align}
\|u_0-\kerr\|_{\sHb^{\NN+2}(\Dspop_{s})}
	&\lesssim_{\NN,\epspower,\Cinmain} \big(\tfrac{s}{\sfix}\big)^{\frac92+\epspower+\NN}
	\|\udata-\kerrdata\|_{\HdataNEW^{\frac52+\epspower,\NN+3}(\Dspdata)}
	\label{eq:u0est}
\end{align}
Proof of \eqref{eq:u0est}: By definition of $u_0$ and $v_0$,
\begin{equation}\label{eq:u0-k}
u_0 -\kerr  = \Pext(\udata-\kerrdata) + c_0
\end{equation}
We use the triangle inequality and estimate the two terms separately.
By \eqref{eq:cestimates} (with $\NN$ replaced by $\NN+2$),
\begin{align*}
\|c_0\|_{\sHb^{\NN+2}(\Dspop_{s})} 
&\le
\Capply{\ClargeGR}_0
\big(\tfrac{s}{s_*}\big)^{\frac{9}{2}+\epspower+\NN}\Vconst_{\NN+2}(v_0) \\
\intertext{
Using the first inequality of \eqref{eq:VN},
and the fact that $\Capply{\ClargeGR}_0$
depends only on $\NN,\epspower,\Cinmain$,}
\|c_0\|_{\sHb^{\NN+2}(\Dspop_{s})} 
&\lesssim_{\NN,\epspower,\Cinmain}
\big(\tfrac{s}{s_*}\big)^{\frac{9}{2}+\epspower+\NN} 
\|\udata-\kerrdata\|_{\HdataNEW^{\frac52+\epspower,\NN+3}(\Dspdata)}
\intertext{
Furthermore, by \eqref{eq:Pu-kbasic} with $k=\NN+2$,
}
\|\Pext(\udata-\kerrdata)\|_{\sHb^{\NN+2}(\Dspop_{s})}
&\lesssim_{\NN}
\big(\tfrac{s}{s_*}\big)^{\frac{9}{2}+\epspower+\NN}
\|\udata-\kerrdata\|_{\HdataNEW^{\frac{5}{2}+\epspower,\NN+3}(\Dspdata)}
\end{align*}
Thus \eqref{eq:u0est} holds. 

\proofheader{Construction of $u$ away from spacelike infinity.}
Fix functions:
\begin{itemize}
\item 
A cutoff $\phi_0\in C^\infty(\overline\diamond_+,[0,1])$
that has the following properties\footnote{%
The property \eqref{eq:cutoffrefl} is motivated by Appendix \ref{ap:ConstructionOnD},
where we construct solutions as in Theorem \ref{thm:main}
on $\diamond$ (not only $\diamond_+$).
}:
\begin{align}
&
\text{$\phi_0 = 1$ on $\Dspop_{\le\frac46\sfix}$
and $\phi_0 = 0$ on $\diamond_{5\sfix}$}
	\label{eq:supp_phi0}\\
&\text{$\phi_0 = \chi(|\vec{y}|)$ 
	on $\Dspop\cap\{\tfrac{y^0}{|\vec{y}|}\le \tfrac{1}{100}\}$
	for some $\chi\in C^\infty(\R,[0,1])$}
	\label{eq:cutoffrefl}
\end{align}
\item 
$\phi_1 = 1-\phi_0$. One has
\begin{align}
\text{$\phi_1 = 0$ on $\Dspop_{\le\frac46\sfix}$
and $\phi_1 = 1$ on $\diamond_{5\sfix}$}
\label{eq:supp_phi1}
\end{align}
\item 
A cutoff $\psi\in C^\infty(\overline{\diamond}_+,[0,1])$
that is equal to $1$ for $\tau\le \frac\pi4$ and equal to $0$ for $\tau\ge\frac\pi3$.
One has $\psi=1$ on $\Dspop_{\le1}$ by Remark \ref{rem:Dstau}, 
hence $\psi\phi_0=\phi_0$.
\end{itemize}

In the following we suppress the dependencies of constants
on the functions $\phi_0,\phi_1,\psi$, because they are fixed once and for all.

Define 
\begin{align}\label{eq:vmaindef}
v = \phi_0 u_0 + \psi \phi_1 \Pextbulk(\udata) \;\in\; \gx^1(\diamond_+)
\end{align}
using the extension operator in Definition \ref{def:bulkextension}.
(Note that $u_0$ is only defined on $\Dspop$.
The product $\phi_0u_0$ is understood to be zero on the complement
of $\Dspop$.)

We will correct $v$ to a solution of \eqref{eq:MCandudatamain}
using Proposition \ref{prop:MainCpt}.
For this it is convenient to first check some basic properties of $v$.

\claimheader{Claim:} For all $k\in\Z_{\ge0}$:
\begin{subequations}\label{eq:vprel}
\begin{align}
v|_{\tau=0} 
	&= \udata \label{eq:v0data}\\
\dg v + \tfrac12[v,v] 
	&=0 \;\;\text{on $\Dspop_{\le\frac\sfix2}$}\label{eq:v0MC}\\
v|_{\tau\ge\frac\pi2} 
	&=0 \label{eq:vzerotimel}\\
\tint_{0}^{\frac\pi2}\|v\|_{\sH^k(\diamond_{\tau,\frac\sfix2})}d\tau
	&\lesssim_{\new{k,\sfix}}
	\tint_{\frac{\sfix}{12}}^{\sfix} 
	\|u_0-\kerr\|_{\sHb^{k+1}(\Dspop_{s})} ds \nonumber \\
	&\qquad+\|\kerr\|_{\nosCb^k(\Dspop)}
	+\|\udata\|_{\H^k(\diamond_{0,\sfix})}
	\label{eq:vHkCpt}
\end{align}
\end{subequations}

Proof of \eqref{eq:v0data}:
By \eqref{eq:u0data}, $\Pextbulk(\udata)|_{\tau=0}=\udata$,
$\psi|_{\tau=0}=1$ and $\phi_0+\phi_1=1$.

Proof of \eqref{eq:v0MC}: 
We have 
$v = u_0$ on $\Dspop_{\le\frac\sfix2}$
by \eqref{eq:supp_phi0}, \eqref{eq:supp_phi1},
now use \eqref{eq:u0MC}.

Proof of \eqref{eq:vzerotimel}: 
Because $\phi_0=\phi_0\psi$ and because and $\psi$ vanishes for $\tau\ge\frac\pi2$.

Proof of \eqref{eq:vHkCpt}:
By \eqref{eq:vmaindef} and the triangle inequality,
and using \eqref{eq:supp_phi0}, \eqref{eq:supp_phi1},
\begin{align*}
\|v\|_{\sH^k(\diamond_{\tau,\frac\sfix2})}
	&\le
	\|\phi_0u_0\|_{\sH^k(\diamond_{\tau,\frac\sfix2})}
	+
	\| \psi \phi_1 \Pextbulk(\udata)\|_{\sH^k(\diamond_{\tau,\frac\sfix2})}
	\\
	&=
	\|\phi_0u_0\|_{\sH^k(\diamond_{\tau,\frac\sfix2}\cap\Dspop)}
	+
	\| \psi \phi_1 \Pextbulk(\udata)\|_{\sH^k(\diamond_{\tau,4\sfix})}
\intertext{
Using $\|\phi_0\|_{\sC^k(\diamond_{\tau,\frac\sfix2}\cap\Dspop)}\lesssim_{k,\sfix}1$
and $\|\psi \phi_1\|_{\sC^k(\diamond_{\tau,4\sfix})}\lesssim_{k,\sfix}1$,
we obtain}
\|v\|_{\sH^k(\diamond_{\tau,\frac\sfix2})}
	&\lesssim_{k,\sfix}
	\|u_0\|_{\sH^k(\diamond_{\tau,\frac\sfix2}\cap\Dspop)}
	+
	\|\Pextbulk(\udata)\|_{\sH^k(\diamond_{\tau,4\sfix})}
\intertext{
By Lemma \ref{lem:Pextbulkest}
we have 
$\|\Pextbulk(\udata)\|_{\sH^k(\diamond_{\tau,4\sfix})}
\le \|\udata\|_{\H^k(\diamond_{0,\sfix})}$.
Thus we obtain }
\tint_{0}^{\frac\pi2} \|v\|_{\sH^k(\diamond_{\tau,\frac\sfix2})}d\tau
	&\lesssim_{k,\sfix}
	\tint_{0}^{\frac\pi2}\|u_0\|_{\sH^k(\diamond_{\tau,\frac\sfix2}\cap\Dspop)}d\tau
	+
	\|\udata\|_{\H^k(\diamond_{0,\sfix})}
\end{align*}
We bound the first term on the right hand side. 
By the triangle inequality,
\begin{align*}
\tint_{0}^{\frac\pi2}
\|u_0\|_{\sH^k(\diamond_{\tau,\frac\sfix2}\cap\Dspop)} d\tau
	&\le
	\tint_{0}^{\frac\pi2}
	\|u_0-\kerr\|_{\sH^k(\diamond_{\tau,\frac\sfix2}\cap\Dspop)} d\tau
	+
	\tint_{0}^{\frac\pi2}
		\|\kerr\|_{\sH^k(\diamond_{\tau,\frac\sfix2}\cap\Dspop)} d\tau
\end{align*}
By Lemma \ref{lem:normsbulktosp} we have
\begin{align}
\tint_{0}^{\frac\pi2}
\|u_0-\kerr\|_{\sH^k(\diamond_{\tau,\frac\sfix2}\cap\Dspop)} d\tau
	&\lesssim_{k,\sfix}
	\tint_{\frac{\sfix}{12}}^{\sfix} \|u_0-\kerr\|_{\sHb^{k+1}(\Dspop_{s})} ds
	\nonumber
\intertext{Further we have}
\tint_{0}^{\frac\pi2}
		\|\kerr\|_{\sH^k(\diamond_{\tau,\frac\sfix2}\cap\Dspop)} d\tau
	&\lesssim_{k,\sfix}
	\|\kerr\|_{\nosC^k(\Dspop\setminus \Dspop_{<\frac{\sfix}{12}})}
	\label{eq:Cspec}\\
	&\lesssim_{k,\sfix}
	\|\kerr\|_{\nosCb^k(\Dspop\setminus \Dspop_{<\frac{\sfix}{12}})}
	\label{eq:Cbspec}\\
	&\le
	\|\kerr\|_{\nosCb^k(\Dspop)}
	\nonumber
\end{align}
(the norm in \eqref{eq:Cspec} is defined analogously
to $\nosC^k(\diamond_{s})$ in Definition \ref{def:bulknorms},
the norm in \eqref{eq:Cbspec} is defined analogously to 
$\nosCb^k(\Dspop_{\le s})$ in Definition \ref{def:norms_i0})
where for the first inequality we use compactness in $\tau$
and the fact that the volume of $\diamond_{\tau,\frac\sfix2}\cap\Dspop$
relative to $\mucylS$ is bounded independently of $\tau,\sfix$,
for the second inequality we use 
the fact that the norms \eqref{eq:Cspec} and
\eqref{eq:Cbspec} are comparable
with a comparability constant that depends only on $k,\sfix$.
This proves \eqref{eq:vHkCpt}.

We use Proposition \ref{prop:MainCpt} with the parameters 
in Table \ref{tab:MainCptApplication}. 
Let $\Capply{\ClargeCpt}_1$, $\epsapply{\CsmallCpt}_1$
be the constants produced by Proposition \ref{prop:MainCpt} 
(called $\ClargeCpt,\CsmallCpt$ there).
They depend only on $\NN$, $\sfix$,
in particular $\Clargemain$, $\Csmallmain$ are allowed to depend on 
$\Capply{\ClargeCpt}_1$, $\epsapply{\CsmallCpt}_1$.

\begin{table}
\centering
\begin{tabular}{cc|c}
	&
	\begin{tabular}{@{}c@{}}
	Parameters \\
	in Proposition \ref{prop:MainCpt}
	\end{tabular}
	& 
	\begin{tabular}{@{}c@{}}
	Parameters\\ 
	used to invoke Proposition \ref{prop:MainCpt}
	\end{tabular}
	 \\
\hline
	Input
	&$\NN$, $s_*$ 
	& $\NN$, $\frac12\sfix$\\
	&$v$ 
	&$v=\phi_0 u_0 + \psi \phi_1 \Pextbulk(\udata)$\\
%
	&$k$, $\CHigherCpt$ (Part 2 only)
	& $k$, $C_{5,k,\epspower,\sfix,\Cinmain,\CinmainHigher}$ \\
\hline
	Output
	&$\ClargeCpt$, $\CsmallCpt$
	&$\Capply{\ClargeCpt}_1$, $\epsapply{\CsmallCpt}_1$
\end{tabular}
\captionsetup{width=115mm}
\caption{The first column lists the input and output parameters of Proposition \ref{prop:MainCpt}. 
The second column specifies the choice of the input parameters used to invoke
Proposition \ref{prop:MainCpt},
in terms of the input parameters of Theorem \ref{thm:main}
and the parameters introduced in this proof.
The output parameters produced by this invocation of
Proposition \ref{prop:MainCpt}
are denoted $\Capply{\ClargeCpt}_1$, $\epsapply{\CsmallCpt}_1$	.
They depend only on the parameters in the first row.}
\label{tab:MainCptApplication}
\end{table}

We check the assumptions of Proposition \ref{prop:MainCpt}.
As required $\NN\ge7$.
\ref{item:vconstraintsD}: By \eqref{eq:v0data}, \ref{item:uconstrMain}.
\ref{item:MCvsp}: By \eqref{eq:v0MC}.
\ref{item:vsupp}: By \eqref{eq:vzerotimel}.
\ref{item:vNewBound}: 
By \eqref{eq:vHkCpt} with $k=\NN+1$, 
\begin{align*}
&\tint_{0}^{\frac\pi2}\|v\|_{\sH^{\NN+1}(\diamond_{\tau,\frac\sfix2})}d\tau\\
	&\qquad\lesssim_{\NN,\sfix}
	\tint_{\frac{\sfix}{12}}^{\sfix} 
	\|u_0-\kerr\|_{\sHb^{\NN+2}(\Dspop_{s})} ds	
	+\|\kerr\|_{\nosCb^{\NN+1}(\Dspop)}
	+\|\udata\|_{\H^{\NN+1}(\diamond_{0,\sfix})}
\end{align*}
We bound the first term on the right using \eqref{eq:u0est}
and $\frac{s}{\sfix}\le1$, which yields 
\begin{align}\label{eq:intvest}
\begin{aligned}
&\tint_{0}^{\frac\pi2}\|v\|_{\sH^{\NN+1}(\diamond_{\tau,\frac\sfix2})}d\tau\\
	&\quad\lesssim_{\NN,\epspower,\sfix,\Cinmain}
	\|\udata-\kerrdata\|_{\HdataNEW^{\frac52+\epspower,\NN+3}(\Dspdata)}
	+\|\kerr\|_{\nosCb^{\NN+1}(\Dspop)}
	+\|\udata\|_{\H^{\NN+1}(\diamond_{0,\sfix})}
	\le3\Csmallmain	
\end{aligned}
\end{align}
where the last step uses 
\ref{item:diffsmallmain},
\ref{item:KsmallMain},
\ref{item:usmallMain}.
This implies \ref{item:vNewBound} under an admissible
smallness assumption on $\Csmallmain$.

We have checked the assumptions 
\ref{item:vconstraintsD},
\ref{item:MCvsp},
\ref{item:vsupp},
\ref{item:vNewBound}
of Proposition \ref{prop:MainCpt}.
Thus there exists a unique $c\in\gxG^{1}(\diamond_+)$ 
(using the gauge in Definition \ref{def:gauge_bulk})
that satisfies \eqref{eq:Cptcprop} 
with $\sfix$ replaced by $\frac12\sfix$. 
Further $c$ satisfies \eqref{eq:cbulkP1}
with $\sfix$ replaced by $\frac12\sfix$.

Define
\begin{equation}\label{eq:defu}
u = v + c \;\in\; \gx^{1}(\diamond_+)
\end{equation}
This satisfies \eqref{eq:MCandudatamain} 
by \eqref{eq:ceqD} respectively \eqref{eq:cdataD} and \eqref{eq:v0data}.

\proofheader{Proof of Part 1.}
\eqref{eq:uspBoundMain}: 
By \eqref{eq:czeroD} we have $c=0$ on $\Dspop_{\le\frac\sfix2}$.
Hence $u=v$ on $\Dspop_{\le\frac\sfix2}$.
Together with \eqref{eq:vmaindef}, \eqref{eq:supp_phi0}, \eqref{eq:supp_phi1} we obtain
\begin{align}\label{eq:uu0}
u=u_0\;\; \text{on}\;\; \Dspop_{\le\frac\sfix2}
\end{align}
Thus \eqref{eq:u0est} yields that for all $s\in(0,\frac\sfix2]$:
\begin{align*}
\|u-\kerr\|_{\sHb^{\NN+2}(\Dspop_{s})}
	\lesssim_{\NN,\epspower,\Cinmain} \big(\tfrac{s}{\sfix}\big)^{\frac92+\epspower+\NN}
	\|\udata-\kerrdata\|_{\HdataNEW^{\frac52+\epspower,\NN+3}(\Dspdata)}
\end{align*}
This implies \eqref{eq:uspBoundMain}, 
by the first item in Lemma \ref{lem:SobolevDandW} with $k=\NN+2$
and $\sfix$ there given by $\sfix/2$ here
(the assumption \eqref{eq:sasspsoblem} holds by 
$s\le\sfix/2$ and by \ref{item:diffsmallmain}),
and by an admissible largeness assumption on $\Clargemain$.
Further Lemma \ref{lem:SobolevDandW} implies that
$u-\kerr$ extends in $C^{\NN-1}$ to \smash{$\Dspcl_{\le\frac\sfix2}$},
where we also use the fact that 
the basis elements \eqref{eq:gbasis_i0} are smooth on $\Dspcl$.
Thus $u$ extends in $C^{\NN-1}$ to \smash{$\Dspcl_{\le\frac\sfix2}$},
using \eqref{eq:KerrMain}.

\eqref{eq:uCptBoundMain}: 
For all $\tau\in [0,\pi)$:
\begin{align}
\|u\|_{\sH^{\NN}(\diamond_{\tau,\sfix})}
&\overset{(1)}{\le}
\|u\|_{\sH^{\NN}(\diamond_{\tau,\frac\sfix2})}
\overset{(2)}{\le}
\|v\|_{\sH^{\NN}(\diamond_{\tau,\frac\sfix2})}
+
\|c\|_{\sH^{\NN}(\diamond_{\tau,\frac\sfix2})}\label{eq:uvcintN}\\
&\overset{(3)}{\lesssim}_{\NN,\sfix}
	\tint_0^{\tau}\|v\|_{\sH^{\NN+1}(\diamond_{\tau',\frac\sfix2})}d\tau'
	+
	\|v\|_{\sH^{\NN}(\diamond_{\tau,\frac\sfix2})}\nonumber\\
&\overset{(4)}{\lesssim}_{\NN,\sfix}
\tint_0^{\frac\pi2}\|v\|_{\sH^{\NN+1}(\diamond_{\tau',\frac\sfix2})}d\tau'
	\nonumber
\end{align}
In (1) we use $\diamond_{\tau,\sfix}\subset \diamond_{\tau,\frac\sfix2}$;
in (2) we use \eqref{eq:defu};
in (3) we use \eqref{eq:SHcEEDNEW}
and the fact that $\Capply{\ClargeCpt}_1$ depends only on $\NN,\sfix$;
and in (4) we use the last statement in Part 1 of Proposition \ref{prop:MainCpt}
and \eqref{eq:vzerotimel}.
Together with \eqref{eq:intvest} we obtain
\begin{align*}
\sup_{\tau\in[0,\pi)}
\|u\|_{\sH^{\NN}(\diamond_{\tau,\sfix})}
&\lesssim_{\NN,\epspower,\sfix,\Cinmain}
	\|\udata-\kerrdata\|_{\HdataNEW^{\frac52+\epspower,\NN+3}(\Dspdata)}\\
	&\qquad\qquad+\|\kerr\|_{\nosCb^{\NN+1}(\Dspop)}
	+\|\udata\|_{\H^{\NN+1}(\diamond_{0,\sfix})}
\end{align*}
This implies \eqref{eq:uCptBoundMain}, 
by the second item in Lemma \ref{lem:SobolevDandW} with $k=\NN$
and $\sfix$ there given by $\sfix$ here
(the assumption \eqref{eq:tauasspsop} holds by \eqref{eq:intvest}),
by 
$
\|u\|_{\nosC^{\NN-3}(\diamond_{\sfix})}
=
\sup_{\tau\in[0,\pi)}\|u\|_{\sC^{\NN-3}(\diamond_{\tau,\sfix})}$, 
and under an admissible largeness assumption on $\Clargemain$.
Moreover Lemma \ref{lem:SobolevDandW} implies that
$u$ extends in \smash{$C^{\NN-3}$} to \smash{$\overline{\diamond}_{\sfix}$},
where we use the fact that the basis elements \eqref{eq:gbasis_bulk}
are smooth on $\cyl$.

We have shown that $u$ extends in $C^{\NN-3}$ to
$\Dspcl_{\le\frac\sfix2}$ and to $\overline{\diamond}_{\sfix}$, 
thus it extends in $C^{\NN-3}$ to 
$\overline{\diamond}_+\setminus\spaceinf
=
\Dspcl_{\le\frac\sfix2}\cup\overline{\diamond}_{\sfix}$.
This concludes the proof of Part 1.

\proofheader{Proof of Part 2.}
Let $k\ge\NN$, $\CinmainHigher>0$ so that 
\ref{item:KHigherBoundMain}, 
\ref{item:uCkHigherBoundMain},
\ref{item:uHkHigherBoundMain},
\ref{item:diffHighersmallmain} hold.

We prove \eqref{eq:uspBoundMain} 
with $\NN$, $\Clargemain$ replaced by $k$, $\ClargemainHigher$,
where we use Part 2 of Proposition \ref{prop:ApplySHS}
with the parameters in Table \ref{tab:ApplyApplySHS}.
We check that the assumptions hold.

\ref{item:GRHigherassp12}:
By \eqref{eq:v0Vconstk} with $k$ replaced by $k+2$,
by $\lfloor\frac{k+3}{2}\rfloor\le k+3$,
and by \eqref{eq:uktruk},
\begin{align}
\Vconst_{k+2}(v_0)
	&\lesssim_k
	\big(1
	+
	\|\kerr\|_{\nosCb^{k+3}(\Dspop)} 
	+
	\|\udata\|_{\Cb^{k+3}(\Dspdata)}
	\big)
	\|\udata-\kerrdata\|_{\HdataNEW^{\frac52+\epspower,k+3}(\Dspdata)}
	\nonumber\\
	&\lesssim
	(1+\CinmainHigher)\|\udata-\kerrdata\|_{\HdataNEW^{\frac52+\epspower,k+3}(\Dspdata)}
	\label{eq:Vk+2}
\end{align}
using \ref{item:KHigherBoundMain}, \ref{item:uCkHigherBoundMain}
for the second inequality.
Together with \ref{item:diffHighersmallmain}, 
we obtain that there exists a constant $\Cauxmain_{3,k,\CinmainHigher}>0$
that depends only on $k$, $\CinmainHigher$, such that 
\begin{align*}
\Vconst_{k+1}(v_0)
\le
\Vconst_{k+2}(v_0) 
\le 
\Cauxmain_{3,k,\CinmainHigher}
\end{align*}
where the first inequality is clear from \eqref{eq:defVkk}.
This proves \ref{item:GRHigherassp12}, using Table \ref{tab:ApplyApplySHS}.

\ref{item:GRHigherassp3}:
By \eqref{eq:v0Ck} and \eqref{eq:uktruk} with $k$ replaced by $k+3$,
\begin{align*}
\|v_0\|_{\nosCb^{k+3}(\Dspop)}
	&\lesssim_{k}
	\|\kerr\|_{\nosCb^{k+3}(\Dspop)}
	+
	\|\udata\|_{\Cb^{k+3}(\Dspdata)}
	\le
	2\CinmainHigher
\end{align*}
using \ref{item:KHigherBoundMain}, \ref{item:uCkHigherBoundMain}
for the second inequality.
Thus there exists a constant 
$\Cauxmain_{4,k,\CinmainHigher}>0$ that depends only on $k$, $\CinmainHigher$,
such that 
\begin{align*}
\|v_0\|_{\nosCb^{k+3}(\Dspop)}
	\le
	\Cauxmain_{4,k,\CinmainHigher}
\end{align*}
This proves \ref{item:GRHigherassp3}, using Table \ref{tab:ApplyApplySHS}.

We have checked the assumptions \ref{item:GRHigherassp12}, \ref{item:GRHigherassp3}
of Part 2 in Proposition \ref{prop:ApplySHS}.
Thus $c_0$ satisfies \eqref{eq:GRHigherconcl} with $k$ replaced by $k+2$,
that is, for all $s\in(0,\sfix]$:
\begin{align}\label{eq:GRHigherconclApp}
\|c_0\|_{\sHb^{\kk+2}(\Dspop_{s})} 
\lesssim_{\kk,\epspower,\Cinmain,\CinmainHigher}
\big(\tfrac{s}{s_*}\big)^{\frac{9}{2}+\epspower+\kk}
\Vconst_{\kk+2}(v_0)
\end{align}
Using \eqref{eq:u0-k}, for all $s\in(0,\sfix]$ we have
\begin{align}
\|u_0-\kerr\|_{\sHb^{k+2}(\Dspop_{s})}
	&\le
	\|\Pext(\udata-\kerrdata) \|_{\sHb^{k+2}(\Dspop_{s})}
	+
	\|c_0\|_{\sHb^{k+2}(\Dspop_{s})} \nonumber
\intertext{
We bound the first term using \eqref{eq:Pu-kbasic}
with $k$ replaced by $k+2$, and the second term 
using \eqref{eq:GRHigherconclApp} and \eqref{eq:Vk+2}.
This yields that for all $s\in(0,\sfix]$:}
\|u_0-\kerr\|_{\sHb^{k+2}(\Dspop_{s})} 
	&\lesssim_{\kk,\epspower,\Cinmain,\CinmainHigher}
	\big(\tfrac{s}{s_*}\big)^{\frac{9}{2}+\epspower+k}
	\|\udata-\kerrdata\|_{\HdataNEW^{\frac{5}{2}+\epspower,k+3}(\Dspdata)}
	\label{eq:u0-K,k+1der}
\end{align}
By \eqref{eq:uu0} this implies that for all $s\in(0,\frac\sfix2]$:
\begin{align*}
\|u-\kerr\|_{\sHb^{k+2}(\Dspop_{s})} 
	&\lesssim_{\kk,\epspower,\Cinmain,\CinmainHigher}
	\big(\tfrac{s}{s_*}\big)^{\frac{9}{2}+\epspower+k}
	\|\udata-\kerrdata\|_{\HdataNEW^{\frac{5}{2}+\epspower,k+3}(\Dspdata)}
\end{align*}
This implies \eqref{eq:uspBoundMain}
with $\NN$, $\Clargemain$ replaced by $k$, $\ClargemainHigher$, 
by the first item in Lemma \ref{lem:SobolevDandW} with $k$, $\sfix$
there given by $k+2$, $\sfix/2$ here 
(the assumption \eqref{eq:sasspsoblem} holds by 
$s\le\sfix/2$ and by \ref{item:diffHighersmallmain}),
and under a largeness assumption on $\ClargemainHigher$
that depends only on $\kk,\epspower,\Cinmain,\CinmainHigher$.
Moreover, Lemma \ref{lem:SobolevDandW} implies that
$u-\kerr$ extends in $C^{k-1}$ to $\Dspcl_{\le\frac\sfix2}$,
and thus $u$ extends in $C^{k-1}$ to $\Dspcl_{\le\frac\sfix2}$ by \eqref{eq:KerrMain}.

We prove \eqref{eq:uCptBoundMain} 
with $\NN$, $\Clargemain$ replaced by $k$, $\ClargemainHigher$,
where we use Part 2 of Proposition \ref{prop:MainCpt}
with the parameters in Table \ref{tab:MainCptApplication}.

We check \ref{item:vHigherNewBound}:
By \eqref{eq:vHkCpt} with $k$ replaced by $k+1$,
\begin{align}
&\tint_{0}^{\frac\pi2}\|v\|_{\sH^{k+1}(\diamond_{\tau,\frac\sfix2})}d\tau
	\nonumber\\
	&\quad
	\lesssim_{\new{k,\sfix}}
	\tint_{\frac{\sfix}{12}}^{\sfix} 
	\|u_0-\kerr\|_{\sHb^{k+2}(\Dspop_{s})} ds 
	+\|\kerr\|_{\nosCb^{k+1}(\Dspop)}
	+\|\udata\|_{\H^{k+1}(\diamond_{0,\sfix})}\nonumber\\
	&\quad
	\lesssim_{k,\epspower,\sfix,\Cinmain,\CinmainHigher}
	\|\udata-\kerrdata\|_{\HdataNEW^{\frac{5}{2}+\epspower,k+3}(\Dspdata)}
	+\|\kerr\|_{\nosCb^{k+1}(\Dspop)}
	+\|\udata\|_{\H^{k+1}(\diamond_{0,\sfix})}
	\label{eq:intvk+1}
\end{align}
where in the second step we use \eqref{eq:u0-K,k+1der}.
By \ref{item:KHigherBoundMain},
\ref{item:uHkHigherBoundMain}, 
\ref{item:diffHighersmallmain}, 
each of the three terms on the right hand side is bounded by $\CinmainHigher$.
Hence there exists a constant $C_{5,k,\epspower,\sfix,\Cinmain,\CinmainHigher}>0$
that depends only on $k,\epspower,\sfix,\Cinmain,\CinmainHigher$,
such that 
\begin{align*}
\tint_{0}^{\frac\pi2}\|v\|_{\sH^{k+1}(\diamond_{\tau,\frac\sfix2})}d\tau
	\le
	C_{5,k,\epspower,\sfix,\Cinmain,\CinmainHigher}
\end{align*}
This proves \ref{item:vHigherNewBound}, using Table \ref{tab:MainCptApplication}.
Hence $c$ satisfies \eqref{eq:bulkh12} with $\sfix$ replaced by $\frac\sfix2$.

Analogously to \eqref{eq:uvcintN} we have, for all $\tau\in[0,\pi)$:
\begin{align*}
\|u\|_{\sH^{k}(\diamond_{\tau,\sfix})}
&\le
\|u\|_{\sH^{k}(\diamond_{\tau,\frac\sfix2})}
\le
\|v\|_{\sH^{k}(\diamond_{\tau,\frac\sfix2})}
+
\|c\|_{\sH^{k}(\diamond_{\tau,\frac\sfix2})}
\intertext{
With \eqref{eq:chigherEECptNEW1},
the last statement in Part 2 of 
Proposition \ref{prop:MainCpt}, and \eqref{eq:vzerotimel}, we obtain 
}
\|u\|_{\sH^{k}(\diamond_{\tau,\sfix})}
&\lesssim_{k,\epspower,\sfix,\Cinmain,\CinmainHigher}
	\tint_0^{\tau}\|v\|_{\sH^{k+1}(\diamond_{\tau',\frac\sfix2})}d\tau'
	+
	\|v\|_{\sH^{k}(\diamond_{\tau,\frac\sfix2})}\\
&\lesssim_{k,\sfix}
\tint_0^{\frac\pi2}\|v\|_{\sH^{k+1}(\diamond_{\tau',\frac\sfix2})}d\tau'
\intertext{With \eqref{eq:intvk+1} this implies}
\sup_{\tau\in[0,\pi)}\|u\|_{\sH^{k}(\diamond_{\tau,\sfix})}
	&\lesssim_{k,\epspower,\sfix,\Cinmain,\CinmainHigher}
	\|\udata-\kerrdata\|_{\HdataNEW^{\frac{5}{2}+\epspower,k+3}(\Dspdata)}\\
	&\qquad\qquad\quad
		+\|\kerr\|_{\nosCb^{k+1}(\Dspop)}
		+\|\udata\|_{\H^{k+1}(\diamond_{0,\sfix})}
\end{align*}
This implies \eqref{eq:uCptBoundMain} with
$\NN$, $\Clargemain$ replaced by $k$, $\ClargemainHigher$, 
by the second item in Lemma \ref{lem:SobolevDandW}
with $k,\sfix$ there given by $k,\sfix$ here
(the assumption \eqref{eq:tauasspsop} holds by 
\ref{item:KHigherBoundMain}, \ref{item:uHkHigherBoundMain}, \ref{item:diffHighersmallmain}),
and under a largeness assumption on $\ClargemainHigher$
that depends only on $k,\epspower,\sfix,\Cinmain,\CinmainHigher$.
Moreover, Lemma \ref{lem:SobolevDandW} implies that
$u$ extends in $C^{k-3}$ to $\overline{\diamond}_{\sfix}$.

We have shown that $u$ extends in $C^{k-3}$
to $\Dspcl_{\le\frac\sfix2}$ and to 
$\overline{\diamond}_{\sfix}$, hence it extends in 
$C^{k-3}$ to 
$\overline{\diamond}_+\setminus\spaceinf
=
\Dspcl_{\le\frac\sfix2}\cup\overline{\diamond}_{\sfix}$.
This concludes the proof of Part 2.

\proofheader{Proof of Part 3.}
We check that the assumptions of Proposition \ref{prop:metricregularity} hold with parameters \eqref{eq:ksu0} given by $\NN-3$, $\sfix$, $\uo$,
respectively by $k-3$, $\sfix$, $\uo$ under the assumptions of Part 2.
The element $u$, and hence $\uo$, extend in $C^{\NN-3}$
to $\overline\diamond_+\setminus\spaceinf$ by Part 1,
respectively in $C^{k-3}$ under the assumptions of Part 2 by Part 2.
We check \ref{item:F1/16assp} and \ref{item:Fdh1/16assp},
note that these assumptions are independent of $k$ in \eqref{eq:ksu0}.
By \eqref{eq:uspBoundMain}, \ref{item:KsmallMain}, \ref{item:diffsmallmain},
respectively by \eqref{eq:uCptBoundMain}, \ref{item:KsmallMain}, 
\ref{item:usmallMain}, \ref{item:diffsmallmain},
\begin{align*}
\|u\|_{\nosCb^0(\Dspop_{\le\frac\sfix2})} 
&\le 
\|u-\kerr\|_{\nosCb^0(\Dspop_{\le\frac\sfix2})} 
+
\|\kerr\|_{\nosCb^0(\Dspop_{\le\frac\sfix2})} 
\le (1+\Clargemain)\Csmallmain\\
\|u\|_{\nosC^0(\diamond_{\sfix})} 
&\le
3\Clargemain\Csmallmain
\end{align*}
Thus \ref{item:F1/16assp} follows from Lemma \ref{lem:FrameInvertibility},
the fact that the change of bases
between $\p_{y^\mu}$ and $\V_{\mu}$ is smooth 
on $\Dspcl_{\le1}\cup\spaceinf$,
and an admissible smallness assumption on $\Csmallmain$.
Further \ref{item:Fdh1/16assp} follows from 
Lemma \ref{lem:Fdhlemma}, the fact that 
the change of bases between $dy^\mu/\s$
and $\Vd^\mu$ is smooth on $\Dspcl_{\le \sfix}\setminus\Dspcl_{<\frac\sfix2}$, and an admissible smallness assumption on $\Csmallmain$.

Now Part 3, except for the statement that $g$ is Ricci-flat, 
follows from Proposition \ref{prop:metricregularity}.
Ricci-flatness follows from Proposition \ref{prop:metricdef}
(with $\diamond$ replaced by $\diamond_+$).
\qed
\end{proof}
\begin{proof}[of Theorem \ref{thm:mainpointwise}]
Instead of specifying $C$ and $\epsO$ up front, 
we will make finitely many admissible largeness assumptions on $C$,
respectively smallness assumptions on $\epsO$, 
where admissible means that they depend only on $\NN,\sfix$.

We use Theorem \ref{thm:main} with the parameters in Table \ref{tab:MainCorollary}.
Let $\Capply{\Clargemain}$, $\epsapply{\Csmallmain}$ be the constants
produced by Theorem \ref{thm:main}
(called $\Clargemain$, $\Csmallmain$ there).
They depend only on $\NN$, $\sfix$,
in particular $C$, $\epsO$ are allowed to depend on 
$\Capply{\Clargemain}$, $\epsapply{\Csmallmain}$.
\begin{table}
\centering
\begin{tabular}{cc|c}
	&
	\begin{tabular}{@{}c@{}}
	Parameters \\
	in Theorem \ref{thm:main}
	\end{tabular}
	& 
	\begin{tabular}{@{}c@{}}
	Parameters\\ 
	used to invoke Theorem \ref{thm:main}
	\end{tabular}
	 \\
\hline
	Input
	&$\NN$, $\epspower$, $s_*$, $\Cinmain$
	& $\NN$, $1/4$, $\sfix$, $1$\\
	&$\kerr$, $\udata$ 
	&$\kerr$, $\udata$\\
	&$k$, $\CinmainHigher$ (Part 2 only)
	&$k$, $\slCinmainHigher$ \\
\hline
	Output
	&$\Clargemain$, $\Csmallmain$
	&$\Capply{\Clargemain}$, $\epsapply{\Csmallmain}$
\end{tabular}
\captionsetup{width=115mm}
\caption{%
The first column lists the input and output parameters of Theorem \ref{thm:main}. 
The second column specifies the choice of the input parameters used to invoke
Theorem \ref{thm:main},
in terms of the input parameters of Theorem \ref{thm:mainpointwise}
and the parameters introduced in this proof.
The output parameters produced by this invocation of
Theorem \ref{thm:main}
are denoted $\Capply{\Clargemain}$, $\epsapply{\Csmallmain}$	.
They depend only on the parameters in the first row.}
\label{tab:MainCorollary}
\end{table}
%
We check that the assumptions of Theorem \ref{thm:main} hold.
\ref{item:KMCmain}: By \ref{item:kerrmc}.
\ref{item:uconstrMain}: By \ref{item:ucons}.
%
%
\ref{item:KBoundMain},
\ref{item:KL1Main},
\ref{item:KsmallMain}:
Note that 
\begin{equation}\label{eq:ys}
\text{$\tfrac{|y|}{\s}$ is homogeneous of degree zero and $|y| \le \s \le \sqrt{6} |y|$}
\end{equation}
The assumption \ref{item:KPts} implies, 
together with the Leibniz rule and \eqref{eq:ys},
$$
|(\s\p_{y})^{\le\NN+3}\kerr|\lesssim_{\NN} \eps \s(1+|\log\s|)
\qquad
\text{on $\Dspop_{\le\sfix}$}
$$
Thus for all $s\in(0,\sfix]$:
\begin{equation*}
\|\kerr\|_{\sCb^{\NN+3}(\Dspop_{s})} \lesssim_{\NN} \eps s (1+|\log s|)
\end{equation*}
Since $\sup_{s\in(0,1]} s(1+|\log s|) = 1$
and $\int_0^{1} s(1+|\log s|) \frac{ds}{s}=2$, this implies 
\begin{subequations}
\begin{align}
\|\kerr\|_{\nosCb^{\NN+3}(\Dspop_{\le\sfix})} &\lesssim_{\NN} \eps 
	\label{eq:kN+3T1}
\\
\tint_{0}^{\sfix} \|\kerr\|_{\sCb^{1}(\Dspop_{s})} \frac{ds}{s} 
	&\lesssim \eps 
	\label{eq:uintT1}
\end{align}
\end{subequations}
Thus \ref{item:KBoundMain},
\ref{item:KL1Main},
\ref{item:KsmallMain} follow under an admissible smallness assumption on $\epsO$,
using $\eps\le\epsO$ 
and Table \ref{tab:MainCorollary}.
\ref{item:uBoundMain}:
By the triangle inequality,
\begin{align}\label{eq:udataCbN}
\begin{aligned}
\|\udata\|_{\Cb^{\NN+3}(\Dspdata_{\le\sfix})}
&\le
\|\udata-\kerrdata\|_{\Cb^{\NN+3}(\Dspdata_{\le\sfix})}
+
\|\kerrdata\|_{\Cb^{\NN+3}(\Dspdata_{\le\sfix})}\\
&\lesssim_{\NN}
\|\udata-\kerrdata\|_{\Cb^{\NN+3}(\Dspdata_{\le\sfix})}
+
\|\kerr\|_{\Cb^{\NN+3}(\Dspop_{\le\sfix})}
\lesssim_{\NN}
\eps
\end{aligned}
\end{align}
where the last inequality holds by \ref{item:K-uPts} and \eqref{eq:kN+3T1}.
This implies \ref{item:uBoundMain} under an admissible smallness assumption 
on $\epsO$, 
using $\eps\le\epsO$
and Table \ref{tab:MainCorollary}.
\ref{item:usmallMain}:
We have
\[ 
\|\udata\|_{\H^{\NN+1}(\diamond_{0,\sfix})} 
\lesssim_{\NN}
\|\udata\|_{C^{\NN+1}(\diamond_{0,\sfix})} 
\]
using the fact that the volume of $\diamond_{0,\sfix}$
with respect to $\mucylS$ is bounded, independently of $\sfix$.
Recall that the $C^{\NN+1}(\diamond_{0,\sfix})$ norm is defined using the
vector fields $\V_1,\V_2,\V_3$.
The change of bases between $\V_1,\V_2,\V_3$ and $\p_{x^1},\p_{x^2},\p_{x^3}$
is smooth on the compact set $\diamond_{0,\sfix}$. 
Hence \ref{item:uxptw} yields 
\begin{equation}\label{eq:uepsT1}
\|\udata\|_{\H^{\NN+1}(\diamond_{0,\sfix})} \lesssim_{\NN,\sfix} \eps
\end{equation}
which implies \ref{item:usmallMain} under an admissible smallness assumption 
on $\epsO$, using $\eps\le\epsO$.
\ref{item:diffsmallmain}: 
For every $s\in(0,\sfix]$ one has 
\begin{align}\label{eq:prelHdataN}
\|\udata-\kerrdata\|_{\Hb^{\NN+3}(\Dspdata_{\frac{s}{\new{3}},s})}
&\lesssim_{\NN}
\|\udata-\kerrdata\|_{\Cb^{\NN+3}(\Dspdata_{\frac{s}{\new{3}},s})}
\lesssim_{\NN}
\eps s^{\NN+5}
\end{align}
where in the first step we use
the fact that the volume of $\Dspdata_{\frac{s}{\new{3}},s}$
with respect to the homogeneous measure $\muDu$ is bounded
independently of $s$, and in the second step we use \ref{item:K-uPts}.
Therefore, with $\epspower=\frac14$ (see Table \ref{tab:MainCorollary}),
\begin{align}
\|\udata-\kerrdata\|_{\HdataNEW^{\frac52+\epspower,\NN+3}(\Dspdata)}
	&\lesssim_{\NN}
	\eps 
	\tint_{0}^{\sfix} \big(\frac{\sfix}{s}\big)^{\frac52+\epspower+(\NN+2)} 
	\left(1+|\log(\tfrac{\sfix}{s})|\right)^{\NN+2}
	s^{\NN+5} \frac{ds}{s} \nonumber\\
	&\lesssim_{\NN,\sfix}
	\eps 
	\tint_{0}^{1} s^{\frac12-\epspower} 
	\left(1+|\log s|\right)^{\NN+2} \frac{ds}{s}
	\lesssim_{\NN}
	\eps\label{eq:u-keps}
\end{align}
Thus \ref{item:diffsmallmain} follows under an admissible smallness assumption on $\eps$, using $\eps\le\epsO$.

We have checked that the assumptions of Theorem \ref{thm:main} hold.
Let $u\in\gx^1(\diamond_+)$ be a solution as in Theorem \ref{thm:main}.
This satisfies \eqref{eq:MCT1} by \eqref{eq:MCandudatamain}

\proofheader{Proof of Part 1.}
By Part 1 of Theorem \ref{thm:main}, $u$ extends in $C^{\NN-3}$ to $\overline\diamond_+\setminus\spaceinf$.
We check \eqref{eq:uconclNptwsp}:
By \eqref{eq:ys}, for all $s\in(0,\sfix/2]$ 
and at every point on $\Dspop_{s}$:
\begin{align}
|y|^{-(\NN+4)} |(|y|\p_y)^{\le\NN}(u-\kerr)| 
	&\lesssim_{\NN}
	s^{-(\NN+4)}\|u-\kerr\|_{\sCb^{\NN}(\Dspop_{s})}
	\nonumber
\intertext{
By \eqref{eq:uspBoundMain},
the fact that $\Capply{\Clargemain}$ depends only on $\NN$, $\sfix$,
and by \eqref{eq:u-keps}, we obtain}
|y|^{-(\NN+4)} |(|y|\p_y)^{\le\NN}(u-\kerr)| 
	&\lesssim_{\NN,\sfix}
	s^{\frac12+\gamma}
	\|\udata-\kerrdata\|_{\HdataNEW^{\frac52+\epspower,\NN+3}(\Dspdata)}
	\lesssim_{\NN,\sfix}
	\eps
	\label{eq:proofp1nT1}
\end{align}
This implies \eqref{eq:uconclNptwsp}
using $\cup_{s\in(0,\sfix/2]}\Dspop_{s}=\Dspop_{\le\sfix/2}$,
and under an admissible largeness assumption on $C$.
We check \eqref{eq:uconclNptwCN}:
By \eqref{eq:uCptBoundMain}
and \eqref{eq:u-keps}, \eqref{eq:kN+3T1}, \eqref{eq:uepsT1},
\begin{align*}
\|u\|_{C^{\NN-3}(\diamond_{\sfix})} 
\lesssim_{\NN,\sfix} \eps
\end{align*}
where we also use the fact that $\Capply{\Clargemain}$
depends only on $\NN$, $\sfix$.
Thus \eqref{eq:uconclNptwCN} holds under an admissible 
largeness assumption on $C$
(recall that $\diamond_{\sfix}=\diamond_+\setminus\Dspop_{<\frac\sfix6}$).

\proofheader{Proof of Part 2.}
Let $k\ge\NN$ such that \eqref{eq:SimpHigher} holds.
As a preliminary, we check that there exists $\slCinmainHigher>0$ such that
\begin{subequations}\label{eq:T1P2basic}
\begin{align}
\|\kerr\|_{\nosCb^{k+3}(\Dspop_{\le\sfix})} 
	&\le \slCinmainHigher
	\label{eq:kerrk+3}\\
\|\udata\|_{\Cb^{k+3}(\Dspdata_{\le\sfix})} 
	&\le \slCinmainHigher
	\label{eq:datak+3}\\
\|\udata\|_{\Hb^{k+1}(\diamond_{0,\sfix})} 
	&\le \slCinmainHigher
	\label{eq:datak+1,H} \\
\|\udata-\kerrdata\|_{\HdataNEW^{\frac52+\epspower,k+3}(\Dspdata_{\le\sfix})} 
	&\le \slCinmainHigher
	\label{eq:diff}
\end{align}
\end{subequations}
Proof of \eqref{eq:T1P2basic}:
\eqref{eq:kerrk+3}:
This follows from \eqref{eq:KPtsHigher}, \eqref{eq:ys},
and $|y|(1+|\log|y||)\lesssim1$ on $\Dspcl_{\le\sfix}$.
\eqref{eq:datak+3}:
Analogously to \eqref{eq:udataCbN}, this follows from the triangle inequality, 
\eqref{eq:kerrk+3} and \eqref{eq:K-uPtsHigher}.
\eqref{eq:datak+1,H}:
This holds because $\udata$ is smooth on $\diamonddata$
and $\diamond_{0,\sfix}\subset\diamonddata$ is compact.
\eqref{eq:diff}:
This is checked analogously to \eqref{eq:prelHdataN} and \eqref{eq:u-keps},
using \eqref{eq:K-uPtsHigher}.

By \eqref{eq:T1P2basic} and Table \ref{tab:MainCorollary}, the assumptions
\ref{item:KHigherBoundMain},
\ref{item:uCkHigherBoundMain},
\ref{item:uHkHigherBoundMain},
\ref{item:diffHighersmallmain}
of Part 2 in Theorem \ref{thm:main} hold.
Thus $u$ extends in $C^{k-3}$ to $\overline\diamond_+\setminus\spaceinf$.
Further by \eqref{eq:uspBoundMain} of Part 2, 
and a calculation analogous to \eqref{eq:proofp1nT1},
one has $\|\tfrac{(|y|\p_y)^{\le k}(u-\kerr)}{|y|^{k+4}}\|_{L^\infty(\Dspop_{\le\frac\sfix2})}<\infty$.

\proofheader{Proof of Part 3.}
This follows from Part 3 of Theorem \ref{thm:main}.\qed
\end{proof}

\appendix

\section{Construction on $\diamond$}
\label{ap:ConstructionOnD}

In Theorem \ref{thm:main} we construct smooth solutions on $\diamond_+$.
Here we explain how the construction can be used to obtain 
smooth solutions on $\diamond$,
including estimates analogous to Theorem \ref{thm:main},
which in particular control the regularity along past null and timelike infinity.

Let $\refl:\cyl\to\cyl$ be the reflection $(\tau,\xi)\mapsto (-\tau,\xi)$.
We denote the restriction of $\refl$ to subsets of $\cyl$ also by $\refl$.
The map $\refl$ naturally induces a map 
$$\refl^{\gx}:\gx(\cyl)\to\gx(\cyl)$$
that commutes with the operations \eqref{eq:gop}, in particular with 
$\dg$ and with $[\cdot,\cdot]$.
Analogously one obtains maps $\gx(\diamond_+)\to\gx(\refl(\diamond_+))$
and $\gx(\Dspcl_{\le\sfix})\to\gx(\refl(\Dspcl_{\le\sfix}))$
for every $\sfix>0$, that we also denote by $\reflg$.
The map $\refl$ acts as the identity on $\diamonddata$, 
and it induces a map $\refldatag:\gxdata(\diamonddata)\to\gxdata(\diamonddata)$,
which in particular maps solutions of the constraints \eqref{eq:1constraints}
to solutions of the constraints.

In Theorem \ref{thm:main}, suppose that one is given an element
\begin{equation*}
\kerr\in\gx^1(\Dspcl_{\le\sfix}\cup\refl(\Dspcl_{\le\sfix}))
\end{equation*}
We claim that 
if \ref{item:KMCmain} holds for this element $\kerr$,
and if \ref{item:KBoundMain}, \ref{item:KL1Main},
\ref{item:KsmallMain} also hold with $\kerr$ replaced by $\reflg(\kerr)$
\new{(the idea is that one may take $\kerr$ to be equal to a Kerr element $\kerr(m,\vec{a})$, as in the case of Theorem \ref{thm:main})},
then there exists
\begin{equation}\label{eq:uonD}
u\in\gx^1(\diamond)
\end{equation} 
for which \eqref{eq:MCandudatamain} holds on $\diamond$,
for which Part 1 and 2 
hold on $\diamond_+$ also
with $u$, $\kerr$ replaced by $\reflg(u)$, $\reflg(\kerr)$ respectively,
and for which Part 3
holds on $\diamond_+$ also with $u$ replaced by $\reflg(u)$.
Hence \eqref{eq:gthm} defines a metric 
$g$ everywhere on $\diamond$ (on $\diamond_+$, the metric $g$
is associated to $u$, and $\refl^*g$ to $\reflg(u)$).
The metric $g$ on $\diamond$ is
null geodesically complete, and the locus of future and past 
null infinity is $\fnullinf$ respectively $\pnullinf$
(this holds because both $g$ and $\refl^*g$ 
satisfy \ref{item:metric_nullcompleteness} on $\diamond_+$).

The element \eqref{eq:uonD} may be constructed as follows.
Apply the construction in the proof of Theorem \ref{thm:main}
once to $\kerr$ and $\udata$, which yields an element $u_+\in\gx^1(\diamond_+)$,
and once to $\reflg(\kerr)$ and $\refldatag(\udata)$, 
which yields an element $u_-\in\gx^1(\diamond_+)$.
Define
\begin{align*}
u = 
\begin{cases}
u_+ & \text{on $\diamond_+$}\\
\reflg(u_-) & \text{on $\refl(\diamond_+)$}
\end{cases}
\end{align*}
using $\diamond_+\cup\refl(\diamond_+)=\diamond$.
It only remains 
to show that $u$ is smooth also along $\diamonddata$:
\begin{itemize}
\item 
Smoothness along $\Dspdata_{\le\frac\sfix2}$:
On $\Dspop_{\le\frac\sfix2}$ one has
$u_\pm = v_{0\pm} + c_{0\pm}$,
where $v_{0\pm}$ and $c_{0\pm}$ are the elements in 
\eqref{eq:v0def} respectively \eqref{eq:c0}.
Define $v_0$ and $c_0$ by 
\begin{align*} 
v_0 = 
\begin{cases}
v_{0+} & \text{on $\Dspop_{\le\frac\sfix2}$}\\
\reflg(v_{0-}) & \text{on $\refl(\Dspop_{\le\frac\sfix2})$}
\end{cases}
&&
c_0 = 
\begin{cases}
c_{0+} & \text{on $\Dspop_{\le\frac\sfix2}$}\\
\reflg(c_{0-}) & \text{on $\refl(\Dspop_{\le\frac\sfix2})$}
\end{cases}
\end{align*}
Then $v_0$ is smooth along $\Dspdata_{\le\frac\sfix2}$ because 
$\kerr$ is smooth, and because,
denoting $\underline{w}=\udata-\kerrdata$, the element defined by 
\begin{align*}
w
=
\begin{cases}
\Pext(\underline{w}) 
	& \text{on $\Dspop_{\le\frac\sfix2}$}\\
(\reflg\circ\Pext\circ\refldatag)(\underline{w})
	& \text{on $\refl(\Dspop_{\le\frac\sfix2})$}
\end{cases}
\end{align*}
is smooth along $\Dspdata_{\le\frac\sfix2}$
(this follows from Definition \ref{def:extensionspinf},
see also example \eqref{eq:ExPext}).
We show that $c_0$ is smooth along $\Dspdata_{\le\frac\sfix2}$.
By construction,
$c_0$ is continuous along $\Dspdata_{\le\frac\sfix2}$
and smooth separately on $\Dspop_{\le\frac\sfix2}$ and on
$\refl(\Dspop_{\le\frac\sfix2})$.
Moreover, separately on $\Dspop_{\le\frac\sfix2}$ and on $\refl(\Dspop_{\le\frac\sfix2})$,
\begin{equation}\label{eq:v0c0D}
\dg (v_0+c_0) + \tfrac12[v_0+c_0,v_0+c_0] = 0
\end{equation}
using the fact that $\reflg$ commutes with $\dg$ and $[\cdot,\cdot]$.
Then smoothness of $c_0$ follows from the fact that 
\eqref{eq:v0c0D}, viewed as an equation for $c_0$
an open neighborhood of $\Dspdata_{\le\frac\sfix2}$ in $\diamond$, 
contains a symmetric hyperbolic subsystem with smooth coefficients. 
To see this, note that the gauge in Definition \ref{def:gauges_i0}
is actually defined on the subset of 
$\Dspcl\cup\refl(\Dspcl)$ given by $\s>0$ (not only on $\Dspcl$),
which contains an open neighborhood of $\Dspdata_{\le\frac\sfix2}$
in $\diamond$, and use the fact that
$c_0$ is contained in the gauge subspace \eqref{eq:gaugespaces_i0},
since the gauge subspace is, by construction,
invariant under $\reflg$ near $\Dspdata_{\le\frac\sfix2}$.
\item 
Smoothness along $\diamond_{0,\sfix}$:
This is shown very similarly to the first item.
Further recall the cutoff functions used in \eqref{eq:vmaindef}.
One must use the fact that the function
defined by $\phi_0$ on $\diamond_+$, and by $\refl^*\phi_0$
on $\refl(\diamond_+)$, is smooth,
which holds by \eqref{eq:cutoffrefl};
analogously for $\phi_1$, $\psi$.
\end{itemize}



\footnotesize
\step

  \noindent\textsc{Department of Mathematics, Stanford University, Stanford CA, USA}\par\nopagebreak
  \noindent\textsc{Present affiliation: 
  Department of Mathematics, EPFL, Lausanne, Switzerland}\par\nopagebreak
 \noindent\textit{Email address:} \texttt{anuetzi@stanford.edu}

\end{document}